\documentclass[showpacs,twocolumn,superscriptaddress,notitlepage]{revtex4-2}
\pdfoutput=1
\usepackage{graphicx}

\usepackage{dcolumn}
\usepackage{bm}
\usepackage{pifont}

\usepackage{suffix}

\usepackage{amsmath}

\usepackage{amssymb}
\usepackage{braket}
\usepackage{physics}
\usepackage{amsthm}
\usepackage{natbib}
\usepackage{mathtools}
\usepackage{diagbox}
\usepackage[utf8]{inputenc}
\usepackage[english]{babel}
\usepackage{scalerel}[2014/03/10]
\usepackage{stackengine}
\usepackage{algorithm}
\usepackage[noend]{algpseudocode}
\usepackage{tikz}
\usetikzlibrary{quantikz}
\usetikzlibrary{shapes}
\usepackage{subfigure}
\newenvironment{tikzequation}
  {\begin{equation}
    \begin{tikzpicture}[line width=0.3mm,baseline={(current bounding box.center)}]
     \begin{scope}
  }
  { \end{scope}
    \end{tikzpicture}
   \end{equation}
  }

\newcommand{\algmargin}{\the\ALG@thistlm}
\makeatother
\newlength{\whilewidth}
\algdef{SE}[parWHILE]{parWhile}{EndparWhile}[1]
  {\parbox[t]{\dimexpr\linewidth-\algmargin}{%
     \hangindent\whilewidth\strut\algorithmicwhile\ #1\ \algorithmicdo\strut}}{\algorithmicend\ \algorithmicwhile}%
\algnewcommand{\parState}[1]{\State%
  \parbox[t]{\dimexpr\linewidth-\algmargin}{\strut #1\strut}}
\tikzstyle WL=[line width=2pt,opacity=1.0]
\tikzstyle WLL=[line width=5pt,opacity=1.0]
\newtheorem{theorem}{Theorem}
\newtheorem{rmk}{Remark}
\newtheorem{coro}{Corollary}
\newtheorem{lemma}{Lemma}
\newtheorem{prop}{Proposition}

\newtheorem{definition}{Definition}

\theoremstyle{definition}

\makeatletter
\@addtoreset{proofpart}{theorem}
\@addtoreset{proofcase}{theorem}
\makeatother
\renewcommand{\exp}[1]{\text{exp}\left( #1 \right)}

\usepackage[bottom]{footmisc}
\usepackage{graphicx}
\usepackage{import}
\usepackage{epstopdf}

\usetikzlibrary{calc}
\usepackage{hyperref}

\usepackage{soul}
\usepackage{cancel}

\begin{document}
\preprint{APS/123-QED}

\title{
Quon Classical Simulation: Unifying Cliffords, Matchgates and Entanglement}

\vspace{0.5cm}

\author{Zixuan Feng}
\thanks{All authors contributed equally to this work. \\Contact Author: liuzhengwei@mail.tsinghua.edu.cn}
\affiliation{Department of Mathematics, Tsinghua University,  Beijing 100084, China}

\author{Zhengwei Liu}
\thanks{All authors contributed equally to this work. \\Contact Author: liuzhengwei@mail.tsinghua.edu.cn}
\affiliation{Department of Mathematics, Tsinghua University,  Beijing 100084, China}
\affiliation{Yau Mathematical Sciences Center, Tsinghua University,  Beijing 100084, China}
\affiliation{Yanqi Lake Beijing Institute of Mathematical Sciences and Applications, Beijing 100407, China}

\author{Fan Lu}
\thanks{All authors contributed equally to this work. \\Contact Author: liuzhengwei@mail.tsinghua.edu.cn}
\affiliation{Department of Mathematics, Tsinghua University,  Beijing 100084, China}

\author{Ningfeng Wang}

\thanks{All authors contributed equally to this work. \\Contact Author: liuzhengwei@mail.tsinghua.edu.cn}
\affiliation{Department of Mathematics, Tsinghua University,  Beijing 100084, China}


\begin{abstract}
We propose a new framework of topological complexity to study the computational complexity of quantum circuits and tensor networks. Within this framework, we establish the Quon Classical Simulation (QCS) for hybrid Clifford-Matchgate circuits, which is efficient for both Clifford circuits and Matchgate circuits, therefore answering a long standing open question on unifying efficient classical simulations. This framework is built upon the Quon language, a 2+1D topological quantum field theory with space-time boundary and defects. Its exponential computation complexity is captured by Magic holes, a topological feature capturing the global long-range entanglement. Both Clifford circuits and Matchgate circuits are free of Magic holes. Efficient classical simulations of Cliffords and Matchgates are implemented by two parallel operations, generalized surgery theory of 3-manifolds and Yang-Baxter relations on the 2D boundary respectively, with additional binary arithmetic properties.

\end{abstract}
\maketitle

\onecolumngrid


\newcommand{\zerozero}[3]{
\pgfmathsetmacro{\s}{#1};
\pgfmathsetmacro{\x}{#2};
\pgfmathsetmacro{\y}{#3};
\begin{scope}[shift={(\x,\y)},scale=\s]
\draw(0,0)arc[start angle=180,end angle=0,radius=.25];
\draw(1,0)arc[start angle=180,end angle=0,radius=.25];
\end{scope}
}
\newcommand{\pnbraiding}[3]{
\pgfmathsetmacro{\s}{#1};
\pgfmathsetmacro{\x}{#2};
\pgfmathsetmacro{\y}{#3};
\begin{scope}[shift={(\x,\y)},scale=\s]
\draw(0,0)--(0,1);
\draw(.5,0)--(1,1);
\draw(1,0)--(.5,1);
\draw(1.5,0)--(1.5,1);
\end{scope}

}

\newcommand{\idid}[3]{
\pgfmathsetmacro{\s}{#1};
\pgfmathsetmacro{\x}{#2};
\pgfmathsetmacro{\y}{#3};
\begin{scope}[shift={(\x,\y)},scale=\s]
\draw(0,0)--(0,1);
\draw(.5,0)--(1,0);
\draw(1,1)--(.5,1);
\draw(1.5,0)--(1.5,1);
\end{scope}

}
\newcommand{\plusplus}[3]{
\pgfmathsetmacro{\s}{#1};
\pgfmathsetmacro{\x}{#2};
\pgfmathsetmacro{\y}{#3};
\begin{scope}[shift={(\x,\y)},scale=\s]
\draw(0,0)arc[start angle=180,end angle=0,radius=.75];
\draw(.5,0)arc[start angle=180,end angle=0,radius=.25];
\end{scope}

}

\newcommand{\hole}[3]{
\pgfmathsetmacro{\s}{#1};
\pgfmathsetmacro{\x}{#2};
\pgfmathsetmacro{\y}{#3};
\begin{scope}[scale=\s, shift={(\x,\y)}]
\draw[blue](.25,-1)arc[start angle=150,end angle=30, radius=.5];
            \draw[blue] (0,-.8)arc[start angle=-150,end angle=-30, radius=.8];
\end{scope}
}
\newcommand{\holenew}[3]{
\pgfmathsetmacro{\s}{#1};
\pgfmathsetmacro{\x}{#2};
\pgfmathsetmacro{\y}{#3};

\begin{scope}[shift={(\x,\y)},scale=\s]
\draw[blue](.25,-1)arc[start angle=150,end angle=30, radius=.5];
\draw[blue] (0,-.8)arc[start angle=-150,end angle=-30, radius=.8];
\end{scope}
}

\newcommand{\noholestring}[3]{
\pgfmathsetmacro{\s}{#1};
\pgfmathsetmacro{\x}{#2};
\pgfmathsetmacro{\y}{#3};

\begin{scope}[shift={(\x,\y)},scale=\s]
\draw(1.7,-1)arc[start angle=0,end angle=180,radius=1]--(1.7-2,-3.5)arc[start angle=180,end angle=360,radius=1]--(1.7,-1);
\end{scope}
\begin{scope}[shift={(\x+.15,\y-.5)},scale=\s*0.5]
\begin{scope}
\draw[blue](1.7,-1)arc[start angle=0,end angle=180,radius=1]--(1.7-2,-3.5)arc[start angle=180,end angle=360,radius=1]--(1.7,-1);
\end{scope}
\end{scope}
}
\newcommand{\holestring}[3]{
\pgfmathsetmacro{\s}{#1};
\pgfmathsetmacro{\x}{#2};
\pgfmathsetmacro{\y}{#3};

\begin{scope}[shift={(\x,\y)},scale=\s]
\draw[blue](.25,-1-1.5)arc[start angle=150,end angle=30, radius=.5];
\draw[blue](0,-.8-1.5)arc[start angle=-150,end angle=-30, radius=.8];
\draw(1.7,-1)arc[start angle=0,end angle=180,radius=1]--(1.7-2,-3.5)arc[start angle=180,end angle=360,radius=1]--(1.7,-1);
\end{scope}
}

\newcommand{\holestringsmall}[3]{
\pgfmathsetmacro{\s}{#1};
\pgfmathsetmacro{\x}{#2};
\pgfmathsetmacro{\y}{#3};

\begin{scope}[shift={(\x,\y)},scale=\s]
\draw[blue](.25,-1-1.5)arc[start angle=150,end angle=30, radius=.5];
\draw[blue](0,-.8-1.5)arc[start angle=-150,end angle=-30, radius=.8];
\draw(1.7,-2.5)arc[start angle=0,end angle=360,radius=1];
\end{scope}
}

\newcommand{\holespin}[3]{
\pgfmathsetmacro{\s}{#1};
\pgfmathsetmacro{\x}{#2};
\pgfmathsetmacro{\y}{#3};

\begin{scope}[shift={(\x,\y)},scale=\s/2]

\draw(-.5,.5)--(0,0);
\draw(-.5,0)--(0,0.5);
\end{scope}

}

\newcommand{\braidleft}[4]{
\pgfmathsetmacro{\yup}{#1};
\pgfmathsetmacro{\ydown}{#2};
\pgfmathsetmacro{\xleft}{#3};
\pgfmathsetmacro{\xright}{#4};
\begin{scope}[scale=\s, shift={(\x,\y)}]
\draw(\xleft,\ydown)--(\xright,\yup);
\draw[white,WL](\xleft,\yup)--(\xright,\ydown);
\draw(\xleft,\yup)--(\xright,\ydown);
\end{scope}
}

\newcommand{\tpspin}[3]{
\pgfmathsetmacro{\ss}{#1};
\pgfmathsetmacro{\xx}{#2};
\pgfmathsetmacro{\yy}{#3};
\begin{scope}[shift={(\xx,\yy)},scale=\ss]
\hole{0.6}{1}{1};
\draw(2,0.4)arc[start angle=15,end angle=345,radius=1]--(1.6,0.4)arc[start angle=30,end angle=330,radius=.6]--(2,.4);
\end{scope}
}

\newcommand{\tpspinnohole}[3]{
\pgfmathsetmacro{\ss}{#1};
\pgfmathsetmacro{\xx}{#2};
\pgfmathsetmacro{\yy}{#3};
\begin{scope}[shift={(\xx,\yy)},scale=\ss]
\draw(2,0.4)arc[start angle=15,end angle=345,radius=1]--(1.6,0.4)arc[start angle=30,end angle=330,radius=.6]--(2,.4);
\end{scope}
}
\newcommand{\crossing}{
\tikz[baseline=(current bounding box.center), line width=0.3mm] {

    \pgfmathsetmacro{\x}{0.5}
    \pgfmathsetmacro{\y}{1}
    \pgfmathsetmacro{\xx}{0.2}
    \draw (-\x,\y) -- (\x,-\y);
    \draw (\x,\y) -- (-\x,-\y);
    
}
}

\newcommand{\crossingleftnode}{
\tikz[baseline=(current bounding box.center), line width=0.3mm] {
\begin{scope}[yscale=.6,xscale=.75]
    \pgfmathsetmacro{\x}{0.5}
    \pgfmathsetmacro{\y}{1}
    \pgfmathsetmacro{\xx}{0.2}
    \draw (-\x,\y) -- (\x,-\y);
    \draw (\x,\y) -- (-\x,-\y);
    \node at (-\xx,0) {$\alpha$};
\end{scope}
}
}
\newcommand{\crossingrightnode}{
\tikz[baseline=(current bounding box.center), line width=0.3mm] {
\begin{scope}[yscale=.6,xscale=.75]
    \pgfmathsetmacro{\x}{0.5}
    \pgfmathsetmacro{\y}{1}
    \pgfmathsetmacro{\xx}{0.2}
    \draw (-\x,\y) -- (\x,-\y);
    \draw (\x,\y) -- (-\x,-\y);
    \node at (\xx,0) {$\alpha$};
\end{scope}
}
}

\newcommand{\crossingleftnodedoubleparameter}{
\tikz[baseline=(current bounding box.center), line width=0.3mm] {
\begin{scope}[yscale=.6,xscale=.75]
    \pgfmathsetmacro{\x}{0.5}
    \pgfmathsetmacro{\y}{1}
    \pgfmathsetmacro{\xx}{0.2}
    \draw (-\x,\y) -- (\x,-\y);
    \draw (\x,\y) -- (-\x,-\y);
    \node at (-4*\xx,0) {$(\alpha,\beta)$};
\end{scope}
}
}

\newcommand{\crossingrightnodedoubleparameter}{
\tikz[baseline=(current bounding box.center), line width=0.3mm] {
\begin{scope}[yscale=.6,xscale=.75]
    \pgfmathsetmacro{\x}{0.5}
    \pgfmathsetmacro{\y}{1}
    \pgfmathsetmacro{\xx}{0.2}
    \draw (-\x,\y) -- (\x,-\y);
    \draw (\x,\y) -- (-\x,-\y);
    \node at (4*\xx,0) {$(\alpha,\beta)$};
\end{scope}
}
}

\newcommand{\crossingleftnodetwo}{
\tikz[baseline=(current bounding box.center), line width=0.3mm] {
\begin{scope}[yscale=.6,xscale=.75]
    \pgfmathsetmacro{\x}{0.5}
    \pgfmathsetmacro{\y}{1}
    \pgfmathsetmacro{\xx}{0.2}
    \draw (-\x,\y) -- (\x,-\y);
    \draw (\x,\y) -- (-\x,-\y);
    \node at (-\xx-.2,0) {$\alpha\beta$};
\end{scope}
}
}

\newcommand{\crossinguppernode}{
\tikz[baseline=(current bounding box.center), line width=0.3mm] {
\begin{scope}[yscale=.6,xscale=.75]
    \pgfmathsetmacro{\x}{0.5}
    \pgfmathsetmacro{\y}{1}
    \pgfmathsetmacro{\yy}{0.4}
    \draw (-\x,\y) -- (\x,-\y);
    \draw (\x,\y) -- (-\x,-\y);
    \node at (0,\yy+.3) {$\frac{1-\alpha}{1+\alpha}$};
\end{scope}
}
}

\newcommand{\crossinguppernodealpha}{
\tikz[baseline=(current bounding box.center), line width=0.3mm] {
\begin{scope}[yscale=.6,xscale=.75]
    \pgfmathsetmacro{\x}{0.5}
    \pgfmathsetmacro{\y}{1}
    \pgfmathsetmacro{\yy}{0.4}
    \draw (-\x,\y) -- (\x,-\y);
    \draw (\x,\y) -- (-\x,-\y);
    \node at (0,\yy+.2) {$\alpha$};
\end{scope}
}
}

\newcommand{\crossinguppernodealphadoubleparameter}{
\tikz[baseline=(current bounding box.center), line width=0.3mm] {
\begin{scope}[yscale=.6,xscale=.75]
    \pgfmathsetmacro{\x}{0.5}
    \pgfmathsetmacro{\y}{1}
    \pgfmathsetmacro{\yy}{0.4}
    \draw (-\x,\y) -- (\x,-\y);
    \draw (\x,\y) -- (-\x,-\y);
    \node at (0,\yy*2) {$(\alpha,\beta)$};
\end{scope}
}
}

\newcommand{\crossingdownernodealphadoubleparameter}{
\tikz[baseline=(current bounding box.center), line width=0.3mm] {
\begin{scope}[yscale=.6,xscale=.75]
    \pgfmathsetmacro{\x}{0.5}
    \pgfmathsetmacro{\y}{1}
    \pgfmathsetmacro{\yy}{0.4}
    \draw (-\x,\y) -- (\x,-\y);
    \draw (\x,\y) -- (-\x,-\y);
    \node at (0,-\yy*2) {$(\alpha,\beta)$};
\end{scope}
}
}

\newcommand{\crossingdownernodealpha}{
\tikz[baseline=(current bounding box.center), line width=0.3mm] {
\begin{scope}[yscale=.6,xscale=.75]
    \pgfmathsetmacro{\x}{0.5}
    \pgfmathsetmacro{\y}{1}
    \pgfmathsetmacro{\yy}{0.4}
    \draw (-\x,\y) -- (\x,-\y);
    \draw (\x,\y) -- (-\x,-\y);
    \node at (0,-\yy-.2) {$\alpha$};
\end{scope}
}
}

\newcommand{\updowncircle}{
\tikz[baseline=(current bounding box.center), line width=0.3mm] {
\begin{scope}[yscale=.6,xscale=.6]
    \pgfmathsetmacro{\rx}{0.5}
    \pgfmathsetmacro{\ry}{0.5}
    \pgfmathsetmacro{\h}{1}
    \draw(0,\h) arc (180:360:{\rx} and {\ry});
    \draw(2*\rx,-\h) arc (0:180:{\rx} and {\ry});
\end{scope}
}
}

\newcommand{\doublestrings}{
\tikz[baseline=(current bounding box.center), line width=0.3mm] {
\begin{scope}[yscale=.6,xscale=.6]
    \pgfmathsetmacro{\rx}{0.5}
    \pgfmathsetmacro{\ry}{0.5}
    \pgfmathsetmacro{\h}{1}
    \draw(0,\h)--(0,-\h);
    \draw(2*\rx,-\h)--(2*\rx,\h);
\end{scope}
}
}
\newcommand{\doublestringscharge}{
\tikz[baseline=(current bounding box.center), line width=0.3mm] {
\begin{scope}[yscale=.6,xscale=.6]
    \pgfmathsetmacro{\rx}{0.5}
    \pgfmathsetmacro{\ry}{0.5}
    \pgfmathsetmacro{\h}{1}
    \draw(0,\h)--(0,-\h);
    \draw(2*\rx,-\h)--(2*\rx,\h);
    \draw[red,decorate,decoration={snake}](0,0)--(2*\rx,0);
    \fill(2*\rx,0) circle(.1);
    \fill(0,0) circle(.1);
\end{scope}
}
}

\newcommand{\doublestringslabel}[1]{
\tikz[baseline=(current bounding box.center), line width=0.3mm] {
\begin{scope}[yscale=.6,xscale=.6]
    \pgfmathsetmacro{\rx}{0.5}
    \pgfmathsetmacro{\ry}{0.5}
    \pgfmathsetmacro{\h}{1}
    \draw(0,\h)--(0,-\h);
    \draw(2*\rx,-\h)--(2*\rx,\h);
    \draw[red,decorate,decoration={snake}](0,0)--(2*\rx,0);
    \node at (\rx,\rx) {$#1$};
    \fill(2*\rx,0) circle(.1);
    \fill(0,0) circle(.1);
\end{scope}
}
}

\newcommand{\Psanbasisone}{
\tikz[baseline=(current bounding box.center), line width=0.3mm] {
    \pgfmathsetmacro{\rx}{0.5}
    \pgfmathsetmacro{\ry}{0.5}
    \pgfmathsetmacro{\h}{1}
    \draw(0,\h) arc (180:360:{\rx} and {\ry});
    \draw(2*\rx,-\h) arc (0:180:{\rx} and {\ry});
    \draw (3*\rx,\h) -- (3*\rx,-\h);
}
}

\newcommand{\updowncirclecharge}{
\tikz[baseline=(current bounding box.center), line width=0.3mm] {
\begin{scope}[yscale=.6,xscale=.6]
    \pgfmathsetmacro{\rx}{0.5}
    \pgfmathsetmacro{\ry}{0.5}
    \pgfmathsetmacro{\h}{1}
    \draw(0,\h) arc (180:360:{\rx} and {\ry});
    \draw(2*\rx,-\h) arc (0:180:{\rx} and {\ry});
    
    \draw[decorate, decoration={snake},red] (\rx,\h-\ry) -- (\rx,-\h+\ry);
    
    \fill[fill]  (\rx,\h-\ry)  circle[radius=2.5pt];
    \fill[fill]  (\rx,-\h+\ry)  circle[radius=2.5pt];
\end{scope}
}
}

\newcommand{\Psanbasistwo}{
\tikz[baseline=(current bounding box.center), line width=0.3mm] {
    \pgfmathsetmacro{\rx}{0.5}
    \pgfmathsetmacro{\ry}{0.5}
    \pgfmathsetmacro{\h}{1}
    \draw(0,\h) arc (180:360:{\rx} and {\ry});
    \draw(2*\rx,-\h) arc (0:180:{\rx} and {\ry});
    
    \draw[decorate, decoration={snake},red] (\rx,\h-\ry) -- (\rx,-\h+\ry);
    
    \fill[fill]  (\rx,\h-\ry)  circle[radius=2.5pt];
    \fill[fill]  (\rx,-\h+\ry)  circle[radius=2.5pt];

    \draw (3*\rx,\h) -- (3*\rx,-\h)
}
}

\newcommand{\Psanbasisthree}{
\tikz[baseline=(current bounding box.center), line width=0.3mm] {
    \pgfmathsetmacro{\rx}{0.5}
    \pgfmathsetmacro{\ry}{0.5}
    \pgfmathsetmacro{\h}{1}
    \draw(0,\h) arc (180:360:{\rx} and {\ry});
    \draw(2*\rx,-\h) arc (0:180:{\rx} and {\ry});
    
    \fill[fill]  (\rx,\h-\ry)  circle[radius=2.5pt];
    \fill[fill]  (3*\rx,0)  circle[radius=2.5pt];

    \draw (3*\rx,\h) -- (3*\rx,-\h);

    \draw[decorate, decoration={snake},red] (\rx,\h-\ry) -- (3*\rx,0);
}
}

\newcommand{\Psanbasisfour}{
\tikz[baseline=(current bounding box.center), line width=0.3mm] {
    \pgfmathsetmacro{\rx}{0.5}
    \pgfmathsetmacro{\ry}{0.5}
    \pgfmathsetmacro{\h}{1}
    \draw(0,\h) arc (180:360:{\rx} and {\ry});
    \draw(2*\rx,-\h) arc (0:180:{\rx} and {\ry});
    
    \fill[fill]  (\rx,-\h+\ry)  circle[radius=2.5pt];
    \fill[fill]  (3*\rx,0)  circle[radius=2.5pt];

    \draw (3*\rx,\h) -- (3*\rx,-\h);

    \draw[decorate, decoration={snake},red] (\rx,-\h+\ry) -- (3*\rx,0);
}
}

\newcommand{\Rone}{%
\tikz[baseline=(current bounding box.center), line width=0.3mm] {
    \pgfmathsetmacro{\rx}{1};
    \pgfmathsetmacro{\ry}{1};
    \pgfmathsetmacro{\rya}{0.2*\ry};
    \pgfmathsetmacro{\theta}{60};
    \pgfmathsetmacro{\rxa}{\rx*(1-cos(\theta))};
    \draw(0,-\rya) arc (180-\theta:180:{\rx} and {\ry});
    \draw(-\rxa,\rya) arc (180:180+\theta:{\rx} and {\ry});
    \draw(0,-\rya) arc (90:-90:{0.5*\ry*sin(\theta)-\rya});

    \node at (-0.5,-0.5) {$\alpha$};
}}

\newcommand{\shu}{
\tikz[baseline=(current bounding box.center), line width=0.3mm] {
  \draw[thick] (0,0.7) -- (0,-0.7);
}
}

\newcommand{\Ronebraidleft}{%
\begin{scope}
    \pgfmathsetmacro{\rx}{1};
    \pgfmathsetmacro{\ry}{1};
    \pgfmathsetmacro{\rya}{0.2*\ry};
    \pgfmathsetmacro{\theta}{60};
    \pgfmathsetmacro{\rxa}{\rx*(1-cos(\theta))};
    \draw(0,-\rya) arc (180-\theta:180:{\rx} and {\ry});
     \draw[white,WLL](-\rxa,\rya) arc (180:180+\theta:{\rx} and {\ry});
    \draw(-\rxa,\rya) arc (180:180+\theta:{\rx} and {\ry});
    \draw(0,-\rya) arc (90:-90:{0.5*\ry*sin(\theta)-\rya});

\end{scope}}

\newcommand{\Ronebraidright}{%
\begin{scope}
    \pgfmathsetmacro{\rx}{1};
    \pgfmathsetmacro{\ry}{1};
    \pgfmathsetmacro{\rya}{0.2*\ry};
    \pgfmathsetmacro{\theta}{60};
    \pgfmathsetmacro{\rxa}{\rx*(1-cos(\theta))};
   
     
    \draw(-\rxa,\rya) arc (180:180+\theta:{\rx} and {\ry});
    \draw[white,WLL](0,-\rya) arc (180-\theta:180:{\rx} and {\ry});

    \draw(0,-\rya) arc (180-\theta:180:{\rx} and {\ry});
    \draw(0,-\rya) arc (90:-90:{0.5*\ry*sin(\theta)-\rya});

\end{scope}}

\newcommand{\Rtwo}{
\tikz[baseline=(current bounding box.center), line width=0.3mm] {
    \pgfmathsetmacro{\x}{0.3}
    \pgfmathsetmacro{\y}{1.2}
    \pgfmathsetmacro{\rx}{1.8}
    \pgfmathsetmacro{\ry}{1.2}
    \pgfmathsetmacro{\yy}{0.7}
    \pgfmathsetmacro{\xx}{0.3}
    \draw(\x,-0) arc (0:60:{\rx} and {\ry});
    \draw(\x,-0) arc (0:-60:{\rx} and {\ry});
    \draw(-\x,0) arc (180:240:{\rx} and {\ry});
    \draw(-\x,0) arc (180:120:{\rx} and {\ry});

    \node at (-\xx,\yy) {$\beta$};
    \node at (-\xx,-\yy) {$\alpha$};
}
}

\newcommand{\Rtwobraidrightleft}{
\begin{scope} 
    \pgfmathsetmacro{\x}{0.3}
    \pgfmathsetmacro{\y}{1.2}
    \pgfmathsetmacro{\rx}{1.8}
    \pgfmathsetmacro{\ry}{1.2}
    \pgfmathsetmacro{\yy}{0.7}
    \pgfmathsetmacro{\xx}{0.3}

    \draw(\x,-0) arc (0:60:{\rx} and {\ry});

      \draw(\x,-0) arc (0:-60:{\rx} and {\ry});

        \draw[white,WLL](-\x,0) arc (180:120:{\rx} and {\ry});
    \draw(-\x,0) arc (180:120:{\rx} and {\ry});

     \draw[white,WLL](-\x,0) arc (180:240:{\rx} and {\ry}); \draw(-\x,0) arc (180:240:{\rx} and {\ry});

\end{scope}
}
\newcommand{\Rtwobraidleftright}{
\begin{scope} 
    \pgfmathsetmacro{\x}{0.3}
    \pgfmathsetmacro{\y}{1.2}
    \pgfmathsetmacro{\rx}{1.8}
    \pgfmathsetmacro{\ry}{1.2}
    \pgfmathsetmacro{\yy}{0.7}
    \pgfmathsetmacro{\xx}{0.3}

    \draw(-\x,0) arc (180:120:{\rx} and {\ry});

    \draw[white,WLL](\x,-0) arc (0:60:{\rx} and {\ry});
    \draw(\x,-0) arc (0:60:{\rx} and {\ry});

     \draw(-\x,0) arc (180:240:{\rx} and {\ry});

      \draw[white,WLL](\x,-0) arc (0:-60:{\rx} and {\ry}); 
      \draw(\x,-0) arc (0:-60:{\rx} and {\ry});

\end{scope}
}

\newcommand{\cylinderonly}[3]{%
\begin{scope}
  \pgfmathsetmacro{\Rx}{#1}
  \pgfmathsetmacro{\Ry}{#2}
  \pgfmathsetmacro{\H}{#3}
  \draw[draw=blue] (0,\H) ellipse ({\Rx} and {\Ry});
  \draw[blue] (-\Rx,0) -- (-\Rx,\H);
  \draw[blue] (\Rx,0) -- (\Rx,\H);
  \draw[blue, dashed] (\Rx,0) arc (0:180:{\Rx} and {\Ry});
  \draw[blue] (-\Rx,0) arc (180:360:{\Rx} and {\Ry});

  \pgfmathsetmacro{\theta}{250}
  \pgfmathsetmacro{\xx}{\Rx*cos(\theta)}
  \pgfmathsetmacro{\yy}{\Ry*sin(\theta)}
\draw[->,blue] (\xx,\yy) -- (\xx+0.01,\yy);
\draw[->,blue] (\xx,\yy+\H) -- (\xx+0.01,\yy+\H);
\end{scope}
}

\newcommand{\cylinderonlybulk}[3]{%
\begin{scope}
  \pgfmathsetmacro{\Rx}{#1}
  \pgfmathsetmacro{\Ry}{#2}
  \pgfmathsetmacro{\H}{#3}
   
   \fill[blue!25] (-\Rx,\H)--(-\Rx,0) arc (180:360:{\Rx} and {\Ry})--(\Rx,\H);
   \fill[blue!18] (0,\H) ellipse ({\Rx} and {\Ry});
  \draw[draw=blue] (0,\H) ellipse ({\Rx} and {\Ry});
  \draw[blue] (-\Rx,0) -- (-\Rx,\H);
  \draw[blue] (\Rx,0) -- (\Rx,\H);
  \draw[blue, dashed] (\Rx,0) arc (0:180:{\Rx} and {\Ry});
 
  \draw[blue] (-\Rx,0) arc (180:360:{\Rx} and {\Ry});

  \pgfmathsetmacro{\theta}{250}
  \pgfmathsetmacro{\xx}{\Rx*cos(\theta)}
  \pgfmathsetmacro{\yy}{\Ry*sin(\theta)}
\draw[->,blue] (\xx,\yy) -- (\xx+0.01,\yy);
\draw[->,blue] (\xx,\yy+\H) -- (\xx+0.01,\yy+\H);
\end{scope}
}

\newcommand{\caponly}[3]{%
\begin{scope}
  \pgfmathsetmacro{\Rx}{#1}
  \pgfmathsetmacro{\Ry}{#2}
  \pgfmathsetmacro{\H}{#3}
  \draw[blue] (-\Rx,0)arc[start angle=180, end angle=0, radius=\Rx];
  \draw[blue, dashed] (\Rx,0) arc (0:180:{\Rx} and {\Ry});
  \draw[blue] (-\Rx,0) arc (180:360:{\Rx} and {\Ry});

  \pgfmathsetmacro{\theta}{250}
  \pgfmathsetmacro{\xx}{\Rx*cos(\theta)}
  \pgfmathsetmacro{\yy}{\Ry*sin(\theta)}
\draw[->,blue] (\xx,\yy) -- (\xx+0.01,\yy);
\end{scope}
}

\newcommand{\caponlybulk}[3]{%
\begin{scope}
  \pgfmathsetmacro{\Rx}{#1}
  \pgfmathsetmacro{\Ry}{#2}
  \pgfmathsetmacro{\H}{#3}
  \fill[blue!25] (-\Rx,0)arc[start angle=180, end angle=0, radius=\Rx];
  \draw[blue] (-\Rx,0)arc[start angle=180, end angle=0, radius=\Rx];
  \draw[blue, dashed] (\Rx,0) arc (0:180:{\Rx} and {\Ry});
  \fill[blue!25] (-\Rx,0) arc (180:360:{\Rx} and {\Ry});
  \draw[blue] (-\Rx,0) arc (180:360:{\Rx} and {\Ry});

  \pgfmathsetmacro{\theta}{250}
  \pgfmathsetmacro{\xx}{\Rx*cos(\theta)}
  \pgfmathsetmacro{\yy}{\Ry*sin(\theta)}
\draw[->,blue] (\xx,\yy) -- (\xx+0.01,\yy);
\end{scope}
}

\newcommand{\cuponly}[3]{%
\begin{scope}
  \pgfmathsetmacro{\Rx}{#1}
  \pgfmathsetmacro{\Ry}{#2}
  \pgfmathsetmacro{\H}{#3}
  \draw[blue] (-\Rx,0)arc[start angle=180, end angle=360, radius=\Rx];
  \draw[blue] (\Rx,0) arc (0:180:{\Rx} and {\Ry});
  \draw[blue] (-\Rx,0) arc (180:360:{\Rx} and {\Ry});

  \pgfmathsetmacro{\theta}{250}
  \pgfmathsetmacro{\xx}{\Rx*cos(\theta)}
  \pgfmathsetmacro{\yy}{\Ry*sin(\theta)}
\draw[->,blue] (\xx,\yy) -- (\xx+0.01,\yy);
\end{scope}
}

\newcommand{\cuponlybulk}[3]{%
\begin{scope}
  \pgfmathsetmacro{\Rx}{#1}
  \pgfmathsetmacro{\Ry}{#2}
  \pgfmathsetmacro{\H}{#3}
   \fill[blue!25] (-\Rx,0)arc[start angle=180, end angle=360, radius=\Rx];
  \draw[blue] (-\Rx,0)arc[start angle=180, end angle=360, radius=\Rx];
 \fill[blue!15] (\Rx,0) arc (0:360:{\Rx} and {\Ry});
  \draw[blue] (\Rx,0) arc (0:180:{\Rx} and {\Ry});
  \draw[blue] (-\Rx,0) arc (180:360:{\Rx} and {\Ry});

  \pgfmathsetmacro{\theta}{250}
  \pgfmathsetmacro{\xx}{\Rx*cos(\theta)}
  \pgfmathsetmacro{\yy}{\Ry*sin(\theta)}
\draw[->,blue] (\xx,\yy) -- (\xx+0.01,\yy);
\end{scope}
}

\newcommand{\yuanzhuleftcharge}[3]{
\tikz[baseline=(current bounding box.center), scale=0.7, line width=0.3mm] {
  \pgfmathsetmacro{\Rx}{#1}
  \pgfmathsetmacro{\Ry}{#2}
  \pgfmathsetmacro{\H}{#3}

  \cylinderonly{\Rx}{\Ry}{\H}

  \pgfmathsetmacro{\cosA}{0.6427876}   
  \pgfmathsetmacro{\cosB}{-0.6427876}  
  \pgfmathsetmacro{\cosC}{0.2588190}   
  \pgfmathsetmacro{\cosD}{-0.2588190}  
  \pgfmathsetmacro{\sinA}{-0.7660444}  
  \pgfmathsetmacro{\sinB}{-0.7660444}  
  \pgfmathsetmacro{\sinC}{-0.9659258}  
  \pgfmathsetmacro{\sinD}{-0.9659258}  

  \pgfmathsetmacro{\xa}{\Rx*\cosA}
  \pgfmathsetmacro{\ya}{\Ry*\sinA}
  \pgfmathsetmacro{\xb}{\Rx*\cosB}
  \pgfmathsetmacro{\yb}{\Ry*\sinB}
  \pgfmathsetmacro{\xc}{\Rx*\cosC}
  \pgfmathsetmacro{\yc}{\Ry*\sinC}
  \pgfmathsetmacro{\xd}{\Rx*\cosD}
  \pgfmathsetmacro{\yd}{\Ry*\sinD}

  \draw[magenta] (\xa,\ya) -- (\xa,\ya+\H);
  \draw[magenta] (\xb,\yb) -- (\xb,\yb+\H);
  \draw (\xc,\yc) -- (\xc,\yc+\H);
  \draw (\xd,\yd) -- (\xd,\yd+\H);

  \pgfmathsetmacro{\yone}{\ya + \H/2}
  \pgfmathsetmacro{\ytwo}{\yc + \H/2}
  \fill[black] (\xa,\yone) circle[radius=2.5pt];
  \fill[black] (\xc,\ytwo) circle[radius=2.5pt];

  \draw[decorate, decoration={snake, segment length=4pt, amplitude=0.6mm}, red]
    (\xa,\yone) -- (\xc,\ytwo);
}
}

\newcommand{\yuanzhurightcharge}[3]{
\tikz[baseline=(current bounding box.center), scale=0.7, line width=0.3mm] {

  \pgfmathsetmacro{\Rx}{#1}
  \pgfmathsetmacro{\Ry}{#2}
  \pgfmathsetmacro{\H}{#3}

  \cylinderonly{\Rx}{\Ry}{\H}

  \pgfmathsetmacro{\cosA}{-0.6427876}   
  \pgfmathsetmacro{\cosB}{0.6427876}    
  \pgfmathsetmacro{\cosC}{-0.2588190}   
  \pgfmathsetmacro{\cosD}{0.2588190}    
  \pgfmathsetmacro{\sinA}{-0.7660444}
  \pgfmathsetmacro{\sinB}{-0.7660444}
  \pgfmathsetmacro{\sinC}{-0.9659258}
  \pgfmathsetmacro{\sinD}{-0.9659258}

  \pgfmathsetmacro{\xa}{\Rx*\cosA}
  \pgfmathsetmacro{\ya}{\Ry*\sinA}
  \pgfmathsetmacro{\xb}{\Rx*\cosB}
  \pgfmathsetmacro{\yb}{\Ry*\sinB}
  \pgfmathsetmacro{\xc}{\Rx*\cosC}
  \pgfmathsetmacro{\yc}{\Ry*\sinC}
  \pgfmathsetmacro{\xd}{\Rx*\cosD}
  \pgfmathsetmacro{\yd}{\Ry*\sinD}

  \draw[magenta] (\xa,\ya) -- (\xa,\ya+\H);
  \draw[magenta] (\xb,\yb) -- (\xb,\yb+\H);
  \draw (\xc,\yc) -- (\xc,\yc+\H);
  \draw (\xd,\yd) -- (\xd,\yd+\H);

  \pgfmathsetmacro{\yone}{\ya + \H/2}
  \pgfmathsetmacro{\ytwo}{\yc + \H/2}
  \fill[black] (\xa,\yone) circle[radius=2.5pt];
  \fill[black] (\xc,\ytwo) circle[radius=2.5pt];

  \draw[decorate, decoration={snake, segment length=4pt, amplitude=0.6mm}, red]
    (\xa,\yone) -- (\xc,\ytwo);

}
}

\newcommand{\yuanzhualpha}[3]{
\tikz[baseline=(current bounding box.center), scale=0.7, line width=0.3mm] {

  \pgfmathsetmacro{\Rx}{#1}
  \pgfmathsetmacro{\Ry}{#2}
  \pgfmathsetmacro{\H}{#3}

  \cylinderonly{\Rx}{\Ry}{\H}

  \pgfmathsetmacro{\cosA}{0.6427876}
  \pgfmathsetmacro{\cosB}{-0.6427876}
  \pgfmathsetmacro{\cosC}{0.2588190}
  \pgfmathsetmacro{\cosD}{-0.2588190}
  \pgfmathsetmacro{\sinA}{-0.7660444}
  \pgfmathsetmacro{\sinB}{-0.7660444}
  \pgfmathsetmacro{\sinC}{-0.9659258}
  \pgfmathsetmacro{\sinD}{-0.9659258}

  \pgfmathsetmacro{\xa}{\Rx*\cosA}
  \pgfmathsetmacro{\ya}{\Ry*\sinA}
  \pgfmathsetmacro{\xb}{\Rx*\cosB}
  \pgfmathsetmacro{\yb}{\Ry*\sinB}
  \pgfmathsetmacro{\xc}{\Rx*\cosC}
  \pgfmathsetmacro{\yc}{\Ry*\sinC}
  \pgfmathsetmacro{\xd}{\Rx*\cosD}
  \pgfmathsetmacro{\yd}{\Ry*\sinD}

  \draw[magenta] (\xa,\ya) -- (\xa,\ya+\H);
  \draw[magenta] (\xb,\yb) -- (\xb,\yb+\H);

  \draw (\xc,\yc) -- (\xd,\yd+\H);
  \draw (\xd,\yd) -- (\xc,\yc+\H);

  \node at (-\xc,\H/2.7) {$\alpha$};

}
}
\newcommand{\yuanzhualphablack}[3]{

\begin{scope} 
  \pgfmathsetmacro{\Rx}{#1}
  \pgfmathsetmacro{\Ry}{#2}
  \pgfmathsetmacro{\H}{#3}

  \cylinderonly{\Rx}{\Ry}{\H}

  \pgfmathsetmacro{\cosA}{0.6427876}
  \pgfmathsetmacro{\cosB}{-0.6427876}
  \pgfmathsetmacro{\cosC}{0.2588190}
  \pgfmathsetmacro{\cosD}{-0.2588190}
  \pgfmathsetmacro{\sinA}{-0.7660444}
  \pgfmathsetmacro{\sinB}{-0.7660444}
  \pgfmathsetmacro{\sinC}{-0.9659258}
  \pgfmathsetmacro{\sinD}{-0.9659258}

  \pgfmathsetmacro{\xa}{\Rx*\cosA}
  \pgfmathsetmacro{\ya}{\Ry*\sinA}
  \pgfmathsetmacro{\xb}{\Rx*\cosB}
  \pgfmathsetmacro{\yb}{\Ry*\sinB}
  \pgfmathsetmacro{\xc}{\Rx*\cosC}
  \pgfmathsetmacro{\yc}{\Ry*\sinC}
  \pgfmathsetmacro{\xd}{\Rx*\cosD}
  \pgfmathsetmacro{\yd}{\Ry*\sinD}

  \draw (\xa,\ya) -- (\xa,\ya+\H);
  \draw (\xb,\yb) -- (\xb,\yb+\H);

  \draw (\xc,\yc) -- (\xd,\yd+\H);
  \draw (\xd,\yd) -- (\xc,\yc+\H);

  \node at (\xc-.25,\H/2+.15) {$\alpha$};
\end{scope}
}

\newcommand{\yuanzhualphascope}[3]{
\begin{scope}[scale=.8]
  \pgfmathsetmacro{\Rx}{#1}
  \pgfmathsetmacro{\Ry}{#2}
  \pgfmathsetmacro{\H}{#3}

  \cylinderonly{\Rx}{\Ry}{\H}

  \pgfmathsetmacro{\cosA}{0.6427876}
  \pgfmathsetmacro{\cosB}{-0.6427876}
  \pgfmathsetmacro{\cosC}{0.2588190}
  \pgfmathsetmacro{\cosD}{-0.2588190}
  \pgfmathsetmacro{\sinA}{-0.7660444}
  \pgfmathsetmacro{\sinB}{-0.7660444}
  \pgfmathsetmacro{\sinC}{-0.9659258}
  \pgfmathsetmacro{\sinD}{-0.9659258}

  \pgfmathsetmacro{\xa}{\Rx*\cosA}
  \pgfmathsetmacro{\ya}{\Ry*\sinA}
  \pgfmathsetmacro{\xb}{\Rx*\cosB}
  \pgfmathsetmacro{\yb}{\Ry*\sinB}
  \pgfmathsetmacro{\xc}{\Rx*\cosC}
  \pgfmathsetmacro{\yc}{\Ry*\sinC}
  \pgfmathsetmacro{\xd}{\Rx*\cosD}
  \pgfmathsetmacro{\yd}{\Ry*\sinD}

  \draw[magenta] (\xa,\ya) -- (\xa,\ya+\H);
  \draw[magenta] (\xb,\yb) -- (\xb,\yb+\H);

  \draw (\xc,\yc) -- (\xd,\yd+\H);
  \draw (\xd,\yd) -- (\xc,\yc+\H);

  \node at (-\xc-.2,\H*.45) {$\alpha$};
\end{scope}
}

\newcommand{\yuanzhualphaleft}[3]{
\begin{scope}
  \pgfmathsetmacro{\Rx}{#1}
  \pgfmathsetmacro{\Ry}{#2}
  \pgfmathsetmacro{\H}{#3}

  \cylinderonly{\Rx}{\Ry}{\H}

  \pgfmathsetmacro{\cosA}{0.6427876}
  \pgfmathsetmacro{\cosB}{-0.6427876}
  \pgfmathsetmacro{\cosC}{0.2588190}
  \pgfmathsetmacro{\cosD}{-0.2588190}
  \pgfmathsetmacro{\sinA}{-0.7660444}
  \pgfmathsetmacro{\sinB}{-0.7660444}
  \pgfmathsetmacro{\sinC}{-0.9659258}
  \pgfmathsetmacro{\sinD}{-0.9659258}

  \pgfmathsetmacro{\xa}{\Rx*\cosA}
  \pgfmathsetmacro{\ya}{\Ry*\sinA}
  \pgfmathsetmacro{\xb}{\Rx*\cosB}
  \pgfmathsetmacro{\yb}{\Ry*\sinB}
  \pgfmathsetmacro{\xc}{\Rx*\cosC}
  \pgfmathsetmacro{\yc}{\Ry*\sinC}
  \pgfmathsetmacro{\xd}{\Rx*\cosD}
  \pgfmathsetmacro{\yd}{\Ry*\sinD}

  \draw(\xb,\yb) -- (\xd,\yd+\H);
  \draw (\xd,\yd) -- (\xb,\yb+\H);

  \draw (\xc,\yc) -- (\xc,\yc+\H);
  \draw (\xa,\ya) -- (\xa,\ya+\H);


\end{scope}
}

\newcommand{\yuanzhualphaleftspin}[3]{
\begin{scope}[scale=0.7]
  \pgfmathsetmacro{\Rx}{#1}
  \pgfmathsetmacro{\Ry}{#2}
  \pgfmathsetmacro{\H}{#3}

  \cylinderonly{\Rx}{\Ry}{\H}

  \pgfmathsetmacro{\cosA}{0.8660254} 
  \pgfmathsetmacro{\cosB}{-0.8660254}
  \pgfmathsetmacro{\cosC}{0.5}
  \pgfmathsetmacro{\cosD}{-0.5}
  \pgfmathsetmacro{\sinA}{-0.5}
  \pgfmathsetmacro{\sinB}{-0.5}
  \pgfmathsetmacro{\sinC}{-0.8660254}
  \pgfmathsetmacro{\sinD}{-0.8660254}

  \pgfmathsetmacro{\xa}{\Rx*\cosA}
  \pgfmathsetmacro{\ya}{\Ry*\sinA}
  \pgfmathsetmacro{\xb}{\Rx*\cosB}
  \pgfmathsetmacro{\yb}{\Ry*\sinB}
  \pgfmathsetmacro{\xc}{\Rx*\cosC}
  \pgfmathsetmacro{\yc}{\Ry*\sinC}
  \pgfmathsetmacro{\xd}{\Rx*\cosD}
  \pgfmathsetmacro{\yd}{\Ry*\sinD}

  \draw(\xb,\yb) -- (\xd,\yd+\H);
  \draw (\xd,\yd) -- (\xb,\yb+\H);

  \draw (\xc,\yc) -- (\xc,\yc+\H);
  \draw (\xa,\ya) -- (\xa,\ya+\H);

  \node at (-\xa-.1,\H/2+.05) {$\alpha$};

  \draw[blue](\xb/2+\xc/2+.14,\H/2-.3)arc[start angle=180,end angle=0, radius=.1];

  \draw[blue](\xb/2+\xc/2+.05,\H/2-.3)arc[start angle=200,end angle=340, radius=.2];
  \draw(\xb/2+\xc/2-.35,\H/2-.3)arc[start angle=180,end angle=540, radius=.6];
  \draw(\xb/2+\xc/2-.15,\H/2-.3)arc[start angle=180,end angle=540, radius=.4];

\end{scope}
}

\newcommand{\yuanzhualphaleftspinreal}[3]{
\begin{scope}[scale=0.7]
  \pgfmathsetmacro{\Rx}{#1}
  \pgfmathsetmacro{\Ry}{#2}
  \pgfmathsetmacro{\H}{#3}

  \cylinderonly{\Rx}{\Ry}{\H}

  \pgfmathsetmacro{\cosA}{0.8660254} 
  \pgfmathsetmacro{\cosB}{-0.8660254}
  \pgfmathsetmacro{\cosC}{0.5}
  \pgfmathsetmacro{\cosD}{-0.5}
  \pgfmathsetmacro{\sinA}{-0.5}
  \pgfmathsetmacro{\sinB}{-0.5}
  \pgfmathsetmacro{\sinC}{-0.8660254}
  \pgfmathsetmacro{\sinD}{-0.8660254}

  \pgfmathsetmacro{\xa}{\Rx*\cosA}
  \pgfmathsetmacro{\ya}{\Ry*\sinA}
  \pgfmathsetmacro{\xb}{\Rx*\cosB}
  \pgfmathsetmacro{\yb}{\Ry*\sinB}
  \pgfmathsetmacro{\xc}{\Rx*\cosC}
  \pgfmathsetmacro{\yc}{\Ry*\sinC}
  \pgfmathsetmacro{\xd}{\Rx*\cosD}
  \pgfmathsetmacro{\yd}{\Ry*\sinD}

  \draw(\xb,\yb) -- (\xb,\yb+\H);
  \draw (\xd,\yd) -- (\xd,\yd+\H);

  \draw (\xc,\yc) -- (\xc,\yc+\H);
  \draw (\xa,\ya) -- (\xa,\ya+\H);

  \node at (\xd+.05,\H/2-.3) {$\alpha$};

  \draw[blue](\xb/2+\xc/2+.14,\H/2-.25)arc[start angle=180,end angle=0, radius=.1];

  \draw[blue](\xb/2+\xc/2+.05,\H/2-.25)arc[start angle=200,end angle=340, radius=.2];
  \draw(\xb/2+\xc/2-.28,\H/2)arc[start angle=150,end angle=-150, radius=.6]--(\xb/2+\xc/2-.1,\H/2-.1);
  \draw(\xb/2+\xc/2-.1,\H/2-.1)arc[start angle=150,end angle=-150, radius=.4]--(\xb/2+\xc/2-.28,\H/2);

\end{scope}
}

\newcommand{\yibiomega}{
\tikz[baseline=(current bounding box.center), line width=0.3mm] {
    \pgfmathsetmacro{\x}{0.5}
    \pgfmathsetmacro{\y}{1}
    \pgfmathsetmacro{\xx}{0.3}
    \draw (-\x,\y) -- (\x,-\y);
    \draw (\x,\y) -- (-\x,-\y);
    \node at (-\xx,0) {$\frac{1}{\omega}$};
}}

\newcommand{\yibiomegacharge}{
\tikz[baseline=(current bounding box.center), line width=0.3mm] {
    \pgfmathsetmacro{\x}{0.5}
    \pgfmathsetmacro{\y}{1}
    \pgfmathsetmacro{\xx}{0.3}
    \draw (-\x,\y) -- (\x,-\y);
    \draw (\x,\y) -- (-\x,-\y);
    \node at (-\xx,0) {$\frac{1}{\omega'}$};

    \fill[fill] (\x/2,\y/2)  circle[radius=2.5pt];
    \fill[fill] (-\x/2,\y/2)  circle[radius=2.5pt];
    \draw[decorate, decoration={snake},red] (-\x/2,\y/2) -- (\x/2,\y/2);
}}


\newcommand{\shadowone}{%
\tikz[baseline=(current bounding box.center), line width=0.3mm] {
    \pgfmathsetmacro{\x}{1}
    \pgfmathsetmacro{\y}{1}
    
    \draw (-\x,\y) -- (\x,-\y);
    \draw (\x,\y) -- (-\x,-\y);
    \draw (0.5*\x,\y) -- (0.5*\x,-\y);
    
    
    \node at (0.7*\x,-0.85*\y) {$c_1$};
    \node at (0.7*\x,0.8*\y+.1) {$c_3$};
    \node at (-0.2*\x,0) {$c_2$};
}}

\newcommand{\shadowtwo}{%
\tikz[baseline=(current bounding box.center), line width=0.3mm] {
    \pgfmathsetmacro{\x}{1}
    \pgfmathsetmacro{\y}{1}
    
    \draw (-\x,\y) -- (\x,-\y);
    \draw (\x,\y) -- (-\x,-\y);
    \draw (-0.5*\x,\y) -- (-0.5*\x,-\y);
    
    
    \node at (-0.7*\x,0.5*\y) {$b_1$};
    \node at (-0.7*\x,-0.5*\y+.1) {$b_3$};
    \node at (0,0.2*\y+.1) {$b_2$};
}}

\newcommand{\shadowtwoprime}{%
\tikz[baseline=(current bounding box.center), line width=0.3mm] {
    \pgfmathsetmacro{\x}{1}
    \pgfmathsetmacro{\y}{1}
    
    \draw (-\x,\y) -- (\x,-\y);
    \draw (\x,\y) -- (-\x,-\y);
    \draw (-0.5*\x,\y) -- (-0.5*\x,-\y);
    
    
    \node at (-0.7*\x,0.5*\y) {$b_1'$};
    \node at (-0.7*\x,-0.5*\y+.1) {$b_3'$};
    \node at (0,0.2*\y+.1) {$b_2'$};
}}

\newcommand{\shadowthree}{%
\tikz[baseline=(current bounding box.center), line width=0.3mm] {
    \pgfmathsetmacro{\x}{1}
    \pgfmathsetmacro{\y}{1}
    
    \draw (-\x,\y) -- (\x,-\y);
    \draw (\x,\y) -- (-\x,-\y);
    \draw (-0.5*\x,\y) -- (-0.5*\x,-\y);
    
    
    \node at (-0.7*\x,0.5*\y) {$b_1'$};
    \node at (-0.7*\x,-0.5*\y) {$b_3'$};
    \node at (0,0.2*\y) {$b_2'$};
}}


\newcommand{\yuanzhuS}[3]{
\begin{scope}[scale=.7]

  \pgfmathsetmacro{\Rx}{#1}
  \pgfmathsetmacro{\Ry}{#2}
  \pgfmathsetmacro{\H}{#3}

  \cylinderonly{\Rx}{\Ry}{\H}

  \pgfmathsetmacro{\cosA}{0.6427876}
  \pgfmathsetmacro{\cosB}{-0.6427876}
  \pgfmathsetmacro{\cosC}{0.2588190}
  \pgfmathsetmacro{\cosD}{-0.2588190}
  \pgfmathsetmacro{\sinA}{-0.7660444}
  \pgfmathsetmacro{\sinB}{-0.7660444}
  \pgfmathsetmacro{\sinC}{-0.9659258}
  \pgfmathsetmacro{\sinD}{-0.9659258}

  \pgfmathsetmacro{\xa}{\Rx*\cosA}
  \pgfmathsetmacro{\ya}{\Ry*\sinA}
  \pgfmathsetmacro{\xb}{\Rx*\cosB}
  \pgfmathsetmacro{\yb}{\Ry*\sinB}
  \pgfmathsetmacro{\xc}{\Rx*\cosC}
  \pgfmathsetmacro{\yc}{\Ry*\sinC}
  \pgfmathsetmacro{\xd}{\Rx*\cosD}
  \pgfmathsetmacro{\yd}{\Ry*\sinD}

  \draw[] (\xa,\ya) -- (\xa,\ya+\H);
  \draw[] (\xc,\yc) -- (\xc,\yc+\H);
 
  \draw (\xd,\yd) -- (\xb,\yb+\H);
   \draw[white,WL] (\xb,\yb) -- (\xd,\yd+\H);
  
   \draw (\xb,\yb) -- (\xd,\yd+\H);


\end{scope}

}
\newcommand{\yuanzhuSS}[3]{
\begin{scope}[scale=.7]

  \pgfmathsetmacro{\Rx}{#1}
  \pgfmathsetmacro{\Ry}{#2}
  \pgfmathsetmacro{\H}{#3}

  \cylinderonly{\Rx}{\Ry}{\H}

  \pgfmathsetmacro{\cosA}{0.6427876}
  \pgfmathsetmacro{\cosB}{-0.6427876}
  \pgfmathsetmacro{\cosC}{0.2588190}
  \pgfmathsetmacro{\cosD}{-0.2588190}
  \pgfmathsetmacro{\sinA}{-0.7660444}
  \pgfmathsetmacro{\sinB}{-0.7660444}
  \pgfmathsetmacro{\sinC}{-0.9659258}
  \pgfmathsetmacro{\sinD}{-0.9659258}

  \pgfmathsetmacro{\xa}{\Rx*\cosA}
  \pgfmathsetmacro{\ya}{\Ry*\sinA}
  \pgfmathsetmacro{\xb}{\Rx*\cosB}
  \pgfmathsetmacro{\yb}{\Ry*\sinB}
  \pgfmathsetmacro{\xc}{\Rx*\cosC}
  \pgfmathsetmacro{\yc}{\Ry*\sinC}
  \pgfmathsetmacro{\xd}{\Rx*\cosD}
  \pgfmathsetmacro{\yd}{\Ry*\sinD}

\draw (\xa,\ya) -- (\xc,\yc+\H);

  \draw[white,WL] (\xc,\yc) -- (\xa,\ya+\H);
  \draw (\xc,\yc) -- (\xa,\ya+\H);

  \draw (\xd,\yd) -- (\xd,\yd+\H);
  
   \draw (\xb,\yb) -- (\xb,\yb+\H);


\end{scope}

}
\newcommand{\yuanzhuSSDG}[3]{
\begin{scope}[scale=.7]

  \pgfmathsetmacro{\Rx}{#1}
  \pgfmathsetmacro{\Ry}{#2}
  \pgfmathsetmacro{\H}{#3}

  \cylinderonly{\Rx}{\Ry}{\H}

  \pgfmathsetmacro{\cosA}{0.6427876}
  \pgfmathsetmacro{\cosB}{-0.6427876}
  \pgfmathsetmacro{\cosC}{0.2588190}
  \pgfmathsetmacro{\cosD}{-0.2588190}
  \pgfmathsetmacro{\sinA}{-0.7660444}
  \pgfmathsetmacro{\sinB}{-0.7660444}
  \pgfmathsetmacro{\sinC}{-0.9659258}
  \pgfmathsetmacro{\sinD}{-0.9659258}

  \pgfmathsetmacro{\xa}{\Rx*\cosA}
  \pgfmathsetmacro{\ya}{\Ry*\sinA}
  \pgfmathsetmacro{\xb}{\Rx*\cosB}
  \pgfmathsetmacro{\yb}{\Ry*\sinB}
  \pgfmathsetmacro{\xc}{\Rx*\cosC}
  \pgfmathsetmacro{\yc}{\Ry*\sinC}
  \pgfmathsetmacro{\xd}{\Rx*\cosD}
  \pgfmathsetmacro{\yd}{\Ry*\sinD}

  \draw (\xc,\yc) -- (\xa,\ya+\H);
  \draw[white,WL] (\xa,\ya) -- (\xc,\yc+\H);
  \draw (\xa,\ya) -- (\xc,\yc+\H);

  \draw (\xd,\yd) -- (\xd,\yd+\H);
  
   \draw (\xb,\yb) -- (\xb,\yb+\H);


\end{scope}

}
\newcommand{\yuanzhuSDG}[3]{
\begin{scope}[scale=.7]
  \pgfmathsetmacro{\Rx}{#1}
  \pgfmathsetmacro{\Ry}{#2}
  \pgfmathsetmacro{\H}{#3}

  \cylinderonly{\Rx}{\Ry}{\H}

  \pgfmathsetmacro{\cosA}{0.6427876}
  \pgfmathsetmacro{\cosB}{-0.6427876}
  \pgfmathsetmacro{\cosC}{0.2588190}
  \pgfmathsetmacro{\cosD}{-0.2588190}
  \pgfmathsetmacro{\sinA}{-0.7660444}
  \pgfmathsetmacro{\sinB}{-0.7660444}
  \pgfmathsetmacro{\sinC}{-0.9659258}
  \pgfmathsetmacro{\sinD}{-0.9659258}

  \pgfmathsetmacro{\xa}{\Rx*\cosA}
  \pgfmathsetmacro{\ya}{\Ry*\sinA}
  \pgfmathsetmacro{\xb}{\Rx*\cosB}
  \pgfmathsetmacro{\yb}{\Ry*\sinB}
  \pgfmathsetmacro{\xc}{\Rx*\cosC}
  \pgfmathsetmacro{\yc}{\Ry*\sinC}
  \pgfmathsetmacro{\xd}{\Rx*\cosD}
  \pgfmathsetmacro{\yd}{\Ry*\sinD}

   \draw[] (\xa,\ya) -- (\xa,\ya+\H);
  \draw[] (\xc,\yc) -- (\xc,\yc+\H);
  \draw (\xb,\yb) -- (\xd,\yd+\H);
  
  \draw[WL,white] (\xd,\yd) -- (\xb,\yb+\H);
  \draw (\xd,\yd) -- (\xb,\yb+\H);

\end{scope}
}

\newcommand{\yuanzhuID}[3]{
\begin{scope}[scale=.7]
  \pgfmathsetmacro{\Rx}{#1}
  \pgfmathsetmacro{\Ry}{#2}
  \pgfmathsetmacro{\H}{#3}

  \cylinderonlybulk{\Rx}{\Ry}{\H}

  \pgfmathsetmacro{\cosA}{0.6427876}
  \pgfmathsetmacro{\cosB}{-0.6427876}
  \pgfmathsetmacro{\cosC}{0.2588190}
  \pgfmathsetmacro{\cosD}{-0.2588190}
  \pgfmathsetmacro{\sinA}{-0.7660444}
  \pgfmathsetmacro{\sinB}{-0.7660444}
  \pgfmathsetmacro{\sinC}{-0.9659258}
  \pgfmathsetmacro{\sinD}{-0.9659258}

  \pgfmathsetmacro{\xa}{\Rx*\cosA}
  \pgfmathsetmacro{\ya}{\Ry*\sinA}
  \pgfmathsetmacro{\xb}{\Rx*\cosB}
  \pgfmathsetmacro{\yb}{\Ry*\sinB}
  \pgfmathsetmacro{\xc}{\Rx*\cosC}
  \pgfmathsetmacro{\yc}{\Ry*\sinC}
  \pgfmathsetmacro{\xd}{\Rx*\cosD}
  \pgfmathsetmacro{\yd}{\Ry*\sinD}

  \draw[] (\xa,\ya) -- (\xa,\ya+\H);
  \draw (\xd,\yd) -- (\xd,\yd+\H);
  \draw (\xc,\yc) -- (\xc,\yc+\H);
  
  \draw[] (\xb,\yb) -- (\xb,\yb+\H);

\end{scope}
}

\newcommand{\yuanzhuddd}[3]{
\begin{scope}[scale=.7]
  \pgfmathsetmacro{\Rx}{#1}
  \pgfmathsetmacro{\Ry}{#2}
  \pgfmathsetmacro{\H}{#3}

  \cylinderonly{\Rx}{\Ry}{\H}

  \pgfmathsetmacro{\cosA}{0.6427876}
  \pgfmathsetmacro{\cosB}{-0.6427876}
  \pgfmathsetmacro{\cosC}{0.2588190}
  \pgfmathsetmacro{\cosD}{-0.2588190}
  \pgfmathsetmacro{\sinA}{-0.7660444}
  \pgfmathsetmacro{\sinB}{-0.7660444}
  \pgfmathsetmacro{\sinC}{-0.9659258}
  \pgfmathsetmacro{\sinD}{-0.9659258}

  \pgfmathsetmacro{\xa}{\Rx*\cosA}
  \pgfmathsetmacro{\ya}{\Ry*\sinA}
  \pgfmathsetmacro{\xb}{\Rx*\cosB}
  \pgfmathsetmacro{\yb}{\Ry*\sinB}
  \pgfmathsetmacro{\xc}{\Rx*\cosC}
  \pgfmathsetmacro{\yc}{\Ry*\sinC}
  \pgfmathsetmacro{\xd}{\Rx*\cosD}
  \pgfmathsetmacro{\yd}{\Ry*\sinD}

  \draw (\xa,\ya) -- (\xa,\ya+\H);
  \draw (\xd,\yd) -- (\xd,\yd+\H);
  \draw (\xc,\yc) -- (\xc,\yc+\H);
  
  \draw (\xb,\yb) -- (\xb,\yb+\H);
  \node at(\xb/2+\xc/2,\xd+\H/2) {$...$};

\end{scope}
}

\newcommand{\yuanzhuNone}[3]{
\begin{scope}[scale=.7]
  \pgfmathsetmacro{\Rx}{#1}
  \pgfmathsetmacro{\Ry}{#2}
  \pgfmathsetmacro{\H}{#3}

  \cylinderonly{\Rx}{\Ry}{\H}

  \pgfmathsetmacro{\cosA}{0.6427876}
  \pgfmathsetmacro{\cosB}{-0.6427876}
  \pgfmathsetmacro{\cosC}{0.2588190}
  \pgfmathsetmacro{\cosD}{-0.2588190}
  \pgfmathsetmacro{\sinA}{-0.7660444}
  \pgfmathsetmacro{\sinB}{-0.7660444}
  \pgfmathsetmacro{\sinC}{-0.9659258}
  \pgfmathsetmacro{\sinD}{-0.9659258}

  \pgfmathsetmacro{\xa}{\Rx*\cosA}
  \pgfmathsetmacro{\ya}{\Ry*\sinA}
  \pgfmathsetmacro{\xb}{\Rx*\cosB}
  \pgfmathsetmacro{\yb}{\Ry*\sinB}
  \pgfmathsetmacro{\xc}{\Rx*\cosC}
  \pgfmathsetmacro{\yc}{\Ry*\sinC}
  \pgfmathsetmacro{\xd}{\Rx*\cosD}
  \pgfmathsetmacro{\yd}{\Ry*\sinD}

  

\end{scope}
}

\newcommand{\yuanzhuIDEMPTY}[3]{
\begin{scope}[scale=.7]
  \pgfmathsetmacro{\Rx}{#1}
  \pgfmathsetmacro{\Ry}{#2}
  \pgfmathsetmacro{\H}{#3}

  \cylinderonly{\Rx}{\Ry}{\H}

  \pgfmathsetmacro{\cosA}{0.6427876}
  \pgfmathsetmacro{\cosB}{-0.6427876}
  \pgfmathsetmacro{\cosC}{0.2588190}
  \pgfmathsetmacro{\cosD}{-0.2588190}
  \pgfmathsetmacro{\sinA}{-0.7660444}
  \pgfmathsetmacro{\sinB}{-0.7660444}
  \pgfmathsetmacro{\sinC}{-0.9659258}
  \pgfmathsetmacro{\sinD}{-0.9659258}

  \pgfmathsetmacro{\xa}{\Rx*\cosA}
  \pgfmathsetmacro{\ya}{\Ry*\sinA}
  \pgfmathsetmacro{\xb}{\Rx*\cosB}
  \pgfmathsetmacro{\yb}{\Ry*\sinB}
  \pgfmathsetmacro{\xc}{\Rx*\cosC}
  \pgfmathsetmacro{\yc}{\Ry*\sinC}
  \pgfmathsetmacro{\xd}{\Rx*\cosD}
  \pgfmathsetmacro{\yd}{\Ry*\sinD}

  \draw[] (\xa,\ya) -- (\xa,\ya+\H);
  \draw (\xd,\yd) -- (\xd,\yd+\H);
  \draw (\xc,\yc) -- (\xc,\yc+\H);
  
  \draw[] (\xb,\yb) -- (\xb,\yb+\H);

\end{scope}
}

\newcommand{\yuanzhuIDEMPTYTWO}[3]{
\begin{scope}[scale=.7]
  \pgfmathsetmacro{\Rx}{#1}
  \pgfmathsetmacro{\Ry}{#2}
  \pgfmathsetmacro{\H}{#3}

  \cylinderonly{\Rx}{\Ry}{\H}

  \pgfmathsetmacro{\cosA}{0.6427876}
  \pgfmathsetmacro{\cosB}{-0.6427876}
  \pgfmathsetmacro{\cosC}{0.2588190}
  \pgfmathsetmacro{\cosD}{-0.2588190}
  \pgfmathsetmacro{\sinA}{-0.7660444}
  \pgfmathsetmacro{\sinB}{-0.7660444}
  \pgfmathsetmacro{\sinC}{-0.9659258}
  \pgfmathsetmacro{\sinD}{-0.9659258}

  \pgfmathsetmacro{\xa}{\Rx*\cosA}
  \pgfmathsetmacro{\ya}{\Ry*\sinA}
  \pgfmathsetmacro{\xb}{\Rx*\cosB}
  \pgfmathsetmacro{\yb}{\Ry*\sinB}
  \pgfmathsetmacro{\xc}{\Rx*\cosC}
  \pgfmathsetmacro{\yc}{\Ry*\sinC}
  \pgfmathsetmacro{\xd}{\Rx*\cosD}
  \pgfmathsetmacro{\yd}{\Ry*\sinD}

  \draw (\xd,\yd) -- (\xd,\yd+\H);
  \draw (\xc,\yc) -- (\xc,\yc+\H);
  

\end{scope}
}

\newcommand{\yuanzhuHDG}[3]{
\begin{scope}[scale=.7]
  \pgfmathsetmacro{\Rx}{#1}
  \pgfmathsetmacro{\Ry}{#2}
  \pgfmathsetmacro{\H}{#3}

  \cylinderonly{\Rx}{\Ry}{\H}

  \pgfmathsetmacro{\cosA}{0.6427876}
  \pgfmathsetmacro{\cosB}{-0.6427876}
  \pgfmathsetmacro{\cosC}{0.2588190}
  \pgfmathsetmacro{\cosD}{-0.2588190}
  \pgfmathsetmacro{\sinA}{-0.7660444}
  \pgfmathsetmacro{\sinB}{-0.7660444}
  \pgfmathsetmacro{\sinC}{-0.9659258}
  \pgfmathsetmacro{\sinD}{-0.9659258}

  \pgfmathsetmacro{\xa}{\Rx*\cosA}
  \pgfmathsetmacro{\ya}{\Ry*\sinA}
  \pgfmathsetmacro{\xb}{\Rx*\cosB}
  \pgfmathsetmacro{\yb}{\Ry*\sinB}
  \pgfmathsetmacro{\xc}{\Rx*\cosC}
  \pgfmathsetmacro{\yc}{\Ry*\sinC}
  \pgfmathsetmacro{\xd}{\Rx*\cosD}
  \pgfmathsetmacro{\yd}{\Ry*\sinD}

  \draw (\xa,\ya) -- (\xa,\ya+\H);
  \draw (\xd,\yd) -- (\xb,\yb+\H/3);
  \draw (\xc,\yc) -- (\xc,\yc+\H/3);
  
  \draw[WL,white] (\xb,\yb) -- (\xd,\yd+\H/3);
  \draw (\xb,\yb) -- (\xd,\yd+\H/3);

    \draw (\xd,\yd+\H*2/3) -- (\xc,\yc+\H/3);
    \draw[WL,white] (\xc,\yc+\H*2/3) -- (\xd,\yd+\H/3);
    \draw (\xc,\yc+\H*2/3) -- (\xd,\yd+\H/3);

    \draw (\xc,\yc+\H*2/3) -- (\xc,\yc+\H);

     \draw (\xb,\yb+\H/3)--(\xb,\yb+\H*2/3);

       \draw(\xb,\yb+\H*2/3)-- (\xd,\yd+\H);


    \draw (\xd,\yd+\H*2/3) -- (\xb,\yb+\H);
    \draw (\xd,\yd+\H*2/3) -- (\xb,\yb+\H);
    
    \draw[WL,white](\xb,\yb+\H*2/3)-- (\xd,\yd+\H);
    \draw(\xb,\yb+\H*2/3)-- (\xd,\yd+\H);




\end{scope}
}

\newcommand{\yuanzhuHTRDG}[3]{
\begin{scope}[scale=.7]
  \pgfmathsetmacro{\Rx}{#1}
  \pgfmathsetmacro{\Ry}{#2}
  \pgfmathsetmacro{\H}{#3}

  \cylinderonly{\Rx}{\Ry}{\H}

  \pgfmathsetmacro{\cosA}{0.6427876}
  \pgfmathsetmacro{\cosB}{-0.6427876}
  \pgfmathsetmacro{\cosC}{0.2588190}
  \pgfmathsetmacro{\cosD}{-0.2588190}
  \pgfmathsetmacro{\sinA}{-0.7660444}
  \pgfmathsetmacro{\sinB}{-0.7660444}
  \pgfmathsetmacro{\sinC}{-0.9659258}
  \pgfmathsetmacro{\sinD}{-0.9659258}

  \pgfmathsetmacro{\xa}{\Rx*\cosA}
  \pgfmathsetmacro{\ya}{\Ry*\sinA}
  \pgfmathsetmacro{\xb}{\Rx*\cosB}
  \pgfmathsetmacro{\yb}{\Ry*\sinB}
  \pgfmathsetmacro{\xc}{\Rx*\cosC}
  \pgfmathsetmacro{\yc}{\Ry*\sinC}
  \pgfmathsetmacro{\xd}{\Rx*\cosD}
  \pgfmathsetmacro{\yd}{\Ry*\sinD}

  \draw (\xa,\ya+\H/3) -- (\xa,\ya+\H);
  \draw (\xd,\yd) -- (\xd,\yd+\H/3);
   \draw (\xb,\yb) -- (\xb,\yb+\H/3);
  
  \draw (\xa,\ya) -- (\xc,\yc+\H/3);
\draw[white,WL] (\xa,\ya+\H/3) -- (\xc,\yc);
\draw (\xa,\ya+\H/3) -- (\xc,\yc);

    \draw (\xd,\yd+\H*2/3) -- (\xc,\yc+\H/3);
    \draw[WL,white] (\xc,\yc+\H*2/3) -- (\xd,\yd+\H/3);
    \draw (\xc,\yc+\H*2/3) -- (\xd,\yd+\H/3);

    \draw (\xc,\yc+\H*2/3) -- (\xc,\yc+\H);

     \draw (\xb,\yb+\H/3)--(\xb,\yb+\H*2/3);

    \draw(\xb,\yb+\H*2/3)-- (\xd,\yd+\H);


    \draw (\xd,\yd+\H*2/3) -- (\xb,\yb+\H);
    \draw (\xd,\yd+\H*2/3) -- (\xb,\yb+\H);
    
    \draw[WL,white](\xb,\yb+\H*2/3)-- (\xd,\yd+\H);
    \draw(\xb,\yb+\H*2/3)-- (\xd,\yd+\H);




\end{scope}
}

\newcommand{\yuanzhuH}[3]{
\begin{scope}[scale=.7]
  \pgfmathsetmacro{\Rx}{#1}
  \pgfmathsetmacro{\Ry}{#2}
  \pgfmathsetmacro{\H}{#3}

  \cylinderonly{\Rx}{\Ry}{\H}

  \pgfmathsetmacro{\cosA}{0.6427876}
  \pgfmathsetmacro{\cosB}{-0.6427876}
  \pgfmathsetmacro{\cosC}{0.2588190}
  \pgfmathsetmacro{\cosD}{-0.2588190}
  \pgfmathsetmacro{\sinA}{-0.7660444}
  \pgfmathsetmacro{\sinB}{-0.7660444}
  \pgfmathsetmacro{\sinC}{-0.9659258}
  \pgfmathsetmacro{\sinD}{-0.9659258}

  \pgfmathsetmacro{\xa}{\Rx*\cosA}
  \pgfmathsetmacro{\ya}{\Ry*\sinA}
  \pgfmathsetmacro{\xb}{\Rx*\cosB}
  \pgfmathsetmacro{\yb}{\Ry*\sinB}
  \pgfmathsetmacro{\xc}{\Rx*\cosC}
  \pgfmathsetmacro{\yc}{\Ry*\sinC}
  \pgfmathsetmacro{\xd}{\Rx*\cosD}
  \pgfmathsetmacro{\yd}{\Ry*\sinD}

  \draw (\xa,\ya) -- (\xa,\ya+\H);
  
  \draw (\xc,\yc) -- (\xc,\yc+\H/3);
  
  \draw (\xb,\yb) -- (\xd,\yd+\H/3);
  \draw[WL,white] (\xd,\yd) -- (\xb,\yb+\H/3);
  \draw (\xd,\yd) -- (\xb,\yb+\H/3);

    \draw (\xc,\yc+\H*2/3) -- (\xd,\yd+\H/3);
    \draw[WL,white] (\xd,\yd+\H*2/3) -- (\xc,\yc+\H/3);
    \draw (\xd,\yd+\H*2/3) -- (\xc,\yc+\H/3);

     \draw[WL,white](\xb,\yb+\H*2/3)-- (\xd,\yd+\H);
    \draw(\xb,\yb+\H*2/3)-- (\xd,\yd+\H);
    
     \draw (\xb,\yb+\H/3)--(\xb,\yb+\H*2/3);
     \draw (\xc,\yc+\H*2/3) -- (\xc,\yc+\H);

    \draw(\xb,\yb+\H*2/3)-- (\xd,\yd+\H);

    \draw(\xb,\yb+\H*2/3)-- (\xd,\yd+\H);
    \draw[WL,white](\xd,\yd+\H*2/3) -- (\xb,\yb+\H);
  \draw(\xd,\yd+\H*2/3) -- (\xb,\yb+\H);
    
   %




\end{scope}
}

\newcommand{\yuanzhuHTR}[3]{
\begin{scope}[scale=.7]
  \pgfmathsetmacro{\Rx}{#1}
  \pgfmathsetmacro{\Ry}{#2}
  \pgfmathsetmacro{\H}{#3}

  \cylinderonly{\Rx}{\Ry}{\H}

  \pgfmathsetmacro{\cosA}{0.6427876}
  \pgfmathsetmacro{\cosB}{-0.6427876}
  \pgfmathsetmacro{\cosC}{0.2588190}
  \pgfmathsetmacro{\cosD}{-0.2588190}
  \pgfmathsetmacro{\sinA}{-0.7660444}
  \pgfmathsetmacro{\sinB}{-0.7660444}
  \pgfmathsetmacro{\sinC}{-0.9659258}
  \pgfmathsetmacro{\sinD}{-0.9659258}

  \pgfmathsetmacro{\xa}{\Rx*\cosA}
  \pgfmathsetmacro{\ya}{\Ry*\sinA}
  \pgfmathsetmacro{\xb}{\Rx*\cosB}
  \pgfmathsetmacro{\yb}{\Ry*\sinB}
  \pgfmathsetmacro{\xc}{\Rx*\cosC}
  \pgfmathsetmacro{\yc}{\Ry*\sinC}
  \pgfmathsetmacro{\xd}{\Rx*\cosD}
  \pgfmathsetmacro{\yd}{\Ry*\sinD}

  \draw (\xd,\yd) -- (\xd,\yd+\H/3);
  
  \draw (\xb,\yb) -- (\xb,\yb+\H/3);
  
  \draw (\xc,\yc) -- (\xa,\ya+\H/3)-- (\xa,\ya+\H);
  \draw[WL,white] (\xa,\ya) -- (\xc,\yc+\H/3);
  \draw (\xa,\ya) -- (\xc,\yc+\H/3);

    \draw (\xc,\yc+\H*2/3) -- (\xd,\yd+\H/3);
    \draw[WL,white] (\xd,\yd+\H*2/3) -- (\xc,\yc+\H/3);
    \draw (\xd,\yd+\H*2/3) -- (\xc,\yc+\H/3);

     \draw[WL,white](\xb,\yb+\H*2/3)-- (\xd,\yd+\H);
    \draw(\xb,\yb+\H*2/3)-- (\xd,\yd+\H);
    
     \draw (\xb,\yb+\H/3)--(\xb,\yb+\H*2/3);
     \draw (\xc,\yc+\H*2/3) -- (\xc,\yc+\H);

    \draw(\xb,\yb+\H*2/3)-- (\xd,\yd+\H);

    \draw(\xb,\yb+\H*2/3)-- (\xd,\yd+\H);
    \draw[WL,white](\xd,\yd+\H*2/3) -- (\xb,\yb+\H);
  \draw(\xd,\yd+\H*2/3) -- (\xb,\yb+\H);
    
   %




\end{scope}
}

\newcommand{\yuanzhuX}[3]{
\begin{scope}[scale=.7]
  \pgfmathsetmacro{\Rx}{#1}
  \pgfmathsetmacro{\Ry}{#2}
  \pgfmathsetmacro{\H}{#3}

  \cylinderonly{\Rx}{\Ry}{\H}

  \pgfmathsetmacro{\cosA}{0.6427876}
  \pgfmathsetmacro{\cosB}{-0.6427876}
  \pgfmathsetmacro{\cosC}{0.2588190}
  \pgfmathsetmacro{\cosD}{-0.2588190}
  \pgfmathsetmacro{\sinA}{-0.7660444}
  \pgfmathsetmacro{\sinB}{-0.7660444}
  \pgfmathsetmacro{\sinC}{-0.9659258}
  \pgfmathsetmacro{\sinD}{-0.9659258}

  \pgfmathsetmacro{\xa}{\Rx*\cosA}
  \pgfmathsetmacro{\ya}{\Ry*\sinA}
  \pgfmathsetmacro{\xb}{\Rx*\cosB}
  \pgfmathsetmacro{\yb}{\Ry*\sinB}
  \pgfmathsetmacro{\xc}{\Rx*\cosC}
  \pgfmathsetmacro{\yc}{\Ry*\sinC}
  \pgfmathsetmacro{\xd}{\Rx*\cosD}
  \pgfmathsetmacro{\yd}{\Ry*\sinD}

  \draw (\xa,\ya) -- (\xa,\ya+\H);
  \draw (\xb,\yb) -- (\xb,\yb+\H);

  \draw (\xc,\yc) -- (\xc,\yc+\H);
  \draw (\xd,\yd) -- (\xd,\yd+\H);
  \draw [red,decorate, decoration={snake}](\xc,\yc+\H/3*2)--(\xd,\yd+\H/3*2);
  
  \fill (\xc,\yc+\H/3*2) circle(.07);
  \fill (\xd,\yd+\H/3*2) circle(.07);

\end{scope}
}

\newcommand{\yuanzhuXX}[3]{
\begin{scope}[scale=.7]
  \pgfmathsetmacro{\Rx}{#1}
  \pgfmathsetmacro{\Ry}{#2}
  \pgfmathsetmacro{\H}{#3}

  \cylinderonly{\Rx}{\Ry}{\H}

  \pgfmathsetmacro{\cosA}{0.6427876}
  \pgfmathsetmacro{\cosB}{-0.6427876}
  \pgfmathsetmacro{\cosC}{0.2588190}
  \pgfmathsetmacro{\cosD}{-0.2588190}
  \pgfmathsetmacro{\sinA}{-0.7660444}
  \pgfmathsetmacro{\sinB}{-0.7660444}
  \pgfmathsetmacro{\sinC}{-0.9659258}
  \pgfmathsetmacro{\sinD}{-0.9659258}

  \pgfmathsetmacro{\xa}{\Rx*\cosA}
  \pgfmathsetmacro{\ya}{\Ry*\sinA}
  \pgfmathsetmacro{\xb}{\Rx*\cosB}
  \pgfmathsetmacro{\yb}{\Ry*\sinB}
  \pgfmathsetmacro{\xc}{\Rx*\cosC}
  \pgfmathsetmacro{\yc}{\Ry*\sinC}
  \pgfmathsetmacro{\xd}{\Rx*\cosD}
  \pgfmathsetmacro{\yd}{\Ry*\sinD}
  
  \draw [red,decorate, decoration={snake}](\xa,\ya+\H/3*2)--(\xb,\yb+\H/3*2);

  \draw (\xa,\ya) -- (\xa,\ya+\H);
  \draw (\xb,\yb) -- (\xb,\yb+\H);

  \draw[white,WL] (\xc,\yc) -- (\xc,\yc+\H);
  \draw[white,WL] (\xd,\yd) -- (\xd,\yd+\H);
  \draw (\xc,\yc) -- (\xc,\yc+\H);
  \draw (\xd,\yd) -- (\xd,\yd+\H);
  
  \fill (\xa,\ya+\H/3*2) circle(.07);
  \fill (\xb,\yb+\H/3*2) circle(.07);

\end{scope}
}

\newcommand{\yuanzhuY}[3]{
\begin{scope}[scale=.7]
  \pgfmathsetmacro{\Rx}{#1}
  \pgfmathsetmacro{\Ry}{#2}
  \pgfmathsetmacro{\H}{#3}

  \cylinderonly{\Rx}{\Ry}{\H}

  \pgfmathsetmacro{\cosA}{0.6427876}
  \pgfmathsetmacro{\cosB}{-0.6427876}
  \pgfmathsetmacro{\cosC}{0.2588190}
  \pgfmathsetmacro{\cosD}{-0.2588190}
  \pgfmathsetmacro{\sinA}{-0.7660444}
  \pgfmathsetmacro{\sinB}{-0.7660444}
  \pgfmathsetmacro{\sinC}{-0.9659258}
  \pgfmathsetmacro{\sinD}{-0.9659258}

  \pgfmathsetmacro{\xa}{\Rx*\cosA}
  \pgfmathsetmacro{\ya}{\Ry*\sinA}
  \pgfmathsetmacro{\xb}{\Rx*\cosB}
  \pgfmathsetmacro{\yb}{\Ry*\sinB}
  \pgfmathsetmacro{\xc}{\Rx*\cosC}
  \pgfmathsetmacro{\yc}{\Ry*\sinC}
  \pgfmathsetmacro{\xd}{\Rx*\cosD}
  \pgfmathsetmacro{\yd}{\Ry*\sinD}

  \draw (\xa,\ya) -- (\xa,\ya+\H);
  \draw (\xb,\yb) -- (\xb,\yb+\H);

  
  \draw [red,decorate, decoration={snake}](\xc,\yc+\H/3*2)--(\xb,\yb+\H/3*2);
    \draw[white,WL] (\xd,\yd) -- (\xd,\yd+\H);
  \draw (\xc,\yc) -- (\xc,\yc+\H);
  \draw (\xd,\yd) -- (\xd,\yd+\H);

  \fill (\xb,\yb+\H/3*2) circle(.1);
  \fill (\xc,\yc+\H/3*2) circle(.1);

\end{scope}
}

\newcommand{\yuanzhuZ}[3]{
\begin{scope}[scale=.7]
  \pgfmathsetmacro{\Rx}{#1}
  \pgfmathsetmacro{\Ry}{#2}
  \pgfmathsetmacro{\H}{#3}

  \cylinderonly{\Rx}{\Ry}{\H}

  \pgfmathsetmacro{\cosA}{0.6427876}
  \pgfmathsetmacro{\cosB}{-0.6427876}
  \pgfmathsetmacro{\cosC}{0.2588190}
  \pgfmathsetmacro{\cosD}{-0.2588190}
  \pgfmathsetmacro{\sinA}{-0.7660444}
  \pgfmathsetmacro{\sinB}{-0.7660444}
  \pgfmathsetmacro{\sinC}{-0.9659258}
  \pgfmathsetmacro{\sinD}{-0.9659258}

  \pgfmathsetmacro{\xa}{\Rx*\cosA}
  \pgfmathsetmacro{\ya}{\Ry*\sinA}
  \pgfmathsetmacro{\xb}{\Rx*\cosB}
  \pgfmathsetmacro{\yb}{\Ry*\sinB}
  \pgfmathsetmacro{\xc}{\Rx*\cosC}
  \pgfmathsetmacro{\yc}{\Ry*\sinC}
  \pgfmathsetmacro{\xd}{\Rx*\cosD}
  \pgfmathsetmacro{\yd}{\Ry*\sinD}

  \draw (\xa,\ya) -- (\xa,\ya+\H);
  \draw (\xb,\yb) -- (\xb,\yb+\H);

  \draw (\xc,\yc) -- (\xc,\yc+\H);
  \draw (\xd,\yd) -- (\xd,\yd+\H);

    \draw[red] (\xb,\yb+\H/3*2)[bend left=90]to(\xb/2+\xd/2,\yb/2+\yd/2+\H/3*2)[bend right=90]to(\xd,\yd+\H/3*2);

  \fill (\xb,\yb+\H/3*2) circle(.07);
  \fill (\xd,\yd+\H/3*2) circle(.07);

\end{scope}
}

\newcommand{\yuanzhuZZZ}[3]{
\begin{scope}[scale=.7]
  \pgfmathsetmacro{\Rx}{#1}
  \pgfmathsetmacro{\Ry}{#2}
  \pgfmathsetmacro{\H}{#3}

  \cylinderonly{\Rx}{\Ry}{\H}

  \pgfmathsetmacro{\cosA}{0.6427876}
  \pgfmathsetmacro{\cosB}{-0.6427876}
  \pgfmathsetmacro{\cosC}{0.2588190}
  \pgfmathsetmacro{\cosD}{-0.2588190}
  \pgfmathsetmacro{\sinA}{-0.7660444}
  \pgfmathsetmacro{\sinB}{-0.7660444}
  \pgfmathsetmacro{\sinC}{-0.9659258}
  \pgfmathsetmacro{\sinD}{-0.9659258}

  \pgfmathsetmacro{\xa}{\Rx*\cosA}
  \pgfmathsetmacro{\ya}{\Ry*\sinA}
  \pgfmathsetmacro{\xb}{\Rx*\cosB}
  \pgfmathsetmacro{\yb}{\Ry*\sinB}
  \pgfmathsetmacro{\xc}{\Rx*\cosC}
  \pgfmathsetmacro{\yc}{\Ry*\sinC}
  \pgfmathsetmacro{\xd}{\Rx*\cosD}
  \pgfmathsetmacro{\yd}{\Ry*\sinD}

  \draw (\xa,\ya) -- (\xa,\ya+\H);
  \draw (\xb,\yb) -- (\xb,\yb+\H);

  \draw (\xc,\yc) -- (\xc,\yc+\H);
  \draw (\xd,\yd) -- (\xd,\yd+\H);

    \draw[red] (\xb,\yb+\H/3*2)[bend left=90]to(\xb/2+\xd/2,\yb/2+\yd/2+\H/3*2)[bend right=90]to(\xd,\yd+\H/3*2);

  \fill (\xb,\yb+\H/3*2) circle(.1);
  \fill (\xd,\yd+\H/3*2) circle(.1);

\end{scope}
}

\newcommand{\yuanzhuZZ}[3]{
\begin{scope}[scale=.7]
  \pgfmathsetmacro{\Rx}{#1}
  \pgfmathsetmacro{\Ry}{#2}
  \pgfmathsetmacro{\H}{#3}

  \cylinderonly{\Rx}{\Ry}{\H}

  \pgfmathsetmacro{\cosA}{0.6427876}
  \pgfmathsetmacro{\cosB}{-0.6427876}
  \pgfmathsetmacro{\cosC}{0.2588190}
  \pgfmathsetmacro{\cosD}{-0.2588190}
  \pgfmathsetmacro{\sinA}{-0.7660444}
  \pgfmathsetmacro{\sinB}{-0.7660444}
  \pgfmathsetmacro{\sinC}{-0.9659258}
  \pgfmathsetmacro{\sinD}{-0.9659258}

  \pgfmathsetmacro{\xa}{\Rx*\cosA}
  \pgfmathsetmacro{\ya}{\Ry*\sinA}
  \pgfmathsetmacro{\xb}{\Rx*\cosB}
  \pgfmathsetmacro{\yb}{\Ry*\sinB}
  \pgfmathsetmacro{\xc}{\Rx*\cosC}
  \pgfmathsetmacro{\yc}{\Ry*\sinC}
  \pgfmathsetmacro{\xd}{\Rx*\cosD}
  \pgfmathsetmacro{\yd}{\Ry*\sinD}

  \draw (\xa,\ya) -- (\xa,\ya+\H);
  \draw (\xb,\yb) -- (\xb,\yb+\H);

  \draw (\xc,\yc) -- (\xc,\yc+\H);
  \draw (\xd,\yd) -- (\xd,\yd+\H);

    \draw[red] (\xa,\ya+\H/3*2)[bend left=90]to(\xa/2+\xc/2,\yc/2+\yc/2+\H/3*2)[bend right=90]to(\xc,\yc+\H/3*2);

  \fill (\xa,\ya+\H/3*2) circle(.07);
  \fill (\xc,\yc+\H/3*2) circle(.07);

\end{scope}
}

\newcommand{\yuanzhuZL}[3]{
\begin{scope}[scale=.7]
  \pgfmathsetmacro{\Rx}{#1}
  \pgfmathsetmacro{\Ry}{#2}
  \pgfmathsetmacro{\H}{#3}

  \cylinderonly{\Rx}{\Ry}{\H}

  \pgfmathsetmacro{\cosA}{0.6427876}
  \pgfmathsetmacro{\cosB}{-0.6427876}
  \pgfmathsetmacro{\cosC}{0.2588190}
  \pgfmathsetmacro{\cosD}{-0.2588190}
  \pgfmathsetmacro{\sinA}{-0.7660444}
  \pgfmathsetmacro{\sinB}{-0.7660444}
  \pgfmathsetmacro{\sinC}{-0.9659258}
  \pgfmathsetmacro{\sinD}{-0.9659258}

  \pgfmathsetmacro{\xa}{\Rx*\cosA}
  \pgfmathsetmacro{\ya}{\Ry*\sinA}
  \pgfmathsetmacro{\xb}{\Rx*\cosB}
  \pgfmathsetmacro{\yb}{\Ry*\sinB}
  \pgfmathsetmacro{\xc}{\Rx*\cosC}
  \pgfmathsetmacro{\yc}{\Ry*\sinC}
  \pgfmathsetmacro{\xd}{\Rx*\cosD}
  \pgfmathsetmacro{\yd}{\Ry*\sinD}

  \draw (\xa,\ya) -- (\xa,\ya+\H);
  \draw (\xb,\yb) -- (\xb,\yb+\H);

  \draw (\xc,\yc) -- (\xc,\yc+\H);
  \draw (\xd,\yd) -- (\xd,\yd+\H);
  \draw [red,decorate, decoration={snake}](\xa,\ya+\H/3*2)--(\xb,\yb+\H/3*2);
  
  \fill (\xb,\yb+\H/3*2) circle(.1);
  \fill (\xa,\ya+\H/3*2) circle(.1);

\end{scope}
}
\newcommand{\yuanzhuALEFT}[3]{
\begin{scope}[scale=.7]
  \pgfmathsetmacro{\Rx}{#1}
  \pgfmathsetmacro{\Ry}{#2}
  \pgfmathsetmacro{\H}{#3}

  \cylinderonly{\Rx}{\Ry}{\H}

  \pgfmathsetmacro{\cosA}{0.6427876}
  \pgfmathsetmacro{\cosB}{-0.6427876}
  \pgfmathsetmacro{\cosC}{0.2588190}
  \pgfmathsetmacro{\cosD}{-0.2588190}
  \pgfmathsetmacro{\sinA}{-0.7660444}
  \pgfmathsetmacro{\sinB}{-0.7660444}
  \pgfmathsetmacro{\sinC}{-0.9659258}
  \pgfmathsetmacro{\sinD}{-0.9659258}

  \pgfmathsetmacro{\xa}{\Rx*\cosA}
  \pgfmathsetmacro{\ya}{\Ry*\sinA}
  \pgfmathsetmacro{\xb}{\Rx*\cosB}
  \pgfmathsetmacro{\yb}{\Ry*\sinB}
  \pgfmathsetmacro{\xc}{\Rx*\cosC}
  \pgfmathsetmacro{\yc}{\Ry*\sinC}
  \pgfmathsetmacro{\xd}{\Rx*\cosD}
  \pgfmathsetmacro{\yd}{\Ry*\sinD}

  \draw (\xa,\ya) -- (\xa,\ya+\H);
  \draw (\xb,\yb) -- (\xb,\yb+\H);

  \draw (\xc,\yc) -- (\xc,\yc+\H);
  \draw (\xd,\yd) -- (\xd,\yd+\H);
   \filldraw[white,dashed] (\xb-1/5,\yb+\H/3*2)--(\xd+1/5,\yd+\H/3*2)--(\xd+1/5,\yd+\H/3*1)--(\xb-1/5,\yb+\H/3*1);
  \draw[dashed] (\xb-1/5,\yb+\H/3*2)--(\xd+1/5,\yd+\H/3*2)--(\xd+1/5,\yd+\H/3*1)--(\xb-1/5,\yb+\H/3*1)--(\xb-1/5,\yb+\H/3*2);
  
  \node at (\xb/2+\xd/2,\yb+\H/2){$A$};

\end{scope}
}

\newcommand{\yuanzhuARIGHT}[3]{
\begin{scope}[scale=.7]
  \pgfmathsetmacro{\Rx}{#1}
  \pgfmathsetmacro{\Ry}{#2}
  \pgfmathsetmacro{\H}{#3}

  \cylinderonly{\Rx}{\Ry}{\H}

  \pgfmathsetmacro{\cosA}{0.6427876}
  \pgfmathsetmacro{\cosB}{-0.6427876}
  \pgfmathsetmacro{\cosC}{0.2588190}
  \pgfmathsetmacro{\cosD}{-0.2588190}
  \pgfmathsetmacro{\sinA}{-0.7660444}
  \pgfmathsetmacro{\sinB}{-0.7660444}
  \pgfmathsetmacro{\sinC}{-0.9659258}
  \pgfmathsetmacro{\sinD}{-0.9659258}

  \pgfmathsetmacro{\xa}{\Rx*\cosA}
  \pgfmathsetmacro{\ya}{\Ry*\sinA}
  \pgfmathsetmacro{\xb}{\Rx*\cosB}
  \pgfmathsetmacro{\yb}{\Ry*\sinB}
  \pgfmathsetmacro{\xc}{\Rx*\cosC}
  \pgfmathsetmacro{\yc}{\Ry*\sinC}
  \pgfmathsetmacro{\xd}{\Rx*\cosD}
  \pgfmathsetmacro{\yd}{\Ry*\sinD}

  \draw (\xa,\ya) -- (\xa,\ya+\H);
  \draw (\xb,\yb) -- (\xb,\yb+\H);

 \draw (\xc,\yc) -- (\xc,\yc+\H);
  \draw (\xd,\yd) -- (\xd,\yd+\H);
   \filldraw[white,dashed] (\xc-1/5,\yc+\H/3*2)--(\xa+1/5,\ya+\H/3*2)--(\xa+1/5,\ya+\H/3*1)--(\xc-1/5,\yc+\H/3*1);
  \draw[dashed]  (\xc-1/5,\yc+\H/3*2)--(\xa+1/5,\ya+\H/3*2)--(\xa+1/5,\ya+\H/3*1)--(\xc-1/5,\yc+\H/3*1)--(\xc-1/5,\yc+\H/3*2);
  \node at (\xa/2+\xc/2,\ya+\H/2){$A$};
\end{scope}
}

\newcommand{\yuanzhuCAPPLUS}[3]{
\begin{scope}[scale=.7]
  \pgfmathsetmacro{\Rx}{#1}
  \pgfmathsetmacro{\Ry}{#2}
  \pgfmathsetmacro{\H}{#3}

  \caponly{\Rx}{\Ry}{\H}

  \pgfmathsetmacro{\cosA}{0.6427876}
  \pgfmathsetmacro{\cosB}{-0.6427876}
  \pgfmathsetmacro{\cosC}{0.2588190}
  \pgfmathsetmacro{\cosD}{-0.2588190}
  \pgfmathsetmacro{\sinA}{-0.7660444}
  \pgfmathsetmacro{\sinB}{-0.7660444}
  \pgfmathsetmacro{\sinC}{-0.9659258}
  \pgfmathsetmacro{\sinD}{-0.9659258}

  \pgfmathsetmacro{\xa}{\Rx*\cosA}
  \pgfmathsetmacro{\ya}{\Ry*\sinA}
  \pgfmathsetmacro{\xb}{\Rx*\cosB}
  \pgfmathsetmacro{\yb}{\Ry*\sinB}
  \pgfmathsetmacro{\xc}{\Rx*\cosC}
  \pgfmathsetmacro{\yc}{\Ry*\sinC}
  \pgfmathsetmacro{\xd}{\Rx*\cosD}
  \pgfmathsetmacro{\yd}{\Ry*\sinD}

  \draw (\xa,\ya) -- (\xa,\ya+\H/3)arc[start angle=180, end angle=360, radius=(\xb-\xa)/2]--(\xb,\yb);
  \draw (\xc,\yc) -- (\xc,\yc+\H/3)arc[start angle=0, end angle=180, radius=(\xc-\xd)/2]--(\xd,\yd);

  

\end{scope}
}

\newcommand{\yuanzhuCAPPLUSBULK}[3]{
\begin{scope}[scale=.7]
  \pgfmathsetmacro{\Rx}{#1}
  \pgfmathsetmacro{\Ry}{#2}
  \pgfmathsetmacro{\H}{#3}

  \caponlybulk{\Rx}{\Ry}{\H}

  \pgfmathsetmacro{\cosA}{0.6427876}
  \pgfmathsetmacro{\cosB}{-0.6427876}
  \pgfmathsetmacro{\cosC}{0.2588190}
  \pgfmathsetmacro{\cosD}{-0.2588190}
  \pgfmathsetmacro{\sinA}{-0.7660444}
  \pgfmathsetmacro{\sinB}{-0.7660444}
  \pgfmathsetmacro{\sinC}{-0.9659258}
  \pgfmathsetmacro{\sinD}{-0.9659258}

  \pgfmathsetmacro{\xa}{\Rx*\cosA}
  \pgfmathsetmacro{\ya}{\Ry*\sinA}
  \pgfmathsetmacro{\xb}{\Rx*\cosB}
  \pgfmathsetmacro{\yb}{\Ry*\sinB}
  \pgfmathsetmacro{\xc}{\Rx*\cosC}
  \pgfmathsetmacro{\yc}{\Ry*\sinC}
  \pgfmathsetmacro{\xd}{\Rx*\cosD}
  \pgfmathsetmacro{\yd}{\Ry*\sinD}

  \draw (\xa,\ya) -- (\xa,\ya+\H/3)arc[start angle=180, end angle=360, radius=(\xb-\xa)/2]--(\xb,\yb);
  \draw (\xc,\yc) -- (\xc,\yc+\H/3)arc[start angle=0, end angle=180, radius=(\xc-\xd)/2]--(\xd,\yd);

  

\end{scope}
}

\newcommand{\yuanzhuCAPPLUSCOLORTIKZ}[3]{
\tikz[baseline=(current bounding box.center), scale=0.7, line width=0.3mm] {
\begin{scope}[scale=.7]
  \pgfmathsetmacro{\Rx}{#1}
  \pgfmathsetmacro{\Ry}{#2}
  \pgfmathsetmacro{\H}{#3}

  \caponly{\Rx}{\Ry}{\H}

  \pgfmathsetmacro{\cosA}{0.6427876}
  \pgfmathsetmacro{\cosB}{-0.6427876}
  \pgfmathsetmacro{\cosC}{0.2588190}
  \pgfmathsetmacro{\cosD}{-0.2588190}
  \pgfmathsetmacro{\sinA}{-0.7660444}
  \pgfmathsetmacro{\sinB}{-0.7660444}
  \pgfmathsetmacro{\sinC}{-0.9659258}
  \pgfmathsetmacro{\sinD}{-0.9659258}

  \pgfmathsetmacro{\xa}{\Rx*\cosA}
  \pgfmathsetmacro{\ya}{\Ry*\sinA}
  \pgfmathsetmacro{\xb}{\Rx*\cosB}
  \pgfmathsetmacro{\yb}{\Ry*\sinB}
  \pgfmathsetmacro{\xc}{\Rx*\cosC}
  \pgfmathsetmacro{\yc}{\Ry*\sinC}
  \pgfmathsetmacro{\xd}{\Rx*\cosD}
  \pgfmathsetmacro{\yd}{\Ry*\sinD}

  \draw [magenta](\xa,\ya) -- (\xa,\ya+\H/3)arc[start angle=180, end angle=360, radius=(\xb-\xa)/2]--(\xb,\yb);
  \draw (\xc,\yc) -- (\xc,\yc+\H/3)arc[start angle=0, end angle=180, radius=(\xc-\xd)/2]--(\xd,\yd);

  

\end{scope}
}
}

\newcommand{\yuanzhuCAPPLUSCOLOR}[3]{
\begin{scope}[scale=.7]
  \pgfmathsetmacro{\Rx}{#1}
  \pgfmathsetmacro{\Ry}{#2}
  \pgfmathsetmacro{\H}{#3}

  \caponly{\Rx}{\Ry}{\H}

  \pgfmathsetmacro{\cosA}{0.6427876}
  \pgfmathsetmacro{\cosB}{-0.6427876}
  \pgfmathsetmacro{\cosC}{0.2588190}
  \pgfmathsetmacro{\cosD}{-0.2588190}
  \pgfmathsetmacro{\sinA}{-0.7660444}
  \pgfmathsetmacro{\sinB}{-0.7660444}
  \pgfmathsetmacro{\sinC}{-0.9659258}
  \pgfmathsetmacro{\sinD}{-0.9659258}

  \pgfmathsetmacro{\xa}{\Rx*\cosA}
  \pgfmathsetmacro{\ya}{\Ry*\sinA}
  \pgfmathsetmacro{\xb}{\Rx*\cosB}
  \pgfmathsetmacro{\yb}{\Ry*\sinB}
  \pgfmathsetmacro{\xc}{\Rx*\cosC}
  \pgfmathsetmacro{\yc}{\Ry*\sinC}
  \pgfmathsetmacro{\xd}{\Rx*\cosD}
  \pgfmathsetmacro{\yd}{\Ry*\sinD}

  \draw [magenta](\xa,\ya) -- (\xa,\ya+\H/3)arc[start angle=180, end angle=360, radius=(\xb-\xa)/2]--(\xb,\yb);
  \draw (\xc,\yc) -- (\xc,\yc+\H/3)arc[start angle=0, end angle=180, radius=(\xc-\xd)/2]--(\xd,\yd);

  

\end{scope}
}

\newcommand{\yuanzhuCAPMINUS}[3]{
\begin{scope}[scale=.7]
  \pgfmathsetmacro{\Rx}{#1}
  \pgfmathsetmacro{\Ry}{#2}
  \pgfmathsetmacro{\H}{#3}

  \caponly{\Rx}{\Ry}{\H}

  \pgfmathsetmacro{\cosA}{0.6427876}
  \pgfmathsetmacro{\cosB}{-0.6427876}
  \pgfmathsetmacro{\cosC}{0.2588190}
  \pgfmathsetmacro{\cosD}{-0.2588190}
  \pgfmathsetmacro{\sinA}{-0.7660444}
  \pgfmathsetmacro{\sinB}{-0.7660444}
  \pgfmathsetmacro{\sinC}{-0.9659258}
  \pgfmathsetmacro{\sinD}{-0.9659258}

  \pgfmathsetmacro{\xa}{\Rx*\cosA}
  \pgfmathsetmacro{\ya}{\Ry*\sinA}
  \pgfmathsetmacro{\xb}{\Rx*\cosB}
  \pgfmathsetmacro{\yb}{\Ry*\sinB}
  \pgfmathsetmacro{\xc}{\Rx*\cosC}
  \pgfmathsetmacro{\yc}{\Ry*\sinC}
  \pgfmathsetmacro{\xd}{\Rx*\cosD}
  \pgfmathsetmacro{\yd}{\Ry*\sinD}

  \draw (\xa,\ya) -- (\xa,\ya+\H/3)arc[start angle=180, end angle=360, radius=(\xb-\xa)/2]--(\xb,\yb);
  \draw (\xc,\yc) -- (\xc,\yc+\H/3)arc[start angle=0, end angle=180, radius=(\xc-\xd)/2]--(\xd,\yd);

  \draw [red,decorate, decoration={snake}](\xb,\yb+\H/5)--(\xd,\yd+\H/5);
  
  \fill (\xb,\yb+\H/5) circle(.05);
  \fill (\xd,\yd+\H/5) circle(.05);

\end{scope}
}

\newcommand{\yuanzhuCAPMINUSCOLOR}[3]{
\begin{scope}[scale=.7]
  \pgfmathsetmacro{\Rx}{#1}
  \pgfmathsetmacro{\Ry}{#2}
  \pgfmathsetmacro{\H}{#3}

  \caponly{\Rx}{\Ry}{\H}

  \pgfmathsetmacro{\cosA}{0.6427876}
  \pgfmathsetmacro{\cosB}{-0.6427876}
  \pgfmathsetmacro{\cosC}{0.2588190}
  \pgfmathsetmacro{\cosD}{-0.2588190}
  \pgfmathsetmacro{\sinA}{-0.7660444}
  \pgfmathsetmacro{\sinB}{-0.7660444}
  \pgfmathsetmacro{\sinC}{-0.9659258}
  \pgfmathsetmacro{\sinD}{-0.9659258}

  \pgfmathsetmacro{\xa}{\Rx*\cosA}
  \pgfmathsetmacro{\ya}{\Ry*\sinA}
  \pgfmathsetmacro{\xb}{\Rx*\cosB}
  \pgfmathsetmacro{\yb}{\Ry*\sinB}
  \pgfmathsetmacro{\xc}{\Rx*\cosC}
  \pgfmathsetmacro{\yc}{\Ry*\sinC}
  \pgfmathsetmacro{\xd}{\Rx*\cosD}
  \pgfmathsetmacro{\yd}{\Ry*\sinD}

  \draw [magenta](\xa,\ya) -- (\xa,\ya+\H/3)arc[start angle=180, end angle=360, radius=(\xb-\xa)/2]--(\xb,\yb);
  \draw (\xc,\yc) -- (\xc,\yc+\H/3)arc[start angle=0, end angle=180, radius=(\xc-\xd)/2]--(\xd,\yd);

  \draw [red,decorate, decoration={snake}](\xb,\yb+\H/5)--(\xd,\yd+\H/5);
  
  \fill (\xb,\yb+\H/5) circle(.05);
  \fill (\xd,\yd+\H/5) circle(.05);

\end{scope}
}

\newcommand{\yuanzhuCAP}[3]{
\begin{scope}[scale=.7]
  \pgfmathsetmacro{\Rx}{#1}
  \pgfmathsetmacro{\Ry}{#2}
  \pgfmathsetmacro{\H}{#3}

  \caponly{\Rx}{\Ry}{\H}

  \pgfmathsetmacro{\cosA}{0.6427876}
  \pgfmathsetmacro{\cosB}{-0.6427876}
  \pgfmathsetmacro{\cosC}{0.2588190}
  \pgfmathsetmacro{\cosD}{-0.2588190}
  \pgfmathsetmacro{\sinA}{-0.7660444}
  \pgfmathsetmacro{\sinB}{-0.7660444}
  \pgfmathsetmacro{\sinC}{-0.9659258}
  \pgfmathsetmacro{\sinD}{-0.9659258}

  \pgfmathsetmacro{\xa}{\Rx*\cosA}
  \pgfmathsetmacro{\ya}{\Ry*\sinA}
  \pgfmathsetmacro{\xb}{\Rx*\cosB}
  \pgfmathsetmacro{\yb}{\Ry*\sinB}
  \pgfmathsetmacro{\xc}{\Rx*\cosC}
  \pgfmathsetmacro{\yc}{\Ry*\sinC}
  \pgfmathsetmacro{\xd}{\Rx*\cosD}
  \pgfmathsetmacro{\yd}{\Ry*\sinD}


  

\end{scope}
}

\newcommand{\yuanzhuCAPZERO}[3]{
\begin{scope}[scale=.7]
  \pgfmathsetmacro{\Rx}{#1}
  \pgfmathsetmacro{\Ry}{#2}
  \pgfmathsetmacro{\H}{#3}

  \caponly{\Rx}{\Ry}{\H}

  \pgfmathsetmacro{\cosA}{0.6427876}
  \pgfmathsetmacro{\cosB}{-0.6427876}
  \pgfmathsetmacro{\cosC}{0.2588190}
  \pgfmathsetmacro{\cosD}{-0.2588190}
  \pgfmathsetmacro{\sinA}{-0.7660444}
  \pgfmathsetmacro{\sinB}{-0.7660444}
  \pgfmathsetmacro{\sinC}{-0.9659258}
  \pgfmathsetmacro{\sinD}{-0.9659258}

  \pgfmathsetmacro{\xa}{\Rx*\cosA}
  \pgfmathsetmacro{\ya}{\Ry*\sinA}
  \pgfmathsetmacro{\xb}{\Rx*\cosB}
  \pgfmathsetmacro{\yb}{\Ry*\sinB}
  \pgfmathsetmacro{\xc}{\Rx*\cosC}
  \pgfmathsetmacro{\yc}{\Ry*\sinC}
  \pgfmathsetmacro{\xd}{\Rx*\cosD}
  \pgfmathsetmacro{\yd}{\Ry*\sinD}

  \draw (\xa,\ya) -- (\xa,\ya+\H/5)arc[start angle=180, end angle=360, radius=(\xc-\xa)/2]--(\xc,\yc);
  \draw (\xd,\yd) -- (\xd,\yd+\H/4)arc[start angle=0, end angle=180, radius=(\xd-\xb)/2]--(\xb,\yb);

  

\end{scope}
}
\newcommand{\yuanzhuCAPZEROBULK}[3]{
\begin{scope}[scale=.7]
  \pgfmathsetmacro{\Rx}{#1}
  \pgfmathsetmacro{\Ry}{#2}
  \pgfmathsetmacro{\H}{#3}

  \caponlybulk{\Rx}{\Ry}{\H}

  \pgfmathsetmacro{\cosA}{0.6427876}
  \pgfmathsetmacro{\cosB}{-0.6427876}
  \pgfmathsetmacro{\cosC}{0.2588190}
  \pgfmathsetmacro{\cosD}{-0.2588190}
  \pgfmathsetmacro{\sinA}{-0.7660444}
  \pgfmathsetmacro{\sinB}{-0.7660444}
  \pgfmathsetmacro{\sinC}{-0.9659258}
  \pgfmathsetmacro{\sinD}{-0.9659258}

  \pgfmathsetmacro{\xa}{\Rx*\cosA}
  \pgfmathsetmacro{\ya}{\Ry*\sinA}
  \pgfmathsetmacro{\xb}{\Rx*\cosB}
  \pgfmathsetmacro{\yb}{\Ry*\sinB}
  \pgfmathsetmacro{\xc}{\Rx*\cosC}
  \pgfmathsetmacro{\yc}{\Ry*\sinC}
  \pgfmathsetmacro{\xd}{\Rx*\cosD}
  \pgfmathsetmacro{\yd}{\Ry*\sinD}

  \draw (\xa,\ya) -- (\xa,\ya+\H/5)arc[start angle=180, end angle=360, radius=(\xc-\xa)/2]--(\xc,\yc);
  \draw (\xd,\yd) -- (\xd,\yd+\H/4)arc[start angle=0, end angle=180, radius=(\xd-\xb)/2]--(\xb,\yb);

  

\end{scope}
}

\newcommand{\yuanzhuCAPZEROTWO}[3]{
\begin{scope}[scale=.7]
  \pgfmathsetmacro{\Rx}{#1}
  \pgfmathsetmacro{\Ry}{#2}
  \pgfmathsetmacro{\H}{#3}

  \caponly{\Rx}{\Ry}{\H}

  \pgfmathsetmacro{\cosA}{0.6427876}
  \pgfmathsetmacro{\cosB}{-0.6427876}
  \pgfmathsetmacro{\cosC}{0.2588190}
  \pgfmathsetmacro{\cosD}{-0.2588190}
  \pgfmathsetmacro{\sinA}{-0.7660444}
  \pgfmathsetmacro{\sinB}{-0.7660444}
  \pgfmathsetmacro{\sinC}{-0.9659258}
  \pgfmathsetmacro{\sinD}{-0.9659258}

  \pgfmathsetmacro{\xa}{\Rx*\cosA}
  \pgfmathsetmacro{\ya}{\Ry*\sinA}
  \pgfmathsetmacro{\xb}{\Rx*\cosB}
  \pgfmathsetmacro{\yb}{\Ry*\sinB}
  \pgfmathsetmacro{\xc}{\Rx*\cosC}
  \pgfmathsetmacro{\yc}{\Ry*\sinC}
  \pgfmathsetmacro{\xd}{\Rx*\cosD}
  \pgfmathsetmacro{\yd}{\Ry*\sinD}

  \draw (\xc,\yc) -- (\xc,\yc+\H/5)arc[start angle=0, end angle=180, radius=(\xc-\xd)/2]--(\xd,\yd);

  

\end{scope}
}

\newcommand{\yuanzhuCAPZEROTWOBULK}[3]{
\begin{scope}[scale=.7]
  \pgfmathsetmacro{\Rx}{#1}
  \pgfmathsetmacro{\Ry}{#2}
  \pgfmathsetmacro{\H}{#3}

  \caponlybulk{\Rx}{\Ry}{\H}

  \pgfmathsetmacro{\cosA}{0.6427876}
  \pgfmathsetmacro{\cosB}{-0.6427876}
  \pgfmathsetmacro{\cosC}{0.2588190}
  \pgfmathsetmacro{\cosD}{-0.2588190}
  \pgfmathsetmacro{\sinA}{-0.7660444}
  \pgfmathsetmacro{\sinB}{-0.7660444}
  \pgfmathsetmacro{\sinC}{-0.9659258}
  \pgfmathsetmacro{\sinD}{-0.9659258}

  \pgfmathsetmacro{\xa}{\Rx*\cosA}
  \pgfmathsetmacro{\ya}{\Ry*\sinA}
  \pgfmathsetmacro{\xb}{\Rx*\cosB}
  \pgfmathsetmacro{\yb}{\Ry*\sinB}
  \pgfmathsetmacro{\xc}{\Rx*\cosC}
  \pgfmathsetmacro{\yc}{\Ry*\sinC}
  \pgfmathsetmacro{\xd}{\Rx*\cosD}
  \pgfmathsetmacro{\yd}{\Ry*\sinD}

  \draw (\xc,\yc) -- (\xc,\yc+\H/5)arc[start angle=0, end angle=180, radius=(\xc-\xd)/2]--(\xd,\yd);

  

\end{scope}
}

\newcommand{\yuanzhuCAPZEROBULKARB}[3]{
\begin{scope}[scale=.7]
  \pgfmathsetmacro{\Rx}{#1}
  \pgfmathsetmacro{\Ry}{#2}
  \pgfmathsetmacro{\H}{#3}

  \caponlybulk{\Rx}{\Ry}{\H}

  \pgfmathsetmacro{\cosA}{0.6427876}
  \pgfmathsetmacro{\cosB}{-0.6427876}
  \pgfmathsetmacro{\cosC}{0.2588190}
  \pgfmathsetmacro{\cosD}{-0.2588190}
  \pgfmathsetmacro{\sinA}{-0.7660444}
  \pgfmathsetmacro{\sinB}{-0.7660444}
  \pgfmathsetmacro{\sinC}{-0.9659258}
  \pgfmathsetmacro{\sinD}{-0.9659258}

  \pgfmathsetmacro{\xa}{\Rx*\cosA}
  \pgfmathsetmacro{\ya}{\Ry*\sinA}
  \pgfmathsetmacro{\xb}{\Rx*\cosB}
  \pgfmathsetmacro{\yb}{\Ry*\sinB}
  \pgfmathsetmacro{\xc}{\Rx*\cosC}
  \pgfmathsetmacro{\yc}{\Ry*\sinC}
  \pgfmathsetmacro{\xd}{\Rx*\cosD}
  \pgfmathsetmacro{\yd}{\Ry*\sinD}

  \draw (\xa,\ya) -- (\xa,\ya+\H/5)arc[start angle=180, end angle=360, radius=(\xd-\xa)/2]--(\xd,\yd);
  \draw (\xb,\yb) -- (\xb,\yb+\H/4)arc[start angle=0, end angle=-180, radius=(\xb-\xc)/2]--(\xc,\yc);
  \node at (\xd/2+\xc/2+.2,1){\small{$v$}};

  

\end{scope}
}

\newcommand{\yuanzhuCAPONE}[3]{
\begin{scope}[scale=.7]
  \pgfmathsetmacro{\Rx}{#1}
  \pgfmathsetmacro{\Ry}{#2}
  \pgfmathsetmacro{\H}{#3}

  \caponly{\Rx}{\Ry}{\H}

  \pgfmathsetmacro{\cosA}{0.6427876}
  \pgfmathsetmacro{\cosB}{-0.6427876}
  \pgfmathsetmacro{\cosC}{0.2588190}
  \pgfmathsetmacro{\cosD}{-0.2588190}
  \pgfmathsetmacro{\sinA}{-0.7660444}
  \pgfmathsetmacro{\sinB}{-0.7660444}
  \pgfmathsetmacro{\sinC}{-0.9659258}
  \pgfmathsetmacro{\sinD}{-0.9659258}

  \pgfmathsetmacro{\xa}{\Rx*\cosA}
  \pgfmathsetmacro{\ya}{\Ry*\sinA}
  \pgfmathsetmacro{\xb}{\Rx*\cosB}
  \pgfmathsetmacro{\yb}{\Ry*\sinB}
  \pgfmathsetmacro{\xc}{\Rx*\cosC}
  \pgfmathsetmacro{\yc}{\Ry*\sinC}
  \pgfmathsetmacro{\xd}{\Rx*\cosD}
  \pgfmathsetmacro{\yd}{\Ry*\sinD}

  \draw (\xa,\ya) -- (\xa,\ya+\H/5)arc[start angle=180, end angle=360, radius=(\xc-\xa)/2]--(\xc,\yc);
  \draw (\xd,\yd) -- (\xd,\yd+\H/4)arc[start angle=0, end angle=180, radius=(\xd-\xb)/2]--(\xb,\yb);

  \draw [red,decorate, decoration={snake}](\xc,\yc+\H/6)--(\xd,\yd+\H/6);
  
  \fill (\xd,\yd+\H/6) circle(.07);
  \fill (\xc,\yc+\H/6) circle(.07);

\end{scope}
}

\newcommand{\yuanzhuCAPONEBULK}[3]{
\begin{scope}[scale=.7]
  \pgfmathsetmacro{\Rx}{#1}
  \pgfmathsetmacro{\Ry}{#2}
  \pgfmathsetmacro{\H}{#3}

  \caponlybulk{\Rx}{\Ry}{\H}

  \pgfmathsetmacro{\cosA}{0.6427876}
  \pgfmathsetmacro{\cosB}{-0.6427876}
  \pgfmathsetmacro{\cosC}{0.2588190}
  \pgfmathsetmacro{\cosD}{-0.2588190}
  \pgfmathsetmacro{\sinA}{-0.7660444}
  \pgfmathsetmacro{\sinB}{-0.7660444}
  \pgfmathsetmacro{\sinC}{-0.9659258}
  \pgfmathsetmacro{\sinD}{-0.9659258}

  \pgfmathsetmacro{\xa}{\Rx*\cosA}
  \pgfmathsetmacro{\ya}{\Ry*\sinA}
  \pgfmathsetmacro{\xb}{\Rx*\cosB}
  \pgfmathsetmacro{\yb}{\Ry*\sinB}
  \pgfmathsetmacro{\xc}{\Rx*\cosC}
  \pgfmathsetmacro{\yc}{\Ry*\sinC}
  \pgfmathsetmacro{\xd}{\Rx*\cosD}
  \pgfmathsetmacro{\yd}{\Ry*\sinD}

  \draw (\xa,\ya) -- (\xa,\ya+\H/5)arc[start angle=180, end angle=360, radius=(\xc-\xa)/2]--(\xc,\yc);
  \draw (\xd,\yd) -- (\xd,\yd+\H/4)arc[start angle=0, end angle=180, radius=(\xd-\xb)/2]--(\xb,\yb);

  \draw [red,decorate, decoration={snake}](\xc,\yc+\H/6)--(\xd,\yd+\H/6);
  
  \fill (\xd,\yd+\H/6) circle(.07);
  \fill (\xc,\yc+\H/6) circle(.07);

\end{scope}
}

\newcommand{\yuanzhuCUPPLUS}[3]{
\begin{scope}[scale=.7]
  \pgfmathsetmacro{\Rx}{#1}
  \pgfmathsetmacro{\Ry}{#2}
  \pgfmathsetmacro{\H}{#3}

  \cuponly{\Rx}{\Ry}{\H}

  \pgfmathsetmacro{\cosA}{0.6427876}
  \pgfmathsetmacro{\cosB}{-0.6427876}
  \pgfmathsetmacro{\cosC}{0.2588190}
  \pgfmathsetmacro{\cosD}{-0.2588190}
  \pgfmathsetmacro{\sinA}{-0.7660444}
  \pgfmathsetmacro{\sinB}{-0.7660444}
  \pgfmathsetmacro{\sinC}{-0.9659258}
  \pgfmathsetmacro{\sinD}{-0.9659258}

  \pgfmathsetmacro{\xa}{\Rx*\cosA}
  \pgfmathsetmacro{\ya}{\Ry*\sinA}
  \pgfmathsetmacro{\xb}{\Rx*\cosB}
  \pgfmathsetmacro{\yb}{\Ry*\sinB}
  \pgfmathsetmacro{\xc}{\Rx*\cosC}
  \pgfmathsetmacro{\yc}{\Ry*\sinC}
  \pgfmathsetmacro{\xd}{\Rx*\cosD}
  \pgfmathsetmacro{\yd}{\Ry*\sinD}

  \draw (\xa,\ya)..controls(\xc+1/3,\yc-\H/3)and(\xd-1/3,\yd-\H/3)..(\xb,\yb);
  \draw (\xc,\yc)arc[start angle=360, end angle=180, radius=(\xc-\xd)/2]--(\xd,\yd);

  

\end{scope}
}
\newcommand{\yuanzhuCUPPLUSBULKARB}[3]{
\begin{scope}[scale=.7]
  \pgfmathsetmacro{\Rx}{#1}
  \pgfmathsetmacro{\Ry}{#2}
  \pgfmathsetmacro{\H}{#3}

  \cuponlybulk{\Rx}{\Ry}{\H}

  \pgfmathsetmacro{\cosA}{0.6427876}
  \pgfmathsetmacro{\cosB}{-0.6427876}
  \pgfmathsetmacro{\cosC}{0.2588190}
  \pgfmathsetmacro{\cosD}{-0.2588190}
  \pgfmathsetmacro{\sinA}{-0.7660444}
  \pgfmathsetmacro{\sinB}{-0.7660444}
  \pgfmathsetmacro{\sinC}{-0.9659258}
  \pgfmathsetmacro{\sinD}{-0.9659258}

  \pgfmathsetmacro{\xa}{\Rx*\cosA}
  \pgfmathsetmacro{\ya}{\Ry*\sinA}
  \pgfmathsetmacro{\xb}{\Rx*\cosB}
  \pgfmathsetmacro{\yb}{\Ry*\sinB}
  \pgfmathsetmacro{\xc}{\Rx*\cosC}
  \pgfmathsetmacro{\yc}{\Ry*\sinC}
  \pgfmathsetmacro{\xd}{\Rx*\cosD}
  \pgfmathsetmacro{\yd}{\Ry*\sinD}
  
\cuponlybulk{\Rx}{\Ry}{\H}
  \draw (\xa,\ya)..controls(\xc+1/3,\yc-\H/3)and(\xd-1/3,\yd-\H/3)..(\xb,\yb);
  \draw (\xc,\yc)arc[start angle=360, end angle=180, radius=(\xc-\xd)/2]--(\xd,\yd);

  

\end{scope}
}

\newcommand{\yuanzhuCUPPLUSCOLOR}[3]{
\begin{scope}[scale=.7]
  \pgfmathsetmacro{\Rx}{#1}
  \pgfmathsetmacro{\Ry}{#2}
  \pgfmathsetmacro{\H}{#3}

  \cuponly{\Rx}{\Ry}{\H}

  \pgfmathsetmacro{\cosA}{0.6427876}
  \pgfmathsetmacro{\cosB}{-0.6427876}
  \pgfmathsetmacro{\cosC}{0.2588190}
  \pgfmathsetmacro{\cosD}{-0.2588190}
  \pgfmathsetmacro{\sinA}{-0.7660444}
  \pgfmathsetmacro{\sinB}{-0.7660444}
  \pgfmathsetmacro{\sinC}{-0.9659258}
  \pgfmathsetmacro{\sinD}{-0.9659258}

  \pgfmathsetmacro{\xa}{\Rx*\cosA}
  \pgfmathsetmacro{\ya}{\Ry*\sinA}
  \pgfmathsetmacro{\xb}{\Rx*\cosB}
  \pgfmathsetmacro{\yb}{\Ry*\sinB}
  \pgfmathsetmacro{\xc}{\Rx*\cosC}
  \pgfmathsetmacro{\yc}{\Ry*\sinC}
  \pgfmathsetmacro{\xd}{\Rx*\cosD}
  \pgfmathsetmacro{\yd}{\Ry*\sinD}

  \draw[magenta] (\xa,\ya)..controls(\xc+1/3,\yc-\H/3)and(\xd-1/3,\yd-\H/3)..(\xb,\yb);
  \draw (\xc,\yc)arc[start angle=360, end angle=180, radius=(\xc-\xd)/2]--(\xd,\yd);

  

\end{scope}
}

\newcommand{\yuanzhuCUPPLUSCOLORTIKZ}[3]{

\begin{scope}[scale=.7]
  \pgfmathsetmacro{\Rx}{#1}
  \pgfmathsetmacro{\Ry}{#2}
  \pgfmathsetmacro{\H}{#3}

  \cuponly{\Rx}{\Ry}{\H}

  \pgfmathsetmacro{\cosA}{0.6427876}
  \pgfmathsetmacro{\cosB}{-0.6427876}
  \pgfmathsetmacro{\cosC}{0.2588190}
  \pgfmathsetmacro{\cosD}{-0.2588190}
  \pgfmathsetmacro{\sinA}{-0.7660444}
  \pgfmathsetmacro{\sinB}{-0.7660444}
  \pgfmathsetmacro{\sinC}{-0.9659258}
  \pgfmathsetmacro{\sinD}{-0.9659258}

  \pgfmathsetmacro{\xa}{\Rx*\cosA}
  \pgfmathsetmacro{\ya}{\Ry*\sinA}
  \pgfmathsetmacro{\xb}{\Rx*\cosB}
  \pgfmathsetmacro{\yb}{\Ry*\sinB}
  \pgfmathsetmacro{\xc}{\Rx*\cosC}
  \pgfmathsetmacro{\yc}{\Ry*\sinC}
  \pgfmathsetmacro{\xd}{\Rx*\cosD}
  \pgfmathsetmacro{\yd}{\Ry*\sinD}

  \draw[magenta] (\xa,\ya)..controls(\xc+1/3,\yc-\H/3)and(\xd-1/3,\yd-\H/3)..(\xb,\yb);
  \draw (\xc,\yc)arc[start angle=360, end angle=180, radius=(\xc-\xd)/2]--(\xd,\yd);

  

\end{scope}
}

\newcommand{\yuanzhuCUPMINUS}[3]{
\begin{scope}[scale=.7]
  \pgfmathsetmacro{\Rx}{#1}
  \pgfmathsetmacro{\Ry}{#2}
  \pgfmathsetmacro{\H}{#3}

  \cuponly{\Rx}{\Ry}{\H}

  \pgfmathsetmacro{\cosA}{0.6427876}
  \pgfmathsetmacro{\cosB}{-0.6427876}
  \pgfmathsetmacro{\cosC}{0.2588190}
  \pgfmathsetmacro{\cosD}{-0.2588190}
  \pgfmathsetmacro{\sinA}{-0.7660444}
  \pgfmathsetmacro{\sinB}{-0.7660444}
  \pgfmathsetmacro{\sinC}{-0.9659258}
  \pgfmathsetmacro{\sinD}{-0.9659258}

  \pgfmathsetmacro{\xa}{\Rx*\cosA}
  \pgfmathsetmacro{\ya}{\Ry*\sinA}
  \pgfmathsetmacro{\xb}{\Rx*\cosB}
  \pgfmathsetmacro{\yb}{\Ry*\sinB}
  \pgfmathsetmacro{\xc}{\Rx*\cosC}
  \pgfmathsetmacro{\yc}{\Ry*\sinC}
  \pgfmathsetmacro{\xd}{\Rx*\cosD}
  \pgfmathsetmacro{\yd}{\Ry*\sinD}

  \draw (\xa,\ya)..controls(\xc+1/3,\yc-\H/3)and(\xd-1/3,\yd-\H/3)..(\xb,\yb);
  \draw (\xc,\yc)arc[start angle=360, end angle=180, radius=(\xc-\xd)/2]--(\xd,\yd);

  
  \fill (\xb+.1,\yb-\H/7) circle(.05);
  \fill (\xd+.04,\yd-.15) circle(.05);
  \draw[decorate,decoration={snake},red] (\xb+.1,\yb-\H/7)--(\xd+.04,\yd-.15);

\end{scope}
}

\newcommand{\yuanzhuCUPMINUSCOLOR}[3]{
\begin{scope}[scale=.7]
  \pgfmathsetmacro{\Rx}{#1}
  \pgfmathsetmacro{\Ry}{#2}
  \pgfmathsetmacro{\H}{#3}

  \cuponly{\Rx}{\Ry}{\H}

  \pgfmathsetmacro{\cosA}{0.6427876}
  \pgfmathsetmacro{\cosB}{-0.6427876}
  \pgfmathsetmacro{\cosC}{0.2588190}
  \pgfmathsetmacro{\cosD}{-0.2588190}
  \pgfmathsetmacro{\sinA}{-0.7660444}
  \pgfmathsetmacro{\sinB}{-0.7660444}
  \pgfmathsetmacro{\sinC}{-0.9659258}
  \pgfmathsetmacro{\sinD}{-0.9659258}

  \pgfmathsetmacro{\xa}{\Rx*\cosA}
  \pgfmathsetmacro{\ya}{\Ry*\sinA}
  \pgfmathsetmacro{\xb}{\Rx*\cosB}
  \pgfmathsetmacro{\yb}{\Ry*\sinB}
  \pgfmathsetmacro{\xc}{\Rx*\cosC}
  \pgfmathsetmacro{\yc}{\Ry*\sinC}
  \pgfmathsetmacro{\xd}{\Rx*\cosD}
  \pgfmathsetmacro{\yd}{\Ry*\sinD}

  \draw[magenta] (\xa,\ya)..controls(\xc+1/3,\yc-\H/3)and(\xd-1/3,\yd-\H/3)..(\xb,\yb);
  \draw (\xc,\yc)arc[start angle=360, end angle=180, radius=(\xc-\xd)/2]--(\xd,\yd);

  
  \fill (\xb+.1,\yb-\H/7) circle(.05);
  \fill (\xd+.04,\yd-.15) circle(.05);
  \draw[decorate,decoration={snake},red] (\xb+.1,\yb-\H/7)--(\xd+.04,\yd-.15);

\end{scope}
}
\newcommand{\yuanzhuCUP}[3]{
\begin{scope}[scale=.7]
  \pgfmathsetmacro{\Rx}{#1}
  \pgfmathsetmacro{\Ry}{#2}
  \pgfmathsetmacro{\H}{#3}

  \cuponly{\Rx}{\Ry}{\H}

  \pgfmathsetmacro{\cosA}{0.6427876}
  \pgfmathsetmacro{\cosB}{-0.6427876}
  \pgfmathsetmacro{\cosC}{0.2588190}
  \pgfmathsetmacro{\cosD}{-0.2588190}
  \pgfmathsetmacro{\sinA}{-0.7660444}
  \pgfmathsetmacro{\sinB}{-0.7660444}
  \pgfmathsetmacro{\sinC}{-0.9659258}
  \pgfmathsetmacro{\sinD}{-0.9659258}

  \pgfmathsetmacro{\xa}{\Rx*\cosA}
  \pgfmathsetmacro{\ya}{\Ry*\sinA}
  \pgfmathsetmacro{\xb}{\Rx*\cosB}
  \pgfmathsetmacro{\yb}{\Ry*\sinB}
  \pgfmathsetmacro{\xc}{\Rx*\cosC}
  \pgfmathsetmacro{\yc}{\Ry*\sinC}
  \pgfmathsetmacro{\xd}{\Rx*\cosD}
  \pgfmathsetmacro{\yd}{\Ry*\sinD}


  

\end{scope}
}

\newcommand{\yuanzhuCUPZERO}[3]{
\begin{scope}[scale=.7]
  \pgfmathsetmacro{\Rx}{#1}
  \pgfmathsetmacro{\Ry}{#2}
  \pgfmathsetmacro{\H}{#3}

  \cuponly{\Rx}{\Ry}{\H}

  \pgfmathsetmacro{\cosA}{0.6427876}
  \pgfmathsetmacro{\cosB}{-0.6427876}
  \pgfmathsetmacro{\cosC}{0.2588190}
  \pgfmathsetmacro{\cosD}{-0.2588190}
  \pgfmathsetmacro{\sinA}{-0.7660444}
  \pgfmathsetmacro{\sinB}{-0.7660444}
  \pgfmathsetmacro{\sinC}{-0.9659258}
  \pgfmathsetmacro{\sinD}{-0.9659258}

  \pgfmathsetmacro{\xa}{\Rx*\cosA}
  \pgfmathsetmacro{\ya}{\Ry*\sinA}
  \pgfmathsetmacro{\xb}{\Rx*\cosB}
  \pgfmathsetmacro{\yb}{\Ry*\sinB}
  \pgfmathsetmacro{\xc}{\Rx*\cosC}
  \pgfmathsetmacro{\yc}{\Ry*\sinC}
  \pgfmathsetmacro{\xd}{\Rx*\cosD}
  \pgfmathsetmacro{\yd}{\Ry*\sinD}

  \draw (\xa,\ya) -- (\xa,\ya-\H/11)arc[start angle=180, end angle=0, radius=(\xc-\xa)/2]--(\xc,\yc);
  \draw (\xd,\yd) -- (\xd,\yd-\H/15)arc[start angle=360, end angle=180, radius=(\xd-\xb)/2]--(\xb,\yb);

  

\end{scope}
}

\newcommand{\yuanzhuCUPZEROBULK}[3]{
\begin{scope}[scale=.7]
  \pgfmathsetmacro{\Rx}{#1}
  \pgfmathsetmacro{\Ry}{#2}
  \pgfmathsetmacro{\H}{#3}

  \cuponlybulk{\Rx}{\Ry}{\H}

  \pgfmathsetmacro{\cosA}{0.6427876}
  \pgfmathsetmacro{\cosB}{-0.6427876}
  \pgfmathsetmacro{\cosC}{0.2588190}
  \pgfmathsetmacro{\cosD}{-0.2588190}
  \pgfmathsetmacro{\sinA}{-0.7660444}
  \pgfmathsetmacro{\sinB}{-0.7660444}
  \pgfmathsetmacro{\sinC}{-0.9659258}
  \pgfmathsetmacro{\sinD}{-0.9659258}

  \pgfmathsetmacro{\xa}{\Rx*\cosA}
  \pgfmathsetmacro{\ya}{\Ry*\sinA}
  \pgfmathsetmacro{\xb}{\Rx*\cosB}
  \pgfmathsetmacro{\yb}{\Ry*\sinB}
  \pgfmathsetmacro{\xc}{\Rx*\cosC}
  \pgfmathsetmacro{\yc}{\Ry*\sinC}
  \pgfmathsetmacro{\xd}{\Rx*\cosD}
  \pgfmathsetmacro{\yd}{\Ry*\sinD}

  \draw (\xa,\ya) -- (\xa,\ya-\H/11)arc[start angle=180, end angle=0, radius=(\xc-\xa)/2]--(\xc,\yc);
  \draw (\xd,\yd) -- (\xd,\yd-\H/15)arc[start angle=360, end angle=180, radius=(\xd-\xb)/2]--(\xb,\yb);

  

\end{scope}
}

\newcommand{\yuanzhuCUPZEROTWO}[3]{
\begin{scope}[scale=.7]
  \pgfmathsetmacro{\Rx}{#1}
  \pgfmathsetmacro{\Ry}{#2}
  \pgfmathsetmacro{\H}{#3}

  \cuponly{\Rx}{\Ry}{\H}

  \pgfmathsetmacro{\cosA}{0.6427876}
  \pgfmathsetmacro{\cosB}{-0.6427876}
  \pgfmathsetmacro{\cosC}{0.2588190}
  \pgfmathsetmacro{\cosD}{-0.2588190}
  \pgfmathsetmacro{\sinA}{-0.7660444}
  \pgfmathsetmacro{\sinB}{-0.7660444}
  \pgfmathsetmacro{\sinC}{-0.9659258}
  \pgfmathsetmacro{\sinD}{-0.9659258}

  \pgfmathsetmacro{\xa}{\Rx*\cosA}
  \pgfmathsetmacro{\ya}{\Ry*\sinA}
  \pgfmathsetmacro{\xb}{\Rx*\cosB}
  \pgfmathsetmacro{\yb}{\Ry*\sinB}
  \pgfmathsetmacro{\xc}{\Rx*\cosC}
  \pgfmathsetmacro{\yc}{\Ry*\sinC}
  \pgfmathsetmacro{\xd}{\Rx*\cosD}
  \pgfmathsetmacro{\yd}{\Ry*\sinD}

  \draw (\xc,\yc) -- (\xc,\yc-\H/13)arc[start angle=0, end angle=-180, radius=(\xc-\xd)/2]--(\xd,\yd);

  

\end{scope}
}

\newcommand{\yuanzhuCUPZEROBULKARB}[3]{
\begin{scope}[scale=.7]
  \pgfmathsetmacro{\Rx}{#1}
  \pgfmathsetmacro{\Ry}{#2}
  \pgfmathsetmacro{\H}{#3}

  \cuponlybulk{\Rx}{\Ry}{\H}

  \pgfmathsetmacro{\cosA}{0.6427876}
  \pgfmathsetmacro{\cosB}{-0.6427876}
  \pgfmathsetmacro{\cosC}{0.2588190}
  \pgfmathsetmacro{\cosD}{-0.2588190}
  \pgfmathsetmacro{\sinA}{-0.7660444}
  \pgfmathsetmacro{\sinB}{-0.7660444}
  \pgfmathsetmacro{\sinC}{-0.9659258}
  \pgfmathsetmacro{\sinD}{-0.9659258}

  \pgfmathsetmacro{\xa}{\Rx*\cosA}
  \pgfmathsetmacro{\ya}{\Ry*\sinA}
  \pgfmathsetmacro{\xb}{\Rx*\cosB}
  \pgfmathsetmacro{\yb}{\Ry*\sinB}
  \pgfmathsetmacro{\xc}{\Rx*\cosC}
  \pgfmathsetmacro{\yc}{\Ry*\sinC}
  \pgfmathsetmacro{\xd}{\Rx*\cosD}
  \pgfmathsetmacro{\yd}{\Ry*\sinD}

  \draw (\xa,\ya)[bend left=120]to(\xd,\yd);
  \draw (\xb,\yb)[bend right=120]to(\xc,\yc);

\node at (\xd/2+\xc/2,0){$v^*$};
  

\end{scope}
}

\newcommand{\yuanzhuCUPONE}[3]{
\begin{scope}[scale=.7]
  \pgfmathsetmacro{\Rx}{#1}
  \pgfmathsetmacro{\Ry}{#2}
  \pgfmathsetmacro{\H}{#3}

  \cuponly{\Rx}{\Ry}{\H}

  \pgfmathsetmacro{\cosA}{0.6427876}
  \pgfmathsetmacro{\cosB}{-0.6427876}
  \pgfmathsetmacro{\cosC}{0.2588190}
  \pgfmathsetmacro{\cosD}{-0.2588190}
  \pgfmathsetmacro{\sinA}{-0.7660444}
  \pgfmathsetmacro{\sinB}{-0.7660444}
  \pgfmathsetmacro{\sinC}{-0.9659258}
  \pgfmathsetmacro{\sinD}{-0.9659258}

  \pgfmathsetmacro{\xa}{\Rx*\cosA}
  \pgfmathsetmacro{\ya}{\Ry*\sinA}
  \pgfmathsetmacro{\xb}{\Rx*\cosB}
  \pgfmathsetmacro{\yb}{\Ry*\sinB}
  \pgfmathsetmacro{\xc}{\Rx*\cosC}
  \pgfmathsetmacro{\yc}{\Ry*\sinC}
  \pgfmathsetmacro{\xd}{\Rx*\cosD}
  \pgfmathsetmacro{\yd}{\Ry*\sinD}

   \draw (\xa,\ya) -- (\xa,\ya-\H/11)arc[start angle=180, end angle=0, radius=(\xc-\xa)/2]--(\xc,\yc);
  \draw (\xd,\yd) -- (\xd,\yd-\H/15)arc[start angle=360, end angle=180, radius=(\xd-\xb)/2]--(\xb,\yb);

  \draw [red,decorate, decoration={snake}](\xc,\yc-.2)--(\xd,\yd-.2);
  
  \fill (\xd,\yd-.2) circle(.07);
  \fill (\xc,\yc-.2) circle(.07);

\end{scope}
}
\newcommand{\diskup}[6]{
    \begin{scope}[shift={(#1+#3,#2)}, rotate=#5]
        \draw[blue, ->] (0,0) ellipse[x radius=#3, y radius=#4];
        \draw[blue,->] (0,-#4)--(0.001,-#4);
    \end{scope}
    \foreach \i/\angle in {0/180,1/240,2/260,3/280,4/300,5/360}{
        \coordinate (#6\i) at ($
            ({#1+#3},{#2}) + 
            ({
                #3*cos(\angle)*cos(#5) - #4*sin(\angle)*sin(#5)
            },{
                #3*cos(\angle)*sin(#5) + #4*sin(\angle)*cos(#5)
            })
        $);
    }
}

\newcommand{\diskuprev}[6]{
    \begin{scope}[shift={(#1+#3,#2)}, rotate=#5]
        \draw[blue] (0,0) ellipse[x radius=#3, y radius=#4];
        \draw[blue,->] (0,-#4)--(-0.001,-#4);
    \end{scope}
    \foreach \i/\angle in {0/180,1/240,2/260,3/280,4/300,5/360}{
        \coordinate (#6\i) at ($
            ({#1+#3},{#2}) + 
            ({
                #3*cos(\angle)*cos(#5) - #4*sin(\angle)*sin(#5)
            },{
                #3*cos(\angle)*sin(#5) + #4*sin(\angle)*cos(#5)
            })
        $);
    }
}

\newcommand{\diskupcolor}[6]{
    \begin{scope}[shift={(#1+#3,#2)}, rotate=#5]
     \fill[blue!15] (0,0) ellipse[x radius=#3, y radius=#4];
        \draw[blue, ->] (0,0) ellipse[x radius=#3, y radius=#4];
        \draw[blue,->] (0,-#4)--(0.001,-#4);
    \end{scope}
    \foreach \i/\angle in {0/180,1/240,2/260,3/280,4/300,5/360}{
        \coordinate (#6\i) at ($
            ({#1+#3},{#2}) + 
            ({
                #3*cos(\angle)*cos(#5) - #4*sin(\angle)*sin(#5)
            },{
                #3*cos(\angle)*sin(#5) + #4*sin(\angle)*cos(#5)
            })
        $);
    }
}

\newcommand{\diskdown}[6]{
    \begin{scope}[shift={(#1+#3,#2)}, rotate=#5]
        \draw[blue] (-#3,0) arc[start angle=180,end angle=360,x radius=#3, y radius=#4];
    \draw[blue,dashed] (#3,0) arc[start angle=0,end angle=180,x radius=#3, y radius=#4];
    \draw[blue,->] (0,-#4)--(0.001,-#4);
    \end{scope}
    \foreach \i/\angle in {0/180,1/240,2/260,3/280,4/300,5/360}{
        \coordinate (#6\i) at ($
            ({#1+#3},{#2}) + 
            ({
                #3*cos(\angle)*cos(#5) - #4*sin(\angle)*sin(#5)
            },{
                #3*cos(\angle)*sin(#5) + #4*sin(\angle)*cos(#5)
            })
        $);
    }
}
\newcommand{\diskupcolorrev}[6]{
    \begin{scope}[shift={(#1+#3,#2)}, rotate=#5]
     \fill[blue!15] (0,0) ellipse[x radius=#3, y radius=#4];
        \draw[blue] (0,0) ellipse[x radius=#3, y radius=#4];
        \draw[blue,->] (0,-#4)--(-0.001,-#4);
    \end{scope}
    \foreach \i/\angle in {0/180,1/240,2/260,3/280,4/300,5/360}{
        \coordinate (#6\i) at ($
            ({#1+#3},{#2}) + 
            ({
                #3*cos(\angle)*cos(#5) - #4*sin(\angle)*sin(#5)
            },{
                #3*cos(\angle)*sin(#5) + #4*sin(\angle)*cos(#5)
            })
        $);
    }
}

\newcommand{\diskdowncolor}[6]{
    \begin{scope}[shift={(#1+#3,#2)}, rotate=#5]
    \fill[blue!25] (-#3,0) arc[start angle=180,end angle=360,x radius=#3, y radius=#4];
    \fill[blue!25] (#3,0) arc[start angle=0,end angle=180,x radius=#3, y radius=#4];
        \draw[blue] (-#3,0) arc[start angle=180,end angle=360,x radius=#3, y radius=#4];
    \draw[blue,dashed] (#3,0) arc[start angle=0,end angle=180,x radius=#3, y radius=#4];
    \draw[blue,->] (0,-#4)--(0.001,-#4);
    \end{scope}
    \foreach \i/\angle in {0/180,1/240,2/260,3/280,4/300,5/360}{
        \coordinate (#6\i) at ($
            ({#1+#3},{#2}) + 
            ({
                #3*cos(\angle)*cos(#5) - #4*sin(\angle)*sin(#5)
            },{
                #3*cos(\angle)*sin(#5) + #4*sin(\angle)*cos(#5)
            })
        $);
    }
}
\newcommand{\twoqubitgate}[3]{

\begin{scope}[scale=1]

\pgfmathsetmacro{\rx}{#1}
\pgfmathsetmacro{\ry}{#2}
\pgfmathsetmacro{\h}{#3}
\pgfmathsetmacro{\gap}{2*\rx}  
\pgfmathsetmacro{\offset}{3*\rx} 
\pgfmathsetmacro{\hy}{\h*0.2} %

    \diskup{0}{\h}{\rx}{\ry}{0}{P1}
    \foreach \i in {1,2,3,4}{
    \node[above left, green] at (P1\i) {P1\i};
    \fill[green] (P1\i) circle (2pt);
    }

    \diskdown{0}{0}{\rx}{\ry}{0}{P2}
    \foreach \i in {1,2,3,4}{
    \node[above left, green] at (P2\i) {P2\i};
    \fill[green] (P2\i) circle (2pt);
    }
    
    \diskup{\offset}{\h}{\rx}{\ry}{0}{P3}
    \foreach \i in {1,2,3,4}{
    \node[above left, green] at (P3\i) {P3\i};
    \fill[green] (P3\i) circle (2pt);
    }
    
    \diskdown{\offset}{0}{\rx}{\ry}{0}{P4}
    \foreach \i in {1,2,3,4}{
    \node[above left, green] at (P4\i) {P4\i};
    \fill[green] (P4\i) circle (2pt);
    }

\draw[blue] (0,0) -- (0,\h);
\draw[blue] (\offset+\gap,0) -- (\offset+\gap,\h);

\draw[blue] (\gap,\h)--++(0,-\hy) arc[start angle=180,end angle=360,x radius=(\offset-\gap)*0.5, y radius=\ry]
--++(0,\hy);
\draw[blue] (\gap,0)--++(0,\hy) arc[start angle=180,end angle=0,x radius=(\offset-\gap)*0.5, y radius=\ry]
--++(0,-\hy);


\end{scope}

}

\newcommand{\twoqubitgateRZZ}[3]{

\begin{scope}[scale=1]
\node at (-1.2,1.2){\textit{$RZZ(\frac{\pi}{2})=$}};
\pgfmathsetmacro{\rx}{#1}
\pgfmathsetmacro{\ry}{#2}
\pgfmathsetmacro{\h}{#3}
\pgfmathsetmacro{\gap}{2*\rx}  
\pgfmathsetmacro{\offset}{3*\rx} 
\pgfmathsetmacro{\hy}{\h*0.1} %

    \diskup{0}{\h}{\rx}{\ry}{0}{P1}
    \foreach \i in {1,2,3,4}{
    }

    \diskdown{0}{0}{\rx}{\ry}{0}{P2}
    \foreach \i in {1,2,3,4}{
    }
    
    \diskup{\offset}{\h}{\rx}{\ry}{0}{P3}
    \foreach \i in {1,2,3,4}{
    }
    
    \diskdown{\offset}{0}{\rx}{\ry}{0}{P4}
    \foreach \i in {1,2,3,4}{
    }

\draw[blue] (0,0) -- (0,\h);
\draw[blue] (\offset+\gap,0) -- (\offset+\gap,\h);

\draw[blue] (\gap,\h)--++(0,-\hy) arc[start angle=180,end angle=360,x radius=(\offset-\gap)*0.5, y radius=\ry]
--++(0,\hy);
\draw[blue] (\gap,0)--++(0,\hy) arc[start angle=180,end angle=0,x radius=(\offset-\gap)*0.5, y radius=\ry]
--++(0,-\hy);

\draw(P14)arc[start angle=180,end angle=360,radius=.612];

\draw(P24)--++(0,.3)arc[start angle=180,end angle=0,radius=.612]--++(0,-.3);

\draw(P23)--++(0,.7)--++(1.58,1.3)--(P32);
\draw[white,WL](P13)--++(0,-.3)--++(1.58,-1.4)--(P42);
\draw(P13)--++(0,-.3)--++(1.58,-1.4)--(P42);
\draw(P11)--(P21);
\draw(P12)--(P22);
\draw(P33)--(P43);
\draw(P34)--(P44);
\end{scope}

}

\newcommand{\twoqubitgateSWAP}[3]{

\begin{scope}[scale=1]
\node at (-1,1.3){\textit{SWAP=}};
\pgfmathsetmacro{\rx}{#1}
\pgfmathsetmacro{\ry}{#2}
\pgfmathsetmacro{\h}{#3}
\pgfmathsetmacro{\gap}{2*\rx}  
\pgfmathsetmacro{\offset}{3*\rx} 
\pgfmathsetmacro{\hy}{\h*0.1} %

    \diskup{0}{\h}{\rx}{\ry}{0}{P1}
    \foreach \i in {1,2,3,4}{
    }

    \diskdown{0}{0}{\rx}{\ry}{0}{P2}
    \foreach \i in {1,2,3,4}{
    }
    
    \diskup{\offset}{\h}{\rx}{\ry}{0}{P3}
    \foreach \i in {1,2,3,4}{
    }
    
    \diskdown{\offset}{0}{\rx}{\ry}{0}{P4}
    \foreach \i in {1,2,3,4}{
    }

\draw[blue] (0,0) -- (0,\h);
\draw[blue] (\offset+\gap,0) -- (\offset+\gap,\h);

\draw[blue] (\gap,\h)--++(0,-\hy) arc[start angle=180,end angle=360,x radius=(\offset-\gap)*0.5, y radius=\ry]
--++(0,\hy);
\draw[blue] (\gap,0)--++(0,\hy) arc[start angle=180,end angle=0,x radius=(\offset-\gap)*0.5, y radius=\ry]
--++(0,-\hy);

\foreach \i in {1,2,3,4}{
\draw(P3\i)--++(0,-.3)--++(-1.8,-1.6)--(P2\i);
}

\foreach \i in {1,2,3,4}{
\draw[white,WL](P1\i)--++(0,-.3)--++(1.8,-1.6)--(P4\i);
\draw(P1\i)--++(0,-.3)--++(1.8,-1.6)--(P4\i);
}

\end{scope}

}


\newcommand{\threequbitgate}[3]{
\begin{scope}[scale=1]

\pgfmathsetmacro{\rx}{#1}
\pgfmathsetmacro{\ry}{#2}
\pgfmathsetmacro{\h}{#3}
\pgfmathsetmacro{\gap}{2*\rx}  
\pgfmathsetmacro{\offset}{3*\rx} 
\pgfmathsetmacro{\hy}{\h*0.2} %

    \diskup{0}{\h}{\rx}{\ry}{0}{P1}
    \foreach \i in {1,2,3,4}{
    \node[above left, red] at (P1\i) {P1\i};
    \fill[red] (P1\i) circle (2pt);
    }

    \diskdown{0}{0}{\rx}{\ry}{0}{P2}
    \foreach \i in {1,2,3,4}{
    \node[above left, red] at (P2\i) {P2\i};
    \fill[red] (P2\i) circle (2pt);
    }
    
    \diskup{\offset}{\h}{\rx}{\ry}{0}{P3}
    \foreach \i in {1,2,3,4}{
    \node[above left, red] at (P3\i) {P3\i};
    \fill[red] (P3\i) circle (2pt);
    }
    
    \diskdown{\offset}{0}{\rx}{\ry}{0}{P4}
    \foreach \i in {1,2,3,4}{
    \node[above left, red] at (P4\i) {P4\i};
    \fill[red] (P4\i) circle (2pt);
    }

    \diskup{2*\offset}{\h}{\rx}{\ry}{0}{P5}
    \foreach \i in {1,2,3,4}{
    \node[above left, red] at (P5\i) {P5\i};
    \fill[red] (P5\i) circle (2pt);
    }
    
    \diskdown{2*\offset}{0}{\rx}{\ry}{0}{P6}
    \foreach \i in {1,2,3,4}{
    \node[above left, red] at (P6\i) {P6\i};
    \fill[red] (P6\i) circle (2pt);
    }

\draw[blue] (0,0) -- (0,\h);
\draw[blue] (2*\offset+\gap,0) -- (2*\offset+\gap,\h);

\foreach \a in {0,\offset}
{
\draw[blue] (\a+\gap,\h)--++(0,-\hy) arc[start angle=180,end angle=360,x radius=(\offset-\gap)*0.5, y radius=\ry]
--++(0,\hy);
\draw[blue] (\a+\gap,0)--++(0,\hy) arc[start angle=180,end angle=0,x radius=(\offset-\gap)*0.5, y radius=\ry]
--++(0,-\hy);
}

\end{scope}
}

\newcommand{\fourqubitstate}[3]{

\begin{scope}[scale=1]

\pgfmathsetmacro{\rx}{#1}
\pgfmathsetmacro{\ry}{#2}
\pgfmathsetmacro{\h}{#3}
\pgfmathsetmacro{\gap}{2*\rx}  
\pgfmathsetmacro{\offset}{3*\rx} 
\pgfmathsetmacro{\hy}{\h*0.2} %

    \diskdown{0}{\h}{\rx}{\ry}{0}{P1}
    \foreach \i in {1,2,3,4}{
    }
    
    \diskdown{\offset}{\h}{\rx}{\ry}{0}{P2}
    \foreach \i in {1,2,3,4}{
    }

    \diskdown{2*\offset}{\h}{\rx}{\ry}{0}{P3}
    \foreach \i in {1,2,3,4}{
    }

    \diskdown{3*\offset}{\h}{\rx}{\ry}{0}{P4}
    \foreach \i in {1,2,3,4}{
    }

\draw[blue](0,\h)  arc[start angle=180, end angle=0, x radius=(3*\offset+\gap)*0.5, y radius=\h];

\draw[blue] (\gap,\h) arc[start angle=180,end angle=0,x radius=(\offset-\gap)*0.5, y radius=\ry];
\draw[blue] (\offset+\gap,\h) arc[start angle=180,end angle=0,x radius=(\offset-\gap)*0.5, y radius=\ry];
\draw[blue] (2*\offset+\gap,\h) arc[start angle=180,end angle=0,x radius=(\offset-\gap)*0.5, y radius=\ry];

\draw[magenta!80](P11)[bend left=90]to(P44);

\draw[magenta!80](P14)arc[start angle=180,end angle=0,radius=1];
\draw[magenta!80](P24)arc[start angle=180,end angle=0,radius=1];
\draw[magenta!80](P34)arc[start angle=180,end angle=0,radius=1];
\end{scope}

}

\newcommand{\threequbitgateCARTHREE}[3]{
\begin{scope}[scale=1]

\pgfmathsetmacro{\rx}{#1}
\pgfmathsetmacro{\ry}{#2}
\pgfmathsetmacro{\h}{#3}
\pgfmathsetmacro{\gap}{2*\rx}  
\pgfmathsetmacro{\offset}{3*\rx} 
\pgfmathsetmacro{\hy}{\h*0.1} %

    \diskup{0}{\h}{\rx}{\ry}{0}{P1}
    \foreach \i in {1,2,3,4}{
    }

    \diskdown{0}{0}{\rx}{\ry}{0}{P2}
    \foreach \i in {1,2,3,4}{
    }
    
    \diskup{\offset}{\h}{\rx}{\ry}{0}{P3}
    \foreach \i in {1,2,3,4}{
    }
    
    \diskdown{\offset}{0}{\rx}{\ry}{0}{P4}
    \foreach \i in {1,2,3,4}{
    }

    \diskup{2*\offset}{\h}{\rx}{\ry}{0}{P5}
    \foreach \i in {1,2,3,4}{
    }
    
    \diskdown{2*\offset}{0}{\rx}{\ry}{0}{P6}
    \foreach \i in {1,2,3,4}{
    }
    
     \draw[red,decorate,decoration={snake}](P42)++(0,+\h*0.6)--++(-3.35,0);
    
    \draw[magenta!80](P11)--(P21);
    \draw[white,WL](P12)--(P22);
    \draw(P12)--(P22);
    \draw[white,WL](P13)--(P23);
    \draw(P13)--(P23);
    \draw[magenta!80](P14)arc[start angle=180, end angle=360,radius=1];
     \draw[magenta!80](P24)--++(0,.5)arc[start angle=180, end angle=0,radius=1]--++(0,-.5);
     \draw(P32)--(P42);
     \draw(P33)--(P43);
      \draw[magenta!80](P34)arc[start angle=180, end angle=360,radius=1];
     \draw[magenta!80](P44)--++(0,.5)arc[start angle=180, end angle=0,radius=1]--++(0,-.5);
      \draw(P52)--(P62);
     \draw(P53)--(P63);
     \draw[magenta!80](P54)--(P64);

\draw[blue] (0,0) -- (0,\h);
\draw[blue] (2*\offset+\gap,0) -- (2*\offset+\gap,\h);

\foreach \a in {0,\offset}
{
\draw[blue] (\a+\gap,\h)--++(0,-\hy) arc[start angle=180,end angle=360,x radius=(\offset-\gap)*0.5, y radius=\ry]
--++(0,\hy);
\draw[blue] (\a+\gap,0)--++(0,\hy) arc[start angle=180,end angle=0,x radius=(\offset-\gap)*0.5, y radius=\ry]
--++(0,-\hy);
}

      \fill(P21)++(0,+\h*0.58) circle(.1);
     \fill(P42)++(0,+\h*0.6) circle(.1);
\end{scope}
}

\newcommand{\threequbitgateCARSIX}[3]{
\begin{scope}[scale=1]

\pgfmathsetmacro{\rx}{#1}
\pgfmathsetmacro{\ry}{#2}
\pgfmathsetmacro{\h}{#3}
\pgfmathsetmacro{\gap}{2*\rx}  
\pgfmathsetmacro{\offset}{3*\rx} 
\pgfmathsetmacro{\hy}{\h*0.1} %

    \diskup{0}{\h}{\rx}{\ry}{0}{P1}
    \foreach \i in {1,2,3,4}{
    }

    \diskdown{0}{0}{\rx}{\ry}{0}{P2}
    \foreach \i in {1,2,3,4}{
    }
    
    \diskup{\offset}{\h}{\rx}{\ry}{0}{P3}
    \foreach \i in {1,2,3,4}{
    }
    
    \diskdown{\offset}{0}{\rx}{\ry}{0}{P4}
    \foreach \i in {1,2,3,4}{
    }

    \diskup{2*\offset}{\h}{\rx}{\ry}{0}{P5}
    \foreach \i in {1,2,3,4}{
    }
    
    \diskdown{2*\offset}{0}{\rx}{\ry}{0}{P6}
    \foreach \i in {1,2,3,4}{
    }
    
     \draw[red,decorate,decoration={snake}](P63)++(0,+\h*0.6)--++(-6.7,0);

    \draw[magenta!80](P11)--(P21);
    \draw[white,WL](P12)--(P22);
    \draw(P12)--(P22);
    \draw[white,WL](P13)--(P23);
    \draw(P13)--(P23);
    \draw[magenta!80](P14)arc[start angle=180, end angle=360,radius=1];
     \draw[magenta!80](P24)--++(0,.5)arc[start angle=180, end angle=0,radius=1]--++(0,-.5);
     \draw[white,WL](P32)--(P42);
     \draw(P32)--(P42);
      \draw[white,WL](P33)--(P43);
     \draw(P33)--(P43);
      \draw[magenta!80](P34)arc[start angle=180, end angle=360,radius=1];
     \draw[magenta!80](P44)--++(0,.5)arc[start angle=180, end angle=0,radius=1]--++(0,-.5);
      \draw[white,WL](P52)--(P62);
      \draw(P52)--(P62);
       \draw[white,WL](P53)--(P63);
     \draw(P53)--(P63);
     \draw[magenta!80](P54)--(P64);

\draw[blue] (0,0) -- (0,\h);
\draw[blue] (2*\offset+\gap,0) -- (2*\offset+\gap,\h);

\foreach \a in {0,\offset}
{
\draw[blue] (\a+\gap,\h)--++(0,-\hy) arc[start angle=180,end angle=360,x radius=(\offset-\gap)*0.5, y radius=\ry]
--++(0,\hy);
\draw[blue] (\a+\gap,0)--++(0,\hy) arc[start angle=180,end angle=0,x radius=(\offset-\gap)*0.5, y radius=\ry]
--++(0,-\hy);
}

 \fill(P21)++(0,+\h*0.58) circle(.1);
    \fill(P63)++(0,+\h*0.6) circle(.1);

\end{scope}
}

\newcommand{\threequbitgateCARTHREEPAULI}[3]{
\begin{scope}[scale=1]

\pgfmathsetmacro{\rx}{#1}
\pgfmathsetmacro{\ry}{#2}
\pgfmathsetmacro{\h}{#3}
\pgfmathsetmacro{\gap}{2*\rx}  
\pgfmathsetmacro{\offset}{3*\rx} 
\pgfmathsetmacro{\hy}{\h*0.1} %

    \diskup{0}{\h}{\rx}{\ry}{0}{P1}
    \foreach \i in {1,2,3,4}{
    }

    \diskdown{0}{0}{\rx}{\ry}{0}{P2}
    \foreach \i in {1,2,3,4}{
    }
    
    \diskup{\offset}{\h}{\rx}{\ry}{0}{P3}
    \foreach \i in {1,2,3,4}{
    }
    
    \diskdown{\offset}{0}{\rx}{\ry}{0}{P4}
    \foreach \i in {1,2,3,4}{
    }

    \diskup{2*\offset}{\h}{\rx}{\ry}{0}{P5}
    \foreach \i in {1,2,3,4}{
    }
    
    \diskdown{2*\offset}{0}{\rx}{\ry}{0}{P6}
    \foreach \i in {1,2,3,4}{
    }
    

     \draw[red,decorate,decoration={snake}](P42)++(0,+\h*0.6)--++(-1,0.35);
     \draw[red,decorate,decoration={snake}](P21)++(0,+\h*0.6)--++(1.5,0.35);

    \draw[magenta!80](P11)--(P21);
    \draw[white,WL](P12)--(P22);
    \draw(P12)--(P22);
    \draw[white,WL](P13)--(P23);
    \draw(P13)--(P23);
    \draw[magenta!80](P14)arc[start angle=180, end angle=360,radius=1];
     \draw[magenta!80](P24)--++(0,.5)arc[start angle=180, end angle=0,radius=1]--++(0,-.5);
     \draw(P32)--(P42);
     \draw(P33)--(P43);
      \draw[magenta!80](P34)arc[start angle=180, end angle=360,radius=1];
     \draw[magenta!80](P44)--++(0,.5)arc[start angle=180, end angle=0,radius=1]--++(0,-.5);
      \draw(P52)--(P62);
     \draw(P53)--(P63);
     \draw[magenta!80](P54)--(P64);

\draw[blue] (0,0) -- (0,\h);
\draw[blue] (2*\offset+\gap,0) -- (2*\offset+\gap,\h);

\foreach \a in {0,\offset}
{
\draw[blue] (\a+\gap,\h)--++(0,-\hy) arc[start angle=180,end angle=360,x radius=(\offset-\gap)*0.5, y radius=\ry]
--++(0,\hy);
\draw[blue] (\a+\gap,0)--++(0,\hy) arc[start angle=180,end angle=0,x radius=(\offset-\gap)*0.5, y radius=\ry]
--++(0,-\hy);
}

     \fill(P21)++(1.5,.3+\h*0.6) circle(.1);

    \fill(P42)++(-1,+\h*0.6+.3)circle(.1);
     
     \fill(P21)++(0,+\h*0.58) circle(.1);
     \fill(P42)++(0,+\h*0.6) circle(.1);
\end{scope}
}

\newcommand{\threequbitgateCARSIXPAULI}[3]{
\begin{scope}[scale=1]

\pgfmathsetmacro{\rx}{#1}
\pgfmathsetmacro{\ry}{#2}
\pgfmathsetmacro{\h}{#3}
\pgfmathsetmacro{\gap}{2*\rx}  
\pgfmathsetmacro{\offset}{3*\rx} 
\pgfmathsetmacro{\hy}{\h*0.1} %

    \diskup{0}{\h}{\rx}{\ry}{0}{P1}
    \foreach \i in {1,2,3,4}{
    }

    \diskdown{0}{0}{\rx}{\ry}{0}{P2}
    \foreach \i in {1,2,3,4}{
    }
    
    \diskup{\offset}{\h}{\rx}{\ry}{0}{P3}
    \foreach \i in {1,2,3,4}{
    }
    
    \diskdown{\offset}{0}{\rx}{\ry}{0}{P4}
    \foreach \i in {1,2,3,4}{
    }

    \diskup{2*\offset}{\h}{\rx}{\ry}{0}{P5}
    \foreach \i in {1,2,3,4}{
    }
    
    \diskdown{2*\offset}{0}{\rx}{\ry}{0}{P6}
    \foreach \i in {1,2,3,4}{
    }
    

     \draw[red,decorate,decoration={snake}](P21)++(0,+\h*0.6)--++(1.5,0.35);
        \draw[red,decorate, decoration={snake}](P42)++(-1+2.4,+\h*0.6+.3)--++(-2.4,0);
         \draw[red,decorate, decoration={snake}](P63)++(0,\h*.6)--++(-1.3,0.33);

    \draw[magenta!80](P11)--(P21);
    \draw[white,WL](P12)--(P22);
    \draw(P12)--(P22);
    \draw[white,WL](P13)--(P23);
    \draw(P13)--(P23);
    \draw[magenta!80](P14)arc[start angle=180, end angle=360,radius=1];
     \draw[magenta!80](P24)--++(0,.5)arc[start angle=180, end angle=0,radius=1]--++(0,-.5);
      \draw[white,WL](P32)--(P42);
     \draw(P32)--(P42);
     \draw[white,WL](P33)--(P43);
     \draw(P33)--(P43);
      \draw[magenta!80](P34)arc[start angle=180, end angle=360,radius=1];
     \draw[magenta!80](P44)--++(0,.5)arc[start angle=180, end angle=0,radius=1]--++(0,-.5);
     \draw[white,WL](P52)--(P62);
     \draw(P52)--(P62);
      \draw[white,WL](P53)--(P63);
     \draw(P53)--(P63);
     \draw[white,WL](P54)--(P64);
     \draw[magenta!80](P54)--(P64);

\draw[blue] (0,0) -- (0,\h);
\draw[blue] (2*\offset+\gap,0) -- (2*\offset+\gap,\h);

\foreach \a in {0,\offset}
{
\draw[blue] (\a+\gap,\h)--++(0,-\hy) arc[start angle=180,end angle=360,x radius=(\offset-\gap)*0.5, y radius=\ry]
--++(0,\hy);
\draw[blue] (\a+\gap,0)--++(0,\hy) arc[start angle=180,end angle=0,x radius=(\offset-\gap)*0.5, y radius=\ry]
--++(0,-\hy);
}

  \fill(P21)++(1.5,.3+\h*0.6) circle(.1);

    \fill(P42)++(-1,+\h*0.6+.3)circle(.1);
     \fill(P42)++(-1+2.4,+\h*0.6+.3)circle(.1);

     \fill(P21)++(0,+\h*0.58) circle(.1);

    \fill(P63)++(0,+\h*0.6) circle(.1);
     \fill(P63)++(-1.3,+\h*0.6+.33) circle(.1);
     
\end{scope}

}


\newcommand{\test}[6]{
\begin{tikzpicture}
    \begin{scope}[shift={(#1+#3,#2)}, rotate=#5]
        \draw[blue, ->] (0,0) ellipse[x radius=#3, y radius=#4];
        \twoqubitgate{1}{0.5}{4};
        \draw[thick,red] (P23)--(P33);
        \draw[thick,red] (P24)--(P31);
    \end{scope}
\end{tikzpicture}
}

\newcommand{\testtwo}[6]{
\begin{tikzpicture}
    \begin{scope}[shift={(#1+#3,#2)}, rotate=#5]
        \twoqubitgate{1}{0.5}{4};
        \draw[thick,red] (P23)--(P33);
        \draw[thick,red] (P24)--(P31);
    \end{scope}
\end{tikzpicture}
}

\newcommand{\threequbitgatetest}[3]{
\begin{tikzpicture}
\begin{scope}[scale=1]

\pgfmathsetmacro{\rx}{#1}
\pgfmathsetmacro{\ry}{#2}
\pgfmathsetmacro{\h}{#3}
\pgfmathsetmacro{\gap}{2*\rx}  
\pgfmathsetmacro{\offset}{3*\rx} 
\pgfmathsetmacro{\hy}{\h*0.1} %

\foreach \x in {0,\offset,2*\offset}{
    \draw[white,WL] (\x,\h) arc[start angle=180,end angle=540,x radius=\rx, y radius=\ry];
    \draw[blue,->] (\x,\h) arc[start angle=180,end angle=540,x radius=\rx, y radius=\ry];
    \draw[blue,dashed] (\x,0) arc[start angle=180,end angle=0,x radius=\rx, y radius=\ry];
    \draw[blue,->] (\x,0) arc[start angle=180,end angle=360,x radius=\rx, y radius=\ry];
}

\draw[blue] (0,0) -- (0,\h);
\draw[blue] (2*\offset+\gap,0) -- (2*\offset+\gap,\h);

\foreach \a in {0,\offset}
{
\draw[blue] (\a+\gap,\h)--++(0,-\hy) arc[start angle=180,end angle=360,x radius=(\offset-\gap)*0.5, y radius=\ry]
--++(0,\hy);
\draw[blue] (\a+\gap,0)--++(0,\hy) arc[start angle=180,end angle=0,x radius=(\offset-\gap)*0.5, y radius=\ry]
--++(0,-\hy);
}

\foreach \i/\angle in {1/240,2/260,3/280,4/300} {
    \coordinate (P\i) at ({\gap/2+\rx*cos(\angle)}, {\h+\ry*sin(\angle)});
}
\foreach \i/\angle in {5/240,6/260,7/280,8/300} {
    \coordinate (P\i) at ({\gap/2+\rx*cos(\angle)}, {\ry*sin(\angle)});
   
}

\foreach \i/\angle in {9/240,10/260,11/280,12/300} {
    \coordinate (P\i) at ({\gap/2+\offset+\rx*cos(\angle)}, {\h+\ry*sin(\angle)});
}
\foreach \i/\angle in {13/240,14/260,15/280,16/300} {
    \coordinate (P\i) at ({\gap/2+\offset+\rx*cos(\angle)}, {\ry*sin(\angle)});
}
\foreach \i/\angle in {17/240,18/260,19/280,20/300} {
    \coordinate (P\i) at ({\gap/2+2*\offset+\rx*cos(\angle)}, {\h+\ry*sin(\angle)});
}
\foreach \i/\angle in {21/240,22/260,23/280,24/300} {
    \coordinate (P\i) at ({\gap/2+2*\offset+\rx*cos(\angle)}, {\ry*sin(\angle)});
}

\draw(P1)--(P5);
\draw(P2)--(P6);
\draw(P3)--(P7);
\draw(P10)--(P14);
\draw(P11)--(P15);
\draw(P18)--(P22);
\draw(P19)--(P23);
\draw(P20)--(P24);

\draw(P4)arc[start angle=180,end angle=360,radius=\offset-\gap];
\draw(P8)--++(0,.5)arc[start angle=180,end angle=0,radius=\offset-\gap]--++(0,-.5);

\draw(P12)arc[start angle=180,end angle=360,radius=\offset-\gap];
\draw(P16)--++(0,.5)arc[start angle=180,end angle=0,radius=\offset-\gap]--++(0,-.5);

\end{scope}
\end{tikzpicture}
}

\newcommand{\copytensorone}[3]{

\begin{scope}[scale=1]

\pgfmathsetmacro{\rx}{#1}
\pgfmathsetmacro{\ry}{#2}
\pgfmathsetmacro{\h}{#3}
\pgfmathsetmacro{\gap}{2*\rx}  
\pgfmathsetmacro{\offset}{3*\rx} 
\pgfmathsetmacro{\hy}{\h*0.2} %

    \diskup{0}{\h}{\rx}{\ry}{0}{P1}
    \foreach \i in {0,1,2,3,4,5}{
    }

    \diskdown{0}{0}{\rx}{\ry}{0}{P2}
    \foreach \i in {0,1,2,3,4,5}{
    }
    
    \diskuprev{\offset}{\h*0.5}{\rx}{\ry}{-90}{P3}
    \foreach \i in {0,1,2,3,4,5}{
    }

\draw[blue] (P20) -- (P10);


\draw[blue] (P15) .. controls +(270:\h*0.2) and +(180:\h*0.2) .. (P30);
\draw[blue] (P25) .. controls +(90:\h*0.2) and +(180:\h*0.2) .. (P35);

\draw(P11)--(P21);
\draw(P12)[bend right=45]to(P33);
\draw(P13)[bend left=30]to(P23);
\draw(P14)[bend right=45]to(P31);
\draw(P24)[bend left=45]to(P34);

\draw(P22)[bend left=45]to(P32);


\end{scope}

}

\newcommand{\copytensorcopy}[3]{

\begin{scope}[scale=1]

\pgfmathsetmacro{\rx}{#1}
\pgfmathsetmacro{\ry}{#2}
\pgfmathsetmacro{\h}{#3}
\pgfmathsetmacro{\gap}{2*\rx}  
\pgfmathsetmacro{\offset}{3*\rx} 
\pgfmathsetmacro{\hy}{\h*0.2} %


    \diskupcolor{0}{\h}{\rx}{\ry}{0}{P1}

    \diskdowncolor{0}{0}{\rx}{\ry}{0}{P2}

    \diskupcolorrev{\offset}{\h*0.5}{\rx}{\ry}{-90}{P3}

\fill[blue!25] (P20)--(P10)--(P15) .. controls +(270:\h*0.2) and +(180:\h*0.2) .. (P30);
\fill[blue!25] (P20)--(P25) .. controls +(90:\h*0.2) and +(180:\h*0.2) .. (P35)--(P30);

    \diskupcolor{0}{\h}{\rx}{\ry}{0}{P1}
    
    \diskupcolorrev{\offset}{\h*0.5}{\rx}{\ry}{-90}{P3}
\draw[blue] (P20) -- (P10);
\draw[blue] (P20) -- (P10);


\draw[blue] (P15) .. controls +(270:\h*0.2) and +(180:\h*0.2) .. (P30);
\draw[blue] (P25) .. controls +(90:\h*0.2) and +(180:\h*0.2) .. (P35);


    \node at (0.3,4.3){T.B.};
     \node at (0.3,-.8){T.B.};

     \node at (5.2,2){T.B.};
     \node at (3,1){S.B.};

\draw(P11)--(P21);
\draw(P12)--(P22);
\draw(P13)[bend right=30]to(P32);

\draw(P14)[bend right=30]to(P31);
\draw(P24)[bend left=45]to(P34);
\draw(P23)[bend left=45]to(P33);


\end{scope}

}

\newcommand{\copytensorpure}[3]{

\begin{scope}[scale=1]

\pgfmathsetmacro{\rx}{#1}
\pgfmathsetmacro{\ry}{#2}
\pgfmathsetmacro{\h}{#3}
\pgfmathsetmacro{\gap}{2*\rx}  
\pgfmathsetmacro{\offset}{3*\rx} 
\pgfmathsetmacro{\hy}{\h*0.2} %

    \diskup{0}{\h}{\rx}{\ry}{0}{P1}
    \foreach \i in {0,1,2,3,4,5}{
    }
    \diskdown{0}{0}{\rx}{\ry}{0}{P2}
    \foreach \i in {0,1,2,3,4,5}{
    }
    
    \diskuprev{\offset}{\h*0.5}{\rx}{\ry}{-90}{P3}
    \foreach \i in {0,1,2,3,4,5}{
    }
\draw[blue] (P20) -- (P10);


\draw[blue] (P15) .. controls +(270:\h*0.2) and +(180:\h*0.2) .. (P30);
\draw[blue] (P25) .. controls +(90:\h*0.2) and +(180:\h*0.2) .. (P35);


\draw(P11)--(P21);
\draw(P12)--(P22);
\draw(P13)[bend right=30]to(P32);

\draw(P14)[bend right=30]to(P31);
\draw(P24)[bend left=45]to(P34);
\draw(P23)[bend left=45]to(P33);


\end{scope}

}

\newcommand{\genus}[2]{ 
  \begin{scope}
  \pgfmathsetmacro{\r}{0.5*#1};
  \pgfmathsetmacro{\rr}{0.8*#1};
  \pgfmathsetmacro{\x}{\r*0.866};
  \pgfmathsetmacro{\xx}{\rr*0.866};
    \draw[blue, thick] 
      (-\x, {-0.1*#1}) arc[start angle=150, end angle=30, radius={\r}];
    \draw[blue, thick] 
      (-\xx, 0.1*#1) arc[start angle=-150, end angle=-30, radius={\rr}];
      \node at (0,0.6*\rr) {#2};
  \end{scope}
}

\newcommand{\hexcell}[1]{ 
  \begin{scope}
    \draw[thick] (30:#1) \foreach \a in {90,150,210,270,330} { -- (\a:#1)} -- cycle;

    \draw[magenta, thick] (0,0) circle({0.6*#1});

    \genus{#1*0.6}

    \foreach \a in {30,90,...,330} {
  \path
    coordinate (A) at ($({#1*cos(\a)},{#1*sin(\a)})!0.5!({#1*cos(\a-60)},{#1*sin(\a-60)})$)
    coordinate (B) at ($({#1*cos(\a)},{#1*sin(\a)})!0.5!({#1*cos(\a+60)},{#1*sin(\a+60)})$);
  \draw[ thick] ($ (A)!-0.5!(B) $) -- ($ (A)!1.5!(B) $);
}

  \end{scope}
}


\newcommand{\hexcellonlyblack}[1]{ 
  \begin{scope}
    \draw[line width=5pt, black!40] (30:#1) \foreach \a in {90,150,210,270,330} { -- (\a:#1)} -- cycle;

  \end{scope}
}

\newcommand{\hexcellwithoutblack}[1]{ 
  \begin{scope}

    \draw[magenta, thick] (0,0) circle({0.6*#1});

    \genus{#1*0.6}

    \foreach \a in {30,90,...,330} {
  \path
    coordinate (A) at ($({#1*cos(\a)},{#1*sin(\a)})!0.5!({#1*cos(\a-60)},{#1*sin(\a-60)})$)
    coordinate (B) at ($({#1*cos(\a)},{#1*sin(\a)})!0.5!({#1*cos(\a+60)},{#1*sin(\a+60)})$);
  \draw[ thick] ($ (A)!-0.5!(B) $) -- ($ (A)!1.5!(B) $);
}

  \end{scope}
}

\newcommand{\hexcellreduced}[1]{ 
  \begin{scope}



    \foreach \a in {30,90,...,330} {
  \path
    coordinate (A) at ($({#1*cos(\a)},{#1*sin(\a)})!0.5!({#1*cos(\a-60)},{#1*sin(\a-60)})$)
    coordinate (B) at ($({#1*cos(\a)},{#1*sin(\a)})!0.5!({#1*cos(\a+60)},{#1*sin(\a+60)})$);
  \draw[ thick] ($ (A)!-0.5!(B) $) -- ($ (A)!1.5!(B) $);
}

  \end{scope}
}

\newcommand{\isingonlyblack}[1]{
  \begin{scope}
  \pgfmathsetmacro{\dy}{sqrt(3)*#1/2} 

  \begin{scope}[shift={(-\dy,0)}]
    \hexcellonlyblack{#1}
  \end{scope}

  \begin{scope}[shift={(\dy,0)}]
    \hexcellonlyblack{#1}
  \end{scope}

  \begin{scope}[shift={(0,#1*1.5)}]
    \hexcellonlyblack{#1}
  \end{scope}

  \end{scope}
}

\newcommand{\isingonlyblackopen}[1]{
  \begin{scope}[scale=#1]
  \pgfmathsetmacro{\dy}{sqrt(3)*0.6/2} 

  \begin{scope}[shift={(-\dy,0)}]
    \hexcellonlyblack{0.6}
  \end{scope}

  \begin{scope}[shift={(\dy,0)}]
    \hexcellonlyblack{0.6}
  \end{scope}

  \begin{scope}[shift={(0,0.6*1.5)}]
    \hexcellonlyblack{0.6}
  \end{scope}
  \draw[black!40,line width=5pt](-.5,1.2)--(-1,1.6);
  \draw[black!40,line width=5pt](.5,1.2)--(1,1.6);
  \draw[black!40,line width=5pt](-.5,-1)--(-.5,-.6);
  \draw[black!40,line width=5pt](.5,-1)--(.5,-.6);
  \draw[black!40,line width=5pt](-1,-.3)--(-1.5,-.6);
  \draw[black!40,line width=5pt](1,-.3)--(1.5,-.6);
  \draw[black!40,line width=5pt](-1,.3)--(-1.5,.6);
  \draw[black!40,line width=5pt](1,.3)--(1.5,.6);
  \draw[black!40,line width=5pt](0,1.4)--(0,1.8);
  \end{scope}
}

\newcommand{\isingMGoriginal}[1]{
  \begin{scope}
  \pgfmathsetmacro{\dy}{sqrt(3)*#1/2} 

  \begin{scope}[shift={(-\dy,0)}]
    \hexcellwithoutblack{#1}
  \end{scope}

  \begin{scope}[shift={(\dy,0)}]
    \hexcellwithoutblack{#1}
  \end{scope}

  \begin{scope}[shift={(0,#1*1.5)}]
    \hexcellwithoutblack{#1}
  \end{scope}

  \end{scope}
}

\newcommand{\disctwod}[3]{
\pgfmathsetmacro{\x}{#1}
\pgfmathsetmacro{\y}{#2}
\pgfmathsetmacro{\r}{#3}
\begin{scope}[shift={(\x,\y)}, scale=\r]
\fill[white](1,0)arc[start angle=0,end angle=360,radius=1];
  \draw[blue,line width=.7pt](1,0)arc[start angle=0,end angle=360,radius=1];
  \draw[blue,->](1,0)--(1,0.01);
\end{scope}
}

\newcommand{\isingMGoriginalopen}[1]{
  \begin{scope}
  \pgfmathsetmacro{\dy}{sqrt(3)*0.6/2} 

  \begin{scope}[shift={(-\dy,0)}]
    \hexcellwithoutblack{.6}
  \end{scope}

  \begin{scope}[shift={(\dy,0)}]
    \hexcellwithoutblack{.6}
  \end{scope}

  \begin{scope}[shift={(0,0.6*1.5)}]
    \hexcellwithoutblack{0.6}
  \end{scope}

  \draw[blue,line width=.7pt](0.25,1.7)[bend right=30]to(.55,1.45);
  \draw[blue,line width=.7pt](-0.25,1.7)[bend left=30]to(-.55,1.45);
  \draw[magenta,line width=.8pt](0.2,1.6)[bend right=30]to(.5,1.4);
  \draw[magenta,line width=.8pt](-0.2,1.6)[bend left=30]to(-.5,1.4);

  \draw[magenta,line width=.8pt](0.65,1.2)[bend right=30]to(1,.5);

  \draw[magenta,line width=.8pt](-0.65,1.2)[bend left=30]to(-1,.5);

   \draw[blue,line width=.7pt](0.75,1.2)[bend right=30]to(1.05,.6);
   \draw[blue,line width=.7pt](-0.75,1.2)[bend left=30]to(-1.05,.6);

    \draw[magenta,line width=.8pt](1.3,.2)[bend right=30]to(1.3,-.2);
   \draw[blue,line width=.7pt](1.4,.2)[bend right=30]to(1.4,-.2);

   \draw[magenta,line width=.8pt](-1.3,.2)[bend left=30]to(-1.3,-.2);
   \draw[blue,line width=.7pt](-1.4,.2)[bend left=30]to(-1.4,-.2);

   \draw[magenta,line width=.8pt](.4,-.8)[bend right=30]to(-.4,-.8);
   \draw[blue,line width=.7pt](.4,-.9)[bend right=30]to(-.4,-.9);

    \draw[magenta,line width=.8pt](.75,-.8)[bend left=30]to(1.2,-.5);
   \draw[blue,line width=.7pt](.75,-.9)[bend left=30]to(1.3,-.6);

    \draw[magenta,line width=.8pt](-.75,-.8)[bend right=30]to(-1.2,-.5);
   \draw[blue,line width=.7pt](-.75,-.9)[bend right=30]to(-1.3,-.6);

  \disctwod{0}{1.8}{.3};
  \disctwod{0.7}{1.3}{.3};
  \disctwod{-0.7}{1.3}{.3};

  \disctwod{1.25}{.4}{.3};
  \disctwod{-1.25}{.4}{.3};
  \disctwod{-1.3}{-.4}{.3};
  \disctwod{1.3}{-.4}{.3};

    \disctwod{.5}{-.8}{.3};
    \disctwod{-.5}{-.8}{.3};

  \end{scope}
}

\newcommand{\isingMGreducedopen}[1]{
  \begin{scope}
  \pgfmathsetmacro{\dy}{sqrt(3)*0.6/2} 

  \begin{scope}[shift={(-\dy,0)}]
    \hexcellreduced{.6}
  \end{scope}

  \begin{scope}[shift={(\dy,0)}]
    \hexcellreduced{.6}
  \end{scope}

  \begin{scope}[shift={(0,0.6*1.5)}]
    \hexcellreduced{0.6}
  \end{scope}

  \draw[blue,line width=.7pt](0.25,1.7)[bend right=30]to(.55,1.45);
  \draw[blue,line width=.7pt](-0.25,1.7)[bend left=30]to(-.55,1.45);
  \draw[magenta,line width=.8pt](0.2,1.6)[bend right=30]to(.5,1.4);
  \draw[magenta,line width=.8pt](-0.2,1.6)[bend left=30]to(-.5,1.4);

  \draw[magenta,line width=.8pt](0.65,1.2)[bend right=30]to(1,.5);

  \draw[magenta,line width=.8pt](-0.65,1.2)[bend left=30]to(-1,.5);

   \draw[blue,line width=.7pt](0.75,1.2)[bend right=30]to(1.05,.6);
   \draw[blue,line width=.7pt](-0.75,1.2)[bend left=30]to(-1.05,.6);

    \draw[magenta,line width=.8pt](1.3,.2)[bend right=30]to(1.3,-.2);
   \draw[blue,line width=.7pt](1.4,.2)[bend right=30]to(1.4,-.2);

   \draw[magenta,line width=.8pt](-1.3,.2)[bend left=30]to(-1.3,-.2);
   \draw[blue,line width=.7pt](-1.4,.2)[bend left=30]to(-1.4,-.2);

   \draw[magenta,line width=.8pt](.4,-.8)[bend right=30]to(-.4,-.8);
   \draw[blue,line width=.7pt](.4,-.9)[bend right=30]to(-.4,-.9);

    \draw[magenta,line width=.8pt](.75,-.8)[bend left=30]to(1.2,-.5);
   \draw[blue,line width=.7pt](.75,-.9)[bend left=30]to(1.3,-.6);

    \draw[magenta,line width=.8pt](-.75,-.8)[bend right=30]to(-1.2,-.5);
   \draw[blue,line width=.7pt](-.75,-.9)[bend right=30]to(-1.3,-.6);

  \disctwod{0}{1.8}{.3};
  \disctwod{0.7}{1.3}{.3};
  \disctwod{-0.7}{1.3}{.3};

  \disctwod{1.25}{.4}{.3};
  \disctwod{-1.25}{.4}{.3};
  \disctwod{-1.3}{-.4}{.3};
  \disctwod{1.3}{-.4}{.3};

    \disctwod{.5}{-.8}{.3};
    \disctwod{-.5}{-.8}{.3};

  \end{scope}
}

\newcommand{\isingMGreduced}[1]{
  \begin{scope}
  \pgfmathsetmacro{\dy}{sqrt(3)*#1/2} 

  \begin{scope}[shift={(-\dy,0)}]
    \hexcellreduced{#1}
  \end{scope}

  \begin{scope}[shift={(\dy,0)}]
    \hexcellreduced{#1}
  \end{scope}

  \begin{scope}[shift={(0,#1*1.5)}]
    \hexcellreduced{#1}
  \end{scope}

  \end{scope}
}

\newcommand{\threequbitroundtensor}[3]{
\begin{scope}[scale=#3]

\pgfmathsetmacro{\R}{2} 

\pgfmathsetmacro{\xa}{0}
\pgfmathsetmacro{\ya}{\R}

\pgfmathsetmacro{\xb}{\R*cos(210)}
\pgfmathsetmacro{\yb}{\R*sin(210)}

\pgfmathsetmacro{\xc}{\R*cos(-30)}
\pgfmathsetmacro{\yc}{\R*sin(-30)}

\diskuprev{\xa-#1}{\ya}{#1}{#2}{0}{up};
\foreach \i in {0,1,2,3,4,5}{
}

\diskuprev{\xb-#1}{\yb}{#1}{#2}{120}{leftdown};
\foreach \i in {0,1,2,3,4,5}{
}

\diskuprev{\xc-#1}{\yc}{#1}{#2}{-120}{rightdown};
\foreach \i in {0,1,2,3,4,5}{
}

\pgfmathsetmacro{\xbb}{\xb+#1*cos(120)}
\pgfmathsetmacro{\ybb}{\yb+#1*sin(120)}
\draw[blue] (\xa-#1,\ya) to[out=-90,in=30] (\xbb,\ybb);

\pgfmathsetmacro{\xbb}{\xb-#1*cos(120)}
\pgfmathsetmacro{\ybb}{\yb-#1*sin(120)}
\pgfmathsetmacro{\xcc}{\xc+#1*cos(-120)}
\pgfmathsetmacro{\ycc}{\yc+#1*sin(-120)}
\draw[blue] (\xbb,\ybb) to[out=30,in=150] (\xcc,\ycc);

\pgfmathsetmacro{\xcc}{\xc+#1*cos(60)}
\pgfmathsetmacro{\ycc}{\yb+#1*sin(60)}
\draw[blue] (\xcc,\ycc) to[out=150,in=-90] (\xa+#1,\ya);

\end{scope}
}


\newcommand{\maxstate}[3]{
\tikz[baseline=(current bounding box.center), line width=0.3mm]{
\begin{scope}[scale=.7]
\threequbitroundtensor{#1}{#2}{#3};
\draw[magenta] (up1) to[out=270, in=30] (leftdown4);
\draw[magenta] (up4) to[out=270, in=150] (rightdown1);
\draw[magenta] (leftdown1) to[out=30, in=150] (rightdown4);

\draw (up2) to[out=270, in=30] (leftdown3);
\draw (up3) to[out=270, in=150] (rightdown2);
\draw (leftdown2) to[out=30, in=150] (rightdown3);
\end{scope}
}
}

\newcommand{\wstate}[3]{
\tikz[baseline=(current bounding box.center), line width=0.3mm] {
\begin{scope}[scale=1]
\threequbitroundtensor{#1}{#2}{#3};
\draw[magenta] (up1) to[out=270, in=30] (leftdown4);
\draw[magenta] (up4) to[out=270, in=150] (rightdown1);
\draw[magenta] (leftdown1) to[out=30, in=150] (rightdown4);

\draw (up3) to[out=300, in=0] (leftdown2);
\draw (up2) to[out=240, in=180] (rightdown3);
\draw (leftdown3) to[out=60, in=120] (rightdown2);

\path ($ (leftdown1)!0.5!(up4) $)++(-.3,0.25) node {\small $e^{i\frac{\pi}{3}}$};
\path ($ (leftdown4)!0.5!(rightdown1) $)++(0,-.35) node {\small $e^{i\frac{\pi}{3}}$};
\path ($ (up1)!0.5!(rightdown4) $)++(.35,0.25) node {\small $e^{i\frac{\pi}{3}}$};
\end{scope}
}
}

\newcommand{\chargerelations}[2]{
\tikz[baseline=(current bounding box.center), line width=0.3mm] {
\begin{scope}[yscale=.8]
    \pgfmathsetmacro{\x}{0.5}
    \pgfmathsetmacro{\y}{1}
    \pgfmathsetmacro{\xx}{0.3}

     \ifthenelse{\equal{#2}{leftup}}{
     \draw[red, decorate,decoration={snake}](-\x*0.5,\y*0.5)--++(-1,0);
      \fill[red] (-\x*0.5,\y*0.5) circle (2pt);
    }{
      \ifthenelse{\equal{#2}{leftdown}}{
      \draw[red, decorate,decoration={snake}](-\x*0.5,-\y*0.5)--++(-1,0);
        \fill[red] (-\x*0.5,-\y*0.5) circle (2pt);
      }{
        \ifthenelse{\equal{#2}{rightup}}{
        \draw[red, decorate,decoration={snake}](\x*0.5,\y*0.5)--++(-1,0);
          \fill[red] (\x*0.5,\y*0.5) circle (2pt);
        }{
          \ifthenelse{\equal{#2}{rightdown}}{
          \draw[red, decorate,decoration={snake}](\x*0.5,-\y*0.5)--++(-1,0);
            \fill[red] (\x*0.5,-\y*0.5) circle (2pt);
          }{
          }
        }
        }
        }
    \draw[white,WLL] (-\x,\y) -- (\x,-\y);
    \draw[white,WLL] (\x,\y) -- (-\x,-\y);
    \draw (-\x,\y) -- (\x,-\y);
    \draw (\x,\y) -- (-\x,-\y);
    \node at (-\xx,0) {$#1$};
    
    \ifthenelse{\equal{#2}{leftup}}{
      \fill (-\x*0.5,\y*0.5) circle (2pt);
    }{
      \ifthenelse{\equal{#2}{leftdown}}{
        \fill (-\x*0.5,-\y*0.5) circle (2pt);
      }{
        \ifthenelse{\equal{#2}{rightup}}{
          \fill (\x*0.5,\y*0.5) circle (2pt);
        }{
          \ifthenelse{\equal{#2}{rightdown}}{
            \fill (\x*0.5,-\y*0.5) circle (2pt);
          }{
          }
        }
        }
        }
\end{scope}
}
}


\newcommand{\maxTN}{
\tikz[baseline=(current bounding box.center), line width=0.3mm] {
    \node[draw, circle, thick, minimum size=1.2cm] (max) at (0,0) {$Max$};
    \draw[thick, blue] (max) -- ++(90:2);
    \draw[thick, blue] (max) -- ++(210:2);
    \draw[thick, blue] (max) -- ++(-30:2);

}
}

\newcommand{\threewTN}{
\tikz[baseline=(current bounding box.center), line width=0.3mm] {
    \node[draw, circle, thick, minimum size=1.2cm] (top) at (0,1) {$w$};
    \node[draw, circle, thick, minimum size=1.2cm] (left) at (-1*0.866, -1*0.5) {$w$};
    \node[draw, circle, thick, minimum size=1.2cm] (right) at (1*0.866, -1*0.5) {$w$};
    
    \draw[thick, blue] (top) -- (left);
    \draw[thick, blue] (top) -- (right);
    \draw[thick, blue] (left) -- (right);
    
    \draw[thick, blue] (top) -- ++(90:1.2);
    \draw[thick, blue] (left) -- ++(210:1.2);
    \draw[thick, blue] (right) -- ++(330:1.2);
}
}

\newcommand{\wTN}{
\tikz[baseline=(current bounding box.center), line width=0.3mm] {
    \node[draw, circle, thick, minimum size=1.2cm] (w) at (0,0) {$w$};
    \draw[thick, blue] (w) -- ++(90:2);
    \draw[thick, blue] (w) -- ++(210:2);
    \draw[thick, blue] (w) -- ++(-30:2);

}
}

\newcommand{\threemaxTN}{
\tikz[baseline=(current bounding box.center), line width=0.3mm] {
  \node[draw, circle, thick, minimum size=1.2cm] (top) at (0,2) {$Max$};
  \node[draw, circle, thick, minimum size=1.2cm] (left) at (-2*0.866,-2*0.5) {$Max$};
  \node[draw, circle, thick, minimum size=1.2cm] (right) at (2*0.866,-2*0.5) {$Max$};

  \path ($ (left)!0.5!(top) $) node[draw, circle, thick, minimum size=0.9cm](tl) {$e^{i\frac{\pi}{6}Z}$};
  \path ($ (right)!0.5!(top) $) node[draw, circle, thick, minimum size=0.9cm](tr) {$e^{i\frac{\pi}{6}Z}$};
  \path ($ (left)!0.5!(right) $) node[draw, circle, thick, minimum size=0.9cm](bottom) {$e^{i\frac{\pi}{6}Z}$};

  \draw[blue, thick] (top) -- (tl);
  \draw[blue, thick] (top) -- (tr);
  \draw[blue, thick] (tl) -- (left);
  \draw[blue, thick] (tr) -- (right);
  \draw[blue, thick] (left) -- (bottom);
  \draw[blue, thick] (right) -- (bottom);

  \draw[blue, thick] (top) -- ++(90:1.2);
  \draw[blue, thick] (left) -- ++(210:1.2);
  \draw[blue, thick] (right) -- ++(330:1.2);
}
}


\newcommand{\MGket}[5]{
  \pgfmathsetmacro{\rx}{#1}
  \pgfmathsetmacro{\ry}{#2}
  \pgfmathsetmacro{\h}{#3}

  \begin{scope}[scale=1, shift={(#4,#5)}]
    \diskup{0}{0}{\rx}{\ry}{0}{P}
    \coordinate(P1up) at ($(P1)+(0,\ry)$);
    \coordinate(P2up) at ($(P2)+(0,\ry)$);
    \coordinate(P3up) at ($(P3)+(0,\ry)$);
    \coordinate(P4up) at ($(P4)+(0,\ry)$);
    \draw[magenta](P1)--(P1up)[bend left=90]to(P4up)--(P4);
    \draw(P2)--(P2up)[bend left=90]to(P3up)--(P3);

    \draw[blue](0,0)arc[start angle=180,end angle=0,radius=\rx];

  \end{scope}
}

\newcommand{\CLketzero}[5]{
  \pgfmathsetmacro{\rx}{#1}
  \pgfmathsetmacro{\ry}{#2}
  \pgfmathsetmacro{\h}{#3}

  \begin{scope}[scale=1, shift={(#4,#5)}]
    \diskup{0}{0}{\rx}{\ry}{0}{P}
    \coordinate(P1up) at ($(P1)+(0,\ry)$);
    \coordinate(P2up) at ($(P2)+(0,\ry)$);
    \coordinate(P3up) at ($(P3)+(0,\ry)$);
    \coordinate(P4up) at ($(P4)+(0,\ry)$);
    \draw(P1)--(P1up)[bend left=90]to(P2up)--(P2);
    \draw(P3)--(P3up)[bend left=90]to(P4up)--(P4);

    \draw[blue](0,0)arc[start angle=180,end angle=0,radius=\rx];

  \end{scope}
}

\newcommand{\CLketone}[5]{
  \pgfmathsetmacro{\rx}{#1}
  \pgfmathsetmacro{\ry}{#2}
  \pgfmathsetmacro{\h}{#3}

  \begin{scope}[scale=1, shift={(#4,#5)}]
    \diskup{0}{0}{\rx}{\ry}{0}{P}
    \coordinate(P1up) at ($(P1)+(0,\ry)$);
    \coordinate(P2up) at ($(P2)+(0,\ry)$);
    \coordinate(P3up) at ($(P3)+(0,\ry)$);
    \coordinate(P4up) at ($(P4)+(0,\ry)$);
    \draw(P1)--(P1up)[bend left=90]to(P2up)--(P2);
    \draw(P3)--(P3up)[bend left=90]to(P4up)--(P4);
    \draw[red, decorate, decoration={snake}](P2up)--(P3up);
    \fill(P2up) circle(0.05);
    \fill(P3up) circle(0.05);

    \draw[blue](0,0)arc[start angle=180,end angle=0,radius=\rx];

  \end{scope}
}

\newcommand{\MGbra}[5]{
  \pgfmathsetmacro{\rx}{#1}
  \pgfmathsetmacro{\ry}{#2}
  \pgfmathsetmacro{\h}{#3}

  \begin{scope}[scale=1, shift={(#4,#5)}]
    \diskup{0}{0}{\rx}{\ry}{0}{P}
    \coordinate(P1down) at ($(P1)+(0,-0.3*\ry)$);
    \coordinate(P2down) at ($(P2)+(0,-0.3*\ry)$);
    \coordinate(P3down) at ($(P3)+(0,-0.3*\ry)$);
    \coordinate(P4down) at ($(P4)+(0,-0.3*\ry)$);
    \draw[magenta](P1)--(P1down)[bend right=90]to(P4down)--(P4);
    \draw(P2)--(P2down)[bend right=90]to(P3down)--(P3);

    \draw[blue](0,0)arc[start angle=180,end angle=360,radius=\rx];

  \end{scope}
}

\newcommand{\CLbra}[5]{
  \pgfmathsetmacro{\rx}{#1}
  \pgfmathsetmacro{\ry}{#2}
  \pgfmathsetmacro{\h}{#3}

  \begin{scope}[scale=1, shift={(#4,#5)}]
    \diskup{0}{0}{\rx}{\ry}{0}{P}
    \coordinate(P1down) at ($(P1)+(0,-0.3*\ry)$);
    \coordinate(P2down) at ($(P2)+(0,-0.3*\ry)$);
    \coordinate(P3down) at ($(P3)+(0,-0.3*\ry)$);
    \coordinate(P4down) at ($(P4)+(0,-0.3*\ry)$);
    \draw(P1)--(P1down)[bend right=90]to(P2down)--(P2);
    \draw(P3)--(P3down)[bend right=90]to(P4down)--(P4);

    \draw[blue](0,0)arc[start angle=180,end angle=360,radius=\rx];

    \draw[red,decorate,decoration={snake}](P3down)--(P2down);
    \fill(P3down) circle(0.05);
    \fill(P2down) circle(0.05);

  \end{scope}
}

\newcommand{\MGonequbitgate}[5]{
\pgfmathsetmacro{\rx}{#1}
\pgfmathsetmacro{\ry}{#2}
\pgfmathsetmacro{\h}{#3}
\pgfmathsetmacro{\gap}{2*\rx}  
\begin{scope}[scale=1, shift={(#4,#5)}]
    \diskup{0}{\h}{\rx}{\ry}{0}{P1}

    \diskdown{0}{0}{\rx}{\ry}{0}{P2}

    \draw[blue] (0,0) -- (0,\h);
    \draw[blue] (\gap,0) -- (\gap,\h); 

    \draw (P11) [magenta]--(P21);
    \draw (P12)--(P23);
    \draw (P13)--(P22);
    \draw(P14)[magenta]--(P24);
\end{scope}
}

\newcommand{\CLrxgate}[5]{
\pgfmathsetmacro{\rx}{#1}
\pgfmathsetmacro{\ry}{#2}
\pgfmathsetmacro{\h}{#3}
\pgfmathsetmacro{\gap}{2*\rx}  
\begin{scope}[scale=1, shift={(#4,#5)}]
    \diskup{0}{\h}{\rx}{\ry}{0}{P1}

    \diskdown{0}{0}{\rx}{\ry}{0}{P2}

    \draw[blue] (0,0) -- (0,\h);
    \draw[blue] (\gap,0) -- (\gap,\h); 

    \draw (P11)--(P21);
    
    \draw (P13)--(P22);
    \draw[white,WL] (P12)--(P23);
    \draw (P12)--(P23);
    \draw(P14)--(P24);
\end{scope}
}
\newcommand{\CLrxgatetwo}[5]{
\pgfmathsetmacro{\rx}{#1}
\pgfmathsetmacro{\ry}{#2}
\pgfmathsetmacro{\h}{#3}
\pgfmathsetmacro{\gap}{2*\rx}  
\begin{scope}[scale=1, shift={(#4,#5)}]
    \diskup{0}{\h}{\rx}{\ry}{0}{P1}

    \diskdown{0}{0}{\rx}{\ry}{0}{P2}

    \draw[blue] (0,0) -- (0,\h);
    \draw[blue] (\gap,0) -- (\gap,\h); 

    \draw (P11)--(P21);
    
    \draw (P13)--(P22);
    \draw[white,WL] (P12)--(P23);
    \draw (P12)--(P23);
    \draw[red](P14)--(P24);
\end{scope}
}

\newcommand{\CLIDgate}[5]{
\pgfmathsetmacro{\rx}{#1}
\pgfmathsetmacro{\ry}{#2}
\pgfmathsetmacro{\h}{#3}
\pgfmathsetmacro{\gap}{2*\rx}  
\begin{scope}[scale=1, shift={(#4,#5)}]
    \diskup{0}{\h}{\rx}{\ry}{0}{P1}

    \diskdown{0}{0}{\rx}{\ry}{0}{P2}

    \draw[blue] (0,0) -- (0,\h);
    \draw[blue] (\gap,0) -- (\gap,\h); 

    \draw (P11)--(P21);
    
    \draw (P13)--(P23);
    \draw (P12)--(P22);
    \draw(P14)--(P24);
\end{scope}
}
\newcommand{\CLIDgatefour}[5]{
\pgfmathsetmacro{\rx}{#1}
\pgfmathsetmacro{\ry}{#2}
\pgfmathsetmacro{\h}{#3}
\pgfmathsetmacro{\gap}{2*\rx}  
\begin{scope}[scale=1, shift={(#4,#5)}]
    \diskup{0}{\h}{\rx}{\ry}{0}{P1}

    \diskdown{0}{0}{\rx}{\ry}{0}{P2}

    \draw[blue] (0,0) -- (0,\h);
    \draw[blue] (\gap,0) -- (\gap,\h); 

    \draw (P11)--(P21);
    
    \draw[red] (P13)--(P23);
    \draw (P12)--(P22);
    \draw(P14)--(P24);
\end{scope}
}

\newcommand{\CLHgate}[5]{
\pgfmathsetmacro{\rx}{#1}
\pgfmathsetmacro{\ry}{#2}
\pgfmathsetmacro{\h}{#3}
\pgfmathsetmacro{\gap}{2*\rx}  
\begin{scope}[scale=1, shift={(#4,#5)}]
    \diskup{0}{\h}{\rx}{\ry}{0}{P1}

    \diskdown{0}{0}{\rx}{\ry}{0}{P2}

    \draw[blue] (0,0) -- (0,\h);
    \draw[blue] (\gap,0) -- (\gap,\h); 

    \draw (P11)--(P21);
    
    \draw (P13)--(P22);
    \draw[white,WL] (P12)--(P23);
    \draw (P12)--(P23);
    \draw(P14)--(P24);
\end{scope}
}

\newcommand{\MGtwoqubitgate}[5]{

\begin{scope}[scale=1, shift={(#4,#5)}]
\pgfmathsetmacro{\rx}{#1}
\pgfmathsetmacro{\ry}{#2}
\pgfmathsetmacro{\h}{#3}
\pgfmathsetmacro{\gap}{2*\rx}  
\pgfmathsetmacro{\offset}{3*\rx} 
\pgfmathsetmacro{\hy}{\h*0.1} %

    \diskup{0}{\h}{\rx}{\ry}{0}{P1}

    \diskdown{0}{0}{\rx}{\ry}{0}{P2}
    
    \diskup{\offset}{\h}{\rx}{\ry}{0}{P3}
    
    \diskdown{\offset}{0}{\rx}{\ry}{0}{P4}

    \draw[blue] (0,0) -- (0,\h);
    \draw[blue] (\offset+\gap,0) -- (\offset+\gap,\h);
    
    \draw[blue] (\gap,\h)--++(0,-\hy) arc[start angle=180,end angle=360,x radius=(\offset-\gap)*0.5, y radius=\ry]
    --++(0,\hy);
    \draw[blue] (\gap,0)--++(0,\hy) arc[start angle=180,end angle=0,x radius=(\offset-\gap)*0.5, y radius=\ry]
    --++(0,-\hy);
    
    \draw(P14)[magenta]arc[start angle=180,end angle=360,radius=.6];
    
    \draw(P24)[magenta]--++(0,.3)arc[start angle=180,end angle=0,radius=.6]--++(0,-.3);

    \draw(P23)--++(0,.7)--++(1.58,1.3)--(P32);
    \draw(P13)--++(0,-.3)--++(1.58,-1.4)--(P42);
    \draw(P11)[magenta]--(P21);
    \draw(P12)--(P22);
    \draw(P33)--(P43);
    \draw(P34)[magenta]--(P44);
\end{scope}
}

\newcommand{\CLtwoqubitgate}[5]{

\begin{scope}[scale=1, shift={(#4,#5)}]
\pgfmathsetmacro{\rx}{#1}
\pgfmathsetmacro{\ry}{#2}
\pgfmathsetmacro{\h}{#3}
\pgfmathsetmacro{\gap}{2*\rx}  
\pgfmathsetmacro{\offset}{3*\rx} 
\pgfmathsetmacro{\hy}{\h*0.1} %

    \diskup{0}{\h}{\rx}{\ry}{0}{P1}

    \diskdown{0}{0}{\rx}{\ry}{0}{P2}
    
    \diskup{\offset}{\h}{\rx}{\ry}{0}{P3}
    
    \diskdown{\offset}{0}{\rx}{\ry}{0}{P4}

    \draw[blue] (0,0) -- (0,\h);
    \draw[blue] (\offset+\gap,0) -- (\offset+\gap,\h);
    
    \draw[blue] (\gap,\h)--++(0,-\hy) arc[start angle=180,end angle=360,x radius=(\offset-\gap)*0.5, y radius=\ry]
    --++(0,\hy);
    \draw[blue] (\gap,0)--++(0,\hy) arc[start angle=180,end angle=0,x radius=(\offset-\gap)*0.5, y radius=\ry]
    --++(0,-\hy);
    
    \draw(P14)arc[start angle=180,end angle=360,radius=.6];
    
    \draw(P24)--++(0,.3)arc[start angle=180,end angle=0,radius=.6]--++(0,-.3);
    
    \coordinate(P42P) at ($(P42)+(0,2*\ry)$);
    \coordinate(P32P) at ($(P32)+(0,-\ry)$);
    
    \draw(P23)--++(0,.7)--(P32P)--(P32);
    \draw[white,WL](P13)--++(0,-.3)--(P42P)--(P42);
    \draw(P13)--++(0,-.3)--(P42P)--(P42);
    \draw(P11)--(P21);
    \draw(P12)--(P22);
    \draw(P33)--(P43);
    \draw(P34)--(P44);
\end{scope}
}

\newcommand{\CLtwoqubitgatefour}[5]{

\begin{scope}[scale=1, shift={(#4,#5)}]
\pgfmathsetmacro{\rx}{#1}
\pgfmathsetmacro{\ry}{#2}
\pgfmathsetmacro{\h}{#3}
\pgfmathsetmacro{\gap}{2*\rx}  
\pgfmathsetmacro{\offset}{3*\rx} 
\pgfmathsetmacro{\hy}{\h*0.1} %

    \diskup{0}{\h}{\rx}{\ry}{0}{P1}

    \diskdown{0}{0}{\rx}{\ry}{0}{P2}
    
    \diskup{\offset}{\h}{\rx}{\ry}{0}{P3}
    
    \diskdown{\offset}{0}{\rx}{\ry}{0}{P4}

    \draw[blue] (0,0) -- (0,\h);
    \draw[blue] (\offset+\gap,0) -- (\offset+\gap,\h);
    
    \draw[blue] (\gap,\h)--++(0,-\hy) arc[start angle=180,end angle=360,x radius=(\offset-\gap)*0.5, y radius=\ry]
    --++(0,\hy);
    \draw[blue] (\gap,0)--++(0,\hy) arc[start angle=180,end angle=0,x radius=(\offset-\gap)*0.5, y radius=\ry]
    --++(0,-\hy);
    
    \draw(P14)arc[start angle=180,end angle=360,radius=.6];
    
    \draw(P24)--++(0,.3)arc[start angle=180,end angle=0,radius=.6]--++(0,-.3);
    
    \coordinate(P42P) at ($(P42)+(0,2*\ry)$);
    \coordinate(P32P) at ($(P32)+(0,-\ry)$);
    
    \draw(P23)--++(0,.7)--(P32P)--(P32);
    \draw[white,WL](P13)--++(0,-.3)--(P42P)--(P42);
    \draw(P13)--++(0,-.3)--(P42P)--(P42);
    \draw(P11)--(P21);
    \draw[red](P12)--(P22);
    \draw(P33)--(P43);
    \draw(P34)--(P44);
\end{scope}
}

\newcommand{\CLtwoqubitgatethree}[5]{

\begin{scope}[scale=1, shift={(#4,#5)}]
\pgfmathsetmacro{\rx}{#1}
\pgfmathsetmacro{\ry}{#2}
\pgfmathsetmacro{\h}{#3}
\pgfmathsetmacro{\gap}{2*\rx}  
\pgfmathsetmacro{\offset}{3*\rx} 
\pgfmathsetmacro{\hy}{\h*0.1} %

    \diskup{0}{\h}{\rx}{\ry}{0}{P1}

    \diskdown{0}{0}{\rx}{\ry}{0}{P2}
    
    \diskup{\offset}{\h}{\rx}{\ry}{0}{P3}
    
    \diskdown{\offset}{0}{\rx}{\ry}{0}{P4}

    \draw[blue] (0,0) -- (0,\h);
    \draw[blue] (\offset+\gap,0) -- (\offset+\gap,\h);
    
    \draw[blue] (\gap,\h)--++(0,-\hy) arc[start angle=180,end angle=360,x radius=(\offset-\gap)*0.5, y radius=\ry]
    --++(0,\hy);
    \draw[blue] (\gap,0)--++(0,\hy) arc[start angle=180,end angle=0,x radius=(\offset-\gap)*0.5, y radius=\ry]
    --++(0,-\hy);
    
    \draw(P14)arc[start angle=180,end angle=360,radius=.6];
    
    \draw(P24)--++(0,.3)arc[start angle=180,end angle=0,radius=.6]--++(0,-.3);
    
    \coordinate(P42P) at ($(P42)+(0,2*\ry)$);
    \coordinate(P32P) at ($(P32)+(0,-\ry)$);
    
    \draw[red](P23)--++(0,.7)--(P32P)--(P32);
    \draw[white,WL](P13)--++(0,-.3)--(P42P)--(P42);
    \draw[red](P13)--++(0,-.3)--(P42P)--(P42);
    \draw(P11)--(P21);
    \draw(P12)--(P22);
    \draw(P33)--(P43);
    \draw(P34)--(P44);
\end{scope}
}

\newcommand{\CLtwoqubitgatetwo}[5]{

\begin{scope}[scale=1, shift={(#4,#5)}]
\pgfmathsetmacro{\rx}{#1}
\pgfmathsetmacro{\ry}{#2}
\pgfmathsetmacro{\h}{#3}
\pgfmathsetmacro{\gap}{2*\rx}  
\pgfmathsetmacro{\offset}{3*\rx} 
\pgfmathsetmacro{\hy}{\h*0.1} %

    \diskup{0}{\h}{\rx}{\ry}{0}{P1}

    \diskdown{0}{0}{\rx}{\ry}{0}{P2}
    
    \diskup{\offset}{\h}{\rx}{\ry}{0}{P3}
    
    \diskdown{\offset}{0}{\rx}{\ry}{0}{P4}

    \draw[blue] (0,0) -- (0,\h);
    \draw[blue] (\offset+\gap,0) -- (\offset+\gap,\h);
    
    \draw[blue] (\gap,\h)--++(0,-\hy) arc[start angle=180,end angle=360,x radius=(\offset-\gap)*0.5, y radius=\ry]
    --++(0,\hy);
    \draw[blue] (\gap,0)--++(0,\hy) arc[start angle=180,end angle=0,x radius=(\offset-\gap)*0.5, y radius=\ry]
    --++(0,-\hy);
    
    \draw(P14)arc[start angle=180,end angle=360,radius=.6];
    
    \draw(P24)--++(0,.3)arc[start angle=180,end angle=0,radius=.6]--++(0,-.3);
    
    \coordinate(P42P) at ($(P42)+(0,2*\ry)$);
    \coordinate(P32P) at ($(P32)+(0,-\ry)$);
    
    \draw(P23)--++(0,.7)--(P32P)--(P32);
    \draw[white,WL](P13)--++(0,-.3)--(P42P)--(P42);
    \draw(P13)--++(0,-.3)--(P42P)--(P42);
    \draw[red](P11)--(P21);
    \draw(P12)--(P22);
    \draw(P33)--(P43);
    \draw(P34)--(P44);
\end{scope}
}

\newcommand{\CLtwoqubitgateone}[5]{

\begin{scope}[scale=1, shift={(#4,#5)}]
\pgfmathsetmacro{\rx}{#1}
\pgfmathsetmacro{\ry}{#2}
\pgfmathsetmacro{\h}{#3}
\pgfmathsetmacro{\gap}{2*\rx}  
\pgfmathsetmacro{\offset}{3*\rx} 
\pgfmathsetmacro{\hy}{\h*0.1} %

    \diskup{0}{\h}{\rx}{\ry}{0}{P1}

    \diskdown{0}{0}{\rx}{\ry}{0}{P2}
    
    \diskup{\offset}{\h}{\rx}{\ry}{0}{P3}
    
    \diskdown{\offset}{0}{\rx}{\ry}{0}{P4}

    \draw[blue] (0,0) -- (0,\h);
    \draw[blue] (\offset+\gap,0) -- (\offset+\gap,\h);
    
    \draw[blue] (\gap,\h)--++(0,-\hy) arc[start angle=180,end angle=360,x radius=(\offset-\gap)*0.5, y radius=\ry]
    --++(0,\hy);
    \draw[blue] (\gap,0)--++(0,\hy) arc[start angle=180,end angle=0,x radius=(\offset-\gap)*0.5, y radius=\ry]
    --++(0,-\hy);
    
    \draw(P14)arc[start angle=180,end angle=360,radius=.6];
    
    \draw[red](P24)--++(0,.3)arc[start angle=180,end angle=0,radius=.6]--++(0,-.3);
    
    \coordinate(P42P) at ($(P42)+(0,2*\ry)$);
    \coordinate(P32P) at ($(P32)+(0,-\ry)$);
    
    \draw(P23)--++(0,.7)--(P32P)--(P32);
    \draw[white,WL](P13)--++(0,-.3)--(P42P)--(P42);
    \draw(P13)--++(0,-.3)--(P42P)--(P42);
    \draw(P11)--(P21);
    \draw(P12)--(P22);
    \draw(P33)--(P43);
    \draw(P34)--(P44);
\end{scope}
}

\newcommand{\CLczgate}[5]{

\begin{scope}[scale=1, shift={(#4,#5)}]
\pgfmathsetmacro{\rx}{#1}
\pgfmathsetmacro{\ry}{#2}
\pgfmathsetmacro{\h}{#3}
\pgfmathsetmacro{\gap}{2*\rx}  
\pgfmathsetmacro{\offset}{3*\rx} 
\pgfmathsetmacro{\hy}{\h*0.1} %

    \diskup{0}{\h}{\rx}{\ry}{0}{P1}

    \diskdown{0}{0}{\rx}{\ry}{0}{P2}
    
    \diskup{\offset}{\h}{\rx}{\ry}{0}{P3}
    
    \diskdown{\offset}{0}{\rx}{\ry}{0}{P4}

    \draw[blue] (0,0) -- (0,\h);
    \draw[blue] (\offset+\gap,0) -- (\offset+\gap,\h);
    
    \draw[blue] (\gap,\h)--++(0,-\hy) arc[start angle=180,end angle=360,x radius=(\offset-\gap)*0.5, y radius=\ry]
    --++(0,\hy);
    \draw[blue] (\gap,0)--++(0,\hy) arc[start angle=180,end angle=0,x radius=(\offset-\gap)*0.5, y radius=\ry]
    --++(0,-\hy);
    
    \coordinate(P23P) at ($(P23)+(0,3*\ry)$);
    \coordinate(P14P) at ($(P14)+(0,-1.5*\ry)$);
    \coordinate(P42P) at ($(P42)+(0,3*\ry)$);
    \coordinate(P31P) at ($(P31)+(0,-1.5*\ry)$);

    \draw(P14)--(P14P)--(P42P)--(P42);
    
    \draw(P11)--(P21);
    \draw(P12)--(P22);
    \draw(P33)--(P43);
    \draw(P34)--(P44);

    \draw[white,WL](P13)[bend right=60]to(P32);
    \draw(P13)[bend right=60]to(P32);

    \draw[white,WL](P23)--(P23P)--(P31P)--(P31);
    \draw(P23)--(P23P)--(P31P)--(P31);
    
    \draw(P24)--++(0,.3)arc[start angle=180,end angle=0,radius=.6]--++(0,-.3);
\end{scope}
}

\newcommand{\CLczgatethree}[5]{

\begin{scope}[scale=1, shift={(#4,#5)}]
\pgfmathsetmacro{\rx}{#1}
\pgfmathsetmacro{\ry}{#2}
\pgfmathsetmacro{\h}{#3}
\pgfmathsetmacro{\gap}{2*\rx}  
\pgfmathsetmacro{\offset}{3*\rx} 
\pgfmathsetmacro{\hy}{\h*0.1} %

    \diskup{0}{\h}{\rx}{\ry}{0}{P1}

    \diskdown{0}{0}{\rx}{\ry}{0}{P2}
    
    \diskup{\offset}{\h}{\rx}{\ry}{0}{P3}
    
    \diskdown{\offset}{0}{\rx}{\ry}{0}{P4}

    \draw[blue] (0,0) -- (0,\h);
    \draw[blue] (\offset+\gap,0) -- (\offset+\gap,\h);
    
    \draw[blue] (\gap,\h)--++(0,-\hy) arc[start angle=180,end angle=360,x radius=(\offset-\gap)*0.5, y radius=\ry]
    --++(0,\hy);
    \draw[blue] (\gap,0)--++(0,\hy) arc[start angle=180,end angle=0,x radius=(\offset-\gap)*0.5, y radius=\ry]
    --++(0,-\hy);
    
    \coordinate(P23P) at ($(P23)+(0,3*\ry)$);
    \coordinate(P14P) at ($(P14)+(0,-1.5*\ry)$);
    \coordinate(P42P) at ($(P42)+(0,3*\ry)$);
    \coordinate(P31P) at ($(P31)+(0,-1.5*\ry)$);

    \draw[red](P14)--(P14P)--(P42P)--(P42);
    
    \draw(P11)--(P21);
    \draw(P12)--(P22);
    \draw(P33)--(P43);
    \draw(P34)--(P44);

    \draw[white,WL](P13)[bend right=60]to(P32);
    \draw(P13)[bend right=60]to(P32);

    \draw[white,WL](P23)--(P23P)--(P31P)--(P31);
    \draw[red](P23)--(P23P)--(P31P)--(P31);
    
    \draw(P24)--++(0,.3)arc[start angle=180,end angle=0,radius=.6]--++(0,-.3);
\end{scope}
}

\newcommand{\CLczgatefour}[5]{

\begin{scope}[scale=1, shift={(#4,#5)}]
\pgfmathsetmacro{\rx}{#1}
\pgfmathsetmacro{\ry}{#2}
\pgfmathsetmacro{\h}{#3}
\pgfmathsetmacro{\gap}{2*\rx}  
\pgfmathsetmacro{\offset}{3*\rx} 
\pgfmathsetmacro{\hy}{\h*0.1} %

    \diskup{0}{\h}{\rx}{\ry}{0}{P1}

    \diskdown{0}{0}{\rx}{\ry}{0}{P2}
    
    \diskup{\offset}{\h}{\rx}{\ry}{0}{P3}
    
    \diskdown{\offset}{0}{\rx}{\ry}{0}{P4}

    \draw[blue] (0,0) -- (0,\h);
    \draw[blue] (\offset+\gap,0) -- (\offset+\gap,\h);
    
    \draw[blue] (\gap,\h)--++(0,-\hy) arc[start angle=180,end angle=360,x radius=(\offset-\gap)*0.5, y radius=\ry]
    --++(0,\hy);
    \draw[blue] (\gap,0)--++(0,\hy) arc[start angle=180,end angle=0,x radius=(\offset-\gap)*0.5, y radius=\ry]
    --++(0,-\hy);
    
    \coordinate(P23P) at ($(P23)+(0,3*\ry)$);
    \coordinate(P14P) at ($(P14)+(0,-1.5*\ry)$);
    \coordinate(P42P) at ($(P42)+(0,3*\ry)$);
    \coordinate(P31P) at ($(P31)+(0,-1.5*\ry)$);

    \draw(P14)--(P14P)--(P42P)--(P42);
    
    \draw(P11)--(P21);
    \draw[red](P12)--(P22);
    \draw(P33)--(P43);
    \draw(P34)--(P44);

    \draw[white,WL](P13)[bend right=60]to(P32);
    \draw(P13)[bend right=60]to(P32);

    \draw[white,WL](P23)--(P23P)--(P31P)--(P31);
    \draw(P23)--(P23P)--(P31P)--(P31);
    
    \draw(P24)--++(0,.3)arc[start angle=180,end angle=0,radius=.6]--++(0,-.3);
\end{scope}
}

\newcommand{\CLczgatefive}[5]{

\begin{scope}[scale=1, shift={(#4,#5)}]
\pgfmathsetmacro{\rx}{#1}
\pgfmathsetmacro{\ry}{#2}
\pgfmathsetmacro{\h}{#3}
\pgfmathsetmacro{\gap}{2*\rx}  
\pgfmathsetmacro{\offset}{3*\rx} 
\pgfmathsetmacro{\hy}{\h*0.1} %

    \diskup{0}{\h}{\rx}{\ry}{0}{P1}

    \diskdown{0}{0}{\rx}{\ry}{0}{P2}
    
    \diskup{\offset}{\h}{\rx}{\ry}{0}{P3}
    
    \diskdown{\offset}{0}{\rx}{\ry}{0}{P4}

    \draw[blue] (0,0) -- (0,\h);
    \draw[blue] (\offset+\gap,0) -- (\offset+\gap,\h);
    
    \draw[blue] (\gap,\h)--++(0,-\hy) arc[start angle=180,end angle=360,x radius=(\offset-\gap)*0.5, y radius=\ry]
    --++(0,\hy);
    \draw[blue] (\gap,0)--++(0,\hy) arc[start angle=180,end angle=0,x radius=(\offset-\gap)*0.5, y radius=\ry]
    --++(0,-\hy);
    
    \coordinate(P23P) at ($(P23)+(0,3*\ry)$);
    \coordinate(P14P) at ($(P14)+(0,-1.5*\ry)$);
    \coordinate(P42P) at ($(P42)+(0,3*\ry)$);
    \coordinate(P31P) at ($(P31)+(0,-1.5*\ry)$);

    \draw(P14)--(P14P)--(P42P)--(P42);
    
    \draw(P11)--(P21);
    \draw(P12)--(P22);
    \draw(P33)--(P43);
    \draw(P34)--(P44);

    \draw[white,WL](P13)[bend right=60]to(P32);
    \draw[red](P13)[bend right=60]to(P32);

    \draw[white,WL](P23)--(P23P)--(P31P)--(P31);
    \draw(P23)--(P23P)--(P31P)--(P31);
    
    \draw(P24)--++(0,.3)arc[start angle=180,end angle=0,radius=.6]--++(0,-.3);
\end{scope}
}

\tableofcontents
\subsection{Introduction}

Classical computing is approaching the limits imposed by Moore’s law \cite{waldrop2016chips}. Quantum computing has emerged as a promising alternative, inspired by the concept of quantum simulation introduced by Manin and Feynman \cite{Manin1980,feynman2018simulating}. A landmark breakthrough was Shor’s discovery of a quantum algorithm for integer factorization, which offers exponential speedup over known classical algorithms \cite{shor1994algorithms,shor1999polynomial}. In the Noisy Intermediate Scale Quantum (NISQ) era \cite{preskill2018quantum,bharti2022noisy,chen2023complexity}, several experimental demonstrations have claimed quantum advantage \cite{arute2019quantum,zhong2020quantum,Ebadi_2022,wu2021strong,huang2022quantum,madsen2022quantum,kim2023evidence,gao2025establishing}. On the other hand, numerous research groups have developed classical simulation techniques to challenge and benchmark these experiments using high-performance classical computation \cite{clifford2018classical,moylett2019classically,zhou2020limits,huang2020classical,pan2021simulating,gray2021hyper,wu2021strong,pan2022solving,oh2022classical,anand2023classical,beguvsic2024fast,tindall2024efficient,shao2024simulating,beguvsic2024fast,gao2024automateddiscoverybranchingrules}. 

Understanding limitations of classical computation relative to quantum computation is a central question in quantum information science. Two well-known classes of quantum circuits admitting efficient classical simulation are Clifford circuits, which can be simulated using stabilizer formalism \cite{gottesman1997stabilizer,gottesman1998heisenberg,aaronson2004improved, NielsenChuang}, 
and planar Matchgates (sometimes referred to simply as Matchgates), which correspond to noninteracting fermionic systems, and can be translated to perfect matching of planar graphs and efficiently computed via Pfaffian \cite{valiant2002quantum,terhal2002classical}.
(The convention of efficiency means computations in a polynomial space and time.)
A further powerful approach involves classical simulation of states with low entanglement in a tensor network. When the entanglement entropy scales according to an area law, as is typical for gapped ground states in one and two dimensions, tensor network methods such as Matrix Product States (MPS) and Projected Entangled Pair States (PEPS) enable efficient simulation by reducing the computational cost to the system's boundary degrees of freedom \cite{eisert2010colloquium,verstraete2008matrix,orus2014practical}.

Valiant introduced the holographic reduction method to study the complexity of the counting problem in \cite{valiant2004holographic}. Motivated by this, Cai et al. developed the Holant framework to classify generating sets of low-degree tensors, allowing efficient classical simulation of the corresponding tensor networks \cite{cai2007holographic,cai2007valiant,cai2011computational}. 
Within this classification, Clifford circuits and planar Matchgates emerge as two major tractable families \cite{cai2007valiant, cai2018clifford}. 
In particular, the union of Clifford gates and planar Matchgates forms a universal gate set for quantum computation, making it crucial to understand the limitations of classical simulation methods for hybrid Clifford–Matchgate circuits.
Considerable progress has been made in extending classical simulation techniques for Clifford circuits to include a limited number of non-Clifford gates, often framed in terms of “Magic” resources relative to the Clifford group \cite{aaronson2004improved,bravyi2005universal,howard2014contextuality, amy2014polynomial,bravyi2016trading,heyfron2018efficient,seddon2019quantifying,bravyi2019simulation,kissinger2020reducing,bnb2020fast,leone2022stabilizer,liu2022many, fux2024disentangling,qian2024augmenting, ruiz2025quantum}. 
Parallel efforts have sought to generalize planar-Matchgate-based simulation methods beyond the regime of non-interacting fermionic systems \cite{little1974extension,vazirani1989nc,cimasoni2007dimers,jozsa2008matchgates, cimasoni2008dimers,cimasoni2009dimers,dijkgraaf2009dimer,bravyi2009contraction,curticapean2014counting,straub2016counting,brod2016efficient,shi2018variational,hebenstreit2019all,likhosherstov2020tractable,cudby2023gaussian,liao2023simulation,bampounis2024matchgate}, 
and to enable efficient simulation of hybrid Clifford-Matchgate circuits with particular structural constraints \cite{mishmash2023hierarchical,projansky2024entanglement,projansky2025gaussianity,huang2025augmenting}.
These developments collectively motivate a comprehensive study of classical simulation strategies in the presence of both Clifford and Matchgate operations.

A longstanding open problem in quantum computation is the development of a classical simulation framework that handles hybrid Clifford-Matchgate circuits, each of which individually admits efficient classical simulation. 
In this letter, we accomplish this goal by establishing a unified classical simulation framework for hybrid Clifford-Matchgate quantum circuits and their associated tensor networks. Our approach builds upon the three-dimensional Quon language introduced in \cite{liu2017quon}, and we refer to our framework as \textbf{Quon Classical Simulation (QCS)}.

Within the QCS framework, we compute the value of a tensor network $\Gamma$ with $n$ internal spins as

    \begin{equation}
    Z(\Gamma) =\left(\frac{1}{\sqrt{2}}\right)^{C(\Gamma)}\sum_{j=1}^{2^k}
    \operatorname{Pfaffian}(M_j(\Gamma)),
    \end{equation}
where $k$ is the number of \textbf{Magic holes} in QCS, a new notion in QCS that captures the essential exponential computational complexity. 
Here $C(\Gamma)$ is an integer (as a result of hole reductions), and $\text{Pfaffian}(M_i(\Gamma))$ is the Pfaffian of a certain matrix associated with $\Gamma$ (coming from crossing reductions), which can all be computed in polynomial time \cite{fisher1961statistical,kasteleyn1961statistics,temperley1961dimer}. 
When $k$ is $O(\log n)$, the values can be computed in polynomial time.
Both Cliffords and Matchgates have no Magic holes in QCS, namely $k=0$, therefore, both admit efficient classical simulations. More precisely, Clifford reductions mainly contribute to $\left(\frac{1}{\sqrt{2}}\right)^{C(\Gamma)}$; and Matchgate reductions mainly contribute to $\text{Pfaffian}(M_1(\Gamma))$.
These results are summarized in Fig.~\ref{fig:QCS complexity}.

\begin{figure}[H]
    \centering
\begin{tikzpicture}
\begin{scope}[scale=.9]
  \fill[blue!10, fill opacity=0.3] (0,1.5) circle (4);
  \draw[blue, thick] (0,1.4) circle (4);
  \node at (0,5.1) {QCS-EXP};
  \node at (0,4.65){(Universal Quantum Computation)};

  \fill[green!35, fill opacity=0.3] (0,1) circle (3.5);
  \draw[blue, thick] (0,1) circle (3.5);
  \node at (0,3.5) {QCS in polynomial time};

  \fill[yellow!25, fill opacity=0.3] (0,0.6) ellipse (3 and 2.5);
  \draw[blue, thick] (0,0.6) ellipse (3 and 2.5);
  \node at (0,2) {$\# Magic\ holes = O(\log(n))$};

  \fill[orange!60, fill opacity=0.3] (0,0.3) ellipse (2 and 1.5);
  \draw[blue, thick] (0,0.3) ellipse (2 and 1.5);
  \node at (0,.8) {$\#Magic\ hole  = 0$};

  \fill[green!30, fill opacity=0.4] (-0.5,0) ellipse (1.1 and 0.6);
  \draw[blue, thick] (-0.5,0) ellipse (1.1 and 0.6);
  \node at (-.85,0) {\small Clifford};

  \fill[magenta!30, fill opacity=0.4] (0.5,0) ellipse (1.1 and 0.6);
  \draw[blue, thick] (0.5,0) ellipse (1.1 and 0.6);
  \node at (0.68,0) {\small Matchgates};
\end{scope}
\end{tikzpicture}
    \caption{QCS complexity}
    \label{fig:QCS complexity}
\end{figure}
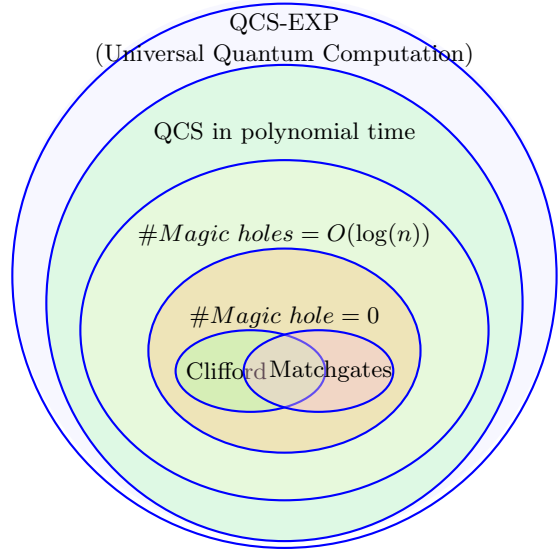

\emph{Quon language}.---
The Quon language is a 2+1D topological quantum field theory (TQFT) \cite{witten1988topological,atiyah1988topological} with space-time boundaries and defects.
TQFT can be regarded as a computational tool as suggested by Witten \cite{witten1988topological}. 
Fruitful computational tools in the Quon language can be derived from a simple partition function which has value $\sqrt{2}^{m}$ on the 3-disc with $m$ non-intersecting defect circles on the space boundary sphere. 
Following the Alterfold theory developed in \cite{liu20233,liu2023alterfold,liu2024functional,liu2024alterfold},
we extend the partition function uniquely to any 3-manifold with space time boundary and defects lines on the space boundary, satisfying four types of handle moves.
These handle moves allow for sewing and factoring of 3-manifolds, enabling a consistent and topologically invariant construction of the partition function.
Given any triangulation or handle decomposition of a 3-manifold with boundary and defects, we systematically remove neighborhoods of $k$-simplices via $k$-handle moves, ultimately obtaining a state-sum expression for the partition function.

Tensor networks have been widely applied in quantum information and quantum computing \cite{schollwock2011density,verstraete2008matrix,markov2008simulating}. In Quon language, we represent a tensor network as a \textbf{Crossing-Decorated 3-Manifold (CDM)} in the TQFT with space-time boundary and defects. The black box in the tensor network becomes a white box in CDM, and the value of the tensor network becomes the partition function of the CDM in the TQFT. More precisely, in QCS a tensor network in the 3D space is represented by its neighboring 3-manifold $M$ with crossing-decorated 2D boundary $\partial M$, see Fig.~\ref{fig:copytensor} as an example of CDM representation of COPY tensor.

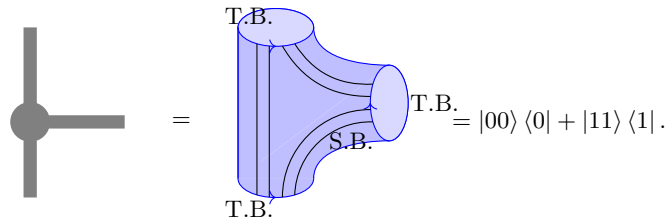
\begin{figure}[H]
    \centering
    \begin{tikzpicture}
        
\begin{scope}[scale=.5]
  \fill[fill=gray] (0,0) circle (15pt);
  \draw[line width=5pt, black!50](0,.5) --(0,.5+2);
  \draw[line width=5pt, black!50](0,-.5) --(0,-.5-1.5);
  \draw[line width=5pt, black!50](.5,0) --(.5+2,0);
 

 \begin{scope}[scale=1,shift={(0,0)}]
\node at (4,0) {=};
 
\node at (14,0) {$=\ket{00}\bra{0}+\ket{11}\bra{1}.$};
  
\end{scope}

\end{scope}
\begin{scope}[scale=.5,shift={(5.5,-1.5)}]
\node at 
    \copytensorcopy{1}{0.5}{4};
\end{scope}
    \end{tikzpicture}
    \caption{Quon representation of the Copy tensor:
    The boundary surface of the 3-manifold are divided into time boundary (TB) and space boundary (SB). The space boundary is decorated by strings. The time boundary is the union of three input/output discs. Every boundary disc has four decorated points corresponds to a two-dimensional Hilbert space in TQFT, namely a single qubit space. In general, a disc with 2m boundary points corresponds to a $2^{m-1}$ dimensional vector space.}
    \label{fig:copytensor}
\end{figure}

Parameterized crossing is defined as a linear sum of two non-intersecting pairs of strings,
$$
\crossingleftnode := \alpha\  \doublestrings + \frac{1-\alpha}{\sqrt{2}} \updowncircle. 
$$
Moreover, a crossing is called \textbf{Clifford}, if the parameter is $\alpha=1,-1,i,-i$. They have better representation in terms of braids and pair of charges illustrated as follows with more efficient computational rules:

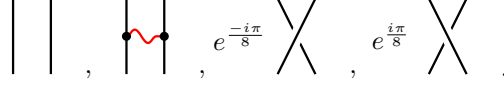
\begin{center}
    \begin{tikzpicture}[line width=.3mm]
        \begin{scope}
            \draw(0,.5)--(0,-.5);
            \draw(.5,.5)--(.5,-.5);
        \end{scope}
        \begin{scope}[shift={(1.5,0)}]
            \draw(0,.5)--(0,-.5);
            \draw(.5,.5)--(.5,-.5);
            \draw[red,decorate,decoration={snake}](0,0)--(.5,0);
            \fill(0,0)circle(.06);
            \fill(.5,0)circle(.06);
        \end{scope}
        \begin{scope}[shift={(3.5,0)}]
        \node at (-.5,0){$e^{\frac{-i\pi}{8}}$};
            \draw(0,.5)--(.5,-.5);
            \draw[white,WLL](.5,.5)--(0,-.5);
            \draw(.5,.5)--(0,-.5);
        \end{scope}

        \begin{scope}[shift={(5.5,0)}]
           \node at (-.5,0){$e^{\frac{i\pi}{8}}$};
            \draw(.5,.5)--(0,-.5);
             \draw[white,WLL](0,.5)--(.5,-.5);
             \draw(0,.5)--(.5,-.5);
        \end{scope}
    \node at (1,-0.5){$,$};
    \node at (2.5,-0.5){$,$};
    \node at (4.5,-.5){$,$};
    \node at (6.5,-.5){$.$};
    \end{tikzpicture}
\end{center}

Strings are world lines of braided anyons $\sigma$ satisfying Ising type fusion rule $\sigma\otimes\sigma=1\oplus\epsilon$ \cite{kitaev2006anyons, nayak2008non,ahlbrecht2009implementation}, where world lines of $\epsilon$ are denoted as red wires. Strings can also be viewed as world lines of Majoranas \cite{majorana, bravyi2002fermionic,sarma2015majorana}, and intersections of strings and wires are charges which behave like Majorana fermions as $\epsilon_i^2=1,\ \epsilon_i\epsilon_j+\epsilon_j\epsilon_i=2\delta_{ij}$, used in topological quantum computation \cite{kitaev2003fault,freedman2003topological, nayak2008non,bravyi2006universal,wang2010topological}.

A tensor network only has 1-skeletons, so the corresponding CDM is a handle body, which is homeomorphic to a solid torus. By isotopy, we move all decorations to the half side of the solid torus, so that we can read full information of the CDM from its 2D projection. This notion of \textbf{holes} are transparent in the 2D projection. In general, a hole is an element in the first homology $H_1$ of the 3-manifold $M$. 
A layer induces an element in the first cohomology $H^1(M,\mathbb{Z}_2)$ according to the parity of the winding number of a layer moving around holes.

\emph{Quon Classical Simulation}.--- Based on parameterized crossings, we represent all 3-qubit tensors and 2-qubit gates as CDM without referring to linear sums in Section 2.A of supplementary, which is important for our classical simulation and efficient computation.
There are elementary ingredients in CDM: (1) braided/charged strings, (2) parameterized
crossings, (3) holes, as shown in Fig.~\ref{fig:Topological Complexity}.
Moreover, in Theorem 1 of supplementary, we prove that any planar Matchgate tensor network can be presented by a CDM with only (1) and (2); and
in Theorem 2 of supplementary, we prove that any Clifford tensor network can be presented by a CDM with only (1) and (3). The converse statements are proved in the theorems as well.

\begin{figure}[H]
\begin{center}
\begin{tikzpicture}[line width=.3mm]
\begin{scope}
\node at (1.8, 1){Clifford};
\node at (3.6, 1){Matchgates};
\draw[purple] (1.5,.5) circle (2 and 1);
\draw[red] (3.5,.5) circle (2 and 1);
\begin{scope}[scale=.5,shift={(.7,2)}]
\hole{1}{0}{0};
\end{scope}
\begin{scope}[shift={(2.2,0)},scale=.6]
\draw[line width=.6pt] (1,0) -- (0,1);
\draw[white,WLL] (0,0) -- (1,1);
\draw[line width=.6pt] (0,0) -- (1,1);
\draw[domain=6:24,smooth,variable=\x,red] plot (\x/30, {sin(\x r)/16+0.2});
\fill (.8,.2) circle (.065);
\fill (.2,.2) circle (.065);
\end{scope}
\begin{scope}[shift={(4,0)}, scale=.6]
\draw (0,1)--(1,0);
\draw (0,0) -- (1,1);
\node at (.2,.5) {$\theta$};
\end{scope}
\end{scope}
\end{tikzpicture}
\end{center}
    \caption{Topological Complexity for Cliffords and Matchgates.}
    \label{fig:Topological Complexity}
\end{figure}
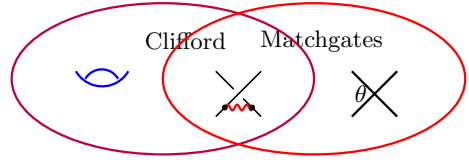

Motivated by these necessary and sufficient topological characterizations of Cliffords and Matchgates, we introduce a novel \textbf{Topological Complexity} of CDM to characterize computation complexity of its partition function. 
Charges have lowest computational complexity. Charges can be re-paired and moved freely along strings up to a phase determined by winding numbers. Braids have the second lowest complexity. One can switch a positive braid to a negative braid by adding a pair of charges:
\begin{equation}
\raisebox{0pt}{
\begin{tikzpicture}[baseline=(current bounding box.center),line width=.3mm]
\begin{scope}[scale=.7]
\draw (0,0) -- (1,1);
\draw[white,WLL] (1,0) -- (0,1);
\draw (1,0) -- (0,1);
\node at (1.75,.5) {$=$};
\begin{scope}[shift={(2.5,0)}]
\draw (1,0) -- (0,1);
\draw[white,WLL] (0,0) -- (1,1);
\draw (0,0) -- (1,1);
\draw[domain=6:24,smooth,variable=\x,red] plot (\x/30, {sin(\x r)/16+0.2});
\fill (.8,.2) circle (.065);
\fill (.2,.2) circle (.065);
\node at (1.5,.5) {.};
\end{scope}
\end{scope}
\end{tikzpicture}
}
\end{equation}
This relation reduces a link to a union of unknots in different \textbf{layers}. Moreover, crossing-decorated connected components can be splited into different layers. 
Holes and crossings have parallel third lowest complexity. 
From the view of resource theory \cite{gour2008resource,veitch2014resource,coecke2016mathematical} where states and operations are evaluated to be free or magic (i.e. resources are limited and costly): if Clifford (1)+(3) are considered to be free, then (2) crossings are 
magic; if Matchgate (1)+(2) are considered to be free, then (3) holes are magic. In both cases, (1) is free.

To realize an efficient unified classical simulation, we show that the source of exponential computational complexity is neither (2) crossings nor (3) holes, but how crossings move around the holes. 
The simplest feature is illustrated as follows, which we call a \textbf{topological spin}:
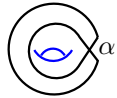
\begin{center}
\begin{tikzpicture}[line width=0.3mm]
\begin{scope}[scale=.6]
\hole{0.6}{1}{1};
\node at (2.2,0.15){$\alpha$};
\draw(2,0.4)arc[start angle=15,end angle=345,radius=1]--(1.6,0.4)arc[start angle=30,end angle=330,radius=.6]--(2,.4);
\end{scope}
\end{tikzpicture}
\end{center}


\emph{QCS Algorithm}.---
Now let us design an algorithm to efficiently eliminate crossings and holes in a CDM. 

First, the parameterized crossings obey the Yang-Baxter relations \cite{liu2016yangbaxterrelationplanaralgebras}, which generalize Reidemeister moves \cite{alexander1926types,reidemeister1927elementare} and the Yang-Baxter equations \cite{yang1967some,baxter1972partition}, and allow systematic simplification of crossing defects. In particular, this simplification ensures that the boundary surface $\partial M$ contains no contractible bigons—regions bounded by two strands between vertices $u$ and 
$v$ \cite{steinitzvorlesungen}, enclosing no genus but potentially involving multiple crossings.  
Here is an example of the bigon reduction process.

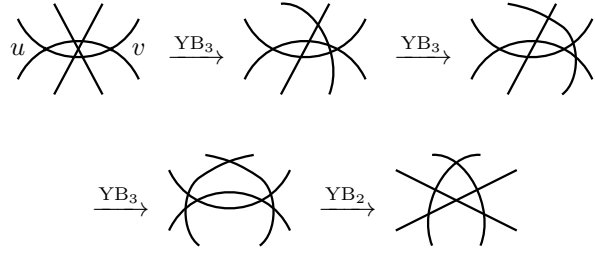
\begin{center}
    \begin{tikzpicture}[line width=.3mm]
        \begin{scope}[xscale=.8]
            \draw(0,.8)[bend right=60]to(2,.8);
            \draw(0,0)[bend left=60]to(2,0);
            \draw(1.4,-.2)--(1-.4,1);
            \draw(1-.4,-.2)--(1+.4,1);
            \fill (0.45,.4)circle(.03);
            \fill (1.55,.4)circle(.03);
            \node at (0,.4){$u$};
            \node at (2,.4){$v$};
        \end{scope}
        
        \begin{scope}[shift={(3,0)},xscale=.8]
    \node at (-.8,0.4){$\xrightarrow{\operatorname{YB}_3}$};
            \draw(0,.8)[bend right=60]to(2,.8);
            \draw(0,0)[bend left=60]to(2,0);
            
            \draw(1.4,-.2)[bend right=60]to(1-.4,1);
            \draw(1-.4,-.2)--(1+.4,1);
        \end{scope}
        \begin{scope}[shift={(6,0)},xscale=.8]
    \node at (-.8,0.4){$\xrightarrow{\operatorname{YB}_3}$};
            \draw(0,.8)[bend right=60]to(2,.8);
            \draw(0,0)[bend left=60]to(2,0);
            
            \draw(1.5,-.2)[bend right=60]to(1.5,.7)[bend right=10]to(1-.4,1);
            \draw(1-.4,-.2)--(1+.4,1);
        \end{scope}
        \begin{scope}[shift={(2,-2)},xscale=.8]
    \node at (-.8,0.4){$\xrightarrow{\operatorname{YB}_3}$};
    
            \draw(0,.8)[bend right=60]to(2,.8);
            \draw(0,0)[bend left=60]to(2,0);

            \draw(1.5,-.2)[bend right=60]to(1.5,.7)[bend right=10]to(1-.4,1);
            \draw(1-.5,-.2)[bend left=60]to(.5,.7)[bend left=10]to(1+.4,1);
        \end{scope}
        \begin{scope}[shift={(5,-2)},xscale=.8]
    \node at (-.8,0.4){$\xrightarrow{\operatorname{YB}_2}$};
            \draw(0,.8)--(2,0);
            \draw(0,0)--(2,.8);
            \draw(1.4,-.2)[bend right=60]to(1-.4,1);
            \draw(1-.4,-.2)[bend left=60]to(1+.4,1);
        \end{scope}
    \end{tikzpicture}
\end{center}

Secondly, the string-genus relation in Eq.~\eqref{string_genus-1} eliminates a hole with a surrounding closed string. For the CDM representation of a planar Matchgate tensor network, every hole is surrounded by a closed string, so all holes can be eliminated. 
The resulting CDM is a 3-disc with crossing-decorations on the boundary sphere, whose partition function can be computed efficiently using the Yang-Baxter relation. 
It is an alternative method of computing planar Matchgate tensor network through Pfaffians. Conversely, one can translate such 3-disc CDM backwards to Matchgate tensor network and compute them efficiently through Pfaffians. 
\begin{align}
\begin{tikzpicture}[baseline=(current bounding box.center),line width=.3mm]
\label{string_genus-1}    
    \begin{scope}[shift={(-1.3,.7)},scale=.8]
    \draw[blue] (1.25,-1)arc[start angle=150,end angle=30, radius=.5];
    \draw[blue] (1,-.8)arc[start angle=-150,end angle=-30, radius=.8];
\end{scope}
        \begin{scope}[yscale=.5]
            \draw(1,0) arc[start angle=0,end angle=360,radius=1];
        \end{scope}
\end{tikzpicture}
=\frac{1}{\sqrt{2}}.
\end{align}

The 3-disc CDM can be computed efficiently through the Yang-Baxter relations or the Pfaffians. The partition function is multiplicative, namely the partition function of multiple 3-discs is the product of individual ones. 
For a general 3-manifold, one can apply the four handle moves as Proposition 3 in supplementary to decompose the 3-manifold into 3-discs. However, the handle move 2 involves a linear sum which causes exponential computational complexity. 
To ensure efficient computations, we can only apply move 2 in logarithm times.
We need other efficient relations to change shape of 3-manifold, particularly to eliminate holes, just as the string-genus relation.

Motivated by the fact that Cliffords are represented by CDM with only (1) Clifford crossings and (3) holes, which encodes Clifford CDM as classical bits,
Biamonte asked for a natural algorithm in the Quon language to efficiently compute Clifford.
The string-genus relation and the Yang-Baxter relation implies all tensor network relations in ZX-calculus \cite{van2020zx,ng2018completeness,hadzihasanovic2018two,jeandel2020completeness,wang2022completeness,backens2016completeness,backens2014zx}. In particular, one can efficiently compute Cliffords by translating the corresponding reduction algorithms in \cite{backens2014zx}.

Here we present a more conceptual algorithm to efficiently compute Clifford CDM. The key point is to introduce better hole-elimination relations, beyond the string-genus relation, which works for a general CDM from hybrid-Clifford-Matchgate tensor networks. 

A hole is called odd, if there is a layer with an odd winding number around the hole. Otherwise the hole is called even. 
We introduce an \textbf{odd hole handle slide} relation to remove the odd hole, see Lemma 3 of Supplementary.
This relation extensively generalizes the string-genus relation, which is an odd hole and a layer of a closed string with winding number one. 
This odd hole handle slide relation can be generalized from qubits to qudits. 

A generating set of representative holes is a basis of $H_1(M)$. 
All layers form a set of vectors in $H^1(M,\mathbb{Z}_2)$. 
After removing odd holes, all layers are zero in $H^1(M,\mathbb{Z}_2)$. 
The remaining even holes are those elements in $H_1(M)$ orthogonal to the layers, which is an intrinsic concept, independent of the odd hole reduction process.

An (even) hole is called Clifford, if all crossings in all layers containing the hole are Clifford. We introduce an \textbf{Clliford hole elimination} relation to remove a Clifford hole, see Theorem 3 of supplementary. This relation is based on the arithmetic properties of the quadratic forms over the binary field, which has been studied as affine signatures by Cai in \cite{cai2018clifford}. It is not known whether such a relation can be generalized from qubits to qudits. 

These hole elimination relations, derived from parity constraints and arithmetic properties over the binary field $\mathbb{F}_2$ extend the simplification toolkit of TQFT beyond conventional handle moves, enabling efficient topological surgery moves of 3-manifolds. In particular, they allow for the elimination of topological holes in the CDM representations of tensor networks, dramatically reducing computational overhead.

The remaining holes are even and non-Clifford, which we call \textbf{Magic Holes}. They capture the global topological entanglement \cite{KitaevTEE} and the long range magic \cite{wei2025long}. Both Cliffords and Matchgates have zero Magic holes in QCS, thus both can be computed efficiently.   
Each (Magic) hole can be eliminated as the cost of a factor 2. So if the number of Magics holes has a logarithm bound, then the partition function of the CDM can be computed efficiently. 
The total number of Magic holes is governed by topological entanglement entropy \cite{KitaevTEE}, rather than conventional entanglement entropy.

\emph{Topological Tensor Networks}.---
If an even hole is surrounded by crossing-decorations without boundary, then it reduces to a topological spin. The essential reason is that in the 2+1D Alterfold TQFT, the vector space associated to such time-boundary is two dimensional and it consists of topological spins. 
This reduction method is a topological entanglement analogue to the low entanglement reduction of tensor network.

The topological insights in QCS motivate us to study the \textbf{Topological Complexity} of the quantum system, which encodes essential computation complexity and quantum advantages.
To further elaborate topological complexity of Magic holes, we propose the notion of \textbf{Topological Tensor Networks} (in Section IV of supplementary), in which we consider a Magic hole (or generally a topological spin) as its spin and a crossing-decorated layer sounding Magic holes as their interactions.

In addition, we introduce a novel global skein relation \textbf{free move of topological spins} (in Section IV B of supplementary) to further reduce the number of Magic holes globally.  
In particular, if two topological spins cannot be distinguished by any layer, then they can merge into one topological spin. This provides a mathematical foundation to establish topological tensor networks.

With this method, there are families of quantum circuits with $O(n)$ width and depth, $O({n^2})$ Clifford gates and Matchgates, and the corresponding CDM has $O(n^2)$ Magic holes; see Section V of supplementary for details. The outcome sampling process with efficient classical simulation or not depending on the choice of measurement basis.

\emph{Summary}.--- Together, the above relations significantly reduce the computational complexity of the partition function in the Quon language. As a result, a broad class of tensor networks—including all Clifford circuits and Matchgate circuits—can be evaluated in polynomial time within the QCS framework.

We highlight that in QCS, Matchgates are free fermions and have efficient reduction on the 2D boundary by Yang-Baxter relations; and Cliffords are bulk 3-manifold with boundary braiding decorations and have efficient reduction by generalized handle slides and handle moves.
The 3D bulk reduction and the 2D boundary reduction are parallel, therefore QCS takes advantages of both Cliffords and Matchgates methods to reduce the computational complexity. It unifies efficient computational tools in both the bulk 2+1 TQFT and the boundary 2D Ising theory, which is a reminiscent of the holographic principle of Quantum gravity \cite{maldacena1999large,witten1998anti,hooft2001holographic,bousso2002holographic}.

Tensor network has been a widely used picture language in simulation of physics systems and engineering projects, while Quon contains all information of tensor networks with finer structure and more efficient computational rules. A 3-manifold arising as a neighborhood of a tensor network naturally forms a handlebody, which can be interpreted as a "fractional" tensor network with refined topological and algebraic decorations. 
This refined structure provides new insights into various foundational concepts in quantum information. Moreover, tensor network representation reveals only 1-dimensional skeleton of the manifolds compared with Quon, and more general 3-manifolds extend beyond conventional tensor networks, offering a broader mathematical framework for simulating quantum systems. These structures have been applied in the study of communication protocols~\cite{jaffe2017constructive}, and quantum error correction~\cite{liu2019quantized,shao2024variational}.
Please refer to \cite{kang20252d}, which presents other mathematical properties of Quon diagrams and applications to classical simulation.

\emph{Acknowledgements}.---
Zhengewei Liu thanks Arthur Jaffe, Xun Gao and Yunxiang Ren for early discussions. The article is partially based on lectures of Zhengwei Liu in Fall 2021 at Tsinghua University with notes taken by Fan Lu, and undergraduate thesis of Ningfeng Wang and Zixuan Feng in June 2022.
Zhengwei Liu was supported by Beijing Natural Science Foundation Key Program (Grant No. Z220002) and by Templeton Religion Trust (TRT159). All authors were supported by Beijing Natural Science Foundation (Grant No. Z221100002722017).

\section{Mathematical foundations}
\subsection{2D Relations in Quon}

In this section, we interpret Quon Language introduced in \cite{liu2017quon} as a 2+1D TQFT \cite{atiyah1988topological} with space-time boundaries and defects, from the viewpoint of Alterfold theory developed in \cite{liu20233,liu2023alterfold,liu2024functional,liu2024alterfold}. In Quon, a decorated 3-manifold with space-time boundaries corresponds to a multi-linear map, which we apply to simulate multi-qubit transformations. 

We first consider the decorations to be non-intersecting strings and established a complete set of linear relations to compute the partition function in TQFT. Any multi-qubit transformation can be represented as a linear sum of decorated 3-manifolds. 
To obtain an efficient classical simulation, we prefer to represent the transformation using a single 3-manifold, without a linear sum. This requires additional decorations of parameterized crossings besides strings. We will show that those crossing-decorated 3-Manifolds (CDM) are enough to represent all 3-qubit tensors and 2-qubit gates and their tensor networks, including all hybird-Clifford-Matchgate quantum circuits.

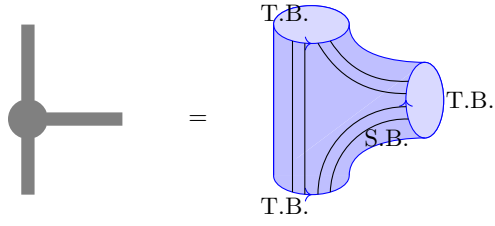
\begin{figure}[H]
    \centering
    \begin{tikzpicture}
        
\begin{scope}[scale=.5]

  \fill[fill=gray] (0,0) circle (15pt);
  \draw[line width=5pt, black!50](0,.5) --(0,.5+2);
  \draw[line width=5pt, black!50](0,-.5) --(0,-.5-1.5);
  \draw[line width=5pt, black!50](.5,0) --(.5+2,0);
 
 \begin{scope}[scale=1,shift={(0,0)}]

\node at (4.5,0) {=};

\end{scope}

\end{scope}
\begin{scope}[scale=.5,shift={(6.5,-1.5)}]
\node at 
    \copytensorcopy{1}{0.5}{4};
\end{scope}
    \end{tikzpicture}
    \label{fig:copy tensor_text}
    \caption{CDM representation of the COPY tensor}
\end{figure}

\input{Quon}
\section{QCS of tensor networks}

In this section, we demonstrate that all generating tensors of Matchgates and Clifford can be represented by Quon diagrams. Based on this observation, we propose a Quon Classical Simulation (QCS) framework for Matchgates and Cliffords.

\subsection{Quon dictionary of qubit tensors}\label{Sec: QCS-3}

In this subsection, we are going to prove that all 2-qubit gates and all 3-qubit tensors can be represented in Quon without linear sum.

\subsubsection{1-qubit states and 1-qubit gates}
In a qubit space, every disc has four points on its boundary. In general, a disc with $2m$ boundary points corresponds to a $2^{m-1}$ dimensional vector space. We simulate the qubit basis vectors as:
\begin{center}
    \begin{tikzpicture}[line width=.3mm]
    \begin{scope}
        \node at (-1.5,0){$\ket{0}=\frac{1}{\sqrt{2}}$};
        \yuanzhuCAPZERO{1}{0.5}{2};
    \end{scope}
    \begin{scope}[shift={(3.5,0)}]
        \node at (-1.5,0){$\ket{1}=\frac{1}{\sqrt{2}}$};
        \yuanzhuCAPONE{1}{0.5}{2};
    \end{scope}
    \begin{scope}[shift={(0,-1.2)}]
        \node at (-1.5,0){$\bra{0}=\frac{1}{\sqrt{2}}$};
        \yuanzhuCUPZERO{1}{0.5}{2};
    \end{scope}
    \begin{scope}[shift={(3.5,-1.2)}]
        \node at (-1.5,0){$\bra{1}=\frac{1}{\sqrt{2}}$};
        \yuanzhuCUPONE{1}{0.5}{2};
    \end{scope}
    \node at (1,0){$,$};
    \node at (4.5,0){$.$};
    \node at (1,-1.5){$,$};
    \node at (4.5,-1.5){$.$};
\end{tikzpicture}
\end{center}

Now we have basic Pauli gates represented by a pair of charges:
\begin{center}
    \begin{tikzpicture}[line width=.3mm]
    
    \begin{scope}[shift={(0,0)}]
    \node at (-1.3,0.5){$X=$};
        \yuanzhuX{1}{0.5}{2};
    \end{scope}
    \begin{scope}[shift={(2.6,0)}]
    \node at (-1.3,0.5){$=$};
        \yuanzhuXX{1}{0.5}{2};
    \end{scope}

    \begin{scope}[shift={(5.2,0)}]
    \node at (-1.3,0.5){$;Z=$};
        \yuanzhuZ{1}{0.5}{2};
    \end{scope}
    \begin{scope}[shift={(7.8,0)}]
    \node at (-1.3,0.5){$=$};
        \yuanzhuZZ{1}{0.5}{2};
    \end{scope}
    \node at (9,0.5){$.$};
\end{tikzpicture}
\end{center}

We can represent $S$ and $H$ gate with braids:
\begin{center}
\begin{tikzpicture}[line width=.3mm]
    \begin{scope}[scale=1]
        \begin{scope}[shift={(0,-3)}]
        \node at (-1.3,0.5){$S=$};
            \yuanzhuS{1}{0.5}{2};
        \end{scope}
    
        \begin{scope}[shift={(2.6,-3)}]
        \node at (-1.3,0.5){$=$};
            \yuanzhuSS{1}{0.5}{2};
        \end{scope}
    
        \begin{scope}[shift={(5.2,-3)}]
        \node at (-1.3,0.5){$;S^\dagger=$};
            \yuanzhuSDG{1}{0.5}{2};
        \end{scope}
        \begin{scope}[shift={(7.8,-3)}]
        \node at (-1.3,0.5){$=$};
            \yuanzhuSSDG{1}{0.5}{2};
        \end{scope}
    \end{scope}

    \begin{scope}[scale=1,shift={(0,-6)}]
        \begin{scope}
        \node at (-1.3,0.5){$H=$};
            \yuanzhuH{1}{0.5}{2};
        \end{scope}
        \begin{scope}[shift={(2.6,0)}]
        \node at (-1.3,0.5){$=$};
            \yuanzhuHDG{1}{0.5}{2};
        \end{scope}
        \begin{scope}[shift={(5.2,0)}]
        \node at (-1.3,0.5){$=$};
            \yuanzhuHTR{1}{0.5}{2};
        \end{scope}
        \begin{scope}[shift={(7.8,0)}]
        \node at (-1.3,0.5){$=$};
            \yuanzhuHTRDG{1}{0.5}{2};
        \end{scope}
    \end{scope}
    
    \node at (9,-3+.5){$.$};
    \node at (9,-6+.5){$.$};
\end{tikzpicture}
\end{center}

More generally, parameterized single qubit gates can be represented by parameterized braids:
\begin{center}
    \begin{tikzpicture}[line width=.3mm]
    \begin{scope}[scale=.7]
    \node at (-4,0.6){$R_X(\theta)=e^{i\theta X}=\sqrt{2}\cos\theta\times$};
    \yuanzhualphablack{1}{0.5}{2};
    \end{scope}

    \begin{scope}[shift={(8.5,0)},scale=.7]
    
    \node at (-4.3,0.6){$R_Z(\theta)=e^{i\theta Z}=\sqrt{2}\cos\theta\times$};
    
    \yuanzhualphaleft{1}{0.5}{2};
     \node at (-0.4,1.3) {$\alpha$};
    \end{scope}
        
    \node at (2,0.4){$,\alpha=i\tan\theta;$};
    \node at (10.6,0.4){$,\alpha=i\tan\theta.$};
    
    \end{tikzpicture}
\end{center}
By Euler decomposition $U=R_X(\theta_1)R_Z(\theta_2)R_X(\theta_3)=R_Z(\theta_4)R_X(\theta_5)R_Z(\theta_6)$, all 1-qubit unitaries are clearly Quon-representable without using linear sums.

\subsubsection{3-qubit tensors}

3-qubit pure states are known (\cite{dur2000three}) to be classified into 6 equivalence classes under local unitary transformations:
\begin{prop} 
    Under local unitary transformations:
    $$
    \ket{\psi}_{ABC} \sim (U_A\otimes U_B \otimes U_C)\ket{\psi}_{ABC},\quad
    U_A,U_B,U_C\in U(2),
    $$
    all three qubits tensors are divided into 6 equivalence classes:
\begin{itemize}
\item GHZ-state and its equivalent states
    $$
    \ket{GHZ} :=\frac{1}{\sqrt{2}}( \ket{000} + \ket{111}).
    $$
\item W-state and its equivalent states
    $$
    \ket{w}:=\frac{1}{\sqrt{2}}(\ket{011} +\ket{101} +\ket{110}).
    $$
    \item 
    A-B entangled and C is separated.
    \item 
    B-C entangled and A is separated.
    \item 
    C-A entangled and B is separated.
    \item 
    A,B,C are separated.
\end{itemize}
\end{prop}

In order for all 3-qubit pure states to be Quon-representable, we only need to consider the representatives of the above 6 equivalence classes. For the 4 separable classes, the result is clear. 
For the representatives of the fully entangled 2 classes, we consider the following X-basis:
\begin{center}
    \begin{tikzpicture}[line width=0.3mm]
    \begin{scope}[scale=.8]
    \node at (-1.9,0){$\ket{0}_X=\frac{1}{\sqrt{2}}$};
        \yuanzhuCAPPLUSCOLOR{1}{0.5}{2};
    \end{scope}

    \begin{scope}[shift={(4,0)},scale=.8]
    \node at (-1.8,0){$\ket{1}_X=\frac{1}{\sqrt{2}}$};
        \yuanzhuCAPMINUSCOLOR{1}{0.5}{2};
    \end{scope}
    
    \begin{scope}[shift={(0,-1.5)},scale=.8]
    \node at (-1.9,0){$\bra{0}_X=\frac{1}{\sqrt{2}}$};
        \yuanzhuCUPPLUSCOLOR{1}{0.5}{2};
    \end{scope}

    \begin{scope}[shift={(4,-1.5)},scale=.8]
    \node at (-1.8,0){$\bra{1}_X=\frac{1}{\sqrt{2}}$};
        \yuanzhuCUPMINUSCOLOR{1}{0.5}{2};
    \end{scope}
    
    \node at (1-.1,0){$,$};
    \node at (5-.1,0){$.$};
    \node at (1-.1,-1.5){$,$};
    \node at (5-.1,-1.5){$.$};
    \end{tikzpicture}
\end{center}
and we have: (up to a scalar)
$$
\ket{w}:=\frac{1}{\sqrt{2}}(\ket{011}_X +\ket{101}_X +\ket{110}_X)=
\wstate{1.3}{0.5}{0.8}.
$$
$$
\ket{Max}_3=\frac{1}{\sqrt{2}}(\ket{000}_X+\ket{011}_X+\ket{101}_X+\ket{110}_X)
=\maxstate{1.3}{0.5}{0.8}
,
$$ 
and similarly for any $n$-qubit $\ket{Max}_n=\frac{1}{\sqrt{2}}\sum_{a_1+...+a_n\equiv 0(mod2)}\ket{a_1\dots a_n}_X$.
Thus,
$$
\ket{GHZ} = H^{\otimes 3} \ket{Max}_3,
$$
is Quon-representable without using linear sum.

\subsubsection{2-qubit gates}

One can represent the control-Z gate $CZ$ as:
\begin{tikzequation}
\begin{scope}[shift={(5.5,0)}]
\node at (-1,1.25){$CZ=$};
    \CLczgate{.6}{.3}{2.5}{0}{0};
\end{scope}
    \node at (9,1.25){$.$};
\end{tikzequation}
Since all 2-qubit gates are generated by $\{R_X(\alpha),S,H,CZ\}$, they can be represented in Quon without linear sum. 

\subsection{Quon representation of CAR algebra}
\label{CARalg}
As discussed in the 2-string model in \cite{jaffe2018holographic}, the charges are Majorana fermions, which can be seen as the generators of the CAR algebra, a $C^*$-algebra with $2n$ generators $\{c_1, c_2, ..., c_{2n}\}$ satisfying 
$$
c_{i}^{2} =1,\ c_{i}^{*}=c_{i},\ 
c_{i}^{} c_{j}^{}=-c_{j}^{}  c_{i}\ (i\ne j) .
$$

The operator $c_i$ is also called a Majorana operator \cite{majorana}. The CAR algebra can be simulated as CDMs in Quon, which illustrates the Jordan-Wigner transformation \cite{JWtrans}. This embedding map of diagrammatic conventions of 2-strings into Quon is a special case of the general embedding map from subfactor planar para algebras \cite{jaffe2017planar} into graph planar algebras.

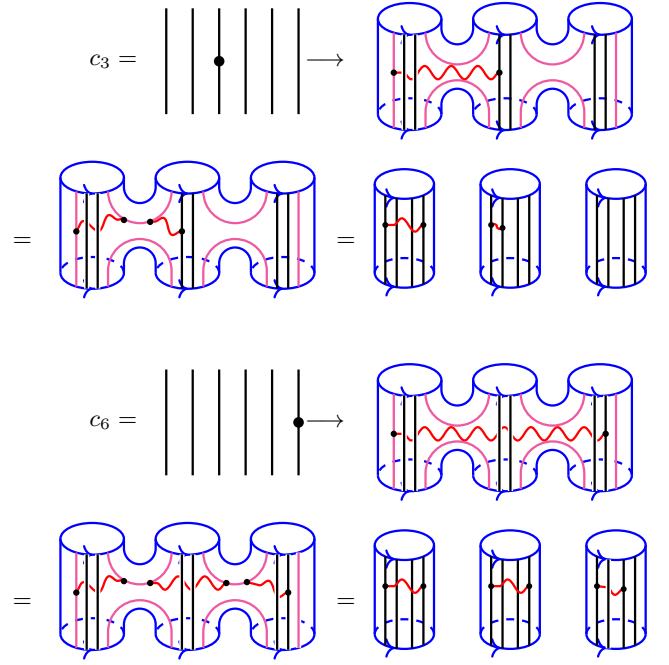
\begin{figure}[H]
    \centering
    \begin{tikzpicture}[line width=0.3mm]
    \begin{scope}[scale=.7]

    \node at (-1,1){$c_3=$};
    \foreach \i in {0,1,2,3,4,5}{
    \draw(\i/2,2)--(\i/2,0);}
    \fill(2/2,1)circle(.1);
    \node at (3,1){$\longrightarrow$};
    
    \begin{scope}[shift={(4,0)},scale=.6]
    
        \threequbitgateCARTHREE{1}{.5}{3};
    \end{scope}

     \begin{scope}[shift={(-2,-3)},scale=.6]
     \node at (-1.2,1){$=$};
     \node at (9,1){$=$};
    \threequbitgateCARTHREEPAULI{1}{.5}{3};
    \end{scope}
    \begin{scope}[shift={(4.5,-3)},scale=.8]
    \yuanzhuZL{1}{.5}{3};
    \end{scope}
     \begin{scope}[shift={(6.5,-3)},scale=.8]
    \yuanzhuZZZ{1}{.5}{3};
    \end{scope}
    \begin{scope}[shift={(8.5,-3)},scale=.8]
    \yuanzhuIDEMPTY{1}{.5}{3};
    \end{scope}

    \end{scope}
        \node at (6.7,-1.5){$.$};
    \end{tikzpicture}
\end{figure}

\begin{figure}[H]
    \centering
    \begin{tikzpicture}[line width=0.3mm]
    \begin{scope}[scale=.7]

    \node at (-1,1){$c_6=$};
    \foreach \i in {0,1,2,3,4,5}{
    \draw(\i/2,2)--(\i/2,0);}
    \fill(5/2,1)circle(.1);
    \node at (3,1){$\longrightarrow$};
    
    \begin{scope}[shift={(4,0)},scale=.6]
    
        \threequbitgateCARSIX{1}{.5}{3};
    \end{scope}
    
   \begin{scope}[shift={(-2,-3)},scale=.6]
     \node at (-1.2,1){$=$};
     \node at (9,1){$=$};
        \threequbitgateCARSIXPAULI{1}{.5}{3};
    \end{scope}
    \begin{scope}[shift={(4.5,-3)},scale=.8]
    \yuanzhuZL{1}{.5}{3};
    \end{scope}
     \begin{scope}[shift={(6.5,-3)},scale=.8]
    \yuanzhuZL{1}{.5}{3};
    \end{scope}
    \begin{scope}[shift={(8.5,-3)},scale=.8]
    \yuanzhuY{1}{.5}{3};
    \end{scope}

    \end{scope}
        \node at (6.7,-1.5){$.$};
    \end{tikzpicture}
\end{figure}

The circuits closely related to the CAR algebra are Matchgates.
Matchgates arise naturally from the physics of free fermions on one-dimensional lattices and constitute a class of quantum circuits that admit efficient classical simulation \cite{terhal2002classical}.
Originally introduced by Valiant \cite{valiant2002quantum}, Matchgates are closely connected to the enumeration of weighted perfect matchings in graphs, which can be efficiently computed using Pfaffian techniques, most notably the FKT algorithm by Fisher, Kasteleyn, and Temperley~\cite{fisher1961statistical,kasteleyn1961statistics,temperley1961dimer}.
So Matchgates have long been regarded as being at the boundary between classically simulable and fully quantum computational models \cite{jozsa2008matchgates,bravyi2002fermionic}.

Planar Matchgates constitute a subgroup of the CAR algebra which are unitaries associated to quadratic Hamiltonians \cite{bravyi2002fermionic,bravyi2004lagrangian,jozsa2008matchgates}:
$$
U = \exp{itH},\quad H = i\sum_{\mu\nu} h_{\mu\nu} c_\mu c_\nu,\quad \text{$h$ is antisymmetric.}
$$
These are called Gaussian operations. 

By Yang-Baxter equations in the CAR algebra and results in \cite{jozsa2008matchgates}, the Matchgate group has a set of generators which can be represented by:
$$
    e^{t c_j c_{j+1}} = \cos t\times I + \sin t\times c_j c_{j+1} = \sqrt{2}\cos t\times
    \tikz[baseline=(current bounding box.center),line width=.3mm]{
    \MGonequbitgate{0.7}{0.3}{1.5}{0}{-2};
    \node at (0.7,-.95){$\alpha$};
    }
    \ \ ,\quad \text{$j$ is odd, $\alpha=\tan t$.}
$$

$$
    e^{t c_j c_{j+1}} = \cos t\times I + \sin t\times c_j c_{j+1} = \sqrt{2}\cos t\times
    \tikz[baseline=(current bounding box.center),line width=.3mm]{ 
    \MGtwoqubitgate{.6}{0.3}{2.5}{0}{-2};
    \node at (1.5,-0.6){$\alpha$};
    }
  \ \   ,\quad \text{$j$ is even, $\alpha=\tan t$.}
$$
which corresponds, under the Jordan-Wigner transformation, exactly to the generating set in circuits under the $X$-basis $\ket{0}_X,\ket{1}_X$:
$$
\quad e^{i\theta Z} , \quad e^{i\theta X\otimes X}.
$$

When they are composed into a circuit, pink lines kill holes by the string-genus relation (Eq.~\eqref{string_genus-1}), resulting in a planar diagram. 
This shows that Matchgate circuits are represented in Quon by planar diagrams. A more general result considering Matchgate tensor networks, and also its converse, is proved in the next section.

\newcommand{\zerodegreenode}[1]{
\tikz[baseline=(current bounding box.center), line width=0.3mm] {
  \draw[thick] (0,0.8) -- (0,-0.8) node[pos=0.5, right] {$#1$};
}
}

\newcommand{\twodegreenode}{
\tikz[baseline=(current bounding box.center), line width=0.3mm] {
  \draw[thick] (0,0.8) -- (0,-0.8);
  \fill (0,0) circle(0.1);
}
}

\newcommand{\onedegreenode}{
\tikz[baseline=(current bounding box.center), line width=0.3mm] {
  \draw[thick] (0,0.4) -- (0,-0.4);
  \fill (0,0.4) circle(0.1);
}
}

\newcommand{\threedegreenode}{
\tikz[baseline=(current bounding box.center), line width=0.3mm] {
  \draw[thick] (0,0.8) -- (0,0);
  \draw[thick] (0,0) -- (-0.6,-0.3);
  \draw[thick] (0,0) -- (0.6,-0.3);
  \fill (0,0) circle(0.1);
}
}

\subsection{QCS of Matchgates and Clifford}
\label{QCS of Matchgates and Clifford}

\subsubsection{QCS of Matchgates}

Now we are going to prove the equivalence between planar Matchgate tensor networks \cite{cai2007valiant,jozsa2008matchgates,bravyi2002fermionic} and Matchgate Quon diagrams. 
This correspondence was first observed by Xun Gao.

\begin{definition}
A tensor network is a planar Matchgate if it is generated by: $\{\ket{0},\bra{0},\ket{Max}, e^{i\theta Z}, X, Y, Z\}$.
\end{definition}

The main result we want to show in this section is that there are three equivalent descriptions of planar Matchgate tensor networks:
\begin{align*}
&\text{planar tensor network generated by }\{\ket{0},\bra{0},\ket{w},\, e^{i\theta Z},\, X,\, Y,\, Z\}
\\
&\Longleftrightarrow
\text{planar tensor network generated by }\{\ket{0},\bra{0},\ket{Max},\, e^{i\theta Z},\, X,\, Y,\, Z\}
\\
&\Longleftrightarrow
\text{genus-zero 2d-CDM in which there is an }
\text{outermost Hamiltonian cycle (colored pink).}
\end{align*}
An illustrative example of this equivalence is shown in Fig.~\ref{fig:matchgate_all} below.

\begin{figure}[H]
    \centering
\subfigure[Matchgate tensor network\label{fig:matchgate_TN}]{
    \begin{tikzpicture}
        \isingonlyblackopen{0.8}
    \end{tikzpicture}
}
\hspace{.5em}
\subfigure[Original Quon diagram\label{fig:matchgate_Quon}]{
    \begin{tikzpicture}
        \isingMGoriginalopen{0.5}
    \end{tikzpicture}
}
\hspace{.5em}
\subfigure[Quon Matchgate diagram\label{fig:matchgate_Quon_reduced}]{
    \begin{tikzpicture}
        \isingMGreducedopen{0.5}
    \end{tikzpicture}
}

\caption{Quon representation of Matchgates}
\label{fig:matchgate_all}
\end{figure}
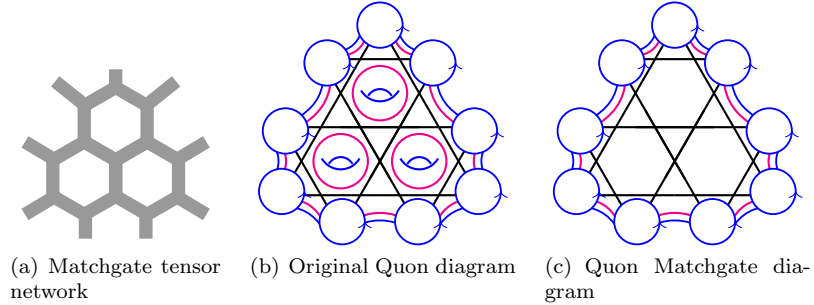

The main idea is to eliminate topological holes using the string-genus relation. Conversely, we insert string-genus back in one set of alternative regions. 
For example, the Matchgate tensor network in Fig.~\ref{fig:matchgate_TN} produces Matchgate Quon diagram Fig.~\ref{fig:matchgate_Quon}, which, by applying the string-genus relation (Eq.~\eqref{string_genus-1}), is reduced to the planar Quon diagram shown in Fig.~\ref{fig:matchgate_Quon_reduced}.

We begin by proving the first two:
\begin{prop}
The planar tensor networks generated by $\{\ket{0},\bra{0},\ket{w}, e^{i\theta Z}, X, Y, Z\}$ $=$ The planar tensor networks generated by $\{\ket{0},\bra{0},\ket{Max}, e^{i\theta Z}, X, Y, Z\}$, under the identification induced solely by the string-genus relation.
\end{prop}
\begin{proof}
One can directly verify that in tensor networks:
$$
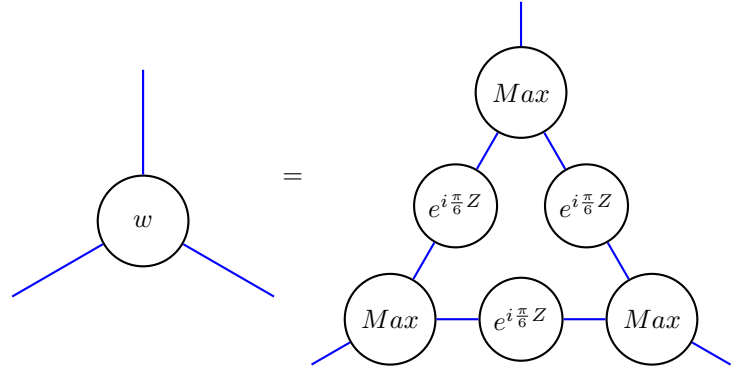
\wTN=\threemaxTN\ .
$$
up to a scalar.
Conversely, it can be directly verified that 
$$
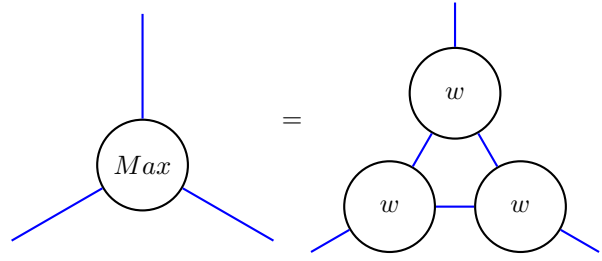
\maxTN=\threewTN\ .
$$
up to s scalar, and the $n$-qubit Max state for $n\ge 4$ can be expressed as a contraction of lower-qubit Max states.
\end{proof}

The equivalent description of planar Matchgates in Quon language was first observed by Xun Gao:

\begin{theorem}
\label{MatchgateThm}
A tensor network is planar Matchgate iff it can be represented as a genus-zero 2d-CDM in which there is an outermost Hamiltonian cycle (colored pink).
\end{theorem}

\begin{proof}
We first prove the necessity. The Quon diagrams of the generating tensors are as follows:
$$
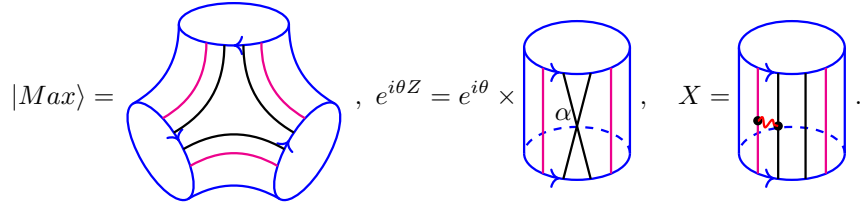
\ket{Max}=\maxstate{1.3}{0.5}{0.8}\ ,
\ 
e^{i\theta Z}=e^{i\theta}\times\yuanzhualpha{1}{0.5}{2}\ ,
\quad
X=\yuanzhurightcharge{1}{0.5}{2}\ .
$$
where $\alpha = e^{-2i\theta}$.
Thus, when contracting the generating tensors, each boundary surface line is accompanied by a surrounding pink string. 

We can apply Eq.~\eqref{fourstringschange} to move charges so that each pink loop carries an even number of charges, and can therefore be neutralized. 
Importantly, during this procedure, the total number of charges on all pink lines is conserved modulo 2.

Apparently, all charges arise from Pauli $X$ operators, and the total number of charges on pink lines equals the total number of Pauli $X$ in the network. 
Therefore, if the number of Pauli $X$ is even, we can always use the string-genus and charge-crossing relations to eliminate all holes. 
If the number is odd, one can easily show that the value of the network must be zero. Hence, without loss of generality, we may assume that the total number of charges on all pink loops is even.
This completes the necessity part.

We now prove the sufficiency.
We show how to add genera, along with surrounding pink lines, to convert a planar Quon diagram with crossings into a tensor network generated by $\{\ket{0},\bra{0},\ket{Max}, e^{i\theta Z}, X, Y, Z\}$. 

\textbf{Step 1.} Suppose we already have a black-and-white coloring of the faces of the planar graph such that any two adjacent faces (i.e., those sharing a common boundary line) have different colors. Then we add a string-genus to each inner black face, pink lines adjacent to the corresponding boundary surface line for each boundary black face, and leave all white faces unchanged.
In this way, it is clear that the resulting tensor network is generated by $\{\ket{0},\bra{0},\ket{Max}, e^{i\theta Z}, X, Y, Z\}$ is produced. 

\textbf{Step 2.} It remains to verify the existence of such a black-and-white coloring. This is equivalent to properly 2-coloring the nodes of the dual graph $H$ such that adjacent nodes receive different colors. 

Suppose $H$ contains no cycle of odd length. 
Starting from an arbitrary node with a chosen color, we traverse $H$. 
At each step, the coloring of the next node is uniquely determined, unless we return to a previously visited node. 
In the absence of odd-length cycles, these constraints are consistent. Thus, a proper 2-coloring exists.

\textbf{Step 3.} We must show that the dual graph $H$ contains no odd-length cycles. 
An edge in $H$ corresponds to a pair of adjacent faces in $G$.
Each step in $H$ represents passing into or out of exactly one cycle in $G$.
Consequently, any cycle in $H$ must have even length.

\end{proof}

As a result, we have two methods to compute a genus-zero 2D-CDM:
\begin{itemize}
    \item Use local Yang-Baxter relations in Quon diagrams.
    \item 
    Convert it into a Matchgate tensor network and reduce it to a perfect matching counting problem.
\end{itemize}

Now we introduce the first method. The idea is to find minimal 1-polygons and 2-polygons. The idea basically originates in \cite{steinitzvorlesungen}.

\begin{prop}\label{thm:Matchgate_YBR}
    Using Yang-Baxter relations, one can evaluate a genus zero 2d-CDM in polynomial time $O(n^3)$, except for a probability zero case where the singularity of the Yang-Baxter equation is encountered. Here n is the number of crossings and charges in the 2d-CDM. 
\end{prop}

\begin{proof}
We sketch the proof. 
Firstly, suppose that there are no charges. In the diagram, one can always find a minimal 1-polygon or a 2-polygon. If it is not pure (i.e., contains internal strings), then Eq.~\eqref{R3} can be used to move the internal strings out. Once the 1-polygon or 2-polygon is pure, one can apply Eq.~\eqref{R1} or Eq.~\eqref{R2} upon it to reduce the number of crossings. 

Next, we consider the case with charges.
Using Eq.~\eqref{charge_crossing_equation} and Eq.\eqref{fourstringschange}, one can remove all charges, reducing the diagram to a genus-zero 2D-CDM without charges, and thereby returning to the previous case.
\end{proof}

Applying Thm.~\ref{MatchgateThm}, we obtain the following corollary.
\begin{coro}
Using Yang-Baxter relations, one can evaluate planar Matchgate tensor networks in polynomial time $O(n^3)$, except for a probability zero case.
\end{coro}

This algorithm is illustrated through an example in  \S\ref{MG_algorithm}.

\subsubsection{Graph-theorectic Matchgate}
\label{mg_and_graph}

We now present the correspondence between graph-theoretic structures and planar Matchgates.
Specifically, we provide an explicit, polynomial-time method to transform between planar graphs in the graph-theoretic sense and genus zero 2d-CDMs such that the relevant data in both settings coincide (i.e., weighted perfect matchings in the former and the partition function in the latter).

A genus-zero 2D-CDM $G$ can be transformed (via black-and-white coloring) into a ``standard form" $\widehat{G}$, which decomposes into elementary planar Matchgate tensors: $\ket{0}$, $\ket{w}, e^{i\theta Z}, X$.
We now give a dictionary $\mathfrak{F}$ that converts a genus-zero 2d-CDM into a weighted planar graph:

\begin{equation}
\label{dictionary_1}
\mathfrak{F}\left(
\yuanzhuCAPPLUSCOLORTIKZ{1}{0.5}{2}
\right)=\sqrt{2} \times\quad
\onedegreenode
,\quad
\mathfrak{F}\left(
\yuanzhurightcharge{0.7}{0.2}{1.9}
\right)= \quad
\twodegreenode
\quad
,\quad
\mathfrak{F}\left(
\yuanzhualpha{0.7}{0.2}{1.8}
\right)
=
\quad
\alpha\times
\zerodegreenode{\frac{1}{\alpha}},
\end{equation}

\begin{equation}
\label{dictionary_2}
\mathfrak{F}\left(
\wstate{1.2}{0.5}{0.8}
\right)=
\frac{\sqrt{3}}{2}i\times
\threedegreenode.
\end{equation}
where the rightmost of Eq.~\eqref{dictionary_1} denotes an edge of weight $\frac{1}{\alpha}$ with a global coefficient $\alpha$; in the remaining equations, each dot represents a vertex in the graph.
We formulate it in the following proposition:
\begin{prop}
We transform a ``standard" Quon diagram $\widehat{G}$ into a degree $\le3$ graph-theoretic planar graph $\Gamma$ by the previous dictionary $\mathfrak{F}$, then we have $\text{Perfect-matching}(\Gamma) = Z(\widehat{G})$. 
\end{prop}
\begin{proof}
The value of the tensor network is defined as a contraction.
In the contraction sum, choosing 0 for a given index can be interpreted as ``choosing this edge," while choosing 1 corresponds to ``not choosing this edge."
Note that the nonzero components of the tensor $e^{i\alpha Z}$ occur when the two sides share the same basis element, those of $X$ occur when the bases differ, and those of $\ket{w}$ occur only when exactly one entry is 0.
Thus, the potentially nonzero terms that appear in the contraction have a one-to-one correspondence with the perfect matchings of the original graph.
Moreover, the weights agree via a straightforward computation.
\end{proof}

As a result, computing the value of a genus-zero 2D-CDM reduces to counting weighted perfect matchings of a planar graph, which can be efficiently evaluated via the FKT algorithm:
$$
Z\left(G \right)
=
\text{Perfect-matching}\left(\mathfrak{F}(\widehat{G}) \right)
=
\text{Pfaffian}\left(K[\mathfrak{F}(\widehat{G})] \right),
$$
where $K$ is the skew-symmetric adjacency matrix under a Pfaffian orientation.

\subsubsection{QCS of Clifford}

Stabilizer states and Clifford gates in quantum computations are useful to design error correction codes. They can be efficiently simulated on classical computers, as established by the 
Gottesman-Knill theorem \cite{gottesman1997stabilizer,gottesman1998heisenberg,aaronson2004improved}. Stabilizer states have been studied with previous diagrammatic approach of tensor network and ZX-calculus \cite{backens2014zx}. Quon language has topological structure of strings, charges and braids to represent tensor network diagrammatically \cite{liu2017quon}. Biomante has proposed an open question in \cite{biamonte2017charged} whether one can give a topological proof of the Gottesman-Knill theorem with graphical language.

In Thm.~\ref{Cliffordhole}, we provide a classical simulation to a larger class of quantum circuits based on general rules of hole eliminations, giving an answer to this question.
Now we establish QCS of Clifford tensors with generators:
\begin{definition}
A crossing is called \textbf{Clifford}, if its crossing coefficient is $\alpha=1,-1,i,-i$, corresponding to the following Quon diagrams:
\begin{center}
    \begin{tikzpicture}[line width=.3mm]
        \begin{scope}
            \draw(0,.5)--(0,-.5);
            \draw(.5,.5)--(.5,-.5);
        \end{scope}
        \begin{scope}[shift={(1.5,0)}]
            \draw(0,.5)--(0,-.5);
            \draw(.5,.5)--(.5,-.5);
            \draw[red,decorate,decoration={snake}](0,0)--(.5,0);
            \fill(0,0)circle(.06);
            \fill(.5,0)circle(.06);
        \end{scope}
        \begin{scope}[shift={(3.5,0)}]
        \node at (-.5,0){$e^{\frac{-i\pi}{8}}$};
            \draw(0,.5)--(.5,-.5);
            \draw[white,WLL](.5,.5)--(0,-.5);
            \draw(.5,.5)--(0,-.5);
        \end{scope}

        \begin{scope}[shift={(5.5,0)}]
           \node at (-.5,0){$e^{\frac{i\pi}{8}}$};
            \draw(.5,.5)--(0,-.5);
             \draw[white,WLL](0,.5)--(.5,-.5);
             \draw(0,.5)--(.5,-.5);
        \end{scope}
    \node at (1,-0.5){$,$};
    \node at (2.5,-0.5){$,$};
    \node at (4.5,-.5){$,$};
    \node at (6.5,-.5){$.$};
    \end{tikzpicture}
\end{center}
\end{definition}

\begin{theorem}
\label{CliffordThm}
A tensor network is Clifford iff it is represented as a crossing-decorated 3-Manifold (CDM) in which all crossings are Clifford crossings.
\end{theorem}

\begin{proof}
As shown in \S \ref{Sec: QCS-3}, the generating Clifford tensors $H$, $S$, $GHZ$ are all CDM with Clifford crossings. Therefore, any Clifford tensor network is a CDM with Clifford crossings.

Conversely, we prove that any CDM $M$ decorated by Clifford crossings and output discs can be decomposed into generating Clifford tensors through tensor and contraction. (Thus it is a stabilizer state.) 
The union of two CDM corresponds to the tensor product of the two tensors, thus we only need to show every connected component is Clifford.

First, we glue handles by Move 1 in Prop.~\ref{alterfold} to $M$, so that the result is a connected handlebody, which is homeomorphic to a solid torus $M'$.
Then $M$ and $M'$ represent the same tensor. So we only need to prove $M'$ to be Clifford. 
A pair of charges can be replaced as a double braid, thus we may assume that $M'$ has no charges for simplicity.

In order to depict the CDM on 2D plane, we arrange the solid torus $H$ as $[0,2g+1]\times [0,3] \times [-1,0] \setminus \cup_{j=1}^g(2j-1,2j)\times (1,2)\times [-1,0]$, for example:

\begin{center}
    \begin{tikzpicture}[line width=.3mm]
        \begin{scope}
            \holenew{.4}{-2}{-.1};
            \draw(-1.25,-1)--(-1.25,0);
            \draw(-1+.25,-1)--(-1+.25,0);
            \holenew{.4}{-.5}{-.1};
            \begin{scope}[shift={(-.3,0)}]
            \draw(.5,-1)--(1,0);
            \draw[white,WLL](1,-1)--(.5,0);
            \draw(1,-1)--(.5,0);
            \draw(1.5,-1)--(1.5,0);
            \draw(2,-1)--(2,0);
             \end{scope}
            \holenew{.4}{1.9}{-.1};
            \begin{scope}[shift={(2.3,0)}]
                \draw(.5,-1)--(.5,0);
            \draw(1,-1)--(1,0);
            \draw(1.5,-1)--(1.5,0);
            \draw(2,-1)--(2,0);
        
            \draw(2.5,-1)--(2.5,0);
            \draw[red,decorate,decoration={snake}](2,-.5)--(2.5,-.5);
            \fill(2,-.5) circle(.06);
            \fill(2.5,-.5) circle(.06);
            \draw(3,-1)--(3,0);
            \end{scope}
             \holenew{.4}{5.5}{-.1};
        \end{scope}
    \end{tikzpicture}
\end{center}

By parity, each vertical cylinder contains even number $2k$ of vertical strings. 
If $2k\leq 2$, then we replace the cylinder through the relations as in ~\ref{cuttwo}. 
If $2k \geq 6$, then we separate cylinder into $k-1$ cylinders with four vertical 
strings through the string genus relation as follows:
\begin{center}
    \begin{tikzpicture}[line width=.3mm]
        \begin{scope}[xscale=.5]
            \draw(.5,0)--(.5,-1);
            \draw(1.5,0)--(1.5,-1);
            \draw(-.5,0)--(-.5,-1);
            \draw(-1.5,0)--(-1.5,-1);
            \draw(2.5,0)--(2.5,-1);
            \draw(-2.5,0)--(-2.5,-1);
            
        \end{scope}
       \holenew{.4}{-2}{-.1};
       \holenew{.4}{1.5}{-.1};

        \node at (2.5,-.5){$\longrightarrow$};
        \begin{scope}[shift={(.5,0)}]
            
       \holenew{.4}{6.6}{-.1};
        \holenew{.4}{2.8}{-.1};
        
        \holestring{.2}{4.9}{0};
        
        \begin{scope}[shift={(5,0)},xscale=.5]
            \draw(1,0)--(1,-1);
            \draw(2,0)--(2,-1);
            \draw(-1,0)--(-1,-1);
            \draw(-2,0)--(-2,-1);
            \draw(3,0)--(3,-1);
            \draw(-3,0)--(-3,-1);
           
        \end{scope}
        \end{scope}
    \end{tikzpicture}
\end{center}
Then every vertical cylinder has four verticle strings. 
We decompose the solid torus as a contraction of the upper half and the lower half.
We need to show each half, a 3-disc CDM decorated by braids, is Clifford.

Without loss of generality, we assume the braid-decorated 4-valent graph of the 3-disc CDM to be connected and deform the graph into the half 2-disc of the boundary surface.
By connectness, each region is a polygon. By parity, the regions of the 4-valent graph admit a check-board shading. We plug in a string-genus in each shaded region and multiply by $\sqrt{2}$. 
For every braid or output disc, we apply a pair of parallel genus cut between the two holes in its nearest shaded regions and the pair of cuts are on two sides of the braid or the output disc.

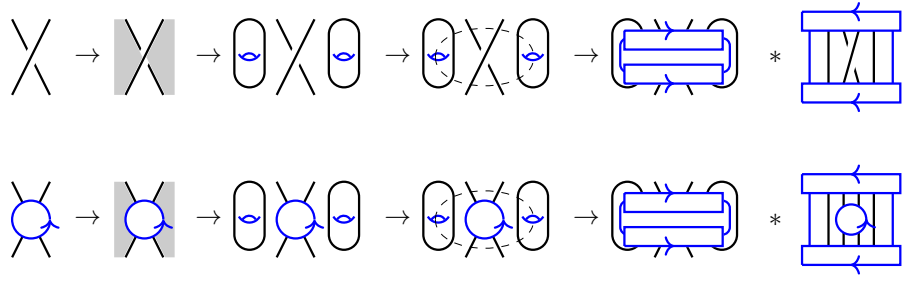
\begin{center}
    \begin{tikzpicture}[line width=.3mm,scale=.5]
        \begin{scope}
            \draw(0,1)--(1,-1);
             \draw[white,WLL](0,-1)--(1,1);
            \draw(0,-1)--(1,1);
        \end{scope}
         \begin{scope}[shift={(3,0)}]
         \node at (-1,0){$\rightarrow$};
            \draw(0,1)--(1,-1);
             \draw[white,WLL](0,-1)--(1,1);
            \draw(0,-1)--(1,1);
            \fill[black!20](-.3,-1)--(-0.05,-1)--(0.45,0)--(-0.05,1)--(-.3,1);

            \fill[black!20](1+.3,-1)--(1+.05,-1)--(0.55,0)--(1+.05,1)--(1+.3,1);
        \end{scope}

        \begin{scope}[shift={(7,0)}]
        \node at (-1.8,0){$\rightarrow$};
        \holestring{.4}{-1}{1};
            \draw(0,1)--(1,-1);
             \draw[white,WLL](0,-1)--(1,1);
            \draw(0,-1)--(1,1);
         \holestring{.4}{1.5}{1};
         
        \end{scope}

        \begin{scope}[shift={(12,0)}]
        \node at (-1.8,0){$\rightarrow$};
        \holestring{.4}{-1}{1};
            \draw(0,1)--(1,-1);
             \draw[white,WLL](0,-1)--(1,1);
            \draw(0,-1)--(1,1);
         \holestring{.4}{1.5}{1};
         \draw[dashed, line width=.1mm](-.8,0)[bend left=90]to(1.8,0);
         \draw[dashed, line width=.1mm](-.8,0)[bend right=90]to(1.8,0);
        \end{scope}

         \begin{scope}[shift={(17,0)}]
        \node at (-1.8,0){$\rightarrow$};
        \noholestring{.4}{-1}{1};
            \draw(0,1)--(1,-1);
             \draw[white,WLL](0,-1)--(1,1);
            \draw(0,-1)--(1,1);
         \noholestring{.4}{1.5}{1};
        \filldraw[white](-.8,0.7)--(1.8,0.7)--(1.8,-.7)--(-.8,-0.7);
        \draw[blue](-.8,0.7)--(1.8,0.7)--(1.8,.2)--(-.8,0.2)--(-.8,0.7);
        \draw[blue,->](0.5,0.7)--(0.51,0.7);

         \draw[blue](-.8,-0.7)--(1.8,-0.7)--(1.8,-.2)--(-.8,-0.2)--(-.8,-0.7);
        \draw[blue,->](0.5,-0.7)--(0.51,-0.7);
        \begin{scope}[shift={(4.7,0)}]
        \node at (-1.5,0){$*$};
        
        \draw(-0.1,0.7)--(-0.1,-.7);
        \draw(0.3,0.7)--(0.7,-.7);
        
        \draw[white,WLL](0.7,0.7)--(0.3,-.7);
        \draw(0.7,0.7)--(0.3,-.7);
        \draw(1.1,0.7)--(1.1,-.7);

         \draw[blue](-.6,1.2-.5)--(-.6,-1.2+.5);
          \draw[blue](1.6,1.2-.5)--(1.6,-1.2+.5);
           \draw[blue](-.8,1.7-.5)--(1.8,1.7-.5)--(1.8,1.2-.5)--(-.8,1.2-.5)--(-.8,1.7-.5);
        \draw[blue,->](0.5,1.7-.5)--(0.5-.01,1.7-.5);

         \draw[blue](-.8,-1.7+.5)--(1.8,-1.7+.5)--(1.8,-1.2+.5)--(-.8,-1.2+.5)--(-.8,-1.7+.5);
        \draw[blue,->](0.5,-1.7+.5)--(0.5-.01,-1.7+.5); 
        \end{scope}
        
        \end{scope}

    \end{tikzpicture}
\end{center}

\begin{center}
    \begin{tikzpicture}[line width=.3mm,scale=.5]
        \begin{scope}
            \draw(0,1)--(1,-1);
             \draw[white,WLL](0,-1)--(1,1);
            \draw(0,-1)--(1,1);
            \fill(0.5,0)[white] circle(.5);
            \draw(1,0)[blue,->]arc[start angle=0,end angle=360,radius=.5];
        \end{scope}
         \begin{scope}[shift={(3,0)}]
         \node at (-1,0){$\rightarrow$};
            \draw(0,1)--(1,-1);
             \draw[white,WLL](0,-1)--(1,1);
            \draw(0,-1)--(1,1);
            \fill[black!20](-.3,-1)--(-0.05,-1)--(0.45,0)--(-0.05,1)--(-.3,1);

            \fill[black!20](1+.3,-1)--(1+.05,-1)--(0.55,0)--(1+.05,1)--(1+.3,1);
            
            \fill(0.5,0)[white] circle(.5);
            \draw(1,0)[blue,->]arc[start angle=0,end angle=360,radius=.5];
        \end{scope}

        \begin{scope}[shift={(7,0)}]
        \node at (-1.8,0){$\rightarrow$};
        \holestring{.4}{-1}{1};
            \draw(0,1)--(1,-1);
             \draw[white,WLL](0,-1)--(1,1);
            \draw(0,-1)--(1,1);
         \holestring{.4}{1.5}{1};
         
            \fill(0.5,0)[white] circle(.5);
            \draw(1,0)[blue,->]arc[start angle=0,end angle=360,radius=.5];
         
        \end{scope}

        \begin{scope}[shift={(12,0)}]
        \node at (-1.8,0){$\rightarrow$};
        \holestring{.4}{-1}{1};
            \draw(0,1)--(1,-1);
             \draw[white,WLL](0,-1)--(1,1);
            \draw(0,-1)--(1,1);
         \holestring{.4}{1.5}{1};
         \draw[dashed, line width=.1mm](-.8,0)[bend left=90]to(1.8,0);
         \draw[dashed, line width=.1mm](-.8,0)[bend right=90]to(1.8,0);
         
            \fill(0.5,0)[white] circle(.5);
            \draw(1,0)[blue,->]arc[start angle=0,end angle=360,radius=.5];
        \end{scope}

        \begin{scope}[shift={(17,0)}]
        \node at (-1.8,0){$\rightarrow$};
        \noholestring{.4}{-1}{1};
            \draw(0,1)--(1,-1);
             \draw[white,WLL](0,-1)--(1,1);
            \draw(0,-1)--(1,1);
         \noholestring{.4}{1.5}{1};
        \filldraw[white](-.8,0.7)--(1.8,0.7)--(1.8,-.7)--(-.8,-0.7);
        \draw[blue](-.8,0.7)--(1.8,0.7)--(1.8,.2)--(-.8,0.2)--(-.8,0.7);
        \draw[blue,->](0.5,0.7)--(0.51,0.7);

         \draw[blue](-.8,-0.7)--(1.8,-0.7)--(1.8,-.2)--(-.8,-0.2)--(-.8,-0.7);
        \draw[blue,->](0.5,-0.7)--(0.51,-0.7);
        \begin{scope}[shift={(4.7,0)}]
        \node at (-1.5,0){$*$};
        
        \draw(-0.1,0.7)--(-0.1,-.7);
        \draw(0.3,0.7)--(0.3,-.7);

        \draw(0.7,0.7)--(0.7,-.7);
        \draw(1.1,0.7)--(1.1,-.7);

        \fill(.5,0)[white]circle(.4);
        \draw(.5,0)[blue]circle(.4);
        \draw[blue,->](.9,0)--(.9,.01);

         \draw[blue](-.6,1.2-.5)--(-.6,-1.2+.5);
          \draw[blue](1.6,1.2-.5)--(1.6,-1.2+.5);
           \draw[blue](-.8,1.7-.5)--(1.8,1.7-.5)--(1.8,1.2-.5)--(-.8,1.2-.5)--(-.8,1.7-.5);
        \draw[blue,->](0.5,1.7-.5)--(0.5-.01,1.7-.5);

         \draw[blue](-.8,-1.7+.5)--(1.8,-1.7+.5)--(1.8,-1.2+.5)--(-.8,-1.2+.5)--(-.8,-1.7+.5);
        \draw[blue,->](0.5,-1.7+.5)--(0.5-.01,-1.7+.5); 
        \end{scope}
        
        \end{scope}

    \end{tikzpicture}
\end{center}
Then the result is a union of three elementary types of CDM: The components containing an unshaded region, an output disc or a braid correspond to three elementary Clifford tensors $k$-qubit $GHZ$, $3$-qubit $Max$ or $S$-gate respectively up to the change of basis by $H$.
Therefore, the CDM $M'$ is Clifford, so is $M$.
\end{proof}

 \color{black}

When the CDM is a 3-disc decorated by Clifford crossings and $n$-output discs, it is a $n$-qubit stabilizer state. Its stabilizer group consists of stabilizers given by {\it cycle operators} on the decorated 4-valent graph on the boundary sphere as shown in \cite{liu2019quantized}.
This observation has been applied to design stabilizer codes \cite{liu2019quantized} and general variational ones \cite{shao2024variational}. Now both Clifford and Matchgates are simulated in QCS, as illustrated in Fig.~\ref{fig:Clifford and Matchgates in QCS}.

\begin{figure}[H]
\begin{center}
\begin{tikzpicture}[line width=.3mm]
\begin{scope}
\node at (1.8, 1){Clifford};
\node at (3.6, 1){Matchgates};

\draw[purple] (1.5,.5) circle (2 and 1);
\draw[red] (3.5,.5) circle (2 and 1);

\begin{scope}[scale=.5,shift={(.7,2)}]
\hole{1}{0}{0};
\end{scope}
\begin{scope}[shift={(2.2,0)},scale=.6]
\draw[line width=.6pt] (1,0) -- (0,1);
\draw[white,WLL] (0,0) -- (1,1);
\draw[line width=.6pt] (0,0) -- (1,1);
\draw[domain=6:24,smooth,variable=\x,red] plot (\x/30, {sin(\x r)/16+0.2});
\fill (.8,.2) circle (.065);
\fill (.2,.2) circle (.065);
\end{scope}
\begin{scope}[shift={(4,0)}, scale=.6]
\draw (0,1)--(1,0);
\draw (0,0) -- (1,1);
\node at (.2,.5) {$\theta$};
\end{scope}
\end{scope}
\end{tikzpicture}
\end{center}
\caption{Clifford and Matchgates in QCS}
\label{fig:Clifford and Matchgates in QCS}
\end{figure}
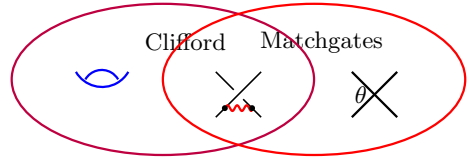

\section{QCS algorithm}
Clifford gates and Matchgates together form a universal gate set, so any quantum circuit or qubit tensor network can be simulated as a CDM in Quon, which has the shape of a handlebody. In this section, we establish the Quon Classical Simulation, in which both Clifford circuits and Matchgates circuits admit efficient classical simulation. The exponential computational complexity comes from a new concept called \textbf{Magic holes}, which combines parameterized crossings and topological holes to capture long-range entanglements.

\subsection{Layers and holes in the 2d-CDM}

We define a connected component of strings as a set of strings that are connected in the following sense ($\alpha\notin \{\pm1,\pm i\}$):
\begin{figure}[H]
\centering
\begin{tikzpicture}[line width=.3mm]
    \begin{scope}
        \draw(0,.5)--(0,-.5);
        \draw(.5,.5)--(.5,-.5);
        \node at (0,.7){$1$};
        \node at (0,-.7){$1$};

        \node at (.5,.7){$2$};
        \node at (.5,-.7){$2$};
    \end{scope}
    \begin{scope}[shift={(1.5,0)}]
        \draw(0,.5)--(0,-.5);
        \draw(.5,.5)--(.5,-.5);
        \draw[red,decorate,decoration={snake}](0,0)--(.5,0);
        \fill(0,0)circle(.06);
        \fill(.5,0)circle(.06);
        \node at (0,.7){$1$};
        \node at (0,-.7){$1$};
        
        \node at (.5,.7){$2$};
        \node at (.5,-.7){$2$};
        
    \end{scope}
    \begin{scope}[shift={(3,0)}]
        \draw(0,.5)--(.5,-.5);
        \draw[white,WLL](.5,.5)--(0,-.5);
        \draw(.5,.5)--(0,-.5);
        \node at (0,.7){$1$};
        \node at (0,-.7){$2$};
        
        \node at (.5,.7){$2$};
        \node at (.5,-.7){$1$};
    \end{scope}

    \begin{scope}[shift={(4.5,0)}]
        \draw(0,.5)--(.5,-.5);
        \draw(.5,.5)--(0,-.5);
        
        \node at (0.5,0){$\alpha$};
        \node at (0,.7){$1$};
        \node at (0,-.7){$1$};
        
        \node at (.5,.7){$1$};
        \node at (.5,-.7){$1$};
    \end{scope}

    \begin{scope}[shift={(6,0)}]
        \draw(0,.5)--(.5,-.5);
        \draw(.5,.5)--(0,-.5);
        
        \fill[white](.5,0)arc[start angle=0,end angle=360,radius=.25];
        \draw[blue,line width=.6pt](.5,0)arc[start angle=0,end angle=360,radius=.25];
        \draw[blue,->](.5,0)--(.5,.01);
        \node at (0,.7){$1$};
        \node at (0,-.7){$1$};
        \node at (.5,.7){$1$};
        \node at (.5,-.7){$1$};
    \end{scope}
\node at (1,-0.5){$,$};
\node at (2.5,-0.5){$,$};
\node at (4,-.5){$,$};
\node at (5.5,-.5){$,$};
\node at (7,-.5){$.$};
\end{tikzpicture}
\caption{Connected component}
\label{connected component}
\end{figure}
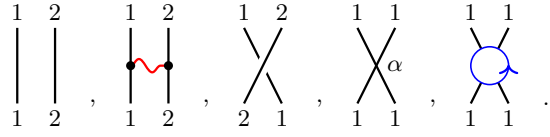

Here strings in the same connected component are labeled with the same number. In the presence of boundaries, all boundary points on an output disc are considered connected.
Since there are braids, different connected components may be linked.
We apply the relation given in Eq.~\eqref{switch} to switch positive and negative braids and to resolve all links at the price of introducing possibly additional charges in this process. 

We now deal with the charges on each connected component. For each component $L$, we re-pair the charges so that all are paired within $L$, except when $L$ contains an odd number of charges—in which case one charge is paired with another outside $L$.
We then apply the charge annihilation relation (Eq.~\eqref{charge_annihilation}) to reduce the number of charges in each layer to at most one.

This process may introduce additional phases, as given in Eqs.~\eqref{relative0}, \eqref{relative1}, \eqref{relative2}, and Eq.~\eqref{charge_crossing_equation}, as well as relative phases from Eq.~\eqref{relativegenus}.
It may also generate wire circles (Eq.~\eqref{wirecircle}) enclosing a hole, as illustrated in Fig.~\ref{bolangquan}.
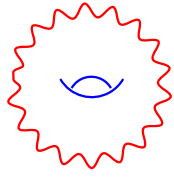
\begin{figure}[H]
    \centering
    \begin{tikzpicture}[line width=.3mm]
        \begin{scope}
        \draw[red,decorate,decoration={snake}](-.5,0)arc[start angle=180,end angle=540,radius=1];
            \hole{0.6}{0.2}{0.8};
        \end{scope}
    \end{tikzpicture}
    \caption{A wire circle enclosing a hole may arise during the charge re-pairing process.}
    \label{bolangquan}
\end{figure}

There remains at most one charge on each connected component. 
To deal with them, we add one auxiliary circle, called the \textbf{parity circle}, and then re-pair charges by Eq.~\eqref{repair}, so that the remaining charges are attached to the parity circle as follows:

%

\begin{tikzequation}
\label{paritycircle}
\begin{scope}[scale=.5]
\draw[domain=-6*pi:0,smooth,variable=\x,red] plot (\x/10, {-sin(\x r)/5});
\draw[white,WL] (-.6*pi+.5,-.4-.5)--(-.6*pi+.5,0-.1-.5) arc (0:180:.5)--++(0,-.3);
\draw (-.6*pi+.5,-.4-.5)--(-.6*pi+.5,0-.1-.5) arc (0:180:.5)--++(0,-.3);

\draw (.5,-.4-.5)--(.5,0-.1-.5) arc (0:180:.5)--++(0,-.3);

\draw[white,WLL](0+.5,1.1)--(-.6*pi-.5,1.1);
\draw(-.9,3)arc[start angle=90,end angle=90+360,radius=1];

\draw[white,WL](0+.5,.8)--(-.6*pi-.5,.8);
\draw[white,WL](0+.5,.5)--(-.6*pi-.5,.5);
\draw(0+.5,.8)--(-.6*pi-.5,.8);
\draw(0+.5,.5)--(-.6*pi-.5,.5);
\node at (-4,1.6){Parity circle};
\node at (0.1,-1.2) {2};
\node at (-.6*pi-.8,-1.2) {Layer 1};
\node at (1.5,.5) {$=$};

\fill (0,-0.1) circle (.1);
\fill (-.6*pi,-0.1) circle (.1);
\end{scope}
\begin{scope}[scale=.5, shift={(5,0)}]

\draw (-.6*pi+.5,-.4-.5)--(-.6*pi+.5,0-.1-.5) arc (0:180:.5)--++(0,-.3);
\draw (.5,-.4-.5)--(.5,0-.1-.5) arc (0:180:.5)--++(0,-.3);
\draw[white,WLL](0+.5,1.1)--(-.6*pi-.5,1.1);
\draw(-.9,3)arc[start angle=90,end angle=90+360,radius=1];

\draw[red,decorate,decoration={snake}](-.6*pi,-0.1)--(-.6*pi+.45,1.1);
\draw[red,decorate,decoration={snake}](0,-0.1)--(-.5,1.1);

\fill (0,-0.1) circle (.1);
\fill (-0.5,1.1) circle (.1);
\fill (-.6*pi,-0.1) circle (.1);
\fill (-.6*pi+.45,1.1) circle (.1);

\draw[white,WL](0+.5,.8)--(-.6*pi-.5,.8);
\draw[white,WL](0+.5,.5)--(-.6*pi-.5,.5);
\draw(0+.5,.8)--(-.6*pi-.5,.8);
\draw(0+.5,.5)--(-.6*pi-.5,.5);
\end{scope}
\node at (3.1,0){$.$};
\end{tikzequation}

After this procedure, we call each connected component a \textbf{Layer}. The wire circle is also considered as a \textbf{Layer}.

\begin{definition}
A hole $g$ is said to be \textbf{not involved} in a given layer $L$ if there exists a curve from $g$ to infinity that does not intersect $L$.
Otherwise $g$ is said to be \textbf{involved} in $L$.
\end{definition}

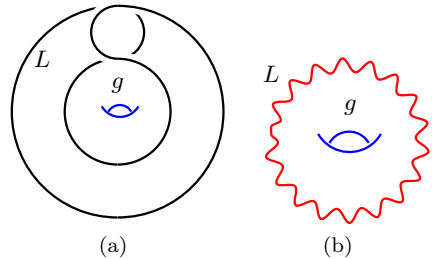
\begin{figure}[H]
\centering
\subfigure[]{
\begin{tikzpicture}[line width=.3mm]
\begin{scope}[scale=.7]
\begin{scope}[shift={(0,0)},scale=.25]
\hole{2}{0.4}{1};
\node at (2,2){$g$};
\node at (-3.7,4){$L$};
\end{scope}
    \begin{scope}
    \draw[white,WLL](.5,1)arc[start angle=-90,end angle=90,radius=.5];
    \draw[white,WLL](.5,-1)arc[start angle=90,end angle=-90,radius=.5];
    \draw[white,WLL](.5,2)arc[start angle=90, end angle=270,radius=2];
    \draw[white,WLL](.5,1)arc[start angle=90,end angle=270,radius=1];
    \draw(.5,1)arc[start angle=-90,end angle=90,radius=.5];
    \draw(.5,2)arc[start angle=90, end angle=270,radius=2];
    \draw(.5,1)arc[start angle=90,end angle=270,radius=1];
    \end{scope}
    
    \begin{scope}[shift={(1,0)},xscale=-1]
    \draw[white,WLL](.5,1)arc[start angle=-90,end angle=90,radius=.5];
    \draw[white,WLL](.5,2)arc[start angle=90, end angle=270,radius=2];
    \draw[white,WLL](.5,1)arc[start angle=90,end angle=270,radius=1];
    \draw(.5,1)arc[start angle=-90,end angle=90,radius=.5];
    \draw(.5,2)arc[start angle=90, end angle=270,radius=2];
    \draw(.5,1)arc[start angle=90,end angle=270,radius=1];
    \end{scope}
\end{scope}
\end{tikzpicture}
}
\subfigure[]{
    \begin{tikzpicture}[line width=.3mm]
        \begin{scope}[shift={(0,3)}]
        \node at (0.55,0.35){$g$};
        \node at (-0.5,0.8){$L$};
        \draw[red,decorate,decoration={snake}](-.5,0)arc[start angle=180,end angle=540,radius=1];
        \hole{0.6}{0.2}{0.8};
        \end{scope}
    \end{tikzpicture}
}

\caption{Examples of a hole $g$ INVOLVED in a layer $L$. 
}
\end{figure}

We then introduce a key notion called the odd hole.
\begin{definition}
A hole $g$ is called an \textbf{odd hole}, if there exists a layer $L$ and a curve $r$ from $g$ to infinity (also denoted by a hole), which does not pass through any crossings, such that $r$ intersects $L$ transversally on strings odd times. 
In this situation, we say that $g$ is an odd hole with respect to the layer $L$.
Otherwise, $g$ is called an \textbf{even hole}.
\end{definition}
Fig.~\ref{fig:odd} gives an example of an odd hole.
\begin{figure}[H]
    \centering
\begin{tikzpicture}[line width=.3mm]
    \begin{scope}[scale=.6]
    \begin{scope}[scale=.55,shift={(-6,1)}]
        
            \draw[blue](.25,-1)arc[start angle=150,end angle=30, radius=.5];
            \draw[blue] (0,-.8)arc[start angle=-150,end angle=-30, radius=.8];
        \end{scope}
    
     \node at (-2.9,1)[blue]{$\infty$};
    \draw(-2,-1.5)--(-2,1.5);
    
        \draw(-1.5,-1.5)--(-1.5,1.5);
        \node at (-1,0){$...$};
        \draw(-.5,-1.5)--(-.5,1.5);
        \draw(0,-1.5)--(0,1.5);
    \begin{scope}[shift={(2,1)}]
    \begin{scope}[shift={(-.3,-.6)}]
        \draw(-.5,0.5)..controls (.8,2)and(2.7,-2)..(3.5,0);
        \draw(3.5,0)..controls (4.5,2)and(-1.5,-2)..(-.5,.5);
    \end{scope}
    \begin{scope}[shift={(1.6,-1.3)}]
        \draw(-.9,-1)..controls(-1,1)and(0,1)..(2,1.3);
        \end{scope}

     \begin{scope}[shift={(-.5,-.6)}]
     \draw[white,WL](-1,0)..controls (-.5,2)and(1.5,2)..(1.5,-1);
        \draw[white,WL](1.5,-1)..controls(2,-4)and(4,2)..(1,-1);
        \draw[white,WL](1,-1)..controls(-.5,-2)and(-1.5,-2)..(-1,0);
        \draw(-1,0)..controls (-.5,2)and(1.5,2)..(1.5,-1);
        \draw(1.5,-1)..controls(2,-4)and(4,2)..(1,-1);
        \draw(1,-1)..controls(-.5,-2)and(-1.5,-2)..(-1,0);
        
    \end{scope}
    \begin{scope}[shift={(3,-1.5)}]
    \end{scope}
        \begin{scope}[scale=.55,shift={(-.5,0.5)}]
        \node at (.5,-.4){$g_0$};
            \draw[blue](.25,-1)arc[start angle=150,end angle=30, radius=.5];
                \draw[blue] (0,-.8)arc[start angle=-150,end angle=-30, radius=.8];
        \end{scope}
        \begin{scope}[scale=.55,shift={(4,0)}]
        \node at (2.6,-.7){$g_{2}$};
            \draw[blue] (.25,-1)arc[start angle=150,end angle=30, radius=.5];
                \draw[blue] (0,-.8)arc[start angle=-150,end angle=-30, radius=.8];
        \end{scope}
        \begin{scope}[scale=.55,shift={(2.1,-1.5)}]
        \node at (.9,-1.67){$g_1$};
            \draw[blue] (.25,-1)arc[start angle=150,end angle=30, radius=.5];
                \draw[blue] (0,-.8)arc[start angle=-150,end angle=-30, radius=.8];
        \end{scope}
    \end{scope}
    \end{scope}
\end{tikzpicture}
    \caption{An example of an odd hole $g_0$}
    \label{fig:odd}
\end{figure}
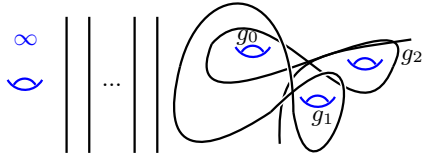

We introduce a notation to deal with holes.
For a 2d-CDM $G$, we can arbitrarily choose a set of holes $h_1,\dots,h_m$, and write each of them into a sum of two terms as in Eq.~\eqref{genus_is_a_variable}, and the partition function is written as a sum of $2^m$ terms:
\begin{equation}
\label{notation_of_genus_as_a_variable}
Z(G) = \frac{1}{2^m}\sum_{g_1,\dots,g_m\in\{0,1\}} Z(G_{h_1=g_1,\dots,h_m=g_m}),
\end{equation}
where the notation $G_{h_1=g_1,\dots,h_m=g_m}$ denotes the 2d-CDM constructed by replacing the hole $h_i$ by one of the two diagrams in Eq.~\eqref{genus_is_a_variable} depending on the value of $g_i\in\{0,1\}$. 
In the rest of the letter we may abuse the notation to write
$
Z(G) = \frac{1}{2^m}\sum_{g_1,\dots,g_m\in\{0,1\}} Z(G_{g_1,\dots,g_m})
$
if no confusions arise.

\begin{definition}
    We say that there is a \textbf{constraint} $g_1+\dots+g_m=x \ (\operatorname{mod}\ 2),x\in\{0,1\}$ in the 2d-CDM $G$, if 
    $$
    Z(G_{g_1,\dots,g_m}) = 0,\quad \text{whenever } g_1+\dots+g_m\ne x.
    $$
\end{definition}
\begin{prop}
Suppose $g_1,...,g_m$ are all odd holes with respect to a layer $L$, then there is a constraint $g_1+...+g_m=x$ for a $x\in\{0,1\}$. 
\end{prop}
\begin{proof}
We perform a genus cut from each $g_i$ to infinity such that $g_i$ lies on each string for odd times. We re-pair, move and annihilate these variables such that there is only one $g_i$ on the layer. Then we use the Yang-Baxter relations to reduce this diagram to a single circle, with variables attached $g_i$ and possibly an additional charge. This circle gives an equation $g_0+...+g_m=x$ where $x=0$ or $1$.
\end{proof}

It is crucial that we can represent the constraint diagrammatically by a double circle surrounding the holes as follows:
\begin{prop}
\label{addtwostrings}
If there is a constraint $g_0+g_1+...+g_m=x$, for $x=0,1$, then:

\begin{center}
    
\begin{tikzpicture}[line width=.3mm]
    \begin{scope}[scale=.6]
        \node at (-.25,0){$...$};
        \draw(-.5,-1)--(-.5,1);
        \draw(0,-1)--(0,1);

    \begin{scope}
    \begin{scope}[scale=.55,shift={(-3,1)}]
        \node at (.6,0){$a$};
        
            \draw[blue](.25,-1)arc[start angle=150,end angle=30, radius=.5];
            \draw[blue] (0,-.8)arc[start angle=-150,end angle=-30, radius=.8];
        \end{scope}
        \begin{scope}[scale=.55,shift={(.9,1)}]
        \node at (.6,0){$g_0$};
        
            \draw[blue](.25,-1)arc[start angle=150,end angle=30, radius=.5];
            \draw[blue] (0,-.8)arc[start angle=-150,end angle=-30, radius=.8];
             \node at (2.3,-1){$...$};
        \end{scope}

    \end{scope}
        
     \begin{scope}[shift={(4,0)}]

        \begin{scope}[scale=.55,shift={(.9,1)}]
        \node at (.6,0){$g_m$};
            \draw[blue](.25,-1)arc[start angle=150,end angle=30, radius=.5];
            \draw[blue] (0,-.8)arc[start angle=-150,end angle=-30, radius=.8];
            
        \draw(3,-2.8)--(3,.8);
        \draw(2,-2.8)--(2,.8);
             \node at (3.6,-1){$...$};
        \end{scope}

        \begin{scope}[scale=.55,shift={(5.5,1)}]
        \node at (.6,0){$b$};
            \draw[blue](.25,-1)arc[start angle=150,end angle=30, radius=.5];
            \draw[blue] (0,-.8)arc[start angle=-150,end angle=-30, radius=.8];
        \end{scope}
        
    \end{scope}

     \begin{scope}[shift={(2,0)}]

        \begin{scope}[scale=.55,shift={(.9,1)}]
        \node at (.6,0){$g_1$};
            \draw[blue](.25,-1)arc[start angle=150,end angle=30, radius=.5];
            \draw[blue] (0,-.8)arc[start angle=-150,end angle=-30, radius=.8];
             \node at (2.3,-1){$...$};
        \end{scope}
        
        \draw(2,-1)--(2,1);
        \draw(1.5,-1)--(1.5,1);
    \end{scope}

     \begin{scope}

        \draw(2,-1)--(2,1);
        \draw(1.5,-1)--(1.5,1);
    \end{scope}
    \end{scope}

     \begin{scope}[scale=.6,shift={(13,0)}]
      \node at (-3,0){$=\frac{1}{2}$};
         \node at (8.5,0){$.$};

    \begin{scope}[shift={(.2,0)},yscale=0.5]
        \draw(0,0)to[bend left=90](5.5,0);\draw(0,0)to[bend right=90](5.5,0);
    \end{scope}
    \begin{scope}[shift={(.43,0)},yscale=0.4,xscale=.9]
        \draw(0,0)to[bend left=90](5.5,0);\draw(0,0)to[bend right=90](5.5,0);
    \end{scope}
       \node at (-.25,0){$...$};
        \draw(-.5,-1)--(-.5,1);
        \draw(0,-1)--(0,1);

    \begin{scope}[scale=.55,shift={(-3,1)}]
        \node at (.6,0){$a$};
        
            \draw[blue](.25,-1)arc[start angle=150,end angle=30, radius=.5];
            \draw[blue] (0,-.8)arc[start angle=-150,end angle=-30, radius=.8];
        \end{scope}

    \begin{scope}

        \begin{scope}[scale=.55,shift={(.9,1)}]
        \node at (.6,0){$g_0$};
            \draw[blue](.25,-1)arc[start angle=150,end angle=30, radius=.5];
            \draw[blue] (0,-.8)arc[start angle=-150,end angle=-30, radius=.8];
             \node at (2.3,-1){$...$};
        \end{scope}
         \draw[white,WL](2,-1)--(2,1);
        \draw[white,WL](1.5,-1)--(1.5,1);
    \end{scope}
        
     \begin{scope}[shift={(4,0)}]

        \begin{scope}[scale=.55,shift={(.9,1)}]
        \node at (.6,0){$g_m$};
            \draw[blue](.25,-1)arc[start angle=150,end angle=30, radius=.5];
            \draw[blue] (0,-.8)arc[start angle=-150,end angle=-30, radius=.8];
            \draw(3,-3)--(3,1);
        \draw(3.8,-3)--(3.8,1);
        
             \node at (4.6,-1){$...$};
        
        \draw[red](1.2,-.5)..controls(1.4,-.1)and(1.6,-.8)..(1.8,-.4);
        \node[red] at (2.5,-.5) {$x$}; 
        \fill(1.2,-.5) circle(.1);
        \fill(1.8,-.4) circle(.1);
        \end{scope}

        \begin{scope}[scale=.55,shift={(6.5,1)}]
        \node at (.6,0){$b$};
            \draw[blue](.25,-1)arc[start angle=150,end angle=30, radius=.5];
            \draw[blue] (0,-.8)arc[start angle=-150,end angle=-30, radius=.8];
        \end{scope}
        
    \end{scope}

     \begin{scope}[shift={(2,0)}]

        \begin{scope}[scale=.55,shift={(.9,1)}]
        \node at (.6,0){$g_1$};
            \draw[blue](.25,-1)arc[start angle=150,end angle=30, radius=.5];
            \draw[blue] (0,-.8)arc[start angle=-150,end angle=-30, radius=.8];
             \node at (2.3,-1){$...$};
        \end{scope}
        \draw[white,WL](2,-1)--(2,1);
        \draw[white,WL](1.5,-1)--(1.5,1);
        
        \draw(2,-1)--(2,1);
        \draw(1.5,-1)--(1.5,1);
    \end{scope}

     \begin{scope}

        \draw(2,-1)--(2,1);
        \draw(1.5,-1)--(1.5,1);
    \end{scope}
        
    \end{scope}

\end{tikzpicture}
\end{center}

\end{prop}
\begin{proof}
Start with the right hand side (RHS): apply genus cut from $g_1,...,g_{m}$ to other genus such that a pair of $g_i$ lies on the circles. The charge $x$ neutralizes the circles. By Eq.~\eqref{evaluation}, the contractible double circles reduce to a scalar factor 2. Multiply by $\frac{1}{2}$, and apply the inverse of genus cut to get the left hand side (LHS). 
\end{proof}

When $x=1$, we add a pair of charges between $g_m$ and $b$ as in Prop.~\ref{addbar}, and we have:
\begin{coro}
\label{addtwostringstwo}
If there is a constraint $g_0+g_1+...+g_m=x$, \begin{center}
    
\begin{tikzpicture}[line width=.3mm]
    \begin{scope}[scale=.6]
        \node at (-.25,0){$...$};
        \draw(-.5,-1)--(-.5,1);
        \draw(0,-1)--(0,1);

    \begin{scope}

    \begin{scope}[scale=.55,shift={(-3,1)}]
        \node at (.6,0){$a$};
        
            \draw[blue](.25,-1)arc[start angle=150,end angle=30, radius=.5];
            \draw[blue] (0,-.8)arc[start angle=-150,end angle=-30, radius=.8];
        \end{scope}
        \begin{scope}[scale=.55,shift={(.9,1)}]
        \node at (.6,0){$g_0$};
        
            \draw[blue](.25,-1)arc[start angle=150,end angle=30, radius=.5];
            \draw[blue] (0,-.8)arc[start angle=-150,end angle=-30, radius=.8];
             \node at (2.3,-1){$...$};
        \end{scope}
    \end{scope}
        
     \begin{scope}[shift={(4,0)}]

        \begin{scope}[scale=.55,shift={(.9,1)}]
        \node at (.6,0){$g_m$};
            \draw[blue](.25,-1)arc[start angle=150,end angle=30, radius=.5];
            \draw[blue] (0,-.8)arc[start angle=-150,end angle=-30, radius=.8];
            
        \draw(3,-2.8)--(3,.8);
        \draw(2,-2.8)--(2,.8);
             \node at (3.6,-1){$...$};
        \end{scope}
        \begin{scope}[scale=.55,shift={(5.5,1)}]
        \node at (.6,0){$b$};
            \draw[blue](.25,-1)arc[start angle=150,end angle=30, radius=.5];
            \draw[blue] (0,-.8)arc[start angle=-150,end angle=-30, radius=.8];
        \end{scope}
        
    \end{scope}

     \begin{scope}[shift={(2,0)}]

        \begin{scope}[scale=.55,shift={(.9,1)}]
        \node at (.6,0){$g_1$};
            \draw[blue](.25,-1)arc[start angle=150,end angle=30, radius=.5];
            \draw[blue] (0,-.8)arc[start angle=-150,end angle=-30, radius=.8];
             \node at (2.3,-1){$...$};
        \end{scope}
        
        \draw(2,-1)--(2,1);
        \draw(1.5,-1)--(1.5,1);
    \end{scope}

     \begin{scope}

        \draw(2,-1)--(2,1);
        \draw(1.5,-1)--(1.5,1);
    \end{scope}
    \end{scope}

     \begin{scope}[scale=.6,shift={(13,0)}]
      \node at (-3,0){\large $=\frac{1}{2}$};
       \node at (8.5,0){$.$};
    \begin{scope}[shift={(.2,0)},yscale=0.5]
        \draw(0,0)to[bend left=90](5.5,0);\draw(0,0)to[bend right=90](5.5,0);
    \end{scope}
    \begin{scope}[shift={(.43,0)},yscale=0.4,xscale=.9]
        \draw(0,0)to[bend left=90](5.5,0);\draw(0,0)to[bend right=90](5.5,0);
    \end{scope}
        \node at (-.25,0){$...$};
        \draw(-.5,-1)--(-.5,1);
        \draw(0,-1)--(0,1);

    \begin{scope}[scale=.55,shift={(-3,1)}]
        \node at (.6,0){$a$};
        
            \draw[blue](.25,-1)arc[start angle=150,end angle=30, radius=.5];
            \draw[blue] (0,-.8)arc[start angle=-150,end angle=-30, radius=.8];
        \end{scope}

    \begin{scope}

        \begin{scope}[scale=.55,shift={(.9,1)}]
        \node at (.6,0){$g_0$};
            \draw[blue](.25,-1)arc[start angle=150,end angle=30, radius=.5];
            \draw[blue] (0,-.8)arc[start angle=-150,end angle=-30, radius=.8];
             \node at (2.3,-1){$...$};
        \end{scope}
         \draw[white,WL](2,-1)--(2,1);
        \draw[white,WL](1.5,-1)--(1.5,1);
        
    \end{scope}
        
     \begin{scope}[shift={(4,0)}]

        \begin{scope}[scale=.55,shift={(.9,1)}]
        \node at (.6,0){$g_m$};
            \draw[blue](.25,-1)arc[start angle=150,end angle=30, radius=.5];
            \draw[blue] (0,-.8)arc[start angle=-150,end angle=-30, radius=.8];
            \draw(3,-3)--(3,1);
        \draw(3.8,-3)--(3.8,1);
        
             \node at (4.6,-1){$...$};
       \draw[red,decorate,decoration={snake}](3,-1)--(3.8,-1);
       \node[red] at (3.4,-.3){$x$};
       \fill(3,-1)circle(.1);
        \fill(3.8,-1)circle(.1);
        \end{scope}

        \begin{scope}[scale=.55,shift={(6.5,1)}]
        \node at (.6,0){$b$};
            \draw[blue](.25,-1)arc[start angle=150,end angle=30, radius=.5];
            \draw[blue] (0,-.8)arc[start angle=-150,end angle=-30, radius=.8];
        \end{scope}

    \end{scope}

     \begin{scope}[shift={(2,0)}]

        \begin{scope}[scale=.55,shift={(.9,1)}]
        \node at (.6,0){$g_1$};
            \draw[blue](.25,-1)arc[start angle=150,end angle=30, radius=.5];
            \draw[blue] (0,-.8)arc[start angle=-150,end angle=-30, radius=.8];
             \node at (2.3,-1){$...$};
        \end{scope}
        \draw[white,WL](2,-1)--(2,1);
        \draw[white,WL](1.5,-1)--(1.5,1);
        
        \draw(2,-1)--(2,1);
        \draw(1.5,-1)--(1.5,1);
    \end{scope}

     \begin{scope}

        \draw(2,-1)--(2,1);
        \draw(1.5,-1)--(1.5,1);
    \end{scope}
        
    \end{scope}

\end{tikzpicture}
\end{center}
where there are $k$ pairs of variables added on previous $2k$ strings between $g_m$ and $b$.
\end{coro}
\subsection{Odd hole handle slides}
\label{handleslidesapp}

Recall that the string-genus relation reduces 1-handle to reshape 3-manifold without an exponential cost, which is an efficient relation to replace Move 2. Now we introduce more such 3D relations to reshape 3-manifolds, inspired by the handle slide in surgery theory and arithmetical properties over the binary field.

The following lemma is a generalized handle slide version of the string genus relation.
\begin{prop}
\label{handleslidesbasic}
If there exists a closed circle that is the lowest layer, i.e. crossing below any other layers on negative braids, then any two strings from the left can be slided to the right.
\begin{align}\label{handleslidesbasic_equation}
    \begin{tikzpicture}[baseline={(current bounding box.center)},line width=.3mm]
\begin{scope}[scale=1.5]
    \begin{scope}
    \begin{scope}[scale=.8]
     \draw(.8,0)arc[start angle=180, end angle=540,radius=.6];
    \begin{scope}[shift={(1.2,0)}]
    \draw[white,WLL](0.1,0.2)--(0.1,1); 
    \draw(0.1,1)--(0.1,.2);
    \draw[white,WLL](-0.1,0.1)--(-0.1,1); 
    \draw(-0.1,1)--(-0.1,.1);
     \draw[white,WLL](0.5,0.1)--(.5,1); 
    \draw(0.5,0.1)--(.5,1); 
    \draw[white,WLL](.3,0.2)--(.3,1); 
    \draw(0.3,0.2)--(.3,1); 
    \end{scope}
       \begin{scope}
         \draw[white,WLL](0.2,1)--(0.2,-1);
         \draw(0.2,1)--(0.2,-1);
        \draw[white,WLL](.5,1)--(.5,-1);
        \draw(.5,1)--(.5,-1);
              \fill[white](1,0)arc[start angle=180, end angle=540,radius=.4];
         \draw(1.8,0)[blue,->]arc[start angle=0, end angle=360,radius=.4];
        \node at (1.4,0){$\forall$};       
    \end{scope}
    \end{scope}
    \end{scope}
\end{scope}
\begin{scope}[shift={(4,0)},scale=1.5]
\node at (-.5,0){$=$};
\node at (2.2,0){$.$};
    \begin{scope}[shift={(0.1,0)},scale=1]
        \draw(0,1)--(0,.5)arc[start angle=180,end angle=360, radius=.4]arc[start angle=180,end angle=0, radius=.5]--++(0,-1)arc[start angle=0,end angle=-180, radius=.5]arc[start angle=0,end angle=180, radius=.4]--++(0,-.5);
    \end{scope}
     \begin{scope}[shift={(.4,0)},scale=1]
         \draw(0,1)--(0,.5)arc[start angle=180,end angle=360, radius=.15]arc[start angle=180,end angle=0, radius=.7]--++(0,-1)arc[start angle=0,end angle=-180, radius=.7]arc[start angle=0,end angle=180, radius=.15]--++(0,-.5);
    \end{scope}
 \begin{scope}[shift={(0,0)}]
       \draw(.88,0)arc[start angle=180, end angle=540,radius=.42];
    \begin{scope}[shift={(1.15,0)},xscale=.7]
   \draw[white,WLL](0.1,0.2)--(0.1,1.3); 
    \draw(0.1,1.3)--(0.1,.2);
    \draw[white,WLL](-0.1,0.1)--(-0.1,1.3); 
    \draw(-0.1,1.3)--(-0.1,.1);
     \draw[white,WLL](0.5,0.1)--(.5,1.3); 
    \draw(0.5,0.1)--(.5,1.3); 
    \draw[white,WLL](.3,0.2)--(.3,1.3); 
    \draw(0.3,0.2)--(.3,1.3); 
    \end{scope}
    \fill[white](1,0)arc[start angle=180, end angle=540,radius=.3]
     \draw(1.6,0)[blue,->]arc[start angle=0, end angle=360,radius=.3];
    \end{scope}
   \node at (1.3,0){$\forall$};
\end{scope}
\end{tikzpicture}
\end{align}
\end{prop}
\begin{proof}
We use unit decomposition of double strings:
\begin{center}
\begin{tikzpicture}[line width=.3mm]
    \begin{scope}[scale=1.5,shift={(0,0)}]
        \node at (-.5,0){$\sqrt{2}LHS=$};
        \node at (1.4,0){$\forall$};
        \draw(0.2,1)--(0.2,.3)arc[start angle=180,end angle=360,radius=.15]--(0.5,1);
        \draw(0.2,-1)--(0.2,-.3)arc[start angle=180,end angle=0,radius=.15]--(0.5,-1);
        \begin{scope}[shift={(.3,0)},scale=.8]
        
            \draw(.8,0)arc[start angle=180, end angle=540,radius=.6];
            \begin{scope}[shift={(1.2,0)}]
            
                \draw[white,WLL](0.1,0.2)--(0.1,1.2); 
                \draw(0.1,1.2)--(0.1,.2);
                \draw[white,WLL](-0.1,0.1)--(-0.1,1.2); 
                \draw(-0.1,1.2)--(-0.1,.1);
                \draw[white,WLL](0.5,0.1)--(.5,1.2); 
                \draw(0.5,0.1)--(.5,1.2); 
                
                \draw[white,WLL](.3,0.2)--(.3,1.2); 
                \draw(0.3,0.2)--(.3,1.2); 
            \end{scope}
            \fill[white](1.8,0)arc[start angle=0, end angle=360,radius=.4];
            \draw(1.8,0)[blue,->]arc[start angle=0, end angle=360,radius=.4];
            \node at (1.4,0){$\forall$};
        \end{scope}
    \end{scope}

    \begin{scope}[scale=1.5,shift={(2.5,0)}]
        \node at (-.25,0){$+$};
        
        \node at (2.1,0){.};
        \node at (1.4,0){$\forall$};
        \draw(0.2,1)--(0.2,.3)arc[start angle=180,end angle=360,radius=.15]--(0.5,1);
        \draw[red,decorate, decoration={snake}](.35,.15)--(.35,-.15);
        \fill(.35,.15) circle(.04);
        \fill(.35,-.15) circle(.04);
        \draw(0.2,-1)--(0.2,-.3)arc[start angle=180,end angle=0,radius=.15]--(0.5,-1);

        \begin{scope}[shift={(.3,0)},scale=.8]
        
            \draw(.8,0)arc[start angle=180, end angle=540,radius=.6];
            \begin{scope}[shift={(1.2,0)}]
                \draw[white,WLL](0.1,0.2)--(0.1,1.2); 
                \draw(0.1,1.2)--(0.1,.2);
                \draw[white,WLL](-0.1,0.1)--(-0.1,1.2); 
                \draw(-0.1,1.2)--(-0.1,.1);
                \draw[white,WLL](0.5,0.1)--(.5,1.2); 
                \draw(0.5,0.1)--(.5,1.2); 
                \draw[white,WLL](.3,0.2)--(.3,1.2); 
                \draw(0.3,0.2)--(.3,1.2); 
            \end{scope}
            \fill[white](1,0)arc[start angle=180, end angle=540,radius=.4];
            \draw(1.8,0)[blue,->]arc[start angle=0, end angle=360,radius=.4];
            \node at (1.4,0){$\forall$};
        \end{scope}
    \end{scope}

\end{tikzpicture}
\end{center}

Let $\sqrt{2}LHS=f_1+f_2$. It is clear that the cap and cup in $f_1$ can be slided to the right by isotopy, so we compute the second term $f_2$:
\begin{center}
    \begin{tikzpicture}[line width=.3mm]
        
\begin{scope}[scale=1.5,shift={(0,0)}]
    \begin{scope}
       \node at (-.25,0){$f_2=$};
        \node at (1.4,0){$\forall$};
    \begin{scope}
         \draw(0.2,1)--(0.2,.3)arc[start angle=180,end angle=360,radius=.15]--(0.5,1);
        \draw[red,decorate, decoration={snake}](.35,.15)--(.35,-.15);
        \fill(.35,.15) circle(.04);
         \fill(.35,-.15) circle(.04);

         \draw(0.2,-1)--(0.2,-.3)arc[start angle=180,end angle=0,radius=.15]--(0.5,-1);
        \draw[white,WLL](1,0)arc[start angle=180, end angle=540,radius=.4];
        
         \draw(1-.05,0)arc[start angle=180, end angle=540,radius=.45];

         \draw[red](.95,.1)arc[start angle=90,end angle=270,radius=.1];
         \fill(.95,.1) circle(.04);
         \fill(.95,-.1) circle(.04);
    \end{scope}
    \begin{scope}[shift={(0.28,0)},scale=.8]
    \begin{scope}[shift={(1.2,0)}]

    \draw[white,WLL](0.1,0.2)--(0.1,1.2); 
    \draw(0.1,1.2)--(0.1,.2);
    \draw[white,WLL](-0.1,0.1)--(-0.1,1.2); 
    \draw(-0.1,1.2)--(-0.1,.1);
     \draw[white,WLL](0.5,0.1)--(.5,1.2); 
    \draw(0.5,0.1)--(.5,1.2); 

    \draw[white,WLL](.3,0.2)--(.3,1.2); 
    \draw(0.3,0.2)--(.3,1.2); 
        
    \end{scope}
    
       \begin{scope}

        \fill[white](1,0)arc[start angle=180, end angle=540,radius=.4];
         \draw(1.8,0)[blue,->]arc[start angle=0, end angle=360,radius=.4];
        \node at (1.4,0){$\forall$};

    \end{scope}

    \end{scope}

    \end{scope}

\end{scope}

\begin{scope}[scale=1.5,shift={(2.5,0)}]
    \begin{scope}
       \node at (-.25,0){$=$};
        \node at (1.4,0){$\forall$};
    \begin{scope}

         \draw(1-.05,0)arc[start angle=180, end angle=540,radius=.45];

         \draw[red,decorate,decoration={snake}](.35,.15)--(1,.1);
         \draw[red,decorate,decoration={snake}](.35,-.15)--(1,-.1);

         \draw[white,WLL](0.2,-1)--(0.2,-.3)arc[start angle=180,end angle=0,radius=.15]--(0.5,-1);
         \draw(0.2,-1)--(0.2,-.3)arc[start angle=180,end angle=0,radius=.15]--(0.5,-1);

         \draw[white,WLL](0.2,1)--(0.2,.3)arc[start angle=180,end angle=360,radius=.15]--(0.5,1);

       \draw(0.2,1)--(0.2,.3)arc[start angle=180,end angle=360,radius=.15]--(0.5,1);

        \fill(.35,.15) circle(.04);
         \fill(.35,-.15) circle(.04);
         \fill(.95,.1) circle(.04);
         \fill(.95,-.1) circle(.04);

    \end{scope}
    
   \begin{scope}[shift={(0.28,0)},scale=.8]
    \begin{scope}[shift={(1.2,0)}]

    \draw[white,WLL](0.1,0.2)--(0.1,1.2); 
    \draw(0.1,1.2)--(0.1,.2);
    \draw[white,WLL](-0.1,0.1)--(-0.1,1.2); 
    \draw(-0.1,1.2)--(-0.1,.1);
     \draw[white,WLL](0.5,0.1)--(.5,1.2); 
    \draw(0.5,0.1)--(.5,1.2); 

    \draw[white,WLL](.3,0.2)--(.3,1.2); 
    \draw(0.3,0.2)--(.3,1.2); 
        
    \end{scope}
    
       \begin{scope}

        \fill[white](1,0)arc[start angle=180, end angle=540,radius=.4];
         \draw(1.8,0)[blue,->]arc[start angle=0, end angle=360,radius=.4];
        \node at (1.4,0){$\forall$};

    \end{scope}

    \end{scope}
    \end{scope}

\end{scope}

\begin{scope}[shift={(8,0)},scale=1.5]
\node at (-.45,0){$=$};

    \begin{scope}[shift={(0,0)}]
        \draw(0,1)--(0,.5)arc[start angle=180,end angle=360, radius=.4]arc[start angle=180,end angle=0, radius=.5]--++(0,-.3)arc[start angle=180,end angle=360,radius=.1];
    \end{scope}

     \begin{scope}[shift={(.3,0)}]
        
         \draw(0,1)--(0,.5)arc[start angle=180,end angle=360, radius=.15]arc[start angle=180,end angle=0, radius=.7]--++(0,-.3);
      \fill(1.6,-0.1)circle(.04);
    \end{scope}
    \draw[white,WL,decorate,decoration=snake](1.7,-0.1)--(1.9,-.1);
    \draw[red,decorate,decoration=snake](1.9,-0.1)--(1.7,-.1);
    \draw[red,decorate,decoration=snake](1.9,0.1)--(1.7,.1);
    \begin{scope}[yscale=-1,shift={(.3,0)}]
        
         \draw(0,1)--(0,.5)arc[start angle=180,end angle=360, radius=.15]arc[start angle=180,end angle=0, radius=.7]--++(0,-.3);
         
        \fill(1.6,-0.1)circle(.04);
    \end{scope}

    \begin{scope}[yscale=-1,shift={(0,0)}]
       \draw[white,WL](0,1)--(0,.5)arc[start angle=180,end angle=360, radius=.4]arc[start angle=180,end angle=0, radius=.5]--++(0,-.3)arc[start angle=180,end angle=360,radius=.1];
       \draw(0,1)--(0,.5)arc[start angle=180,end angle=360, radius=.4]arc[start angle=180,end angle=0, radius=.5]--++(0,-.3)arc[start angle=180,end angle=360,radius=.1];
    \end{scope}

    \fill(1.7,-.1)circle(.04);
     \fill(1.7,.1)circle(.04);
     
    \fill(1.9,-.1)circle(.04);
     \fill(1.9,.1)circle(.04);
    
    \begin{scope}[shift={(0,0)}]
       \draw(.88,0)arc[start angle=180, end angle=540,radius=.42];
    \begin{scope}[shift={(1.15,0)},xscale=.7]

   \draw[white,WLL](0.1,0.2)--(0.1,1.3); 
    \draw(0.1,1.3)--(0.1,.2);
    \draw[white,WLL](-0.1,0.1)--(-0.1,1.3); 
    \draw(-0.1,1.3)--(-0.1,.1);
     \draw[white,WLL](0.5,0.1)--(.5,1.3); 
    \draw(0.5,0.1)--(.5,1.3); 

    \draw[white,WLL](.3,0.2)--(.3,1.3); 
    \draw(0.3,0.2)--(.3,1.3); 
        
    \end{scope}
    \fill[white](1,0)arc[start angle=180, end angle=540,radius=.3]
      \draw(1.6,0)[blue,->]arc[start angle=0, end angle=360,radius=.3];

     \node at (1.3,0){$\forall$};
    \end{scope}

\end{scope}

\begin{scope}[shift={(12,0)},scale=1.5]
\node at (-.25,0){$=$};

    \begin{scope}[shift={(0,0)}]
        \draw(0,1)--(0,.5)arc[start angle=180,end angle=360, radius=.4]arc[start angle=180,end angle=0, radius=.5]--++(0,-.3)arc[start angle=180,end angle=360,radius=.1];
    \end{scope}

    \begin{scope}[yscale=-1,shift={(0,0)}]
        \draw(0,1)--(0,.5)arc[start angle=180,end angle=360, radius=.4]arc[start angle=180,end angle=0, radius=.5]--++(0,-.3)arc[start angle=180,end angle=360,radius=.1];
    \end{scope}

    \draw(1.7,-.1)[red,bend right=90]to(1.75,.1);
    \fill(1.7,-.1)circle(.04);
     \fill(1.7,.1)circle(.04);
     
    \draw(1.9,-0.1)[red,bend left=45]to(1.9,0);
    \draw(1.9,0)[red,bend right=45]to(1.9,.1);

     \begin{scope}[shift={(.3,0)}]
        
         \draw(0,1)--(0,.5)arc[start angle=180,end angle=360, radius=.15]arc[start angle=180,end angle=0, radius=.7]--++(0,-.3);
      \fill(1.6,-0.1)circle(.04);
    \end{scope}
    
    \begin{scope}[yscale=-1,shift={(.3,0)}]
        
         \draw(0,1)--(0,.5)arc[start angle=180,end angle=360, radius=.15]arc[start angle=180,end angle=0, radius=.7]--++(0,-.3);
         
        \fill(1.6,-0.1)circle(.04);
    \end{scope}

    \begin{scope}[shift={(0,0)}]
       \draw(.88,0)arc[start angle=180, end angle=540,radius=.42];
    \begin{scope}[shift={(1.15,0)},xscale=.7]

   \draw[white,WLL](0.1,0.2)--(0.1,1.3); 
    \draw(0.1,1.3)--(0.1,.2);
    \draw[white,WLL](-0.1,0.1)--(-0.1,1.3); 
    \draw(-0.1,1.3)--(-0.1,.1);
     \draw[white,WLL](0.5,0.1)--(.5,1.3); 
    \draw(0.5,0.1)--(.5,1.3); 

    \draw[white,WLL](.3,0.2)--(.3,1.3); 
    \draw(0.3,0.2)--(.3,1.3); 
        
    \end{scope}
    \fill[white](1,0)arc[start angle=180, end angle=540,radius=.3]
      \draw(1.6,0)[blue,->]arc[start angle=0, end angle=360,radius=.3];

     \node at (1.3,0){$\forall$};
      \node at (2.3,0){;};
    \end{scope}
   
\end{scope}
    \end{tikzpicture}
\end{center}
\begin{center}
\begin{tikzpicture}[line width=.3mm]

\begin{scope}[shift={(0,-6)},scale=1.5]

\node at (-.5,0){$\sqrt{2}LHS=$};
\node at (5.5,0){$=\sqrt{2}RHS.$};
    \begin{scope}[shift={(0,0)}]
        \draw(0,1)--(0,.5)arc[start angle=180,end angle=360, radius=.4]arc[start angle=180,end angle=0, radius=.5]--++(0,-.3)arc[start angle=180,end angle=360,radius=.1];
    \end{scope}

    \begin{scope}[yscale=-1,shift={(0,0)}]
        \draw(0,1)--(0,.5)arc[start angle=180,end angle=360, radius=.4]arc[start angle=180,end angle=0, radius=.5]--++(0,-.3)arc[start angle=180,end angle=360,radius=.1];
    \end{scope}

    \begin{scope}[shift={(.3,0)}]
        
         \draw(0,1)--(0,.5)arc[start angle=180,end angle=360, radius=.15]arc[start angle=180,end angle=0, radius=.7]--++(0,-.3);
    \end{scope}

    \begin{scope}[yscale=-1,shift={(.3,0)}]
        
         \draw(0,1)--(0,.5)arc[start angle=180,end angle=360, radius=.15]arc[start angle=180,end angle=0, radius=.7]--++(0,-.3);
    \end{scope}

      \begin{scope}[shift={(0,0)}]
       \draw(.88,0)arc[start angle=180, end angle=540,radius=.42];
    \begin{scope}[shift={(1.15,0)},xscale=.7]

   \draw[white,WLL](0.1,0.2)--(0.1,1.3); 
    \draw(0.1,1.3)--(0.1,.2);
    \draw[white,WLL](-0.1,0.1)--(-0.1,1.3); 
    \draw(-0.1,1.3)--(-0.1,.1);
     \draw[white,WLL](0.5,0.1)--(.5,1.3); 
    \draw(0.5,0.1)--(.5,1.3); 

    \draw[white,WLL](.3,0.2)--(.3,1.3); 
    \draw(0.3,0.2)--(.3,1.3); 
        
    \end{scope}
    \fill[white](1,0)arc[start angle=180, end angle=540,radius=.3]
      \draw(1.6,0)[blue,->]arc[start angle=0, end angle=360,radius=.3];

     \node at (1.3,0){$\forall$};
    \end{scope}
   
\end{scope}

\begin{scope}[shift={(4,-6)},scale=1.5]

\node at (-.25,0){$+$};
    \begin{scope}[shift={(0,0)}]
        \draw(0,1)--(0,.5)arc[start angle=180,end angle=360, radius=.4]arc[start angle=180,end angle=0, radius=.5]--++(0,-.3)arc[start angle=180,end angle=360,radius=.1];
    \end{scope}

    \begin{scope}[yscale=-1,shift={(0,0)}]
        \draw(0,1)--(0,.5)arc[start angle=180,end angle=360, radius=.4]arc[start angle=180,end angle=0, radius=.5]--++(0,-.3)arc[start angle=180,end angle=360,radius=.1];
    \end{scope}

     \begin{scope}[shift={(.3,0)}]
        
         \draw(0,1)--(0,.5)arc[start angle=180,end angle=360, radius=.15]arc[start angle=180,end angle=0, radius=.7]--++(0,-.3);
      \fill(1.6,-0.1)circle(.04);
    \end{scope}
    
    \draw(1.9,-0.1)[red,bend left=45]to(1.9,0);
    \draw(1.9,0)[red,bend right=45]to(1.9,.1);
    \begin{scope}[yscale=-1,shift={(.3,0)}]
        
         \draw(0,1)--(0,.5)arc[start angle=180,end angle=360, radius=.15]arc[start angle=180,end angle=0, radius=.7]--++(0,-.3);
         
        \fill(1.6,-0.1)circle(.04);
    \end{scope}

      \begin{scope}[shift={(0,0)}]
       \draw(.88,0)arc[start angle=180, end angle=540,radius=.42];
    \begin{scope}[shift={(1.15,0)},xscale=.7]

   \draw[white,WLL](0.1,0.2)--(0.1,1.3); 
    \draw(0.1,1.3)--(0.1,.2);
    \draw[white,WLL](-0.1,0.1)--(-0.1,1.3); 
    \draw(-0.1,1.3)--(-0.1,.1);
     \draw[white,WLL](0.5,0.1)--(.5,1.3); 
    \draw(0.5,0.1)--(.5,1.3); 

    \draw[white,WLL](.3,0.2)--(.3,1.3); 
    \draw(0.3,0.2)--(.3,1.3); 
        
    \end{scope}
    \fill[white](1,0)arc[start angle=180, end angle=540,radius=.3]
      \draw(1.6,0)[blue,->]arc[start angle=0, end angle=360,radius=.3];

     \node at (1.3,0){$\forall$};
    \end{scope}
   
\end{scope}
\end{tikzpicture}
\end{center}

\end{proof}

This handle slides can be generated to cases where there is a constraint by previous corollary. So we obtain the following key lemma to eliminate an odd hole:
\begin{lemma}[Odd hole Handle slides]
\label{handleslideslem}
Given an odd hole $g_0$ and pick a constraint $g_0+g_1+...+g_m=x$, $x=0\ or\ 1$, then $g_0$ can be eliminated by the following identity:
\begin{center}
\begin{tikzpicture}[line width=.3mm]

    \begin{scope}[scale=.6]
    
        \node at (-.25,0){$...$};
        \draw(-.5,-1)--(-.5,1);
        \draw(0,-1)--(0,1);

    \begin{scope}
    
    \begin{scope}[scale=.55,shift={(-3,1)}]
     \node at (.6,0){$a$};
        
            \draw[blue](.25,-1)arc[start angle=150,end angle=30, radius=.5];
            \draw[blue] (0,-.8)arc[start angle=-150,end angle=-30, radius=.8];
        \end{scope}
        \begin{scope}[scale=.55,shift={(.9,1)}]
        \node at (.6,0){$g_0$};
        
            \draw[blue](.25,-1)arc[start angle=150,end angle=30, radius=.5];
            \draw[blue] (0,-.8)arc[start angle=-150,end angle=-30, radius=.8];
             \node at (2.3,-1){$...$};
        \end{scope}
    \end{scope}
        
     \begin{scope}[shift={(4,0)}]

        \begin{scope}[scale=.55,shift={(.9,1)}]
        \node at (.6,0){$g_m$};
            \draw[blue](.25,-1)arc[start angle=150,end angle=30, radius=.5];
            \draw[blue] (0,-.8)arc[start angle=-150,end angle=-30, radius=.8];
            
        \draw(3,-2.8)--(3,.8);
        \draw(2,-2.8)--(2,.8);
             \node at (2.6,-1){$...$};
        \end{scope}

    \end{scope}

     \begin{scope}[shift={(2,0)}]

        \begin{scope}[scale=.55,shift={(.9,1)}]
        \node at (.6,0){$g_1$};
            \draw[blue](.25,-1)arc[start angle=150,end angle=30, radius=.5];
            \draw[blue] (0,-.8)arc[start angle=-150,end angle=-30, radius=.8];
             \node at (2.3,-1){$...$};
        \end{scope}
        
        \draw(2,-1)--(2,1);
        \draw(1.5,-1)--(1.5,1);
    \end{scope}

     \begin{scope}

        \draw(2,-1)--(2,1);
        \draw(1.5,-1)--(1.5,1);
    \end{scope}
        \begin{scope}[scale=.55,shift={(12,1)}]
        \node at (.6,0){$b$};
            \draw[blue](.25,-1)arc[start angle=150,end angle=30, radius=.5];
            \draw[blue] (0,-.8)arc[start angle=-150,end angle=-30, radius=.8];
        \end{scope}
    \end{scope}
    
    \begin{scope}[shift={(6,0)},scale=.6]
    
     \node at (-1,0){$=\frac{1}{2}$};
      \node at (9.5,0){$.$};
    \begin{scope}[scale=.55,shift={(0,1)}]
        \node at (.6,0){$a$};
        
            \draw[blue](.25,-1)arc[start angle=150,end angle=30, radius=.5];
            \draw[blue] (0,-.8)arc[start angle=-150,end angle=-30, radius=.8];
        \end{scope}
        \node at (.4,-2){$...$};
        \node at (.4,2){$...$};
      \draw(-.5,-3)..controls(1,1)and(5.8,-3)..(6,-1)--(6,1);
      \draw(-.5,3)..controls(1,-1)and(5.8,3)..(6,1);

        \draw(0,-3)..controls(1,0)and(6,-3.5)..(6.5,-1)--(6.5,1)..controls(6,3.5)and(1,0)..(0,3);

    \begin{scope}

        \begin{scope}[scale=.55,shift={(.9,1)}]
             \node at (2.3,-1){$...$};
        \end{scope}
        \draw[white,WL](2,-3)--(2,3);
        \draw[white,WL](1.5,-3)--(1.5,3);
        \draw(2,-3)--(2,3);
        \draw(1.5,-3)--(1.5,3);
    \end{scope}
        
     \begin{scope}[shift={(4,0)}]

        \begin{scope}[scale=.55,shift={(.9,1)}]
        \node at (.6,0){$g_m$};
            \draw[blue](.25,-1)arc[start angle=150,end angle=30, radius=.5];
            \draw[blue] (0,-.8)arc[start angle=-150,end angle=-30, radius=.8];
         \draw(5.4,-6.3)--(5.4,5);
        \draw(4.5,-6.3)--(4.5,5);
        
             \node at (6.5,-1){$...$};
             \node at (6,-.3){$x$};
             \draw[red,decorate,decoration={snake}](4.5,-.3)--(5.4,-.3);
             \fill (4.5,-.3) circle(.1);
              \fill (5.4,-.3) circle(.1);
        
        \end{scope}
        \begin{scope}[scale=.55,shift={(8,1)}]
        \node at (.6,0){$b$};
            \draw[blue](.25,-1)arc[start angle=150,end angle=30, radius=.5];
            \draw[blue] (0,-.8)arc[start angle=-150,end angle=-30, radius=.8];
        \end{scope}

    \end{scope}

     \begin{scope}[shift={(2,0)}]

        \begin{scope}[scale=.55,shift={(.9,1)}]
        \node at (.6,0){$g_1$};
            \draw[blue](.25,-1)arc[start angle=150,end angle=30, radius=.5];
            \draw[blue] (0,-.8)arc[start angle=-150,end angle=-30, radius=.8];
             \node at (2.3,-1){$...$};
        \end{scope}
        \draw[white,WL](2,-3)--(2,3);
        \draw[white,WL](1.5,-3)--(1.5,3);
        
        \draw(2,-3)--(2,3);
        \draw(1.5,-3)--(1.5,3);
    \end{scope}

     \begin{scope}
        \draw(2,-1)--(2,1);
        \draw(1.5,-1)--(1.5,1);
    \end{scope}
    \end{scope}

\end{tikzpicture}
\end{center}where there are $k$ pairs of variables added on previous $2k$ strings between $g_m$ and $b$.
\end{lemma}

\begin{proof}
Add double circles by Cor.~\ref{addtwostringstwo},
and then by Prop.~\ref{handleslidesbasic} we can slide all strings between $g_0$ and $a$ except for the two strings given by double circles.

\begin{center}
    
\begin{tikzpicture}[line width=.3mm]
\begin{scope}[scale=.6]
      \node at (-3.5,0){\large{$LHS=\frac{1}{2}$}};
    \begin{scope}[shift={(.2,0)},yscale=0.5]
        \draw(0,0)to[bend left=90](5.5,0);\draw(0,0)to[bend right=90](5.5,0);
    \end{scope}
    \begin{scope}[shift={(.43,0)},yscale=0.4,xscale=.9]
        \draw(0,0)to[bend left=90](5.5,0);\draw(0,0)to[bend right=90](5.5,0);
    \end{scope}
        \node at (-.25,0){$...$};
        \draw(-.5,-1)--(-.5,1);
        \draw(0,-1)--(0,1);

    \begin{scope}[scale=.55,shift={(-3,1)}]
        \node at (.6,0){$a$};
        
            \draw[blue](.25,-1)arc[start angle=150,end angle=30, radius=.5];
            \draw[blue] (0,-.8)arc[start angle=-150,end angle=-30, radius=.8];
        \end{scope}

        \begin{scope}[scale=.55,shift={(.9,1)}]
        \node at (.6,0){$g_0$};
            \draw[blue](.25,-1)arc[start angle=150,end angle=30, radius=.5];
            \draw[blue] (0,-.8)arc[start angle=-150,end angle=-30, radius=.8];
             \node at (2.3,-1){$...$};
        \end{scope}
         \draw[white,WL](2,-1)--(2,1);
        \draw[white,WL](1.5,-1)--(1.5,1);
        
     \begin{scope}[shift={(4,0)}]

        \begin{scope}[scale=.55,shift={(.9,1)}]
        \node at (.6,0){$g_m$};
            \draw[blue](.25,-1)arc[start angle=150,end angle=30, radius=.5];
            \draw[blue] (0,-.8)arc[start angle=-150,end angle=-30, radius=.8];
            \draw(3,-3)--(3,1);
        \draw(3.8,-3)--(3.8,1);
        
             \node at (4.6,-1){$...$};
       \draw[red,decorate,decoration={snake}](3,-1)--(3.8,-1);
       \node[red] at (3.4,-.3){$x$};
       \fill(3,-1)circle(.1);
        \fill(3.8,-1)circle(.1);
        \end{scope}

        \begin{scope}[scale=.55,shift={(6.5,1)}]
        \node at (.6,0){$b$};
            \draw[blue](.25,-1)arc[start angle=150,end angle=30, radius=.5];
            \draw[blue] (0,-.8)arc[start angle=-150,end angle=-30, radius=.8];
        \end{scope}

    \end{scope}

     \begin{scope}[shift={(2,0)}]

        \begin{scope}[scale=.55,shift={(.9,1)}]
        \node at (.6,0){$g_1$};
            \draw[blue](.25,-1)arc[start angle=150,end angle=30, radius=.5];
            \draw[blue] (0,-.8)arc[start angle=-150,end angle=-30, radius=.8];
             \node at (2.3,-1){$...$};
        \end{scope}
        \draw[white,WL](2,-1)--(2,1);
        \draw[white,WL](1.5,-1)--(1.5,1);
        
        \draw(2,-1)--(2,1);
        \draw(1.5,-1)--(1.5,1);
    \end{scope}

        \draw(2,-1)--(2,1);
        \draw(1.5,-1)--(1.5,1);
        
    \end{scope}
    \begin{scope}[scale=.6,shift={(12,0)}]
    
     \node at (-2.7,0){\large{$=\frac{1}{2}$}};
     
      \node at (9.5,0){$.$};
    \begin{scope}[shift={(.2,0)},yscale=0.5]
        \draw(0,0)to[bend left=90](5.5,0);\draw(0,0)to[bend right=90](5.5,0);
    \end{scope}
    \begin{scope}[shift={(.43,0)},yscale=0.4,xscale=.9]
        \draw(0,0)to[bend left=90](5.5,0);\draw(0,0)to[bend right=90](5.5,0);
    \end{scope}
    \begin{scope}[scale=.55,shift={(-3,1)}]
        \node at (.6,0){$a$};
        
            \draw[blue](.25,-1)arc[start angle=150,end angle=30, radius=.5];
            \draw[blue] (0,-.8)arc[start angle=-150,end angle=-30, radius=.8];
        \end{scope}
        \node at (.5,-2){$...$};
        \node at (.5,2){$...$};
      \draw(-.5,-3)..controls(1,1)and(5.8,-3)..(6,-1)--(6,1);
      \draw(-.5,3)..controls(1,-1)and(5.8,3)..(6,1);
        \begin{scope}[yscale=.6]
         \end{scope}

        \draw(0,-3)..controls(1,0)and(6,-3.5)..(6.5,-1)--(6.5,1)..controls(6,3.5)and(1,0)..(0,3);

    \begin{scope}

        \begin{scope}[scale=.55,shift={(.9,1)}]
        \node at (.6,0){$g_0$};
            \draw[blue](.25,-1)arc[start angle=150,end angle=30, radius=.5];
            \draw[blue] (0,-.8)arc[start angle=-150,end angle=-30, radius=.8];
             \node at (2.3,-1){$...$};
        \end{scope}
        \draw[white,WL](2,-3)--(2,3);
        \draw[white,WL](1.5,-3)--(1.5,3);
        \draw(2,-3)--(2,3);
        \draw(1.5,-3)--(1.5,3);
    \end{scope}
        
     \begin{scope}[shift={(4,0)}]

        \begin{scope}[scale=.55,shift={(.9,1)}]
        \node at (.6,0){$g_m$};
            \draw[blue](.25,-1)arc[start angle=150,end angle=30, radius=.5];
            \draw[blue] (0,-.8)arc[start angle=-150,end angle=-30, radius=.8];
         \draw(5.4,-6.3)--(5.4,5);
        \draw(4.5,-6.3)--(4.5,5);
        
             \node at (6.5,-1){$...$};
        \end{scope}

    \end{scope}

     \begin{scope}[shift={(2,0)}]

        \begin{scope}[scale=.55,shift={(.9,1)}]
        \node at (.6,0){$g_1$};
            \draw[blue](.25,-1)arc[start angle=150,end angle=30, radius=.5];
            \draw[blue] (0,-.8)arc[start angle=-150,end angle=-30, radius=.8];
             \node at (2.3,-1){$...$};
        \end{scope}
        \draw[white,WL](2,-3)--(2,3);
        \draw[white,WL](1.5,-3)--(1.5,3);
        
        \draw(2,-3)--(2,3);
        \draw(1.5,-3)--(1.5,3);
    \end{scope}

     \begin{scope}

        \draw(2,-1)--(2,1);
        \draw(1.5,-1)--(1.5,1);
    \end{scope}
        \begin{scope}[scale=.55,shift={(15.3,1)}]
        \node at (.6,0){$b$};
            \draw[blue](.25,-1)arc[start angle=150,end angle=30, radius=.5];
            \draw[blue] (0,-.8)arc[start angle=-150,end angle=-30, radius=.8];
        \end{scope}

        \begin{scope}[shift={(1.6,-1.3)}]
        \end{scope}
        \node at (7.2,0.5){$x$};
        \draw[red,decorate,decoration={snake}](6.96,0)--(7.45,0);
        \fill (6.97,0) circle(.05);
        \fill (7.48,0) circle(.05);

    \end{scope}

\end{tikzpicture}
\end{center}

Then, since there are only two strings between $g_0$ and $a$, the genus $g_0$ can be eliminated by Eq.~\eqref{cuttwo}, multiplied by $\frac{1}{\sqrt{2}}$. The double-circle is cut into a contractible circle and reduces to a scalar $\sqrt{2}$, which gives RHS.

\end{proof}

\input{QCS_even}
\subsection{QCS algorithm}

We are now ready to present the unified algorithm of this work: the \textbf{Quon Classical Simulation (QCS)} algorithm.
It proceeds in three fundamental steps, each designed to systematically reduce the topological complexity of a given 2d-CDM.

QCS algorithm:
\begin{itemize}
\item
\textbf{Step 1:} Component identification and layer construction. 

Given a 2d-CDM, identify all connected components and resolve the links between them. Then insert a parity circle and re-pair charges to it to construct layers, as described in Eq.~\eqref{paritycircle}.

\item
\textbf{Step 2:} Odd hole elimination. 

Apply Lem.~\ref{handleslideslem} to an odd hole. Repeat this process until no odd holes remain.

\item
\textbf{Step 3:} Even Clifford hole elimination. 

If all remaining holes are even, search for a Clifford cut and apply Thm.~\ref{Cliffordhole}. If the reduction falls into the case described by Eq.~\eqref{C1}, new odd holes may be introduced, in which case return to Step 2.
\end{itemize}
At this point, the closed 2d-CDM is reduced to a configuration where all remaining holes are Magic holes.

For open tensor networks $T_1$ and $T_2$ with time boundaries, we represent them as 2d-CDMs and apply the QCS procedure, resulting in simplified CDMs $Q(T_1)$ and $Q(T_2)$, each consisting of layers and holes. 
Denote their parity circles by $C_1$ and $C_2$, respectively. To contract the two tensors, we first re-pair charges on parity circles such that each pair of charges is attached to different layers. 
Then we glue their time boundaries and connect the corresponding strings. 
After contraction, new holes may be introduced, and we apply the QCS procedure again for further reduction, during which a new parity circle is created. 

Contracting two Clifford tensors may create odd holes or even Clifford holes, while contracting two Matchgate tensors produces holes that can be eliminated by the string-genus relation. 
Neither operation generates Magic holes.
However, contracting a Clifford tensor with a Matchgate tensor may result in the formation of Magic holes.

After the QCS procedure, the original 2d-CDM is reduced to one containing only Magic holes. 
A straightforward approach is to perform genus cuts (Lem.~\ref{stringgenusn}), introducing binary variables at the cost of doubling the number of diagrams with each cut. This results in a collection of $2^k$ 2D-CDMs, denoted $G_i, i=1,\cdots,2^k$ where $k$ is the number of Magic holes. Each $G_i$ is a hole-free CDM, i.e., a Matchgate Quon diagram. According to subsubsection \S \ref{mg_and_graph}, the value of each such diagram can be efficiently computed using the FKT algorithm. This leads to the central result of this work:

\begin{theorem}
\label{mainthmapp}
    Let $Q$ be a closed 2d-CDM. Then its value is given by
    \begin{equation}
    \label{maineq}
    Z(Q) =\left(\frac{1}{2}\right)^{|O|+|C1|+\frac{|C2|}{2}} \sum_{i=1}^{2^k}
    \text{Pfaffian}\left(\mathrm{Adj}[\mathfrak{F}(\widehat{G}_i(O;C))] \right),
    \end{equation}
    where:
    \begin{itemize}
    \item
    $O$ is the set of odd holes eliminated via Lem.~\ref{handleslideslem};
    \item
    $C=C_1\sqcup C_2$ are sets of even Clifford holes eliminated via Thm.~\ref{Cliffordhole}, corresponding to cases in Eq.~\eqref{C1} and~\eqref{C2} respectively;
    \item
    $k$ is the number of Magic holes after QCS;
    \item
    $G_i(O;C)$ denotes the $i$-th diagram after Magic hole genus cut;
    \item
    $\mathfrak{F}(\hat{\cdot})$ is the translation from a Matchgate Quon diagram to a graph-theoretic representation;
    \item
    $\mathrm{Adj}(G)$ is the adjacency matrix of $G$ under a Pfaffian orientation.
    \end{itemize}
The time complexity of evaluating this expression is $O(n^32^k)$, where $n$ denotes the total number of crossings, charges, and holes.
\end{theorem}

\begin{coro}
Clifford and Matchgate tensor networks can be classically simulated in polynomial time, as they contain no Magic holes—that is, $k=0$.
\end{coro}

There exists a much broader class of quantum circuits that can be classically simulated beyond Clifford and Matchgates,
provided that they contain at most $O(\log n)$ Magic holes. 
This is because such circuits can be expanded as a sum over a polynomial number of Clifford diagrams. 
After this transformation, all remaining holes become Clifford, making the entire diagram classically simulable.
\begin{coro}
    Any 2d-CDM with at most $O(\log n)$ Magic holes can be classically simulated in polynomial time, where $n$ is an upper bound on the total number of braids, charges, and crossings in the diagram.
\end{coro}

In \S~\ref{ttn}, we present additional relations to reduce the number of Magic holes.

\subsection{Layers and holes in 3D handlebodies}
\label{homologysection}
For readers interested in holes and layers in 3D handlebodies, we give the following formal interpretations. 
A hole is an element of the first \textbf{homology}, and a layer is an element of the first \textbf{cohomology}.

\label{Homology}
 Given a crossing-decorated 3-manifold $M$, if $M$ comes from a tensor network, there is a graph $\Gamma$, where degree-$k$ tensors are degree-$k$ vertices and their contractions are edges. Then $M$ is 3-dimensional neighborhood of $\Gamma$, denoted by $M=N(\Gamma)$, while tensor network representation uses only 0 and 1-skeleton of the manifold $M$. In this manifold, a ``hole" is a cycle of $\Gamma$. Since two cycles with a common edge merge into another cycle, we can define addition of two cycles. Then we choose a basis of cycles. The number of holes that we choose is equal to the first Betti number of the graph $\Gamma$: $|E|+|V|-1$, which is the rank of the first homology group $H_1$. The detailed definitions are as follows:

The standard $0$-simplex is a vertex; the standard $1$-simplex is an edge; the standard $2$-simplex is a triangle; the $3$-simplex is a tetrahedron. They are all defined to sit in a ``standard" position.

A $k$-simplex in a manifold $M$ is a continuous map from the standard $k$-simplex to $M$. 
Denote a $k$-simplex in $M$ as $[v_0,...,v_{k}]$. 
Let $C_k(M)$ be the formal $\mathbb{Z}_2$-linear combination of all $k$-simplices in $M$, which is called the $k$-chain group. 
Define the boundary map $\partial_k:C_k(M)\rightarrow C_{k-1}(M)$, which maps each $k$-simplex to the sum of its $(k-1)$-dimensional faces: 
$$
\partial_k([v_0,v_1,...,v_{k}]):=\sum\limits_{i=0}^k[v_0,...,\hat{v_i},...,v_{k}],
$$
where $\hat{v_i}$ denotes omission of the $i$-th vertex.

We have $\partial_k\circ \partial_{k-1}=0$. 
We define the homology groups of $M$ as
$$
H_k(M) = \frac{\ker(\partial_k)}{\operatorname{im}(\partial_{k+1})}\ .
$$
We also define the groups of cochains: $C^k(M; \mathbb{Z}_2) := \text{Hom}(C_k(M), \mathbb{Z}_2)$.
The coboundary map $\delta^k: C^k(M;\mathbb{Z}_2) \to C^{k+1}(M;\mathbb{Z}_2)$ is defined by dualizing the boundary operator on chains: $$(\delta^k \varphi)(\sigma) = \varphi(\partial_{k+1} \sigma),
$$
We also have $\delta^{k+1} \circ \delta^k = 0$ for all $k$.
Therefore, $\operatorname{im}(\delta^{k-1})\subset \operatorname{ker}(\delta^k).$ 
Then we can define the $k$-th cohomology groups:
$$
H^k(M;\mathbb{Z}_2) = \frac{\ker(\delta^k)}{\operatorname{im}(\delta^{k-1})}\ ,
$$
which is the space of $k$-cocycles (cochains with zero coboundary) modulo $k$-coboundaries (those cochains that are coboundaries of $(k-1)$-cochains).

Given a CDM from taking a neighborhood of a tensor network $M=N(\Gamma)$, with only non-trivial $C_0(M)$ and $C_1(M)$. 
We choose a basis of $H_1$ as the ``holes" of our interest. 
A string $S$ is an element of $H^1(M;\mathbb{Z}_2)$, whose action on a hole $g$ is the winding number mod 2:
$$
S \in H^1(M;\mathbb{Z}_2);S(g)\in\mathbb{Z}_2.
$$
These 3D concepts are intrinsically well defined, which naturally extend to 3-manifolds beyond handlebodies, and the computation in 3D is efficient.

\section{Topological tensor network}
\label{ttn}

\newcommand{\topologicalmultispin}{

\hole{0.25}{1}{1};
\hole{0.25}{-1.5}{1};
\hole{0.25}{3.5}{1};
\draw(.82,1.3)--(.82,-1.1);
\draw(.68,1.3)--(.68,-1.1);
\draw(.82-.6,1.3)--(.82-.6,-1.1);
\draw(.68-.6,1.3)--(.68-.6,-1.1);

\draw[white,WLL]
(2,0.4)arc[start angle=15,end angle=90,radius=1]--++(-1,0)arc[start angle=90,end angle=270,radius=1]--++(1,0)arc[start angle=270,end angle=345,radius=1]--(1.6,0.4)
(1.6,0.4)arc[start angle=30,end angle=90,radius=.6]--++(-1,0)arc[start angle=90,end angle=270,radius=.6]--++(1,0)arc[start angle=270,end angle=330,radius=.6]--(2,.4);

\node at (2.5,0.2){$(\alpha,\beta)$};
\draw[red]
(2,0.4)arc[start angle=15,end angle=90,radius=1]--++(-1,0)arc[start angle=90,end angle=270,radius=1]--++(1,0)arc[start angle=270,end angle=345,radius=1]--(1.6,0.4)
(1.6,0.4)arc[start angle=30,end angle=90,radius=.6]--++(-1,0)arc[start angle=90,end angle=270,radius=.6]--++(1,0)arc[start angle=270,end angle=330,radius=.6]--(2,.4);
}

\newcommand{\topologicalmultispinA}{

\begin{scope}[shift={(-.9,0)}]

\hole{0.25}{1}{1};
\hole{0.25}{-1.5}{1};
\hole{0.25}{3.5}{1};

\draw(.82,1.3)--(.82,-1.3);
\draw(.68,1.3)--(.68,-1.3);
\draw(.82-.6,1.3)--(.82-.6,-1.3);
\draw(.68-.6,1.3)--(.68-.6,-1.3);
  \end{scope}
  
  \draw[white,WLL]
    (1,0) arc[start angle=0,end angle=90,radius=1]
    --++(-1,0)
    arc[start angle=90,end angle=270,radius=1]
    --++(1,0)
    arc[start angle=270,end angle=360,radius=1];
  \draw[red]
    (1,0) arc[start angle=0,end angle=90,radius=1]
    --++(-1,0)
    arc[start angle=90,end angle=270,radius=1]
    --++(1,0)
    arc[start angle=270,end angle=360,radius=1];

  \draw[white,WLL]
    (0.6,0) arc[start angle=0,end angle=90,radius=0.6]
    --++(-1,0)
    arc[start angle=90,end angle=270,radius=0.6]
    --++(1,0)
    arc[start angle=270,end angle=360,radius=0.6];
  \draw[red]
    (0.6,0) arc[start angle=0,end angle=90,radius=0.6]
    --++(-1,0)
    arc[start angle=90,end angle=270,radius=0.6]
    --++(1,0)
    arc[start angle=270,end angle=360,radius=0.6];
}

\newcommand{\topologicalmultispinB}{

\begin{scope}[shift={(-.9,0)}]

\hole{0.25}{1}{1};
\hole{0.25}{-1.5}{1};
\hole{0.25}{3.5}{1};

\draw(.82,1.3)--(.82,-1.3);
\draw(.68,1.3)--(.68,-1.3);
\draw(.82-.6,1.3)--(.82-.6,-1.3);
\draw(.68-.6,1.3)--(.68-.6,-1.3);
  \end{scope}
  
  \draw[white,WLL]
    (1,0) arc[start angle=0,end angle=90,radius=1]
    --++(-1,0)
    arc[start angle=90,end angle=270,radius=1]
    --++(1,0)
    arc[start angle=270,end angle=360,radius=1];
  \draw[red]
    (1,0) arc[start angle=0,end angle=90,radius=1]
    --++(-1,0)
    arc[start angle=90,end angle=270,radius=1]
    --++(1,0)
    arc[start angle=270,end angle=360,radius=1];

  
  \draw[white,WLL]
    (0.6,0) arc[start angle=0,end angle=90,radius=0.6]
    --++(-1,0)
    arc[start angle=90,end angle=270,radius=0.6]
    --++(1,0)
    arc[start angle=270,end angle=360,radius=0.6];
  \draw[red]
    (0.6,0) arc[start angle=0,end angle=90,radius=0.6]
    --++(-1,0)
    arc[start angle=90,end angle=270,radius=0.6]
    --++(1,0)
    arc[start angle=270,end angle=360,radius=0.6];

    \draw[red,decorate,decoration={snake}](1,0)--(0.6,0);
    \fill(1,0) circle(.06);
    \fill(.6,0) circle(.06);
}

Clifford and Matchgate gates together form a universal gate set. Our QCS algorithm is capable of simulating hybird-Clifford-Matchgate circuits. This leads to a topological criterion--namely, the presence of \emph{Magic holes}--which delineates the boundary between classically simulable models and those with potential quantum advantage.
QCS eliminates all odd holes as well as even Clifford holes, both of which reflect \textbf{global} topological properties of the tensor network, as captured in the Quon framework.
The remaining Magic holes constitute the core source of quantum complexity.
We further introduce global techniques to reduce the number of Magic holes, applicable only when all remaining holes are even.

\subsection{Topological spins}

By Eq.~\eqref{Equ:kernel element}, any planar diagram in a $S^1\times [1,2]$ reduces to a linear sum of the two diagrams:

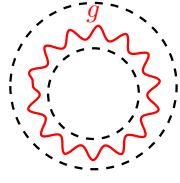
\begin{center}
    \begin{tikzpicture}[line width=.3mm]
        \begin{scope}
        \node at (0,.96)[red]{$g$};
            \draw(-1.1,0)[dashed]arc[start angle=180,end angle=540,radius=1.1];
             \draw(-.8,0)[red,decorate,decoration={snake}]arc[start angle=180,end angle=540,radius=.8];

            \draw(-.6,-.1)[dashed]arc[start angle=180,end angle=540,radius=.6];
        \end{scope}
    \end{tikzpicture}
\end{center}
for $g=0,1$, or a multiple of
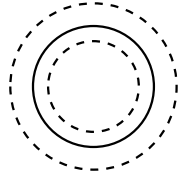
\begin{center}
    \begin{tikzpicture}[line width=.3mm]
        \begin{scope}
            \draw(-1.1,0)[dashed]arc[start angle=180,end angle=540,radius=1.1];
             \draw(-.8,0)arc[start angle=180,end angle=540,radius=.8];
            \draw(-.6,0)[dashed]arc[start angle=180,end angle=540,radius=.6];
        \end{scope}
    \end{tikzpicture}
\end{center}
according to the parity of the winding number. Moreover, the composition of the three annular diagrams verifies the Ising type fusion rule.

In particular, if a layer involves only one hole, 
then it shrinks to the small neighborhood of the hole. 
If the winding number is odd, then they reduce to a scalar. 
If the winding number is even, then they reduce to a topological spin:

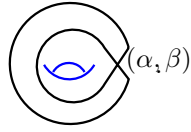
\begin{center}
    \begin{tikzpicture}[line width=0.3mm]
        
        \begin{scope}[shift={(3,0)},scale=.8]
        \hole{0.6}{1}{1};
        \node at (2.5,0.2){$(\alpha,\beta)$};
        \draw(2,0.4)arc[start angle=15,end angle=345,radius=1]--(1.6,0.4)arc[start angle=30,end angle=330,radius=.6]--(2,.4);
        \end{scope}
        \node at (5,0){$.$};
    \end{tikzpicture}
\end{center}

All topological spins form a two-dimensional vector space.
When we replace the crossing by a Pauli $X$-, $Y$- or $Z$-basis, we obtain the corresponding basis for topological spins. 
In particular, we have the following basis, for $g=0,1$:
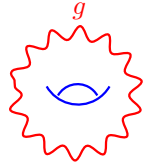
\begin{center}
  \begin{tikzpicture}[line width=.3mm]
  \begin{scope}[shift={(0,0)}]
        \node at (0.55,1)[red]{$g$};
        \draw[red,decorate,decoration={snake}](-.3,0)arc[start angle=180,end angle=540,radius=.8];
        \hole{0.6}{0.2}{0.8};
        \end{scope}
    \end{tikzpicture}
\end{center}

\begin{prop}
\label{reducetotpsp}
If the boundary of a 2-disc does not intersect with any string, then the internal diagram reduces to a scalar or a topological spin according to the parity of the winding number. 
\begin{align}\label{lastproposition}
\begin{tikzpicture}[baseline=(current bounding box.center),line width=.3mm]
    \begin{scope}
    \draw(-1.5,0)[dashed]arc[start angle=180,end angle=540,radius=1.5];
    \draw[white,WLL](1,1)--(-1,1)--(-1,-1)--(1,-1);
    \newcommand{\epsl}{.3};
        \draw(1,1-\epsl-.14)--(1,1-\epsl)[bend right=45]to(1-\epsl,1)--(-1+\epsl,1)[bend right=45]to(-1,1-\epsl)--(-1,-1+\epsl)[bend right=45]to(-1+\epsl,-1)--(1-\epsl,-1)[bend right=45]to(1,-1+\epsl)--++(0,1.3);
        \draw(1,1-\epsl)[bend right=45]to(1-\epsl,1)--(-1+\epsl,1)[bend right=45]to(-1,1-\epsl)--(-1,-1+\epsl)[bend right=45]to(-1+\epsl,-1)--(1-\epsl,-1)[bend right=45]to(1,-1+\epsl);
    \end{scope}
    \begin{scope}[shift={(-1.2,-0.5)},scale=0.25]
        \hole{1}{1}{1};
    \end{scope}
    \begin{scope}[shift={(-0.5,-0.5)},scale=0.25]
        \hole{1}{1}{1};
    \end{scope}
    \begin{scope}[shift={(0.2,-0.5)},scale=0.25]
        \hole{1}{1}{1};
    \end{scope}
    \draw(-0.5,0.8)--(-0.5,-0.8);
    \draw(-0.4,0.8)--(-0.4,-0.8);
    \draw(0.3,0.8)--(0.3,-0.8);
    \draw(0.4,0.8)--(0.4,-0.8);
\end{tikzpicture}
=
\begin{tikzpicture}[baseline=(current bounding box.center),line width=.3mm]
    \begin{scope}
    \draw(-.5,0)[dashed]arc[start angle=180,end angle=540,radius=1.5];
        \hole{0.6}{1}{1};
        \draw[white,WLL](2,0.4)arc[start angle=15,end angle=345,radius=1]--(1.6,0.4)arc[start angle=30,end angle=330,radius=.6]--(2,.4);
        \draw(2,0.4)arc[start angle=15,end angle=345,radius=1]--(1.6,0.4)arc[start angle=30,end angle=330,radius=.6]--(2,.4);
        \node at (2.5,0.2){$(\alpha,\beta)$};
    \end{scope}
\end{tikzpicture}.
\end{align}
\end{prop}

\begin{proof}
By Eq.~\eqref{Equ:kernel element}, the number of strings between two holes can be reduced to at most 2. Then by Eq.~\eqref{two holes meet}, Eq.~\eqref{cuttwo}, the two holes merge into one hole. Thus, all holes reduce to at most one hole. Then the diagram inside the disc reduces to a scalar multiple of a topological spin.      
\end{proof}

The coefficients can be computed as the inner products between the diagram and the basis vectors. If the number of crossings or holes is $O(\log n)$ inside the 2-disc, then the coefficients can be computed efficiently.
Prop.~\ref{reducetotpsp} is a topological analogue of the small entanglement rank reduction method in tensor networks \cite{liao2023simulation}. 
In general, if the boundary circle has $2m$ points, then the vector space is $2^m$ dimensional. Comparing to Prop.~\ref{prop:dim}, the actual factor 2 is contributed by the hole.

\subsection{Free move of topological spin and topological tensor network}

\begin{lemma}
Two topological spins can be eliminated when they have the same connectivity with respect to the layers, based on the following cases: 
\begin{enumerate}
    \item Merging: Two topological spins merge into one if there is no string between them:
\begin{tikzequation}
    \label{merge}
    \begin{scope}[scale=.8]

    \begin{scope}
    \hole{0.6}{1}{1};
    \node at (2.2,0.15){$\alpha_1$};
    \draw(2,0.4)arc[start angle=15,end angle=345,radius=1]--(1.6,0.4)arc[start angle=30,end angle=330,radius=.6]--(2,.4);
    \end{scope}
    \begin{scope}[shift={(2.5,0)}]
    \hole{0.6}{1}{1};
    \node at (2.2,0.15){$\alpha_2$};
    \draw(2,0.4)arc[start angle=15,end angle=345,radius=1]--(1.6,0.4)arc[start angle=30,end angle=330,radius=.6]--(2,.4);
    \end{scope}
    \begin{scope}[shift={(5.5,0)}]
    \hole{0.6}{1}{1};
    \node at (-.25,0){$=$};
    \node at (2.5,0.15){${\alpha_1}{\alpha_2}$};
    \draw(2,0.4)arc[start angle=15,end angle=345,radius=1]--(1.6,0.4)arc[start angle=30,end angle=330,radius=.6]--(2,.4);
    \end{scope}
    
    \end{scope}
    \node at (7,0){$.$};
\end{tikzequation}

    \item Free move: If a layer involves only even holes and has no charge attached to the parity circle, and a topological spin is not involved in that layer, then the topological spin can move freely through the layer.

\begin{tikzequation}
    \label{Equ: Magic hole free move}
    \begin{scope}[scale=.8]       
    \begin{scope}[shift={(6,0)}]
    
    \node at (-3.5,0){$=$};
    \begin{scope}[shift={(1.35,0)},scale=.5]
    \hole{0.6}{1}{1};
    \draw(2,0.4)arc[start angle=15,end angle=345,radius=1]--(1.6,0.4)arc[start angle=30,end angle=330,radius=.6]--(2,.4);
    \end{scope}

    \begin{scope}[scale=1.1]
    \newcommand{\epsl}{.3};
    
    \hole{0.26}{-2.7}{3.8};
    \hole{0.26}{-2.7}{-2.1};
    \draw(1,1-\epsl-.14)--(1,1-\epsl)[bend right=45]to(1-\epsl,1)--(-1+\epsl,1)[bend right=45]to(-1,1-\epsl)--(-1,-1+\epsl)[bend right=45]to(-1+\epsl,-1)--(1-\epsl,-1)[bend right=45]to(1,-1+\epsl)--++(0,1.3);
    
    \draw(1,1-\epsl)[bend right=45]to(1-\epsl,1)--(-1+\epsl,1)[bend right=45]to(-1,1-\epsl)--(-1,-1+\epsl)[bend right=45]to(-1+\epsl,-1)--(1-\epsl,-1)[bend right=45]to(1,-1+\epsl);
    \node at (0,0){$a\ layer$};
    \end{scope}

    \begin{scope}[shift={(-.3,1)},yscale=1.2]
    \draw[white,WL](0,1)--(0,-.2)[bend left=90]to(-.5,-.2)[bend right=60]to(-2,0.3)[bend right=60]to(-2,-1.9)[bend right=50]to(-.5,-1.4)[bend left=90]to (0,-1.4)--(0,-2.5);
    \draw(0,1)--(0,-.2)[bend left=90]to(-.5,-.2)[bend right=60]to(-2,0.3)[bend right=60]to(-2,-1.9)[bend right=50]to(-.5,-1.4)[bend left=90]to (0,-1.4)--(0,-2.5);
    \begin{scope}[shift={(.3,0)}]
    \draw[white,WL](0,1)--(0,-.2)[bend left=90]to(-1.1,-.2)[bend right=60]to(-2,0.3)[bend right=60]to(-2,-1.8)[bend right=60]to(-1.1,-1.4)[bend left=90]to (0,-1.4)--(0,-2.5);
    \draw(0,1)--(0,-.2)[bend left=90]to(-1.1,-.2)[bend right=60]to(-2,0.3)[bend right=60]to(-2,-1.8)[bend right=60]to(-1.1,-1.4)[bend left=90]to (0,-1.4)--(0,-2.5);
    \end{scope}
    \end{scope}
    \end{scope}
    
    \begin{scope}
    \begin{scope}[scale=1.1]
    \newcommand{\epsl}{.3};
    \hole{0.26}{-2.7}{3.8};
    \hole{0.26}{-2.7}{-2.1};
    \draw(1,1-\epsl-.14)--(1,1-\epsl)[bend right=45]to(1-\epsl,1)--(-1+\epsl,1)[bend right=45]to(-1,1-\epsl)--(-1,-1+\epsl)[bend right=45]to(-1+\epsl,-1)--(1-\epsl,-1)[bend right=45]to(1,-1+\epsl)--++(0,1.3);
    
    \draw(1,1-\epsl)[bend right=45]to(1-\epsl,1)--(-1+\epsl,1)[bend right=45]to(-1,1-\epsl)--(-1,-1+\epsl)[bend right=45]to(-1+\epsl,-1)--(1-\epsl,-1)[bend right=45]to(1,-1+\epsl);
    \node at (0,0){$a\ layer$};
    \end{scope}
    \begin{scope}[shift={(-2.2,0)},scale=.5]
    \hole{0.6}{1}{1};
    \draw(2,0.4)arc[start angle=15,end angle=345,radius=1]--(1.6,0.4)arc[start angle=30,end angle=330,radius=.6]--(2,.4);
    \end{scope}
    
    \begin{scope}[shift={(-.3,1)},yscale=1.2]
    \draw[white,WL](0,1)--(0,-.2)[bend left=90]to(-.5,-.2)[bend right=60]to(-2,0.3)[bend right=60]to(-2,-1.9)[bend right=50]to(-.5,-1.4)[bend left=90]to (0,-1.4)--(0,-2.5);
    \draw(0,1)--(0,-.2)[bend left=90]to(-.5,-.2)[bend right=60]to(-2,0.3)[bend right=60]to(-2,-1.9)[bend right=50]to(-.5,-1.4)[bend left=90]to (0,-1.4)--(0,-2.5);
    \begin{scope}[shift={(.3,0)}]
    \draw[white,WL](0,1)--(0,-.2)[bend left=90]to(-1.1,-.2)[bend right=60]to(-2,0.3)[bend right=60]to(-2,-1.8)[bend right=60]to(-1.1,-1.4)[bend left=90]to (0,-1.4)--(0,-2.5);
    \draw(0,1)--(0,-.2)[bend left=90]to(-1.1,-.2)[bend right=60]to(-2,0.3)[bend right=60]to(-2,-1.8)[bend right=60]to(-1.1,-1.4)[bend left=90]to (0,-1.4)--(0,-2.5);
    \end{scope}
    \end{scope}
    
    \end{scope}
    \end{scope}
    \node at (7,0){$.$};
\end{tikzequation}
\end{enumerate}
\end{lemma}

\begin{proof}
The merging relation follows from Eq.~\eqref{fourstringschange}, together with Eq.~\eqref{R2} for topological spins with double crossings and Eq.~\eqref{string_genus-1} for those without crossings. 

For the free move relation, as the layer involves only even holes, by Eqs.~\eqref{bar_cap} and \eqref{bar_even}, we have that  
\begin{tikzequation}
\label{layerbar}
\newcommand{\epsl}{.3};
\begin{scope}[scale=.8]   

    \draw(1,1-\epsl-.14)--(1,1-\epsl)[bend right=45]to(1-\epsl,1)--(-1+\epsl,1)[bend right=45]to(-1,1-\epsl)--(-1,-1+\epsl)[bend right=45]to(-1+\epsl,-1)--(1-\epsl,-1)[bend right=45]to(1,-1+\epsl)--++(0,1.3);
    
    \draw(1,1-\epsl)[bend right=45]to(1-\epsl,1)--(-1+\epsl,1)[bend right=45]to(-1,1-\epsl)--(-1,-1+\epsl)[bend right=45]to(-1+\epsl,-1)--(1-\epsl,-1)[bend right=45]to(1,-1+\epsl);
    \node at (0,0.5){$a\ layer$};
\end{scope}
\node at (2,0){$=$};
\begin{scope}[shift={(4,0)}]
    \begin{scope}[scale=.8]
    \draw[red](-1.2,-0.2)--(1.2,-0.2);
    \draw(1,1-\epsl-.14)--(1,1-\epsl)[bend right=45]to(1-\epsl,1)--(-1+\epsl,1)[bend right=45]to(-1,1-\epsl)--(-1,-1+\epsl)[bend right=45]to(-1+\epsl,-1)--(1-\epsl,-1)[bend right=45]to(1,-1+\epsl)--++(0,1.3);
    
    \draw(1,1-\epsl)[bend right=45]to(1-\epsl,1)--(-1+\epsl,1)[bend right=45]to(-1,1-\epsl)--(-1,-1+\epsl)[bend right=45]to(-1+\epsl,-1)--(1-\epsl,-1)[bend right=45]to(1,-1+\epsl);
    \node at (0,.5){$a\ layer$};
    \end{scope}
\end{scope}
\end{tikzequation}
Then by Lem.~\ref{stringgenusn} and by Reidemeister moves, we can move the topological spin along the wire. 
\end{proof}

\begin{rmk}
The example in Fig.~\ref{fig:not free move} shows that the free move may not hold for a layer involving odd holes.
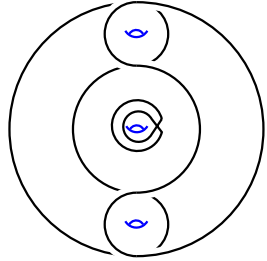
\begin{figure}[H]
\begin{center}
\begin{tikzpicture}[line width=.3mm,scale=1.2]
\begin{scope}[scale=.7]
\begin{scope}[shift={(0.1,0)},scale=.4]
\hole{0.6}{1}{1};
\draw(2,0.4)arc[start angle=15,end angle=345,radius=1]--(1.6,0.4)arc[start angle=30,end angle=330,radius=.6]--(2,.4);
\end{scope}
\hole{.25}{1.3}{7};
\hole{.25}{1.3}{-5};
    \begin{scope}
    \draw[white,WLL](.5,1)arc[start angle=-90,end angle=90,radius=.5];
    \draw[white,WLL](.5,-1)arc[start angle=90,end angle=-90,radius=.5];
    \draw[white,WLL](.5,2)arc[start angle=90, end angle=270,radius=2];
    \draw[white,WLL](.5,1)arc[start angle=90,end angle=270,radius=1];
    \draw(.5,1)arc[start angle=-90,end angle=90,radius=.5];
    \draw(.5,-1)arc[start angle=90,end angle=-90,radius=.5];\draw(.5,2)arc[start angle=90, end angle=270,radius=2];
    \draw(.5,1)arc[start angle=90,end angle=270,radius=1];
    \end{scope}
    
    \begin{scope}[shift={(1,0)},xscale=-1]
    \draw[white,WLL](.5,1)arc[start angle=-90,end angle=90,radius=.5];
    \draw[white,WLL](.5,-1)arc[start angle=90,end angle=-90,radius=.5];
    \draw[white,WLL](.5,2)arc[start angle=90, end angle=270,radius=2];
    \draw[white,WLL](.5,1)arc[start angle=90,end angle=270,radius=1];
    \draw(.5,1)arc[start angle=-90,end angle=90,radius=.5];
    \draw(.5,-1)arc[start angle=90,end angle=-90,radius=.5];
    \draw(.5,2)arc[start angle=90, end angle=270,radius=2];
    \draw(.5,1)arc[start angle=90,end angle=270,radius=1];
    \end{scope}
\end{scope}
\end{tikzpicture}
\end{center}
\caption{The topological spin in the middle cannot move out, even though it is not involved in the two layers.}
\label{fig:not free move}
\end{figure}

\end{rmk}

When all holes are even, we regard holes as topological spins, and every layer characterizes interactions of inside topological spins, then we get a network and call it a \textbf{topological tensor network}.
If two topological spins are not separated by any single layer, then after the free move, then they merge into one topological spin.
In general, two holes can be contracted by performing a genus cut between them by Lem.~\ref{stringgenusn}. 
This will increase the computational complexity by a factor 2, namely the contraction of two holes results in a sum of two terms, which corresponds to topological entanglement entropy rather than conventional entanglement entropy.

To compute a tensor network classically, one can reduce the bulk spins of a local part to boundary spins, when the number of spins is small, corresponding to small entanglement entropy. If the contraction dimension is kept under a computable bound, say $2^{31}$, in the reduction process, then the tensor network can be evaluated classically. There are multiple methods to minimize dimensions appeared in contraction. For example, the tree-shaped networks can be efficiently evaluated by contracting from the leaves toward the root \cite{markov2008simulating,tree2019}.

Given bulk spins, the other spins interacting with them are boundary spins.
As an analogue in topological tensor network, if we consider a set of holes as bulk topological spins, then the other holes within common layers are boundary topological spins. One can reduce the bulk ones to the boundary ones when the number of topological spins is small, corresponding to small topological entanglement entropy. 
If the number is kept small, then the topological tensor network can be evaluated classically.

\section{Examples of QCS}

\newcommand{\smallbox}[4]{
\pgfmathsetmacro{\x}{#1};
\pgfmathsetmacro{\y}{#2};
\pgfmathsetmacro{\rx}{#3};
\pgfmathsetmacro{\ry}{#4};

\begin{scope}[shift={(\x,\y)},xscale=\rx,yscale=\ry]
\fill[white](1,-1)--(1,1)--(-1,1)--(-1,-1)--(1,-1);
\draw(1,-1)--(1,1)--(-1,1)--(-1,-1)--(1,-1);
\end{scope}
}

\subsection{An example of the Clifford circuit}
\label{Cliffordalg}
Now we provide an example of the Clifford circuit to illustrate the QCS algorithm.
We want to compute the value $\bra{q_1q_2q_3}U\ket{001}$ as shown in Fig.~\ref{Clifford_example_1}. We draw the corresponding 2d-CDM (Fig.~\ref{Cliffford_example_2}) horizontally to conform to the convention of quantum circuits.
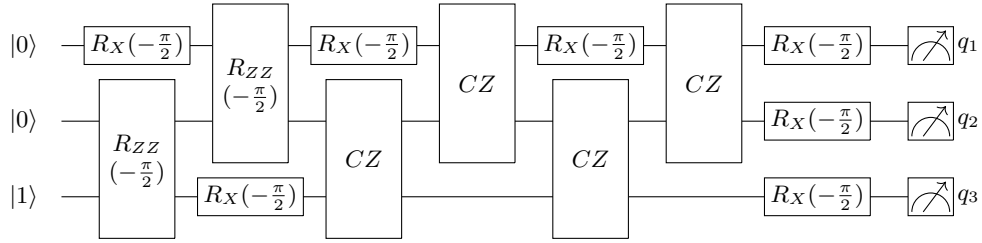
\begin{figure}[H]
\centering
    \begin{tikzpicture}
    \begin{scope}
        \foreach \i in {0,1,2}{
        \draw(0,\i)--(11.5,\i);
        }
        
        \node at (-.5,2){$\ket{0}$};
        \node at (-.5,1){$\ket{0}$};
        \node at (-.5,0){$\ket{1}$};

        \node at (12,2){$q_1$};
        \node at (12,1){$q_2$};
        \node at (12,0){$q_3$};
        
        \smallbox{1}{.5}{.5}{1.05};
        \node at (1,.7){$R_{ZZ}$};
        \node at (1,.3){$(-\frac{\pi}{2})$};
        \smallbox{1}{2}{.7}{.25};
        \node at (1,2){$R_X(-\frac{\pi}{2})$};

        \smallbox{2.5}{1.5}{.5}{1.05};
        \node at (2.5,1+.7){$R_{ZZ}$};
        \node at (2.5,1+.3){$(-\frac{\pi}{2})$};
        \smallbox{2.5}{0}{.7}{.25};
        \node at (2.5,0){$R_X(-\frac{\pi}{2})$};

        \smallbox{4}{.5}{.5}{1.05};
        \node at (4,.5){$CZ$};
        \smallbox{4}{2}{.7}{.25};
        \node at (4,2){$R_X(-\frac{\pi}{2})$};

        \smallbox{5.5}{1.5}{.5}{1.05};
        \node at (5.5,1+.5){$CZ$};

        \smallbox{7}{.5}{.5}{1.05};
        \node at (7,.5){$CZ$};
        \smallbox{7}{2}{.7}{.25};
        \node at (7,2){$R_X(-\frac{\pi}{2})$};

        \smallbox{8.5}{1.5}{.5}{1.05};
        \node at (8.5,1+.5){$CZ$};

        \smallbox{10}{2}{.7}{.25};
        \node at (10,2){$R_X(-\frac{\pi}{2})$};

        \smallbox{10}{1}{.7}{.25};
        \node at (10,1){$R_X(-\frac{\pi}{2})$};

        \smallbox{10}{0}{.7}{.25};
        \node at (10,0){$R_X(-\frac{\pi}{2})$};

        \smallbox{11.5}{2}{.3}{.25};

        \smallbox{11.5}{1}{.3}{.25};

        \smallbox{11.5}{0}{.3}{.25};

        \begin{scope}[shift={(11.5,0)}]
        \draw(-.25,-.2)arc[start angle=180,end angle=0,radius=.25];
        \draw[->](-.1,-.2)--(.2,.2);
        \end{scope}

        \begin{scope}[shift={(11.5,1)}]
        \draw(-.25,-.2)arc[start angle=180,end angle=0,radius=.25];
        \draw[->](-.1,-.2)--(.2,.2);
        \end{scope}

        \begin{scope}[shift={(11.5,2)}]
        \draw(-.25,-.2)arc[start angle=180,end angle=0,radius=.25];
        \draw[->](-.1,-.2)--(.2,.2);
        \end{scope}

    \end{scope}
        
    \end{tikzpicture}
    \caption{An example of the Clifford circuit}
\label{Clifford_example_1}
\end{figure}

\begin{figure}[H]
\centering
\begin{tikzpicture}[line width=.3mm]
\begin{scope}[scale=1.15,rotate=90]
    \CLketone{.6}{0.3}{2}{0}{2};
    \CLketzero{.6}{0.3}{2}{1.8}{2};
    \CLketzero{.6}{0.3}{2}{3.6}{2};
    \CLtwoqubitgateone{.6}{0.3}{2}{0}{0};
    \CLtwoqubitgatetwo{.6}{0.3}{2}{1.8}{-2};
    \CLczgatethree{.6}{0.3}{2}{0}{-4};
    \CLczgatefour{.6}{0.3}{2}{1.8}{-6};
    \CLczgatefive{.6}{0.3}{2}{0}{-8};
    \CLczgate{.6}{0.3}{2}{1.8}{-10};
    
    \CLrxgate{0.6}{0.3}{2}{3.6}{0};
    \CLrxgatetwo{0.6}{0.3}{2}{0}{-2};
    \CLrxgate{0.6}{0.3}{2}{3.6}{-4};
    
    \CLIDgatefour{0.6}{0.3}{2}{0}{-6};
    \CLHgate{0.6}{0.3}{2}{3.6}{-8};
    \CLIDgate{0.6}{0.3}{2}{0}{-10};
    \CLrxgate{0.6}{0.3}{2}{0}{-12};
    \CLrxgate{0.6}{0.3}{2}{1.8}{-12};
    \CLrxgate{0.6}{0.3}{2}{3.6}{-12};
    
    \CLbra{.6}{0.3}{2}{0}{-12};
    \CLbra{.6}{0.3}{2}{1.8}{-12};
    \CLbra{.6}{0.3}{2}{3.6}{-12};

    \node at (.8,-13)[orange]{$q_3$};
    \node at (1.8+.8,-13)[orange]{$q_2$};
    \node at (3.6+.8,-13)[orange]{$q_1$};
\end{scope}
\end{tikzpicture}
\caption{Quon diagram of the Clifford circuit}
\label{Cliffford_example_2}
\end{figure}
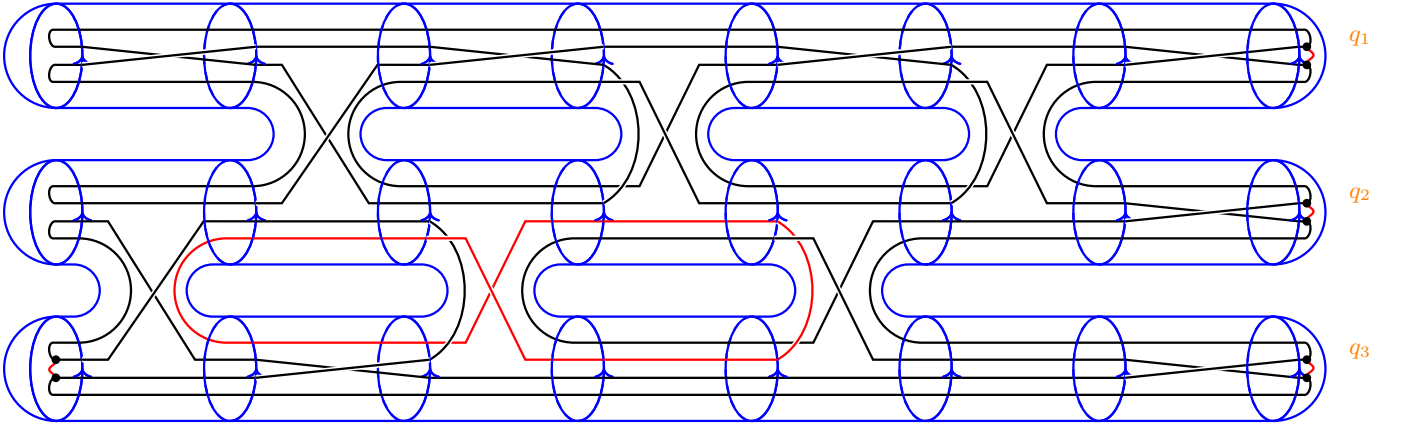

A pair of odd holes is involved in the red layer. We apply handle slides to the left hole and slide all four strings from the left side of the red layer to the right. After this transformation, all remaining holes become even Clifford holes. By Thm.~\ref{Cliffordhole}, they can be eliminated one by one.

\subsection{An example of the Matchgate circuit}
\label{MG_algorithm}
We have introduced an algorithm that evaluates planar Matchgate Quon diagrams in polynomial time using Yang-Baxter relations. 
The following is an illustrative example. Note that we choose $\ket{0}_X=\ket{+}_Z,\ket{1}_X=\ket{-}_Z$ as our basis, instead of the usual Z-basis.

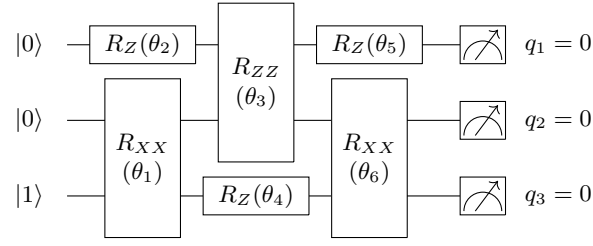
\begin{figure}[H]
\centering
    \begin{tikzpicture}
    \begin{scope}
        \foreach \i in {0,1,2}{
        \draw(0,\i)--(5.5,\i);
        }
        
        \node at (-.5,2){$\ket{0}$};
        \node at (-.5,1){$\ket{0}$};
        \node at (-.5,0){$\ket{1}$};

        \node at (6.5,2){$q_1=0$};
        \node at (6.5,1){$q_2=0$};
        \node at (6.5,0){$q_3=0$};
        
        \smallbox{1}{.5}{.5}{1.05};
        \node at (1,.7){$R_{XX}$};
        \node at (1,.3){$(\theta_1)$};
        \smallbox{1}{2}{.7}{.25};
        \node at (1,2){$R_Z(\theta_2)$};

        \smallbox{2.5}{1.5}{.5}{1.05};
        \node at (2.5,1+.7){$R_{ZZ}$};
        \node at (2.5,1+.3){$(\theta_3)$};
        \smallbox{2.5}{0}{.7}{.25};
        \node at (2.5,0){$R_Z(\theta_4)$};

        \smallbox{4}{.5}{.5}{1.05};
        \node at (4,.7){$R_{XX}$};
        \node at (4,.3){$(\theta_6)$};
        \smallbox{4}{2}{.7}{.25};
        \node at (4,2){$R_Z(\theta_5)$};

        \smallbox{5.5}{2}{.3}{.25};

        \smallbox{5.5}{1}{.3}{.25};

        \smallbox{5.5}{0}{.3}{.25};

        \begin{scope}[shift={(5.5,0)}]
        \draw(-.25,-.2)arc[start angle=180,end angle=0,radius=.25];
        \draw[->](-.1,-.2)--(.2,.2);
        \end{scope}

        \begin{scope}[shift={(5.5,1)}]
        \draw(-.25,-.2)arc[start angle=180,end angle=0,radius=.25];
        \draw[->](-.1,-.2)--(.2,.2);
        \end{scope}

        \begin{scope}[shift={(5.5,2)}]
        \draw(-.25,-.2)arc[start angle=180,end angle=0,radius=.25];
        \draw[->](-.1,-.2)--(.2,.2);
        \end{scope}

    \end{scope}
        
    \end{tikzpicture}
\caption{An example of the Matchgate circuit}
\label{Matchgate_example_1}
\end{figure}

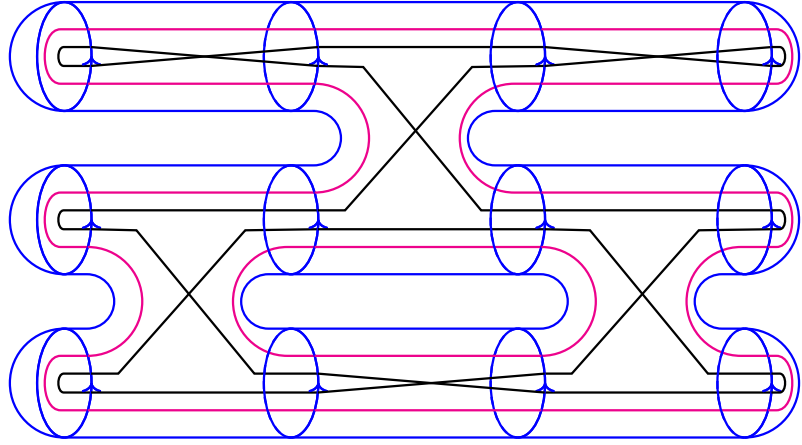
\begin{figure}[H]
\centering
\begin{tikzpicture}[line width=.3mm]
\begin{scope}[scale=1.2, rotate=90]
    \MGket{.6}{0.3}{2.5}{0}{2.5};
    \MGket{.6}{0.3}{2.5}{1.8}{2.5};
    \MGket{.6}{0.3}{2.5}{3.6}{2.5};
    \MGtwoqubitgate{.6}{0.3}{2.5}{0}{0};
    \MGtwoqubitgate{.6}{0.3}{2.5}{1.8}{-2.5};
    \MGtwoqubitgate{.6}{0.3}{2.5}{0}{-5};
    \MGonequbitgate{0.6}{0.3}{2.5}{3.6}{0};
    \MGonequbitgate{0.6}{0.3}{2.5}{0}{-2.5};
    \MGonequbitgate{0.6}{0.3}{2.5}{3.6}{-5};
    \MGbra{.6}{0.3}{2.5}{0}{-5};
    \MGbra{.6}{0.3}{2.5}{1.8}{-5};
    \MGbra{.6}{0.3}{2.5}{3.6}{-5};
\end{scope}
\end{tikzpicture}
\caption{Quon diagram of the Matchgate circuit}
\label{Matchgate_example_2}
\end{figure}

We can transform the Matchgate circuit into its Quon diagram in Fig.~\ref{Matchgate_example_2}. The procedure for computing this Quon diagram is shown in Fig.~\ref{Matchgate_example_3}.

\begin{figure}[H]
\centering
    
\begin{center}
\begin{tikzpicture}[line width=.3mm]

\node at (-1.8,0.5) {$\Longrightarrow$};
\begin{scope}[rotate=90]
    \draw(-0.5,1)--(1,-1);
    \draw(-0.5,-1)--(1,1);
    \draw(-0.25,1.25)arc[start angle=90,end angle=180,radius=0.25];
    \draw(-0.25,-1.25)arc[start angle=-90,end angle=-180,radius=0.25];
    \draw(-0.25,1.25)arc[start angle=90,end angle=-90,radius=1.25];

    \draw(1,1)arc[start angle=180,end angle=0,radius=0.25];
    \draw(1,-1)arc[start angle=-180,end angle=0,radius=0.25];

    \draw(1.5,1)--(3,-1);
    \draw(1.5,-1)--(3,1);

    \draw(3,1)arc[start angle=0,end angle=180,radius=0.25];
    \draw(3,-1)arc[start angle=0,end angle=-180,radius=0.25];
    \draw(2.5,1)--(2.5,-1);
\end{scope}
\begin{scope}[shift={(5,0)}]
    
\node at (-1.8,0.5) {$\Longrightarrow$};
\begin{scope}[rotate=90]
    \draw(-0.5,1)--(1,-1);
    \draw(-0.5,-1)--(1,1);
    \draw(-0.25,1.25)arc[start angle=90,end angle=180,radius=0.25];
    \draw(-0.25,-1.25)arc[start angle=-90,end angle=-180,radius=0.25];
    \draw(-0.25,1.25)arc[start angle=90,end angle=-90,radius=1.25];

    \draw(1,1)arc[start angle=180,end angle=0,radius=0.25];
    \draw(1,-1)arc[start angle=-180,end angle=0,radius=0.25];

    \draw(1.5,1)--(2.5,-0.333);
    \draw(1.5,-1)--(2.5,0.333);

    \draw(2.5,0.333)arc[start angle=90,end angle=-90,radius=0.333];
\end{scope}
\end{scope}
\begin{scope}[shift={(10,0)}]
\node at (-1.8,0.5) {$\Longrightarrow$};
\begin{scope}[rotate=90]
    \draw(-0.5,1)--(1,-1);
    \draw(-0.5,-1)--(1,1);
    \draw(-0.25,1.25)arc[start angle=90,end angle=180,radius=0.25];
    \draw(-0.25,-1.25)arc[start angle=-90,end angle=-180,radius=0.25];
    \draw(-0.25,1.25)arc[start angle=90,end angle=-90,radius=1.25];

    \draw(1,1)arc[start angle=90,end angle=-90,radius=1];
\end{scope}
\end{scope}
\end{tikzpicture}
\end{center}

\begin{center}
\begin{tikzpicture}[line width=.3mm]
\begin{scope}[shift={(-1,0)}]
\node at (-1.8,0) {$\Longrightarrow$};
    \draw(0,0.5)arc[start angle=180,end angle=540,radius=1];
    \draw(0,-0.5)arc[start angle=180,end angle=540,radius=1];
\end{scope}
\begin{scope}[shift={(4,0)}]
\node at (-1.8,0) {$\Longrightarrow$};
\begin{scope}
\draw(0,0)--(0.5,0.5);
\draw(0,0)--(0.5,-0.5);
\draw(0,0)--(-0.5,0.5);
\draw(0,0)--(-0.5,-0.5);
\draw(0.5,0.5)arc[start angle=0,end angle=180,radius=0.5];
\draw(0.5,-0.5)arc[start angle=0,end angle=-180,radius=0.5];
\end{scope}
\begin{scope}[shift={(3.5,0)}]
\node at (-1.4,0) {$\Longrightarrow$};
    \draw(1,0)arc[start angle=0,end angle=360,radius=0.5];
\end{scope}
\end{scope}
\end{tikzpicture}
\end{center}
\caption{Computation procedure for the Matchgate circuit}
\label{Matchgate_example_3}
\end{figure}
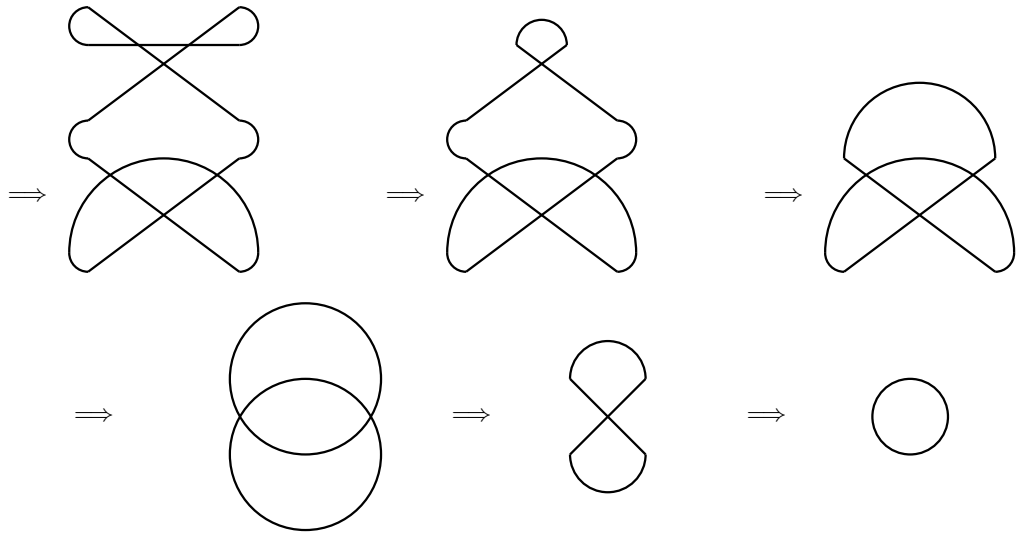

\subsection{Examples of circuits with Magic holes}

We present two families of examples of quantum circuits with $O(n)$ width and depth, $O({n^2})$ Clifford gates and Matchgates, which are classically simulable by the free move and merging Magic holes. We simplify our notation: 
\begin{center}
\begin{tikzpicture}

\begin{scope}[shift={(-6,3)}]
    \draw (-1,0)[blue] circle(.1);
    \node at (-.5,0){$=$};
    \holestringsmall{.4}{0}{1};
\end{scope}

\begin{scope}[shift={(-3,3)}]
    \node at (-1,0) {$\times$};
    \node at (-.5,0){$=$};
    \tpspin{.5}{0}{0};
\end{scope}

\begin{scope}[shift={(0,3)}]
    \node at (-1,0) {$?$};
    \node at (-.5,0){$=$};
    \tpspin{.5}{0}{0};
\end{scope}

\end{tikzpicture}
\end{center}
\begin{itemize}
    \item Each blue circle represents a string-genus;
    \item Each small wrong mark $\times$ or question mark $?$ in the plaquette represents a topological spin;
    \item Intersections of strings are positive or negative braids.
\end{itemize}
Whenever there are four strings between every pair of neighboring holes, we can translate the Quon diagram to a quantum circuit.

\begin{figure}[H]
\begin{center}
    
\begin{tikzpicture}
\begin{scope}[scale=1,rotate=90]
\node at (7.25,11){$\ket{+}$};
\node at (6.25,11){$\ket{0}$};
\node at (11.75,11){$\ket{0}$};
\node at (9,11){$\ket{0}$};
\node at (5,11){$\ket{+}$};
\node at (10.75,11){$\ket{+}$};
\node at (12.9,11){$\ket{+}$};
  \foreach \i in {1,...,8} {
    \draw[thick] ({(\i+19)*0.5},{(19-\i)*0.5})-- +(-5.5,-5.5);
  }

  \foreach \i in {1,...,8} {
    \draw[thick] ({(17-\i)*0.5},{(19-\i)*0.5})-- +(+5.5,-5.5);
  }


  \foreach \i in {1,...,8} {
    \foreach \j in {1,...,8} {
      \pgfmathtruncatemacro{\numerator}{\i + 18 - \j}
      \pgfmathsetmacro{\x}{0.5*(\i + 18 - \j)}
      \pgfmathsetmacro{\y}{0.5*(18 - \j - \i)}
      \pgfmathtruncatemacro{\paritytest}{mod(\numerator,2)}
      \pgfmathtruncatemacro{\ijparity}{mod(\i+\j,2)}
      
        \ifnum\i=3
        \ifnum\j=5
          \node at (\x,{\y + 0.4}) {\textcolor{black}{$\times$}};
      \fi
      \fi
      \ifnum\i=5
        \ifnum\j=3
          \node at (\x,{\y + 0.4}) {\textcolor{black}{$\times$}};
      \fi
      \fi
      \ifnum\i=7
        \ifnum\j=5
          \node at (\x,{\y + 0.4}) {\textcolor{black}{$\times$}};
      \fi
      \fi
      \ifnum\i=5
        \ifnum\j=7
          \node at (\x,{\y + 0.4}) {\textcolor{black}{$\times$}};
      \fi
      \fi
      
      \ifnum\paritytest=1
        \ifnum\ijparity=1
        \ifnum\j<8
        \ifnum\i<8
          \node at (\x,{\y - 0.4}) {\textcolor{blue}{$\circ$}};
        \fi
        \fi
        \fi
      \fi
    }
  }

\foreach \i in {1,2,3,4,5,6,7,8} {
    \draw[thick]({(19+\i)*0.5},{(19-\i)*0.5}) -- ({(19+\i)*0.5},10);
}

\foreach \i in {1,2,3,4,5,6,7,8} {
    \draw[thick]({(17-\i)*0.5},{(19-\i)*0.5}) -- ({(17-\i)*0.5},10);
}
\foreach \i in {1,2,3,4,5,6,7,8} {
    \draw[thick]({(19+\i)*0.5},{9-(19-\i)*0.5}) -- ({(19+\i)*0.5},-1);
}

\foreach \i in {1,2,3,4,5,6,7,8} {
    \draw[thick]({(17-\i)*0.5},{9-(19-\i)*0.5}) -- ({(17-\i)*0.5},-1);
}
\draw (4.5,10) arc (180:0:{(5.6-4.5)*.5});
\draw (5.6,10) -- (5.6,6.9);

\draw (5,10) arc (180:0:{(5.5-5)*.5});

\draw (5.9,10) arc (180:0:{(6-5.9)*.5});
\draw (5.9,10)--(5.9,6.9);

\draw (6.5,10) arc (180:0:{(6.6-6.5)*.5});
\draw (6.6,10) -- (6.6,7.9);

\draw (6.9,10) arc (180:0:{(7.6-6.9)*.5});
\draw (6.9,10)--(6.9,7.9);
\draw (7.6,10)--(7.6,8.9);

\draw (7,10) arc (180:0:{(7.5-7)*.5});

\draw (7.9,10) arc (180:0:{(8-7.9)*.5});
\draw (7.9,10)--(7.9,8.9);


\draw (10,10) arc (180:0:{(10.1-10)*.5});
\draw (10.1,10)--(10.1,  8.9);

\draw (10.4,10) arc (180:0:{(11.1-10.4)*.5});

\draw (10.4,10)--(10.4,  8.9);
\draw (11.1,10)--(11.1,  7.9);

\draw (10.5,10) arc (180:0:{(11-10.5)*.5});

\draw (11.4,10) arc (180:0:{(11.5-11.4)*.5});
\draw (11.4,10)--(11.4,  7.9);

\draw (12,10) arc (180:0:{(12.1-12)*.5});
\draw (12.1,10)--(12.1,  6.9);

\draw (12.4,10) arc (180:0:{(13.5-12.4)*.5});
\draw (12.4,10)--(12.4,  6.9);

\draw (12.5,10) arc (180:0:{(13-12.5)*.5});
\draw (5.9,6.9) arc (360:180:{(5.9-5.6)*.5});
\draw (6.9,7.9) arc (360:180:{(6.9-6.6)*.5});
\draw (7.9,8.9) arc (360:180:{(7.9-7.6)*.5});

\draw (10.1,8.9) arc (180:360:{(10.4-10.1)*.5});
\draw (11.1,7.9) arc (180:360:{(11.4-11.1)*.5});
\draw (12.1,6.9) arc (180:360:{(12.4-12.1)*.5});

\draw (5.6,{9-10}) -- (5.6,{9-6.9});
\draw (5.9,{9-10})--(5.9,{9-6.9});
\draw (6.6,{9-10}) -- (6.6,{9-7.9});
\draw (6.9,{9-10})--(6.9,{9-7.9});
\draw (7.6,{9-10})--(7.6,{9-8.9});
\draw (7.9,{9-10})--(7.9,{9-8.9});

\draw (10.1,10)--(10.1,  8.9);
\draw (10.1,{9-10})--(10.1,  {9-8.9});
\draw (10.4,{9-10})--(10.4,  {9-8.9});
\draw (11.1,{9-10})--(11.1,  {9-7.9});
\draw (11.4,{9-10})--(11.4,  {9-7.9});
\draw (12.1,{9-10})--(12.1,  {9-6.9});
\draw (12.4,{9-10})--(12.4,  {9-6.9});

\draw (5.9,{9-6.9}) arc (0:180:{(5.9-5.6)*.5});
\draw (6.9,{9-7.9}) arc (0:180:{(6.9-6.6)*.5});
\draw (7.9,{9-8.9}) arc (0:180:{(7.9-7.6)*.5});

\draw (10.1,{9-8.9}) arc (180:0:{(10.4-10.1)*.5});
\draw (11.1,{9-7.9}) arc (180:0:{(11.4-11.1)*.5});
\draw (12.1,{9-6.9}) arc (180:0:{(12.4-12.1)*.5});

\end{scope}

\end{tikzpicture}
\end{center}
\caption{Basis-independent simulability of circuits with Magic holes}
\label{fig:big_example_1}
\end{figure}

The first family of examples is shown in Fig.~\ref{fig:big_example_1}.
The blue circles are removed by the string-genus relation, and the topological spins are moved to the left side and merged with the $\ket{+}$ state. As a result, we obtain a Quon diagram without holes; therefore, which allows for efficient classical simulation of measurement sampling in both the $Z$- and $X$-bases.

\begin{figure}[H]
    \centering
\begin{tikzpicture}
\begin{scope}[scale=0.9,rotate=90]
  \foreach \i in {1,...,8} {
    \draw({(\i+19)*0.5-.1},{(19-\i)*0.5})-- +(-5.5,-5.5);
    \draw({(\i+19)*0.5+.1},{(19-\i)*0.5})-- +(-5.5,-5.5);
  }

  \foreach \i in {1,...,8} {
    \draw ({(17-\i)*0.5-.1},{(19-\i)*0.5})-- +(+5.5,-5.5);
    \draw ({(17-\i)*0.5+.1},{(19-\i)*0.5})-- +(+5.5,-5.5);
  }


\foreach \i in {0,1,4,5}{
 \draw[blue] (5.25+\i*.5, 4.75+\i*.5) circle(.05);
 \draw[blue] (7.5+5.25-\i*.5, 4.75+\i*.5) circle(.05);

 }
\foreach \i in {-1,0,1,2,3,4,5,6,7}{

\foreach \j in {0,1,2,3,4,5,6,7}{
 \draw[blue] (.5+5.25+\i*.5+\j*.5, 4.75+\i*.5-\j*.5) circle(.05);
}

\foreach \j in {0,1,2,3,4,5,6,7}{
 \draw[blue] (7+5.25-\i*.5-\j*.5, 4.75+\i*.5-\j*.5) circle(.05);
 }
}

 \node at (5.5, 5) {\textcolor{black}{$\times$}};
  \node at (6.5, 6) {\textcolor{black}{$\times$}};
  \node at (7.5, 7) {\textcolor{black}{$\times$}};
  \node at (8.5, 8) {\textcolor{black}{$\times$}};
  \node at (9.5, 8) {\textcolor{black}{$\times$}};
  \node at (10.5, 7) {\textcolor{black}{$\times$}};
  \node at (11.5, 6) {\textcolor{black}{$\times$}};
  \node at (12.5, 5) {\textcolor{black}{$\times$}};
    \node at (12.5, 4) {\textcolor{black}{$\times$}};
    \node at (11.5, 3) {\textcolor{black}{$\times$}};
    \node at (10.5, 2) {\textcolor{black}{$\times$}};
    \node at (9.5, 1) {\textcolor{black}{$\times$}};
    \node at (8.5, 1) {\textcolor{black}{$\times$}};
    \node at (7.5, 2) {\textcolor{black}{$\times$}};
    \node at (6.5, 3) {\textcolor{black}{$\times$}};
    \node at (5.5, 4) {\textcolor{black}{$\times$}};

    \node at (6.5, 4) {\textcolor{black}{$?$}};
    \node at (7.5, 5) {\textcolor{black}{$?$}};
    \node at (8.5, 6) {\textcolor{black}{$?$}};
    \node at (9.5, 7) {\textcolor{black}{$?$}};
    \node at (10.5, 6) {\textcolor{black}{$\times$}};
     \node at (11.5, 5) {\textcolor{black}{$?$}};
     \node at (10.5, 4) {\textcolor{black}{$?$}};
     \node at (9.5, 3) {\textcolor{black}{$?$}};
     \node at (8.5, 2) {\textcolor{black}{$?$}};

     \node at (7.5, 3) {\textcolor{black}{$\times$}};
     \node at (8.5, 4) {\textcolor{black}{$\times$}};
     \node at (9.5, 5) {\textcolor{black}{$\times$}};

     \node at (7.5, 6) {\textcolor{black}{$\times$}};
     \node at (8.5, 5) {\textcolor{black}{$\times$}};
     \node at (9.5, 4) {\textcolor{black}{$\times$}};
     \node at (10.5, 3) {\textcolor{black}{$\times$}};

    \node at (6.5, 5) {\textcolor{black}{$?$}};
    \node at (8.5, 7) {\textcolor{black}{$?$}};
    
    \node at (7.5, 4) {\textcolor{black}{$?$}};
    \node at (9.5, 6) {\textcolor{black}{$?$}};

    \node at (8.5, 3) {\textcolor{black}{$?$}};
    \node at (10.5, 5) {\textcolor{black}{$?$}};

    \node at (9.5, 2) {\textcolor{black}{$?$}};
    \node at (11.5, 4) {\textcolor{black}{$?$}};
\foreach \i in {1,2,3,4,5,6,7,8} {
    \draw({(19+\i)*0.5-.1},{(19-\i)*0.5}) -- ({(19+\i)*0.5-.1},10);
    \draw({(19+\i)*0.5+.1},{(19-\i)*0.5}) -- ({(19+\i)*0.5+.1},10);
}

\foreach \i in {1,2,3,4,5,6,7,8} {
    \draw({(17-\i)*0.5-.1},{(19-\i)*0.5}) -- ({(17-\i)*0.5-.1},10);
    \draw({(17-\i)*0.5+.1},{(19-\i)*0.5}) -- ({(17-\i)*0.5+.1},10);
}
\foreach \i in {1,2,3,4,5,6,7,8} {
    \draw({(19+\i)*0.5-.1},{9-(19-\i)*0.5}) -- ({(19+\i)*0.5-.1},-1);
    \draw({(19+\i)*0.5+.1},{9-(19-\i)*0.5}) -- ({(19+\i)*0.5+.1},-1);
}

\foreach \i in {1,2,3,4,5,6,7,8} {
    \draw({(17-\i)*0.5-.1},{9-(19-\i)*0.5}) -- ({(17-\i)*0.5-.1},-1);
    \draw({(17-\i)*0.5+.1},{9-(19-\i)*0.5}) -- ({(17-\i)*0.5+.1},-1);
}
\draw (4.5-.1,10) arc (180:0:{(6-4.5)*.5+.1});
\draw (4.5+.1,10) arc (180:0:{(6-4.5)*.5-.1});

\draw (5-.1,10) arc (180:0:{(5.5-5)*.5+.1});
\draw (5+.1,10) arc (180:0:{(5.5-5)*.5-.1});

\draw (6.5-.1,10) arc (180:0:{(6-4.5)*.5+.1});
\draw (6.5+.1,10) arc (180:0:{(6-4.5)*.5-.1});

\draw (7-.1,10) arc (180:0:{(5.5-5)*.5+.1});
\draw (7+.1,10) arc (180:0:{(5.5-5)*.5-.1});

\draw (4.5-.1,10) arc (180:0:{(6-4.5)*.5+.1});
\draw (4.5+.1,10) arc (180:0:{(6-4.5)*.5-.1});

\draw (10-.1,10) arc (180:0:{(6-4.5)*.5+.1});
\draw (10+.1,10) arc (180:0:{(6-4.5)*.5-.1});

\draw (10.5-.1,10) arc (180:0:{(5.5-5)*.5+.1});
\draw (10.5+.1,10) arc (180:0:{(5.5-5)*.5-.1});

\draw (12-.1,10) arc (180:0:{(6-4.5)*.5+.1});
\draw (12+.1,10) arc (180:0:{(6-4.5)*.5-.1});

\draw (12.5-.1,10) arc (180:0:{(5.5-5)*.5+.1});
\draw (12.5+.1,10) arc (180:0:{(5.5-5)*.5-.1});
\begin{scope}[yscale=-1,shift={(0,-8)}]
\draw (4.5-.1,10) arc (180:0:{(6-4.5)*.5+.1});
\draw (4.5+.1,10) arc (180:0:{(6-4.5)*.5-.1});

\draw (5-.1,10) arc (180:0:{(5.5-5)*.5+.1});
\draw (5+.1,10) arc (180:0:{(5.5-5)*.5-.1});

\draw (6.5-.1,10) arc (180:0:{(6-4.5)*.5+.1});
\draw (6.5+.1,10) arc (180:0:{(6-4.5)*.5-.1});

\draw (7-.1,10) arc (180:0:{(5.5-5)*.5+.1});
\draw (7+.1,10) arc (180:0:{(5.5-5)*.5-.1});

\draw (4.5-.1,10) arc (180:0:{(6-4.5)*.5+.1});
\draw (4.5+.1,10) arc (180:0:{(6-4.5)*.5-.1});

\draw (10-.1,10) arc (180:0:{(6-4.5)*.5+.1});
\draw (10+.1,10) arc (180:0:{(6-4.5)*.5-.1});

\draw (10.5-.1,10) arc (180:0:{(5.5-5)*.5+.1});
\draw (10.5+.1,10) arc (180:0:{(5.5-5)*.5-.1});

\draw (12-.1,10) arc (180:0:{(6-4.5)*.5+.1});
\draw (12+.1,10) arc (180:0:{(6-4.5)*.5-.1});

\draw (12.5-.1,10) arc (180:0:{(5.5-5)*.5+.1});
\draw (12.5+.1,10) arc (180:0:{(5.5-5)*.5-.1});
\end{scope}
\begin{scope}[shift={(0,-.3)}]

\draw (13-.1,-1) arc (180:360:{(.5)*.5+.1});
\draw (13+.1,-1) arc (180:360:{(.5)*.5-.1});

\draw (12-.1,-1) arc (180:360:{(.5)*.5+.1});
\draw (12+.1,-1) arc (180:360:{(.5)*.5-.1});

\draw (11-.1,-1) arc (180:360:{(.5)*.5+.1});
\draw (11+.1,-1) arc (180:360:{(.5)*.5-.1});

\draw (10-.1,-1) arc (180:360:{(.5)*.5+.1});
\draw (10+.1,-1) arc (180:360:{(.5)*.5-.1});

\draw (7.5-.1,-1) arc (180:360:{(.5)*.5+.1});
\draw (7.5+.1,-1) arc (180:360:{(.5)*.5-.1});
\draw (6.5-.1,-1) arc (180:360:{(.5)*.5+.1});
\draw (6.5+.1,-1) arc (180:360:{(.5)*.5-.1});
\draw (5.5-.1,-1) arc (180:360:{(.5)*.5+.1});
\draw (5.5+.1,-1) arc (180:360:{(.5)*.5-.1});
\draw (4.5-.1,-1) arc (180:360:{(.5)*.5+.1});
\draw (4.5+.1,-1) arc (180:360:{(.5)*.5-.1});
\end{scope}

\end{scope}
\end{tikzpicture}
    
\caption{Basis-dependent simulability of circuits with Magic holes}
\label{fig:big_example_2}
\end{figure}

In the second family of examples, shown in Fig.~\ref{fig:big_example_2}, the efficient simulability depends on the measurement basis:
The blue circles are removed by the string-genus relation, and the topological spins at the wrong mark positions are moved to the left side and merged with the initial state.
Under $X$-basis measurement, the topological spins in plaquettes with the question mark ``?" are free to move to the left, resulting in a diagram without holes and thus enabling efficient classical simulation of measurement sampling.
However, under Bell-Basis measurement, these topological spins in plaquettes with the question mark $?$ cannot move. QCS fails to efficiently simulate measurement sampling.

\section{Table of calculation rules}

Relations of braids:

\begin{tikzequation}
\tag{R1}
    \begin{scope}[shift={(0,2.5)}]
        \begin{scope}[shift={(-.1,0)}]
        
            \Ronebraidleft
            \node at (1,-0.5){$=e^{-i\frac{\pi}{8}}$};
            \draw (2,-1)--(2,0);
            \node at (2.5,-.6){,};
        \end{scope}
        
        \begin{scope}[shift={(4.4,0)}]
            \Ronebraidright
            \node at (1,-0.2){$=e^{i\frac{\pi}{8}}$};
            \draw (2,0)--(2,-1);
            \node at (2.5,-.6){.};
        \end{scope}
    \end{scope}
  \end{tikzequation}  
\begin{tikzequation}
    \tag{R2}

    \begin{scope}

    \begin{scope}[scale=.8]
            \Rtwobraidleftright
        \node at (1,0){$=$};
        \begin{scope}[scale=.8]
            \draw(2,1)--(2,-1);
            \draw(2.8,1)--(2.8,-1);
            \node at (3.5,-.8){,};
        \end{scope}
        \end{scope}

        \begin{scope}[scale=.8,shift={(5,0)}]
            \Rtwobraidrightleft
            \node at (1,0){$=$};
        \begin{scope}[scale=.8]
            \draw(2,1)--(2,-1);
            \draw(2.8,1)--(2.8,-1);
            \node at (3.5,-.8){.};
        \end{scope}
        
        \end{scope}
    \end{scope}
\end{tikzequation}

\begin{tikzequation}
    
\tag{R3}

    \begin{scope}[scale=.8, shift={(0,-5)}]
        \begin{scope}[shift={(-.5,0)}]
            \draw(0,3)--(1,2);
            \draw(2,3)--(2,2);
            \draw[white,WLL](1,3)--(0,2);
            \draw(1,3)--(0,2);
            
            \draw(0,2)--(0,1);
            
            \draw(1,2)--(2,1);
            \draw[white,WLL](1,1)--(2,2);
            \draw(1,1)--(2,2);
            \draw(2,1)--(2,0);
            \draw(1,1)--(0,0);
            \draw[white,WLL](0,1)--(1,0);
            \draw(0,1)--(1,0);
            \node at (2.5,1.5){$=$};
            \node at (6,0){,};
        \end{scope}

        \begin{scope}[shift={(5,0)},xscale=-1]
            \draw(0,3)--(1,2);
            \draw(2,3)--(2,2);
            \draw[white,WLL](1,3)--(0,2);
            \draw(1,3)--(0,2);
            
            \draw(0,2)--(0,1);
            \draw(1,1)--(2,2); \draw[white,WLL](1,2)--(2,1);
            \draw(1,2)--(2,1);
            \draw(2,1)--(2,0);
            \draw(1,1)--(0,0);
            \draw[white,WLL](0,1)--(1,0);
            \draw(0,1)--(1,0);
        \end{scope}
    \end{scope}

    \begin{scope}[scale=.8,shift={(6,-5)}]
        \begin{scope}
            \draw(0,3)--(1,2);
            \draw(2,3)--(2,2);
            \draw[white,WLL](1,3)--(0,2);
            \draw(1,3)--(0,2);
            \draw(0,2)--(0,1);
            \draw(1,1)--(2,2);
            \draw[white,WLL](1,2)--(2,1); 
            \draw(1,2)--(2,1);
            \draw(2,1)--(2,0);
            \draw(1,1)--(0,0);
            \draw[white,WLL](0,1)--(1,0);
            \draw(0,1)--(1,0);
            \node at (2.5,1.5){$=$};
            \node at (5.8,0){.};
        \end{scope}
        
        \begin{scope}[shift={(5,0)},xscale=-1]
            \draw(0,3)--(1,2);
            \draw(2,3)--(2,2);
            \draw[white,WLL](1,3)--(0,2);
            \draw(1,3)--(0,2);
            \draw(0,2)--(0,1);
            \draw(1,2)--(2,1);
            \draw[white,WLL](1,1)--(2,2); 
            \draw(1,1)--(2,2); 
            \draw(2,1)--(2,0);
            \draw(1,1)--(0,0);
            \draw[white,WLL](0,1)--(1,0);
            \draw(0,1)--(1,0);
        \end{scope}
        
    \end{scope}
\end{tikzequation}

\begin{tikzequation}
\tag{\ref{underbraid}}
\begin{scope}[scale=.8]
        \begin{scope}
            \draw(1,1)--(-1,-1);
            \draw(1.2,1)--(-1+.2,-1);
            \draw(1.4,1)--(-1+.4,-1);
            \draw(1.6,1)--(-1+.6,-1);
            \begin{scope}[shift={(-.7,-.5)}]
            \fill[white](0,0)arc[start angle=180,end angle=540,radius=.4];
            \draw(0,0)arc[start angle=180,end angle=540,radius=.4];
            \node at (0.4,0){$\forall$};
            \end{scope}
            \draw[white,WLL](-1,1)--(1,-1);
            \draw(-1,1)--(1,-1);
        \end{scope}
        \begin{scope}[shift={(4,0)}]
        \node at (-1.5,0){$=$};
            \draw(1,1)--(-1,-1);
            \draw(1.2,1)--(-1+.2,-1);
            \draw(1.4,1)--(-1+.4,-1);
            \draw(1.6,1)--(-1+.6,-1);
            \begin{scope}[shift={(0.5,.5)}]
            
            \fill[white](0,0)arc[start angle=180,end angle=540,radius=.4];
            \draw(0,0)arc[start angle=180,end angle=540,radius=.4];
            \node at (0.4,0){$\forall$};
            \end{scope}
            \draw[white,WLL](-1,1)--(1,-1);
            \draw(-1,1)--(1,-1);
        \end{scope}
        \node at (5.8,-.2){.};
    \end{scope}
\end{tikzequation}

Relations of charges:
\begin{tikzequation}
    \tag{\ref{charge_annihilation}}
    \begin{scope}[scale=.5]
    \draw[domain=0:4*pi,smooth,variable=\x,red] plot ( {cos(\x r)/8-0.2}, {\x/10-.1});

    \draw(0,-1)--(0,2);
    
    \fill (0,-0.1) circle (.1);
    \fill (0,1.1) circle (.1);
    \node at (1.5, .5) {$=$};
    
    \node at (3.6, .4) {$.$};
    \end{scope}
    
    \begin{scope}[scale=.5, shift ={(5,0)}]
    \draw(-1.9,-1)--(-1.9,2);
    
    \end{scope}
\end{tikzequation}

\begin{equation}
    \tag{\ref{repair}}
    \raisebox{0pt}{
    \begin{tikzpicture}[baseline=(current bounding box.center),line width=.3mm]
    \begin{scope}[scale=.5]
    \draw[domain=-6*pi:0,smooth,variable=\x,red] plot (\x/10, {-sin(\x r)/5});
    \draw[domain=-6*pi:0,smooth,variable=\x,red] plot (\x/10, {sin(\x r)/5+1});
    \node at (1.5,.5) {$=$};
    \end{scope}
    \begin{scope}[scale=.5, shift={(5,0)}]
    \draw[domain=0:4*pi,smooth,variable=\x,red] plot ( {-sin(\x r)/5}, {\x/10-.1});
    \draw[domain=0:4*pi,smooth,variable=\x,red] plot ({sin(\x r)/5-1.9}, {\x/10-.1});
    \end{scope}
    
    \begin{scope}[scale=.5, shift ={(6.5,1)}]
    \node at (0,-.5){$=(-1)$};
    \draw[red,decorate,decoration={snake}](1.5,0.2)--(3,-1.2);
    \draw[red,decorate,decoration={snake}](1.5,-1.2)--(3,.2);
    
    \node at (3.9,-.6){.};
    \end{scope}
    \end{tikzpicture} 
    }
\end{equation}

\begin{tikzequation}
\tag{\ref{relative0}}
\begin{scope}[scale=.6]
    \begin{scope}
        \draw[red,decorate,decoration={snake}](-1,1)--(1,-1);
        \draw[white,WLL](-1,-1)--(1,1);
        \draw(-1,-1)--(1,1);
    \end{scope}
    \begin{scope}[shift={(4,0)}]
        \node at (-2,0){$=(-1)$};
        \node at (1.6,0){$.$};
        \draw(-1,-1)--(1,1);
        \draw[white,WLL,decorate,decoration={snake}](-1,1)--(1,-1);
        
        \draw[red,decorate,decoration={snake}](-1,1)--(1,-1);
    \end{scope}
    \end{scope}
\end{tikzequation}

\begin{tikzequation}
    
\tag{\ref{relative1}}

\begin{scope}[scale=.6, shift={(0,-1.1)}]
\draw (0,-.2)--(0,0) arc (180:0:.5)--++(0,-.2);
\draw[domain=-20:0,smooth,variable=\x,red] plot (\x/10, {-sin(\x r)/5});
\fill (0,0) circle (.1);
\node at (2.1,.15) {\small{$= \sqrt{-1}$}};
\node at (5.9,-.1) {.};
\draw[domain=-20:0,smooth,variable=\x, red] plot (\x/10+5.5, {-sin(\x r)/5});

\draw[white,WLL] (4.5,-.2)--(4.5,0) arc (180:90:.5)--++(0,-.2);
\draw (4.5,-.2)--(4.5,0) arc (180:0:.5)--++(0,-.2);
\fill (5.5,0) circle (.1);
\end{scope}

\begin{scope}[scale=.6, shift={(0,-2.2)}]
\draw (0,.2)--(0,0) arc (-180:0:.5)--++(0,.2);
\draw[domain=-20:0,smooth,variable=\x,red] plot (\x/10, {sin(\x r)/5});
\fill (0,0) circle (.1);
\node at (2.4,.15) {\small{$=\sqrt{-1}^3$}};

\node at (5.9,.05) {.};
\draw[domain=-20:0,smooth,variable=\x, red] plot (\x/10+5.5, {sin(\x r)/5});

\draw[white,WLL] (4.5,.2)--(4.5,0) arc (-180:-90:.5)--++(0,.2);
\draw (4.5,.2)--(4.5,0) arc (-180:0:.5)--++(0,.2);
\fill (5.5,0) circle (.1);
\end{scope}
\end{tikzequation}

\begin{tikzequation}
    \tag{\ref{relative2}}
\begin{scope}[shift={(-1,-3.5)}]
        \begin{scope}[scale=.6]
        \draw[red, decorate,decoration={snake}](0,-.5)--(2,-.5);
        \draw[red, decorate,decoration={snake}](0,.5)--(1,.5);
        \draw[white,WLL](0,1)--(0,-1);
        \draw(0,1)--(0,-1);
        \draw[white,WLL](1,1)--(1,-1);
	   \draw(1,1)--(1,-1);
       \draw[white,WLL](2,1)--(2,-1);
        \draw(2,1)--(2,-1);

        \fill (1,.5)circle(.08);
        \fill (2,-.5)circle(.08);
        \fill (0,.5)circle(.08);
        \fill (0,-.5)circle(.08);
        \end{scope}

         \begin{scope}[shift={(3,0)},scale=.6]
         \node at (-1.5,0){$=(-1)$};
         
\node at (2.5,-.3) {.};
         
        \draw[red, decorate,decoration={snake}](0,.5)--(2,.5);
        \draw[red, decorate,decoration={snake}](0,-.5)--(1,-.5);
        \draw[white,WLL](0,1)--(0,-1);
        \draw(0,1)--(0,-1);
        \draw[white,WLL](1,1)--(1,-1);
	   \draw(1,1)--(1,-1);
       \draw[white,WLL](2,1)--(2,-1);
        \draw(2,1)--(2,-1);

        \fill (1,-.5)circle(.08);
        \fill (2,.5)circle(.08);
        \fill (0,-.5)circle(.08);
        \fill (0,.5)circle(.08);
        \end{scope}
\end{scope}
    \begin{scope}[scale=.6,shift={(0,-4.2)}]
\draw[domain=-20:0,smooth,variable=\x,red] plot (\x/10+0.2, {sin(\x r)/5+0.2});

\draw (1,0) -- (0,1);
\draw[white,WLL] (-.2,-.2) -- (1,1);
\draw (0,0) -- (1,1);

\fill (.2,.2) circle (.1);
\node at (1.75,.5) {$=(-1)$};

\node at (4.75,.5) {.};
\begin{scope}[shift={(3.5,0)}]
\draw (1,0) -- (0,1);
\draw[domain=-20:0,smooth,variable=\x,red] plot (\x/10+0.8, {sin(\x r)/5+0.8});
\draw[white,WLL] (1,0) -- (0,1);
\draw (1,0) -- (0,1);
\draw[white,WLL] (0,0) -- (1,1);

\draw (0,0) -- (1,1);
\fill (.8,.8) circle (.1);
\end{scope}
\end{scope}
\end{tikzequation}

\begin{equation}
\tag{\ref{switch}}
\raisebox{0pt}{
\begin{tikzpicture}[baseline=(current bounding box.center),line width=.3mm]
\begin{scope}[scale=.7]
\draw (0,0) -- (1,1);
\draw[white,WLL] (1,0) -- (0,1);
\draw (1,0) -- (0,1);
\node at (1.75,.5) {$=$};

\node at (3.75,.5) {$.$};
\begin{scope}[shift={(2.5,0)}]
\draw (1,0) -- (0,1);
\draw[white,WLL] (0,0) -- (1,1);
\draw (0,0) -- (1,1);
\draw[domain=6:24,smooth,variable=\x,red] plot (\x/30, {sin(\x r)/16+0.2});
\fill (.8,.2) circle (.065);
\fill (.2,.2) circle (.065);
\end{scope}
\end{scope}
\end{tikzpicture}
}
\end{equation}

Relations of crossings:

\begin{equation}
    \crossingleftnode=
\frac{1+\alpha}{\sqrt{2}} 
\crossinguppernode,\quad\alpha\ne -1.
\tag{YB0}
\end{equation}
\begin{equation}
\Rone = \frac{1+\alpha}{\sqrt{2}}
\shu\  .
\tag{YB1}
\end{equation}
\begin{equation}
\Rtwo=
\crossingleftnodetwo.
\tag{YB2}
\end{equation}
\begin{equation}
\shadowone =  k\times\shadowtwo.
\tag{YB3}
\end{equation}

\begin{equation}
\tag{\ref{charge_crossing_equation}}
\begin{aligned}
\chargerelations{b}{leftup}
=&
\chargerelations{-b}{rightup}
\times
(-\sqrt{-1}),
\quad\quad\quad\quad
\chargerelations{b}{leftdown}
&&=
\chargerelations{-b}{rightdown}
\times
\sqrt{-1}.
\\
\chargerelations{b}{leftup}
=&
\chargerelations{-\frac{1}{b}}{rightdown}
\times
 b\sqrt{-1},
 \quad\quad\quad\quad\quad
 \chargerelations{b}{leftdown}
&&=
\chargerelations{-\frac{1}{b}}{rightup}
\times
(-b\sqrt{-1}).
\\
\chargerelations{b}{leftup}
=&
\chargerelations{\frac{1}{b}}{leftdown}
\times
b,
\quad\quad\quad\quad\quad\quad
\chargerelations{b}{leftdown}
&&=
\chargerelations{\frac{1}{b}}{leftup}
\times
b.
\end{aligned}
\end{equation}

Relations of holes:

\begin{tikzequation}\tag{\ref{string_genus-1}}
    \begin{scope}[shift={(-1.3,.7)},scale=.8]
    \draw[blue] (1.25,-1)arc[start angle=150,end angle=30, radius=.5];
                \draw[blue] (1,-.8)arc[start angle=-150,end angle=-30, radius=.8];
\end{scope}
        \begin{scope}[yscale=.5]
            \draw(1,0) arc[start angle=0,end angle=360,radius=1];
        \node at (1.6,.3){$=\frac{1}{\sqrt{2}}.$};
        \end{scope}
\end{tikzequation}

\begin{align}\tag{\ref{stringgenusn_equation}}
    \begin{tikzpicture}[baseline={(current bounding box.center)},line width=.3mm]
    \begin{scope}
    \node at (-.6,.5){$g_1$};
    \node at (1.5,.5){$g_2$};
        \hole{0.5}{-1.8}{1};
        \draw(0,0.5)--(0,-.5);
        \draw(0.25,0.5)--(0.25,-.5);
        \draw(0.75,0.5)--(0.75,-.5);
        \draw(1,0.5)--(1,-.5);
        \node at (.5,0){$...$};
        \hole{0.5}{2.3}{1};
    \end{scope} 
    \end{tikzpicture}
    =
    \begin{tikzpicture}[baseline={(current bounding box.center)},line width=.3mm]
        \begin{scope}[shift={(3.3,0)}]
    \node at (-.6,.5){$g_1$};
    \node at (0.5,.2)[red]{$g_2$};
        \hole{0.5}{-2}{1};
        \draw(0,0.5)--(0,-.5);
        \draw(0.25,0.5)--(0.25,-.5);
        \draw(0.75,0.5)--(0.75,-.5);
        \draw(1,0.5)--(1,-.5);
        \node at (.5,-.15){$...$};
        \draw[red](-.25,0)--(1.25,0);
    \end{scope} 
    \end{tikzpicture}
    =
    \begin{tikzpicture}[baseline={(current bounding box.center)},line width=.3mm]
        \begin{scope}[shift={(5.5,0)}]
    \node at (0.5,0.2)[red]{$g_1$};
    \node at (1.6,.5){$g_2$};
        \draw(0,0.5)--(0,-.5);
        \draw(0.25,0.5)--(0.25,-.5);
        \draw(0.75,0.5)--(0.75,-.5);
        \draw(1,0.5)--(1,-.5);
        \node at (.5,-.15){$...$};
        \draw[red](-.25,0)--(1.25,0);
        \hole{0.5}{2.6}{1};
    \end{scope}
    \end{tikzpicture}\ .
\end{align}

\begin{tikzequation}
\tag{\ref{barholemove}}

    \begin{scope}
        \hole{0.5}{-1.8}{1};
        \draw(0,0.5)--(0,-.5);
        \draw(0.25,0.5)--(0.25,-.5);
        \draw(0.75,0.5)--(0.75,-.5);
        \draw(1,0.5)--(1,-.5);
        \node at (.5,0){$...$};
        \hole{0.5}{2.3}{1};
    \end{scope} 

\begin{scope}[shift={(3.3,0)}]
        \hole{0.5}{-2}{1};
        \node at (-1.2,0){$=$};
        \draw(0,0.5)--(0,-.5);
        \draw(0.25,0.5)--(0.25,-.5);
        \draw(0.75,0.5)--(0.75,-.5);
        \draw(1,0.5)--(1,-.5);
        \node at (.5,-.15){$...$};
        \draw[red](-.25,0)--(1.25,0);
    \end{scope} 
 \begin{scope}[shift={(5.5,0)}]
        \node at (-.6,0){$=$};
        \draw(0,0.5)--(0,-.5);
        \draw(0.25,0.5)--(0.25,-.5);
        \draw(0.75,0.5)--(0.75,-.5);
        \draw(1,0.5)--(1,-.5);
        \node at (.5,-.15){$...$};
        \draw[red](-.25,0)--(1.25,0);
        \hole{0.5}{2.6}{1};
        \node at (2.2,0){.};
    \end{scope} 

\end{tikzequation}


        
\begin{align}
\tag{\ref{two holes meet}}
    \begin{tikzpicture}[line width=.3mm]
    \begin{scope}
    \hole{.5}{1}{1};
    \hole{.5}{4}{1};
    \node at (3.2,0){$=$};
    \hole{.5}{7}{1};
    \end{scope}
    \end{tikzpicture}.
\end{align}

\begin{align}
\tag{\ref{cuttwo}}
    \begin{tikzpicture}[line width=.3mm]
    \begin{scope}
    \hole{.5}{1}{1};
    \hole{.5}{4}{1};
    \draw(1.4,.5)--(1.4,-.5);
    \draw(1.75,.5)--(1.75,-.5);
    \end{scope}
    \begin{scope}[shift={(3.5,0)}]
    \hole{.5}{1}{1};
    \hole{.5}{4}{1};
    \node at (-.2,0){$=\frac{1}{\sqrt{2}}$};
    \draw(1.4,.5)[bend right=90]to(1.75,.5);
    \draw(1.4,-.5)[bend left=90]to(1.75,-.5);
    \node at (3,0){$=$};
    \end{scope}
    \begin{scope}[shift={(7,0)}]
    \hole{.5}{1}{1};
    \node at (-.2,0){$=\frac{1}{\sqrt{2}}$};
    \draw(1.4,.5)[bend right=90]to(1.75,.5);
    \draw(1.4,-.5)[bend left=90]to(1.75,-.5);
    \end{scope}
    \end{tikzpicture}\ .
\end{align}

\begin{align}
\tag{\ref{fourstringschange}}
    \begin{tikzpicture}[baseline={(current bounding box.center)},line width=.3mm]
    \begin{scope}
        \hole{0.5}{-2.5}{1};
        \draw(-0.3,1)--(-0.3,-1);
        \draw(0.2,1)--(0.2,-1);
        \draw(0.7,1)--(0.7,-1);
        \draw(1.2,1)--(1.2,-1);
        \hole{0.5}{3}{1};
    \end{scope}
    \fill[white] (-0.4,-0.3) rectangle (0.3,0.3);
    \draw (-0.4,-0.3) rectangle (0.3,0.3); 
    \node at (-.1,0) {$\forall$};
    \end{tikzpicture}
    =
    \begin{tikzpicture}[baseline={(current bounding box.center)},line width=.3mm]
    \begin{scope}
        \hole{0.5}{-2.5}{1};
        \draw(-0.3,1)--(-0.3,-1);
        \draw(0.2,1)--(0.2,-1);
        \draw(0.7,1)--(0.7,-1);
        \draw(1.2,1)--(1.2,-1);
        \hole{0.5}{3}{1};
    \end{scope}
    \fill[white] (0.6,-0.3) rectangle (1.3,0.3);
    \draw (0.6,-0.3) rectangle (1.3,0.3);
    \node at (.95,0) {$\forall$};
    \end{tikzpicture},
\end{align}

\begin{align}\tag{\ref{bar_cap}}
    \begin{tikzpicture}[baseline={(current bounding box.center)},line width=.3mm]
        \begin{scope}[scale=0.4]
            \draw[red](-3.5,0.5)--(3.5,0.5);
            \draw(-3,2)--(-3,-2);
            \node at (-2.5,0) {$...$};
            \draw(-2,2)--(-2,-2);
            \draw(-1,1)--(-1,-2);
            \draw(3,2)--(3,-2);
            \node at (2.5,0) {$...$};
            \draw(2,2)--(2,-2);
            \draw(1,1)--(1,-2);
            \draw(-1,1)arc[start angle=180, end angle = 0, radius = 1];
        \end{scope}
    \end{tikzpicture}
    =
    \begin{tikzpicture}[baseline={(current bounding box.center)},line width=.3mm]
        \begin{scope}[scale=0.4]
            \draw[red](-3.5,0.5)--(3.5,0.5);
            \draw(-3,2)--(-3,-2);
            \node at (-2.5,0) {$...$};
            \draw(-2,2)--(-2,-2);
            \draw(3,2)--(3,-2);
            \node at (2.5,0) {$...$};
            \draw(2,2)--(2,-2);
            \draw(-1,-2)arc[start angle=180, end angle = 0, radius = 1];
        \end{scope}
    \end{tikzpicture}.
\end{align}


If there are even number of strings between holes:
\begin{align}
\tag{\ref{bar_even}}
\begin{tikzpicture}[baseline={(current bounding box.center)},line width=.3mm]
\begin{scope}
    \hole{0.5}{-.65}{0};
    \draw(-1.8,-1)--(-1.8,0);
    \draw(-2.3,-1)--(-2.3,0);
    \draw(-1,-1)--(-1,0);
    \draw(-.5,-1)--(-.5,0);
    \draw(1,-1)--(1,0);
    \draw(.5,-1)--(.5,0);
    \draw(1.8,-1)--(1.8,0);
    \draw(2.3,-1)--(2.3,0);
    \node at (-1.4,-.5){$\dots$};
    \node at (1.4,-.5){$\dots$};
    \draw[red](-2.5,-.8)--(2.5,-.8);
\end{scope}
\end{tikzpicture}
=
\begin{tikzpicture}[baseline={(current bounding box.center)},line width=.3mm]
\begin{scope}[shift={(7,0)}]
\node at (3,-.8){.};
    \hole{0.5}{-.65}{0};
    \draw(-1.8,-1)--(-1.8,0);
    \draw(-2.3,-1)--(-2.3,0);
    \draw(-1,-1)--(-1,0);
    \draw(-.5,-1)--(-.5,0);
    \draw(1,-1)--(1,0);
    \draw(.5,-1)--(.5,0);
    \draw(1.8,-1)--(1.8,0);
    \draw(2.3,-1)--(2.3,0);
    \node at (-1.4,-.5){$\dots$};
    \node at (1.4,-.5){$\dots$};
    \draw[red](-2.5,-.4)--(-.4,-.4);
    \draw[red](2.5,-.4)--(.4,-.4);
\end{scope}
\end{tikzpicture}
\end{align}

\begin{align}\tag{\ref{bolangquan_genus_cut}}
\begin{tikzpicture}[baseline={(current bounding box.center)},line width=.3mm]
    \begin{scope}
    \node at (0.7,1.6)[red]{$v$};
    \node at (0.7,.9){$g$};
    \begin{scope}      
    \draw[red,decorate,decoration={snake}](0,0.7)arc[start angle=180,end angle=540,radius=0.7];
    \hole{0.4}{1}{2.3};
    \end{scope}
    \node at (2.5,.6){$=(-1)^{gv}$};
    \begin{scope}[shift={(3.2,0)}]
    \node at (0.6,.9){$g$};
    \hole{0.4}{1}{2.3};
    \end{scope}
\end{scope}
\end{tikzpicture}
\end{align}

\begin{tikzequation}
\tag{\ref{relativegenus}}
\begin{scope}
    \begin{scope}[shift={(0,0)}]
    \node at (-1.3,.4){$v$};
    \draw[red,decorate,decoration={snake}](-1,.4)--(0,.4);
        \fill(-1,.4)circle(.05);
        \fill(0,.4)circle(.05);
        \draw(-1,.7)--(-1,-.7);
         \draw(0,.7)--(0,-.7);
         \begin{scope}[scale=.5,shift={(-2.7,0.9)}]
         \node at (1.6,-.4){$g$};
                \draw[blue] (1.25,-1)arc[start angle=150,end angle=30, radius=.5];
                \draw[blue] (1,-.8)arc[start angle=-150,end angle=-30, radius=.8];
            \end{scope}
        
    \end{scope}
    \begin{scope}[shift={(3,0)}]
    \node at (-2,0){$=(-1)^{gv}$};
    \node at (0.3,0){.};
        \draw(-1,.7)--(-1,-.7);
         \draw(0,.7)--(0,-.7);
        \node at (-1.3,-.4){$v$};
    \draw[red,decorate,decoration={snake}](-1,-.4)--(0,-.4);
        \fill(-1,-.4)circle(.05);
        \fill(0,-.4)circle(.05);
         \begin{scope}[scale=.5,shift={(-2.7,0.9)}]
         \node at (1.6,-.4){$g$};
                \draw[blue] (1.25,-1)arc[start angle=150,end angle=30, radius=.5];
                \draw[blue] (1,-.8)arc[start angle=-150,end angle=-30, radius=.8];
            \end{scope}
        
    \end{scope}
\end{scope}
\end{tikzequation}

If there are odd number of strings between holes:
\begin{align}
\begin{tikzpicture}[baseline={(current bounding box.center)},line width=.3mm]
\tag{\ref{bar_odd}}
\begin{scope}
\begin{scope}[shift={(0,-0.5)}]
    \genus{0.5}{$g$};
\end{scope}
\node at (1,-.6)[red]{$v$};
    \begin{scope}[shift={(0.5,0)}]
        \draw(-1.8,-1)--(-1.8,0);
    \draw(-2.3,-1)--(-2.3,0);
    \node at (-1.4,-.5){$\dots$};
    \draw(-1,-1)--(-1,0);
    \end{scope}
    \begin{scope}[shift={(-0.5,0)}]
    \draw(1,-1)--(1,0);
    \draw(2.3,-1)--(2.3,0);
    \node at (1.4,-.5){$\dots$};
    \draw(1.8,-1)--(1.8,0);
    \end{scope}
    \draw[red](-2.5,-.8)--(2.5,-.8);
\end{scope}
\end{tikzpicture}
=(-1)^{gv}
\begin{tikzpicture}[baseline={(current bounding box.center)},line width=.3mm]
\begin{scope}
\node at (1,-.05)[red]{$v$};
    \begin{scope}[shift={(0,-0.5)}]
    \genus{0.5}{$g$};
\end{scope}
    \begin{scope}[shift={(0.5,0)}]
        \draw(-1.8,-1)--(-1.8,0);
    \draw(-2.3,-1)--(-2.3,0);
    \node at (-1.4,-.5){$\dots$};
    \draw(-1,-1)--(-1,0);
    \end{scope}
    \begin{scope}[shift={(-0.5,0)}]
    \draw(1,-1)--(1,0);
    \draw(2.3,-1)--(2.3,0);
    \node at (1.4,-.5){$\dots$};
    \draw(1.8,-1)--(1.8,0);
    \end{scope}
    \draw[red](-2.5,-.2)--(2.5,-.2);
\end{scope}
\end{tikzpicture}\ .
\end{align}

\begin{tikzequation}
\tag{\ref{handleslidesbasic_equation}}
\begin{scope}[scale=1.5]
    \begin{scope}
    \begin{scope}[scale=.8]
        
     \draw(.8,0)arc[start angle=180, end angle=540,radius=.6];
    \begin{scope}[shift={(1.2,0)}]

    \draw[white,WLL](0.1,0.2)--(0.1,1); 
    \draw(0.1,1)--(0.1,.2);
    \draw[white,WLL](-0.1,0.1)--(-0.1,1); 
    \draw(-0.1,1)--(-0.1,.1);
     \draw[white,WLL](0.5,0.1)--(.5,1); 
    \draw(0.5,0.1)--(.5,1); 

    \draw[white,WLL](.3,0.2)--(.3,1); 
    \draw(0.3,0.2)--(.3,1); 
        
    \end{scope}
    
       \begin{scope}
         \draw[white,WLL](0.2,1)--(0.2,-1);
         \draw(0.2,1)--(0.2,-1);
         
        \draw[white,WLL](.5,1)--(.5,-1);
        \draw(.5,1)--(.5,-1);

        \fill[white](1,0)arc[start angle=180, end angle=540,radius=.4];
         \draw(1.8,0)[blue,->]arc[start angle=0, end angle=360,radius=.4];
        \node at (1.4,0){$\forall$};

    \end{scope}

    \end{scope}
          
    \end{scope}

\end{scope}

\begin{scope}[shift={(4,0)},scale=1.5]
\node at (-.5,0){$=$};
\node at (2.2,0){.};

    \begin{scope}[shift={(0.1,0)},scale=1]
       
        \draw(0,1)--(0,.5)arc[start angle=180,end angle=360, radius=.4]arc[start angle=180,end angle=0, radius=.5]--++(0,-1)arc[start angle=0,end angle=-180, radius=.5]arc[start angle=0,end angle=180, radius=.4]--++(0,-.5);
    \end{scope}

     \begin{scope}[shift={(.4,0)},scale=1]
    
         \draw(0,1)--(0,.5)arc[start angle=180,end angle=360, radius=.15]arc[start angle=180,end angle=0, radius=.7]--++(0,-1)arc[start angle=0,end angle=-180, radius=.7]arc[start angle=0,end angle=180, radius=.15]--++(0,-.5);
    \end{scope}

 \begin{scope}[shift={(0,0)}]
       \draw(.88,0)arc[start angle=180, end angle=540,radius=.42];
    \begin{scope}[shift={(1.15,0)},xscale=.7]

   \draw[white,WLL](0.1,0.2)--(0.1,1.3); 
    \draw(0.1,1.3)--(0.1,.2);
    \draw[white,WLL](-0.1,0.1)--(-0.1,1.3); 
    \draw(-0.1,1.3)--(-0.1,.1);
     \draw[white,WLL](0.5,0.1)--(.5,1.3); 
    \draw(0.5,0.1)--(.5,1.3); 

    \draw[white,WLL](.3,0.2)--(.3,1.3); 
    \draw(0.3,0.2)--(.3,1.3); 
        
    \end{scope}
    \fill[white](1,0)arc[start angle=180, end angle=540,radius=.3]
     \draw(1.6,0)[blue,->]arc[start angle=0, end angle=360,radius=.3];
    \end{scope}
   \node at (1.3,0){$\forall$};
        
\end{scope}
\end{tikzequation}

If $g_0$ is a Clifford hole, it can be substituted with some Clifford crossing:

\begin{tikzequation}
\tag{\ref{Cliffordhole_equation}}
        \begin{scope}[scale=.8]
        \node at (-.25,0){$...$};
        \draw(-.5,-1)--(-.5,1);
        \draw(0,-1)--(0,1);
    \begin{scope}[shift={(3,-1.5)}]
    \end{scope}

    \begin{scope}

        \begin{scope}[scale=.55,shift={(.9,1)}]
        \node at (.6,0){$g_0$};
        
            \draw[blue](.25,-1)arc[start angle=150,end angle=30, radius=.5];
            \draw[blue] (0,-.8)arc[start angle=-150,end angle=-30, radius=.8];
             \node at (2.3,-1){$...$};
        \end{scope}
    \end{scope}
        
     \begin{scope}[shift={(4,0)}]

        \begin{scope}[scale=.55,shift={(.9,1)}]
        \node at (.6,0){$g_m$};
            \draw[blue](.25,-1)arc[start angle=150,end angle=30, radius=.5];
            \draw[blue] (0,-.8)arc[start angle=-150,end angle=-30, radius=.8];
            
        \draw(3,-2.8)--(3,.8);
        \draw(2,-2.8)--(2,.8);
             \node at (2.6,-1){$...$};
        \end{scope}
        
    \end{scope}

     \begin{scope}[shift={(2,0)}]

        \begin{scope}[scale=.55,shift={(.9,1)}]
        \node at (.6,0){$g_1$};
            \draw[blue](.25,-1)arc[start angle=150,end angle=30, radius=.5];
            \draw[blue] (0,-.8)arc[start angle=-150,end angle=-30, radius=.8];
             \node at (2.3,-1){$...$};
        \end{scope}
        
        \draw(2,-1)--(2,1);
        \draw(1.5,-1)--(1.5,1);
    \end{scope}

     \begin{scope}

        \draw(2,-1)--(2,1);
        \draw(1.5,-1)--(1.5,1);
    \end{scope}
    \end{scope}
    \begin{scope}[shift={(8,0)},scale=.8]
      \node at (-2,0){$=\frac{1}{c}$};
      \node at (7,-.2){.};
    \begin{scope}[shift={(.2,0)},yscale=0.5]
        \draw(2,0)to[bend left=90](5.5,0.5);\draw(2,0)to[bend right=90](5.5,-.5)--(5.2,.4)[bend right=90]to(2.2,0)[bend right=90]to(5.2,-.4)--(5.5,0.5);
        \node at (5.7,0){$\alpha$};
    \end{scope}
    \begin{scope}[shift={(.43,0)},yscale=0.4,xscale=.9]
    \end{scope}
        \node at (-.25,0){$...$};
        \draw(-.5,-1)--(-.5,1);
        \draw(0,-1)--(0,1);

    \begin{scope}

        \begin{scope}[scale=.55,shift={(.9,1)}]
             \node at (2.3,-1){$...$};
        \end{scope}
         \draw[white,WL](2,-1)--(2,1);
        \draw[white,WL](1.5,-1)--(1.5,1);
    \end{scope}
        
     \begin{scope}[shift={(4,0)}]
    
        \begin{scope}[scale=.55,shift={(.9,1)}]
        \node at (.6,0){$g_m$};
            \draw[blue](.25,-1)arc[start angle=150,end angle=30, radius=.5];
            \draw[blue] (0,-.8)arc[start angle=-150,end angle=-30, radius=.8];
            \draw(3,-3)--(3,1);
        \draw(4,-3)--(4,1);
            \node at (3.5,-1){$...$};
        \end{scope}
        
    \end{scope}

     \begin{scope}[shift={(2,0)}]

        \begin{scope}[scale=.55,shift={(.9,1)}]
        \node at (.6,0){$g_1$};
            \draw[blue](.25,-1)arc[start angle=150,end angle=30, radius=.5];
            \draw[blue] (0,-.8)arc[start angle=-150,end angle=-30, radius=.8];
             \node at (2.3,-1){$...$};
        \end{scope}
        \draw[white,WL](2,-1)--(2,1);
        \draw[white,WL](1.5,-1)--(1.5,1);
        
        \draw(2,-1)--(2,1);
        \draw(1.5,-1)--(1.5,1);
    \end{scope}

     \begin{scope}

        \draw(2,-1)--(2,1);
        \draw(1.5,-1)--(1.5,1);
    \end{scope}
        
    \end{scope}
\end{tikzequation}

If the boundary of a 2-disc does not intersect with any string, then the internal diagram reduces to a scalar or a topological spin:
\begin{align}\tag{\ref{lastproposition}}
\begin{tikzpicture}[baseline=(current bounding box.center),line width=.3mm]
    \begin{scope}
    \draw(-1.5,0)[dashed]arc[start angle=180,end angle=540,radius=1.5];
    \draw[white,WLL](1,1)--(-1,1)--(-1,-1)--(1,-1);
    \newcommand{\epsl}{.3};
        \draw(1,1-\epsl-.14)--(1,1-\epsl)[bend right=45]to(1-\epsl,1)--(-1+\epsl,1)[bend right=45]to(-1,1-\epsl)--(-1,-1+\epsl)[bend right=45]to(-1+\epsl,-1)--(1-\epsl,-1)[bend right=45]to(1,-1+\epsl)--++(0,1.3);
        \draw(1,1-\epsl)[bend right=45]to(1-\epsl,1)--(-1+\epsl,1)[bend right=45]to(-1,1-\epsl)--(-1,-1+\epsl)[bend right=45]to(-1+\epsl,-1)--(1-\epsl,-1)[bend right=45]to(1,-1+\epsl);
    \end{scope}
    \begin{scope}[shift={(-1.2,-0.5)},scale=0.25]
        \hole{1}{1}{1};
    \end{scope}
    \begin{scope}[shift={(-0.5,-0.5)},scale=0.25]
        \hole{1}{1}{1};
    \end{scope}
    \begin{scope}[shift={(0.2,-0.5)},scale=0.25]
        \hole{1}{1}{1};
    \end{scope}
    \draw(-0.5,0.8)--(-0.5,-0.8);
    \draw(-0.4,0.8)--(-0.4,-0.8);
    \draw(0.3,0.8)--(0.3,-0.8);
    \draw(0.4,0.8)--(0.4,-0.8);
\end{tikzpicture}
=
\begin{tikzpicture}[baseline=(current bounding box.center),line width=.3mm]
    \begin{scope}
    \draw(-.5,0)[dashed]arc[start angle=180,end angle=540,radius=1.5];
        \hole{0.6}{1}{1};
        \draw[white,WLL](2,0.4)arc[start angle=15,end angle=345,radius=1]--(1.6,0.4)arc[start angle=30,end angle=330,radius=.6]--(2,.4);
        \draw(2,0.4)arc[start angle=15,end angle=345,radius=1]--(1.6,0.4)arc[start angle=30,end angle=330,radius=.6]--(2,.4);
        \node at (2.5,0.2){$(\alpha,\beta)$};
    \end{scope}
\end{tikzpicture}.
\end{align}

Magic hole free move:
If a layer involves only even holes and has no charge attached to the parity circle, and a topological spin is not involved in that layer, then the topological spin can move freely through the layer:
\begin{tikzequation}
    \tag{\ref{Equ: Magic hole free move}}
    \begin{scope}[scale=.8]       
    \begin{scope}[shift={(6,0)}]
    
    \node at (-3.5,0){$=$};
    \begin{scope}[shift={(1.35,0)},scale=.5]
    \hole{0.6}{1}{1};
    \draw(2,0.4)arc[start angle=15,end angle=345,radius=1]--(1.6,0.4)arc[start angle=30,end angle=330,radius=.6]--(2,.4);
    \end{scope}

    \begin{scope}[scale=1.1]
    \newcommand{\epsl}{.3};
    
    \hole{0.26}{-2.7}{3.8};
    \hole{0.26}{-2.7}{-1.8};
    \draw(1,1-\epsl-.14)--(1,1-\epsl)[bend right=45]to(1-\epsl,1)--(-1+\epsl,1)[bend right=45]to(-1,1-\epsl)--(-1,-1+\epsl)[bend right=45]to(-1+\epsl,-1)--(1-\epsl,-1)[bend right=45]to(1,-1+\epsl)--++(0,1.3);
    
    \draw(1,1-\epsl)[bend right=45]to(1-\epsl,1)--(-1+\epsl,1)[bend right=45]to(-1,1-\epsl)--(-1,-1+\epsl)[bend right=45]to(-1+\epsl,-1)--(1-\epsl,-1)[bend right=45]to(1,-1+\epsl);
    \node at (0,0){$a\ layer$};
    \end{scope}

    \begin{scope}[shift={(-.3,1)},yscale=1.2]
    \draw[white,WL](0,1)--(0,-.2)[bend left=90]to(-.5,-.2)[bend right=60]to(-2,0.3)[bend right=60]to(-2,-1.9)[bend right=50]to(-.5,-1.4)[bend left=90]to (0,-1.4)--(0,-2.5);
    \draw(0,1)--(0,-.2)[bend left=90]to(-.5,-.2)[bend right=60]to(-2,0.3)[bend right=60]to(-2,-1.9)[bend right=50]to(-.5,-1.4)[bend left=90]to (0,-1.4)--(0,-2.5);
    \begin{scope}[shift={(.3,0)}]
    \draw[white,WL](0,1)--(0,-.2)[bend left=90]to(-1.1,-.2)[bend right=60]to(-2,0.3)[bend right=60]to(-2,-1.8)[bend right=60]to(-1.1,-1.4)[bend left=90]to (0,-1.4)--(0,-2.5);
    \draw(0,1)--(0,-.2)[bend left=90]to(-1.1,-.2)[bend right=60]to(-2,0.3)[bend right=60]to(-2,-1.8)[bend right=60]to(-1.1,-1.4)[bend left=90]to (0,-1.4)--(0,-2.5);
    \end{scope}
    \end{scope}
    \end{scope}
    
    \begin{scope}
    \begin{scope}[scale=1.1]
    \newcommand{\epsl}{.3};
    \hole{0.26}{-2.7}{3.8};
    \hole{0.26}{-2.7}{-1.8};
    \draw(1,1-\epsl-.14)--(1,1-\epsl)[bend right=45]to(1-\epsl,1)--(-1+\epsl,1)[bend right=45]to(-1,1-\epsl)--(-1,-1+\epsl)[bend right=45]to(-1+\epsl,-1)--(1-\epsl,-1)[bend right=45]to(1,-1+\epsl)--++(0,1.3);
    
    \draw(1,1-\epsl)[bend right=45]to(1-\epsl,1)--(-1+\epsl,1)[bend right=45]to(-1,1-\epsl)--(-1,-1+\epsl)[bend right=45]to(-1+\epsl,-1)--(1-\epsl,-1)[bend right=45]to(1,-1+\epsl);
    \node at (0,0){$a\ layer$};
    \end{scope}
    \begin{scope}[shift={(-2.2,0)},scale=.5]
    \hole{0.6}{1}{1};
    \draw(2,0.4)arc[start angle=15,end angle=345,radius=1]--(1.6,0.4)arc[start angle=30,end angle=330,radius=.6]--(2,.4);
    \end{scope}
    
    \begin{scope}[shift={(-.3,1)},yscale=1.2]
    \draw[white,WL](0,1)--(0,-.2)[bend left=90]to(-.5,-.2)[bend right=60]to(-2,0.3)[bend right=60]to(-2,-1.9)[bend right=50]to(-.5,-1.4)[bend left=90]to (0,-1.4)--(0,-2.5);
    \draw(0,1)--(0,-.2)[bend left=90]to(-.5,-.2)[bend right=60]to(-2,0.3)[bend right=60]to(-2,-1.9)[bend right=50]to(-.5,-1.4)[bend left=90]to (0,-1.4)--(0,-2.5);
    \begin{scope}[shift={(.3,0)}]
    \draw[white,WL](0,1)--(0,-.2)[bend left=90]to(-1.1,-.2)[bend right=60]to(-2,0.3)[bend right=60]to(-2,-1.8)[bend right=60]to(-1.1,-1.4)[bend left=90]to (0,-1.4)--(0,-2.5);
    \draw(0,1)--(0,-.2)[bend left=90]to(-1.1,-.2)[bend right=60]to(-2,0.3)[bend right=60]to(-2,-1.8)[bend right=60]to(-1.1,-1.4)[bend left=90]to (0,-1.4)--(0,-2.5);
    \end{scope}
    \end{scope}
    
    \end{scope}
    \end{scope}
    \node at (7,0){$.$};
\end{tikzequation}



\bibliography{aps_preprint}

\clearpage

\end{document}